\documentclass[a4paper,11pt]{book}
\usepackage[utf8]{inputenc}		
\usepackage[T1]{fontenc}		
\usepackage[pdftex]{graphicx}
\usepackage{amsmath}			
\usepackage{amssymb}			
\usepackage{amsthm}			
\usepackage{array}				
\usepackage{cases}
\usepackage{color}
\usepackage{footnote}
\usepackage{hyperref}
\usepackage{indentfirst} 		
\usepackage{latexsym}			
\usepackage{lscape}
\usepackage{makecell}
\usepackage{makeidx}				
\usepackage{minitoc}
\usepackage{multicol}			
\usepackage{multirow}
\usepackage{placeins} 
\usepackage[authoryear,round]{natbib}			
\usepackage{stmaryrd}
\usepackage{vmargin}				
\usepackage{xspace}	
\usepackage{algorithm,algpseudocode}
\usepackage{csquotes}
\usepackage{fancybox}
\usepackage{xifthen}
\usepackage{makecell}
\usepackage{subcaption}
\usepackage{fontawesome}
\usepackage{marginnote}\reversemarginpar
\usepackage{tikz}
\usetikzlibrary{positioning}
\usetikzlibrary{shapes, positioning, calc}

\algblock{ForEach}{EndForEach}
\algnewcommand\algorithmicforeach{\textbf{foreach}}
\algnewcommand\algorithmicforeachdo{\textbf{do}}
\algnewcommand\algorithmicendforeach{\textbf{end\ foreach}}
\algrenewtext{ForEach}[1]{\algorithmicforeach\ #1\ \algorithmicforeachdo}
\algrenewtext{EndForEach}{\algorithmicendforeach}
\usepackage{fancyhdr}
\newcommand\nph{$\mathcal{NP}$-hard}
\newcommand\npc{$\mathcal{NP}$-complete}
\newcommand\phn[1]{\underset{#1}{\bullet}}
\newcommand\ph{\bullet}
\newcommand{\bigo}[1]{\mathcal{O}(#1)}

\newcommand{\ostar}[1]{\mathcal{O}^*(#1)}
\newcommand{\dpa}{\textit{Dynamic Programming}}

\newcommand{\pf}{\textit{Pareto\ Front}}
\newcommand{\pp}{\textit{Pareto\ Permutations}}

\newcommand{\optperm}{Opt}
\newcommand{\minperm}{\minvec}
\newcommand{\minvec}{\mathcal{M}in}
\newcommand{\multiline}[1]{\begin{tabular}{@{}c@{}}#1\end{tabular}}

\newcommand{\pbtt}{$1||\sum{T_i}$}
\newcommand{\pbft}{$F3||C_{max}$}
\newcommand{\pbfmax}{$F3\|f_{max}$}
\newcommand{\pbsumfi}{$F3\|\sum f_i$}
\newcommand{\pbrisumc}{$1|r_i|\sum C_i$}
\newcommand{\pbfsumc}{$F2||\sum C_i$}
\newcommand{\pbdtilde}{$1|\tilde{d_i}|\sum w_iC_i$}

\newcommand{\pbss}{\textsc{Subset Sum}}
\newcommand{\pbsac}{\textsc{Knapsack}}
\newcommand{\pbxsat}{\textsc{Exact Satifiability}}
\newcommand{\pbhp}{\textsc{Hamiltonian Path}}
\newcommand{\pbbp}{\textsc{Bin Packing}}

\newcommand{\techbb}{\textit{Branch \& Bound}}
\newcommand{\techbc}{\textit{Branch \& Cut}}
\newcommand{\techbm}{\textit{Branch \& Merge}}
\newcommand{\techbr}{\textit{Branch \& Reduce}}
\newcommand{\techbp}{\textit{Branch \& Price}}
\newcommand{\techmemo}{\textit{Memorization}}
\newcommand{\techbmemo}{\textit{Branch \& Memorize}}
\newcommand{\techss}{\textit{Sort \& Search}}
\newcommand{\techdc}{\textit{Divide \& Conquer}}
\newcommand{\techie}{\textit{Inclusion \& Exclusion}}
\newcommand{\techdp}{\textit{Dynamic Programming}}
\newcommand{\techls}{\textit{Local Search}}
\newcommand{\techffc}{\textit{Fast Subset Convolution}}

\newcommand{\fb}{\textit{forward branching}}
\newcommand{\bb}{\textit{backward branching}}
\newcommand{\db}{\textit{decomposition branching}}
\newcommand{\df}{\textit{depth first}}
\newcommand{\best}{\textit{best first}}
\newcommand{\wf}{\textit{breadth first}}
\newcommand{\an}{\textit{active nodes}}
\newcommand{\en}{\textit{explored nodes}}
\newcommand{\pbmemo}{\textit{passive node memorization}}
\newcommand{\abmemo}{\textit{predictive node memorization}}
\newcommand{\fmemo}{\textit{solution memorization}}

\newcommand{\subsubsubsection}[1]{\vspace{0.5cm}\noindent\textbf{#1}\\}

\newtheorem{theorem}{Theorem}
\newtheorem{coro}{Corollary}
\newtheorem{definition}{Definition}
\newtheorem{lemma}{Lemma}
\newtheorem{proper}{Property}
\newtheorem{propos}{Proposition}
\newtheorem{para}{Paradox}
\newtheorem{dec}{Decomposition}

\makeindex
\begin{document}

\setmarginsrb{25mm}{0mm}{10mm}{0mm}{0mm}{0mm}{0mm}{0mm}

\begin{tabular}{ p{3cm} p{9cm} p{3cm}}
		\begin{minipage}{3cm}
			\includegraphics[width=3cm]{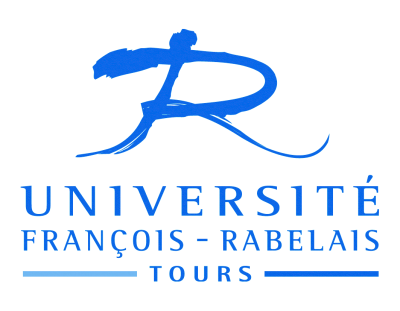} 
		\end{minipage}
	&
		\begin{minipage}{9cm}
			\begin{center}
				\includegraphics[width=3cm]{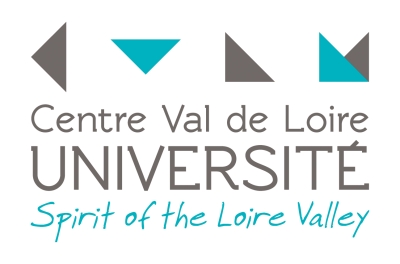} 
			\end{center}
		\end{minipage}
	&
		\begin{minipage}{3cm}
			\includegraphics[width=3cm]{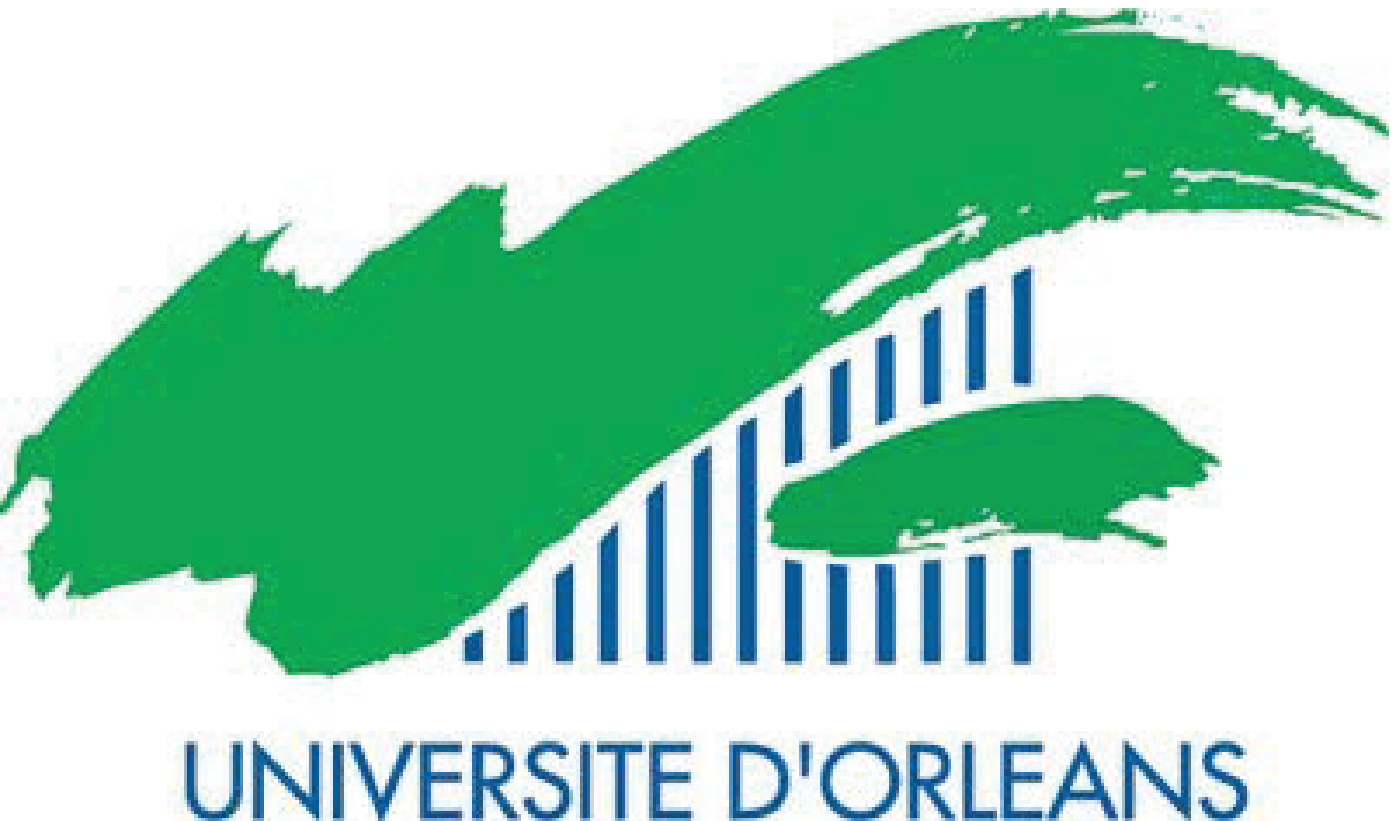} 
		\end{minipage}
\end{tabular}

\vspace{0.5cm}

\begin{minipage}{16cm}
	\begin{center}
		\textbf{\LARGE{UNIVERSIT\'E FRAN\c{C}OIS~RABELAIS DE TOURS}}
	\end{center}
\end{minipage}
	
\vspace{1cm}
	
\begin{minipage}{16cm}
	\begin{center}
		\'ECOLE DOCTORALE MIPTIS \\ \vspace{0.2cm}
		$\ $ Laboratoire d'Informatique (EA 6300) 
	\end{center}
\end{minipage}
	
\vspace{1cm}
	
\begin{minipage}{16cm}
	\begin{center}
		\LARGE \textbf{TH\`ESE} \normalsize pr\'esent\'ee par : \\ \vspace{0.2cm}
		\large \textbf{Lei SHANG}\\ \vspace{0.2cm}
		soutenue le : 30/11/2017
	\end{center}
\end{minipage}
	
\vspace{1cm}
	
\begin{minipage}{16cm}
	\begin{center}
		pour obtenir le grade de : Docteur de l'Université François-Rabelais de Tours \\ \vspace{0.2cm}
		Discipline/ Spécialité : INFORMATIQUE
	\end{center}
\end{minipage}
	
\vspace{0.8cm}

\fbox{
	\begin{minipage}{15cm}
		\vspace{0.39cm}
		\begin{center}
				\Large \textbf{Exact Algorithms With Worst-case Guarantee For Scheduling: From Theory to Practice}\\ \vspace{0.2cm}
		\end{center}
		\vspace{0.01cm}
	\end{minipage}
}

\vspace{1cm}
		 
\textsc{\large Th\`ese dirig\'ee par : } $\ $ \vspace{0.2cm}\\ 
\begin{tabular}{l p{0.3cm} p{11cm}}
	\textsc{T'KINDT} Vincent & $\ $ &  Professeur des Universités, Universit\'e Fran\c{c}ois Rabelais de Tours\\
\end{tabular}

\vspace{0.2cm}
 
\textsc{RAPPORTEURS :}  $\ $ \vspace{0.1cm}\\
\begin{tabular}{l p{0.3cm} p{11cm}}
\textsc{ARTIGUES} Christian & $\ $ &  Directeur de Recherche, LAAS-CNRS, Toulouse\\
\textsc{PASCHOS} Vangelis & $\ $ &  Professeur des Universités, Universit\'e Paris-Dauphine\\
\end{tabular}

\vspace{0.2cm}

\textsc{JURY : }  $\ $ \vspace{0.1cm} \\
\begin{tabular}{l p{0.3cm} p{10.2cm}}
\textsc{ARTIGUES} Christian & $\ $ &  Directeur de Recherche, LAAS-CNRS, Toulouse\\
\textsc{DELLA CROCE} Federico & $\ $ &  Professeur, Politechnico di Torino, Italie\\
\textsc{LENTE} Christophe & $\ $ &  Maître de Conférences, HDR, Universit\'e Fran\c{c}ois Rabelais de Tours\\
\textsc{LIEDLOFF} Mathieu & $\ $ &  Maître de Conférences, HDR, Universit\'e d'Orléans\\
\textsc{PASCHOS} Vangelis & $\ $ &  Professeur des Universités, Universit\'e Paris-Dauphine\\
\textsc{PINSON} Eric & $\ $ &  Professeur des Universités, Institut de Mathématiques Appliquées (UCO)\\
\textsc{T'KINDT} Vincent & $\ $ &  Professeur des Universités, Universit\'e Fran\c{c}ois Rabelais de Tours\\
\end{tabular}
\setlength{\voffset}{0pt}                    
\thispagestyle{empty}
\setlength{\voffset}{0pt} 				
\setlength{\topmargin}{0pt}				
\setlength{\headheight}{30pt}			
\setlength{\headsep}{30pt}				
\setlength{\textheight}{620pt}		
\setlength{\footskip}{30pt}				
\setlength{\hoffset}{0pt} 				
\setlength{\oddsidemargin}{20pt}	
\setlength{\evensidemargin}{0pt}	
\setlength{\textwidth}{420pt}			
\setlength{\marginparsep}{10pt}		
\setlength{\marginparwidth}{40pt}	
\parskip=3pt											
\pagestyle{fancy}

\newcolumntype{V}{>{\centering\arraybackslash}m{2.3cm}}
\newcolumntype{A}{>{\centering\arraybackslash}m{2.0cm}}
\newcolumntype{B}{>{\centering\arraybackslash}m{1.9cm}}
\newcolumntype{C}{>{\centering\arraybackslash}m{1.5cm}}
\newcolumntype{D}{>{\centering\arraybackslash}m{1.8cm}}
\newcolumntype{E}{>{\arraybackslash}m{6.0cm}}
\newcolumntype{F}{>{\arraybackslash}m{7.2cm}}
\newcolumntype{T}{>{\arraybackslash}m{8cm}}
\newcolumntype{G}{>{\centering\arraybackslash}m{2.6cm}}
\newcolumntype{K}{>{\centering\arraybackslash}m{3.0cm}}
\newcolumntype{L}{>{\arraybackslash}m{4.3cm}}
\newcolumntype{M}{>{\arraybackslash}m{4.6cm}}
\newcolumntype{N}{>{\arraybackslash}m{4.0cm}}
\newcolumntype{O}{>{\arraybackslash}m{3.5cm}}
\newcolumntype{Y}{>{\arraybackslash}m{2.6cm}}

\DeclareGraphicsExtensions{.pdf,.jpeg,.png}
\setcounter{minitocdepth}{1} 
\setlength{\mtcindent}{24pt} 
\setcounter{secnumdepth}{3}
\renewcommand{\mtcfont}{\small\rm} 
\renewcommand{\mtcSfont}{\small\bf} 
\newenvironment{vcenterpage}
\newenvironment{abstract}{\rightskip1in\itshape}{}

\chapter*{Acknowledgment}
\markright{\MakeUppercase{Acknowledgment}}
I really enjoy the writing of this ``chapter'', which provides me an opportunity to review my work and my life during the PhD study. 
The first person that appears in my mind, to whom I feel really thankful, is my advisor Vincent T'Kindt. It has been an excellent experience to work with him. He always has many ideas and it was not always easy for me to follow his thoughts during meetings. His ideas, which do not always work (of course!), always inspire. He has always been available when I need guidance. Moreover, he trusts me and allows me to have much freedom during the PhD, without which I would not enjoyed these three years so much. I am not qualified to evaluate other aspects but he is at least an excellent supervisor. Merci Vincent.

I also want to thank all my co-authors, namely Christophe Lenté, Mathieu Liedloff, Federico Della Croce and Michele Garraffa, without whom all my work would not have been possible. Special thanks go to Michele, with whom I worked everyday during several months, struggling to search for a beautiful structure in a tree. We often spent several hours in a meeting room without being aware. It is a very nice souvenir.

Now let me get out of the professional context. My wife Ting sacrificed her career in Beijing in order to come and stay with me during my PhD. We met each other in 2004 and she is half of my life, if not all. I cannot imagine what could I do without her support. 
My parents have had strong insistence on the importance of my education. 
I would not be here without their encouragement and support. They are always proud of me and I am also proud of being their son and proud of the fact that they are proud of me. 
I also want to thank my grandmother, great Buddhist philosopher of the family, who has been leading me since my young age and has helped me to establish my view of life, of world and of values. She taught me to seek happiness not from outside but from inside. By well maintaining his inner peace, one can always stay happy in whatever environment. 
Though I am not (yet) wise enough to be well enlightened, I've already benefited a lot from her words. 
I am so lucky to born in such family environment!

Now let's talk about Polytech! The beginning of my  adventure in France started at 2011. I thank the Mundus program which ``imported'' me from China. I still remember the day of interview in Beijing, with Jean-Louis Bouquard and Marjolaine Martin. Later plus Audrey Perez-Prada, we've had many nice moments together. The teachers of Mundus department helped me to learn French and the teachers of DI helped me to become a good engineer on Computer Science, which enabled the possibility of my PhD study.

The PhD study lasts only three years but is a long way, sometimes really exciting but sometime pretty despairing. I am so glad to have many PhD friends around me, with whom I had many interesting discussions, many happy breaks and many rounds of baby foot. I am often curious to know where will we be 10 years later. I hope we can always keep contact even though the world changes fast.

I thought of making a name list of people who helped me a lot, but honestly that will be too long and incomplete. Therefore, I stop here, being aware that I've never been alone. 

Thank you all.

\chapter*{R\'esum\'e}
\markright{\MakeUppercase{R\'esum\'e}}
Cette thèse synthétise les travaux de recherches réalisés pendant les études doctorales de l'auteur. L'objectif de ces travaux est de proposer des algorithmes exacts qui ont une meilleure complexité, temporelle ou spatiale, dans le pire des cas pour des problèmes d'ordonnancement qui sont $\mathcal{NP}$-difficiles. En plus, on s'intéresse aussi à évaluer leurs performances en pratique.

Trois contributions principales sont rapportées. La première concerne un algorithme du type \techdp{} qui résout le  problème \pbft{}  en $\ostar{3^n}$ en temps et en espace. L'algorithme est généralisé facilement à d'autres problèmes du type Flowshop, y compris les problèmes $F2||f$ et $F3||f$, et aux problèmes d'ordonnancement à une seule machine tels que les problèmes $1|r_i|f$, avec $f\in\{f_{max}, f_i\}$. 

La seconde contribution porte sur l'élaboration d'une méthode arborescente appelée \techbm{} pour  résoudre le problème  \pbtt{} en $\ostar{(2+\epsilon)^n}$ en temps avec $\epsilon>0$ arbitrairement petit et en espace polynomial. Le travail se base sur l'observation que de nombreux sous-problèmes identiques apparaissent répétitivement pendant la résolution du problème global. A partir de ça, une opération appelée \textit{merge} est proposée, qui fusionne les sous-problèmes (les noeuds dans l'arbre de recherche) identiques autant que possible. Cette méthode doit pouvoir être généralisée à d'autres problèmes.

Le but de la troisième contribution est d'améliorer les performances en pratique des algorithmes exacts procédant par parcours d'un arbre de recherche. D'abord nous avons aperçu qu'une meilleure façon d'implémenter l'idée de  \techbm{} est d'utiliser une technique appelée \techmemo. Avec la découverte d'un nouveau \textit{paradoxe} et la mis en place d'une stratégie de nettoyage de mémoire,  notre algorithme a résolu les instances qui ont 300 tâches de plus par rapport à l'algorithme de référence pour le problème \pbtt. 
Avec ce succès, nous avons testé \techmemo{} sur trois autres problèmes d'ordonnancement notés \pbrisumc, \pbdtilde{} et \pbfsumc, précédemment traités par \cite{tkindt2004revisiting}. Les résultats finaux des quatre problèmes ont montré la puissance de  \techmemo{} appliquée aux problèmes d'ordonnancement. Nous nommons ce paradigme \techbmemo{} afin de  promouvoir la considération systématique de l'intégration de \techmemo{} dans les algorithmes de branchement comme \techbb, en tant qu'un composant essentiel. 
La méthode peut aussi être généralisée pour résoudre d'autres problèmes qui ne sont pas forcément des problèmes d'ordonnancement.

\paragraph{}


\chapter*{Abstract}
\markright{\MakeUppercase{Abstract}}
This thesis summarizes the author's PhD research works on the design of exact algorithms that provide a worst-case (time or space)
guarantee for \nph{} scheduling problems. Both theoretical and practical aspects are considered with three main results reported. 

The first one is about a \techdp{} algorithm which solves the \pbft{} problem in $\ostar{3^n}$ time and space. 
The algorithm is easily generalized to other flowshop problems including $F2||f$ and $F3||f$, and single machine scheduling problems like $1|r_i|f$, with $f\in\{f_{max}, f_i\}$. 

The second contribution is about a search tree method called \techbm{} which solves the \pbtt{} problem with the time complexity converging to $\ostar{2^n}$ and in polynomial space. The work is based on the observation that many identical subproblems appear during the solution of the input problem. An operation called \textit{merge} is then derived, which merges all identical nodes to one whenever possible and hence yields a better complexity.

Our third contribution
 aims to improve the practical efficiency of exact search tree algorithms solving scheduling problems. First we realized that a better way to implement the idea of  \techbm{}  is to use a technique called \techmemo. By the finding of a new  \textit{algorithmic paradox} and the implementation of a memory cleaning strategy, the method succeeded to solve instances with 300 more jobs with respect to the state-of-the-art algorithm for the \pbtt{} problem. Then the treatment is extended to another three problems \pbrisumc, \pbdtilde{} and \pbfsumc{} previously addressed by \cite{tkindt2004revisiting}. The results of the four problems all together show the power of the \techmemo{} paradigm when applied on sequencing problems. We name it \techbmemo{} to promote a systematic consideration of \techmemo{} as an essential building block in branching algorithms like \techbb. The method can surely also be used to solve other problems, which are not necessarily scheduling problems.

\paragraph{}

\tableofcontents
\listoftables
\listoffigures

\chapter*{Introduction}\label{ch0}
\markright{\MakeUppercase{Introduction}}

This thesis summarizes some research works on the design of exact algorithms that provide a worst-case (time or space)
guarantee for \nph{} scheduling problems. By \textit{exact}, we mean algorithms that search for optimal solutions of a given optimization problem. By \textit{worst-case guarantee}, we aim to minimize the time or space complexity of algorithms in worst-case scenarios. Finally, by \textit{scheduling}, we restrict our target problems to scheduling problems, especially sequencing/permutation problems.

Even though it is always difficult to clearly separate research domains due to the existence of intersections, we regard this work as a connection of two  fields: \textit{Scheduling Theory} and \textit{Exact Exponential Algorithms} (EEA for short). 

``Scheduling is a decision-making process that is used on a regular basis in many manufacturing and services industries. It deals with the allocation of resources to tasks over given time periods and its goal is to optimize one or more objectives.''\citep{pinedo2008scheduling}. Scheduling theory can be traced back to early 20's, before it really started to form a standalone field since the 60's. Now it is a well established area within operational research and it plays an essential role in industries. 

On the other side, we regard EEA as a field which groups research efforts on designing exact algorithms with worst-case complexity guarantee for solving \nph{} problems. The research is not limited to a specific type of problems (like scheduling), but guided by the ambition to explore the intrinsic hardness of \nph{} problems. Traced back to 1956, this ambition was expressed in the famous letter of G\"{o}del to Von Neumann as the following question: ``...It would be interesting to know, ...how strongly in general the number of steps in finite combinatorial problems can be reduced with respect to simple exhaustive search.''\citep{sipser1992history}. 
Such research can provide an idea on the relative hardness of problems: all \nph{} problems do not seem to have the same hardness. 
In fact, since the early stage of computer science, some of these problems appeared to be solvable with a lower exponential complexity than others belonging to the same complexity class. 
For instance, the MIS (\textit{Maximum Independent Set}) problem, which asks to find from a given undirected graph, a largest subset of non-adjacent vertices, can be solved in $\bigo{2^n}$ time by enumerating all subsets. A series of improvements have been proposed in the literature, which allow to solve the problems in $\bigo{1.26^n}$ \citep{tarjan1977finding}, $\bigo{1.2346^n}$ \citep{jian19862}, $\bigo{1.2278^n}$ (or $\bigo{1.2109^n}$ with exponential space) \citep{robson1986algorithms}, $\bigo{1.2202^n}$ \citep{fomin2006measure}, $\bigo{1.2132^n}$ \citep{kneis2009fine}, $\bigo{1.2114^n}$ \citep{bourgeois2010bottom} and most recently $\bigo{1.1996^n}$ \citep{xiao2017exact} time.
In contrast, the famous TSP (\textit{Traveling Salesman Problem}) requires $\bigo{n!}$ time to be solved in a brute-force way, when the size $n$ is chosen as  the number of cities. But it can be solved in $\bigo{n2^n}$ time by \techdp{} \citep{bellman1962dynamic}, which is a considerable improvement. However, since then, during more than 50 years, no faster algorithm has been proposed. The question is, can other problems be solved faster than enumeration, if yes, to what extent? If no, are we in front of the unbreakable wall of the intrinsic hardness of the problem? 
For a survey on the most effective techniques in designing EEAs, readers are kindly referred to the paper of  \cite{woeginger2003exact} and to the book by  \cite{fomin2010exact}. 

The initiation of this research is motivated by the following observations. First, EEA has been well developed during recent decades. 
For a number of \nph{} problems, new algorithms have been proposed consecutively, each time improving little by little the worst-case complexity of the considered problem. As a result, the community now possesses a powerful ``toolbox'' containing a number of problem solving techniques that can provide much guidance on solving a given problem and analyzing their worst-case complexity. 
In contrast, \textit{Scheduling} as a major field within Operation Research, has not benefited much from the development of EEA. In fact, many scheduling problems have been proved as \nph, and then most considerations are oriented to the practical efficiency instead of the theoretical worst-case guarantee of algorithms. For example, we can find a large number of state-of-the-art algorithms that are \techbb{} algorithms with many problem-dependent properties incorporated, tested on random generated instances to prove their computational efficiency but without their worst-case complexity discussed. This may partially because of the lack of techniques and computation power at that time, which made EEAs inefficient in practice. Today, with the active development on EEA, it becomes natural, even essential to revisit \nph{} scheduling problems on searching a worst-case guarantee for scheduling problems. Some work have been done by \cite{lente2013extension}, who introduced the so-called class of multiple constraint problems and showed that all problems fitting into that class could be tackled by means of the \techss{} technique. 
Further, they showed that several known scheduling problems are part of that class. 
However, all these problems require assignment decisions only and none of them require the solution of a pure sequencing problem.

\section*{Positioning of the Thesis}
The thesis is positioned on the connection of \textit{\nph{} Scheduling Problems} and \textit{Exact Exponential Algorithms}, aiming at proposing EEAs with better time/space complexity for solving \nph{} Scheduling Problems. We put special focus on the following points.
\begin{itemize}
    \item The tackled scheduling problems are sequencing/permutation problems, that is, a solution is represented by an ordering of operations. 
    \item Besides of the search of a better theoretical complexity, we are also interested in the practical efficiency of the proposed algorithms.
    \item The proposed algorithms are expected to be generalizable.
\end{itemize}

\section*{Outline}
As already presented, the results on EEAs for \nph{} scheduling problems are pretty limited. In this thesis we add several results on this subject. Several classic scheduling (sequencing) problems are treated, notably the three machine flowshop problem (denoted by \pbft) and the single machine total tardiness problem (denoted by  \pbtt). On the treated problems, our algorithms yield the currently best complexity guarantee; on the techniques used, our techniques, which could be generalized to other problems, provide new hints on the design of exact algorithms solving scheduling problems.

Chapter \ref{ch1} presents the fundamentals of EEA and scheduling theory. The historical evolution of the scheduling theory is reviewed, with the main development in each decade summarized, based on the interesting review of \cite{potts2009fifty}. Some  well known techniques for constructing EEAs are presented including some results on scheduling problems. 
Basic concepts that will be used later are introduced. We also add some notes on Parameterized Algorithms, even though it is not considered during the thesis.

Chapter \ref{ch2} presents our first contribution: a \techdp{} algorithm which solves the \pbft{} problem to minimize the makespan, with a running time guarantee. The algorithm can be easily generalized to some other flowshop scheduling problems including the $F2||f$ and $F3||f$ problems, with $f\in\{f_{max},\sum f_i\}$, and also some single machine scheduling problems with release date ($1|r_i|f$). The framework is named as \textit{Pareto Dynamic Programming}. We also discuss the practical efficiency of the algorithm and the complexity lower bound of the \pbft{} problem based on ETH (\textit{Exponential Time Hypothesis}).

Chapter \ref{ch3} reports our second contribution: a new method  which solves the \pbtt{} problem with a worst-case complexity improving existing ones. By analyzing the search tree of a basic \techbr{} algorithm, constructed based on some well known properties, we observed that many nodes (i.e. subproblems) are identical. By avoiding repetitive solution of identical nodes, a new method called \techbm{} is conceived. Identical nodes are \textit{merged} to one. By carefully designing the structure of merging, we are able to analyze and prove the complexity of the resulting algorithm. 
The method runs in polynomial space and can be generalized to other problems that verify certain properties.

Chapter \ref{ch4} extends the work of Chapter \ref{ch3} by considering the practical efficiency of algorithms. During the study, another powerful technique called \techmemo{} drew our interest. 
We discuss the link between \techbm{} and \techmemo{}. We also found some new computational properties of the \pbtt{} problem, which finally allows to solve instances with 300 more jobs than the state-of-the-art algorithm. Then the same treatment is performed on another three problems \pbrisumc, \pbfsumc{} and \pbdtilde. The computational results are provided, compared and analyzed.

Finally, in the conclusion chapter, we discuss the inner link between the involved techniques and provide our perspectives.

\chapter{Fundamentals of Scheduling and Exact Exponential Algorithms}\label{ch1}
\section{Scheduling Theory}

\begin{quote}
``Scheduling is a decision-making process that is used on a regular basis in many manufacturing and services industries. It deals with the allocation of resources to tasks over given time periods and its goal is to optimize one or more  objectives.''\hfil---\cite{pinedo2008scheduling}
\end{quote}

Scheduling plays an essential role in industries. It is undoubted that in any industrial process, some kinds of scheduling decisions must be made to well organize the process of different tasks. {The decision can be an assignment of resources to tasks, an ordering of tasks to be processed or the exact time intervals of tasks during which these tasks are processed, etc.} 
In this sense, scheduling covers a really large variety of problems. Some of these problems could be uncritical or just enough simple (the size of the problem is so small that the optimal solution is obvious), while others are extremely difficult and need much research effort to break through. {As an example, consider a problem denoted as \pbrisumc: a set of tasks need to be processed sequentially on one machine, each job has a release date and a processing time, the objective is to find an ordering of tasks which minimizes the sum of completion time of each job. The problem description being simple, this problem is strongly \nph.}

Today, scheduling has become a major field within operational research and it also benefits contributions from other angles such as Mathematics, Computer Science and Economics. When and how did all this start? What are the applications of scheduling? How are different scheduling problems classified?   What are the common techniques to deal with them? In this chapter we try to make a brief review on scheduling theory and clarify these points.

\subsection{History}\label{sec:his}
Arguably, scheduling as a research field can be traced back to the early twenties when Gantt, as a management consultant, discusses the organization of work to maximize profit in his book ``Work, wages and profit''\citep{gantt1916work}. He proposed to give the foreman each day a list of jobs to be done and to well coordinate the jobs so that there is no (or least) interference between jobs.~\footnote{The scenario that Gantt considered is now called \textit{jobshop}}
However, it is about forty years later, when a collection of millstone papers, among which the seminal paper by \cite{johnson1954optimal}, are published that scheduling started to be considered as an independent research area within operational research. Since then, it has been attracting more and more research attention.

\cite{potts2009fifty} reviewed the milestones of scheduling over fifty years. In an elegant way, the paper provides the reader a snapshot for each decade. Starting from mid fifties, the first decade witnessed the research contributions on combinatorial analysis of scheduling problems. Various problem dependant properties are incorporated into algorithms to restrict the search space. The most typical property would be the priority or precedence relation between jobs which generally states something like ``at least one optimal schedule exist such that job A precedes job B''.  The most common technique to achieve that is some deduction based on \textit{Pairwise Interchange}. It is interesting to notice that most well-known algorithms proposed during this period are polynomial-time algorithms, even though concepts on computational complexity were not yet clear at that time. The edited collection of \cite{muth1963industrial} and the research monograph by \cite{conway1967theory} helped on the expansion of the scheduling community.

Works on combinatorial analysis continued during the second decade but the resulting properties are more complex, for example a precedence relation may only be valid under certain conditions. Not being directly useful to solve the problem, those properties induce dominance conditions and serve particularly to build  \techbb{} algorithms. Problems that are typically treated by \techbb{} during this period include flowshop scheduling and single machine total weighted tardiness problem. A new book by \cite{baker1974introduction} stepped onto the stage and replaced the one of \cite{conway1967theory} as the main scheduling textbook. 

Two topics were highlighted during the third decade: computational complexity and classification scheme. Since the seminal paper of  \cite{edmonds1965paths}, polynomial algorithms have been generally considered as {``good''} algorithms and therefore became the search target of schedulers. However, it seems impossible to design this kind of algorithm for certain problems, due to their ``inherent hardness''. This feeling was justified by \cite{cook1971complexity}  which stated the existence of ``hard'' problems for which polynomial algorithms are unlikely to exist. \cite{karp1972reducibility} followed up by formulating this discovery in a more formal way and by adding the first list of \npc{} problems. Further, the monograph of \cite{garey1979computers} demonstrated various techniques for proving $\mathcal{NP}$-hardness and discussed the complexity issue of 300 problems from 12 areas including scheduling. The book remains a classical text on the topic. About the classification scheme, the most important work was done by \cite{graham1979optimization}, in which the well-known three-field notation $\alpha|\beta|\gamma$ is proposed. Field $\alpha$  describes the machine environment (single machine, parallel machine, etc); field $\beta$ specifies problem-dependent parameters (release time, due dates, etc) and field $\gamma$ is the objective functions to minimize. More details on this can be found in section \ref{sec:class}.

During the fourth decade, the complexity status of many classic scheduling problems are known. Besides of a more frequent appearance of \textit{Column Generation}, it was still lack of weapons to attack $\mathcal{NP}$-hardness in an exact way. More efforts were put on approximation and heuristic algorithms. Apart from specific algorithms depending on problem properties, various approaches were well discussed. Local Search, a method to improve solution quality by searching the \textit{neighborhood} space, which can be traced back to late 50's \citep{croes1958method}, starts to attract the attention of schedulers \citep{nicholson1967sequential}. The early publications on Simulated Annealing \citep{kirkpatrick1983optimization,cerny1985thermodynamical} made Local Search a charming research area, after which a number of new developments start to pop up. Without being exhaustive, Tabu Search is introduced by Glover \citep{glover1986future,glover1989tabu,glover1990tabu}. At the same period, Genetic Algorithms \citep{goldberg1989genetic} also started to gain attention. Ant Colony Optimization was introduced in early 90's by  \cite{dorigo1991positive, dorigo1992optimization,dorigo1996ant}, and a survey on this metaheuristic is provided by \cite{dorigo2005ant}. A recent survey on main metaheuristics is provided by \cite{boussaid2013survey}.

The development of scheduling  afterwards (the 00's) was veritably diverse. Classic scheduling models were largely extended by adding new constraints coming from realistic industrial environments. Enhanced models include, for instance, \textit{Online Scheduling} \citep{pruhs2004} in which  jobs to schedule arrive over time; \textit{Scheduling With Batching} \citep{potts1992integrating,potts2000scheduling} in which a set of jobs can be processed as a batch;  
\textit{Scheduling With Machine Availability Constraints} \citep{lee1996machine,sanlaville1998machine,schmidt2000scheduling,lee2004machine}, in which the machines may not be available during some intervals of time due to the need of reparation or maintenance. 
\textit{Multicriteria Scheduling} \citep{t2006multicriteria}, \textit{Robust Scheduling} and
\textit{Scheduling with Variable Processing Time}, etc, are also attracting research efforts.

Today (August 2017), searching the term ``scheduling problem'' in Google Scholar for publications published since 2010 gives more than 530000 results, which affirms that \textit{Scheduling} is now a well established domain. More and more new results and solving tools are being added to this field, while the challenges in front of us remain numerous.

\subsection{Applications}\label{sec:ordoapp}
It is not exaggerated to say that scheduling is involved in all industrial activities and it has a critical role in the decision making process. Indeed, in any well established industry, it is much natural to have decisions to make, on the allocation of resources, schedule of activities, etc. 
According to the application areas, we may encounter scheduling problems from manufacturing production, project management, computer systems, transportation and distribution and those from service industries, without being exhaustive.


Some application examples are introduced by \cite{pinedo2008scheduling}, such as \textit{Airport Gate Assignment Problem}, can be considered as in service industry, in which planes must be assigned to boarding gates according to their arrival and departure time to minimize delays suject to a number of known constraints (e.x. some gates are only accessible for small planes) and also uncertainties (weather, technical problems); \textit{CPU Task Scheduling Problems}, which are encountered during the evolution of computers and is about scheduling task execution on CPU to maximize system performance and user experience.  Other applications include Nurse Rostering Problems in hospitals, University Timetabling, Vehicule Routing Problems, Aircraft Scheduling Problems and Crew Scheduling Problems, Bus and Train Driver Scheduling, Sports Scheduling Problems, etc \citep{leung2004handbook}. Also, \cite{harj2014scope} have done an excellent review on the industrial aspect of scheduling: they present briefly the approaches that are adopted by industries, discuss their strength and weakness and also the gap between academic research models and real industrial problem environments.

\subsection{Classification and Notation}\label{sec:class}
Scheduling problems can be classified according to different characteristics. For instances a given scheduling problem can be:
\begin{itemize}
    \item \textit{deterministic} or \textit{stochastic}: a scheduling problem is deterministic if all the data (processing times, release dates, etc) of the problems are well defined and it is stochastic in contrary if some of these characteristic values are not known precisely but are known to follow some probability laws.
    \item \textit{static} or \textit{dynamic}: a scheduling problem is static if all the data are known and will not change during the solution, while it is dynamic if the data may change during the solution and the schedule has to be adjusted in real time. The data changes could be, for instances, new jobs arrive during the solution or the actual processing time of jobs are different than what was expected, etc.
    \item \textit{unitary} or \textit{repetitive}: a scheduling problem is repetitive if its operations appear to be cyclical and it is unitary if each operation corresponds to one unique product. 
\end{itemize}

A more common way to classify  scheduling problems is based on the machine environment and job characteristics. This is well adopted  especially in the context of machine and processor scheduling problems. In the following, we present this classification scheme together with the notation, knowing that we focus on machine scheduling problems in the scope of this thesis, hence specific classifications or notations are not provided on some other scheduling areas like Resource-Constrained Project Scheduling Problem (RCPSP), even though the notation presented below can be extended to this problem \citep{brucker1999resource}.

As already mentioned in the previous section, the most important work on the classification, especially the notation, of scheduling problems was done by \cite{graham1979optimization} in which a three-field notation ($\alpha|\beta|\gamma$) was proposed. This notation has become a well-known language of schedulers, called ``Scheduleese'' by E.L. Lawler according to \cite{lenstra1998mystical}. The adoption of this notation makes the formulation, representation and comparison of scheduling problems very clear and concise.

The detail of this notation can be found directly in the original paper of Graham et al. or in some classic scheduling textbooks like the one of \cite{brucker2007scheduling}. However, for the completeness of the content of thesis, we present this classification scheme through the sections below for reference purpose.

The generic sequencing problem can be viewed as to process $n$ jobs $J_i$, $i=1,...,n$ over $m$ machines $M_j$, $j=1,...,m$. A \textit{schedule} determines for each job, the time intervals and machines for its execution. We normally suppose that each machine can only execute one job at a time and one job can only be processed by one machine at a time. For instance, the \pbdtilde{} problem asks to schedule jobs on one machine in order to minimize the sum of weighted completion time of jobs, subject to the deadlines of jobs. Note that the solution is often expressed as an ordering of jobs, because jobs are often supposed to start as early as possible in most situations, hence the time intervals of jobs are determined automatically. This is the case for all problems treated in this thesis.

\subsubsection{Job Data}
A given job $J_i$ may consist of multiple operations which are noted as $O_{ij}$ (the $j$-th operation of $J_i$). $\mu_{ij}\subset \{M_1,...,M_m\}$ defines the set of candidate machines that can process $O_{ij}$. The most common data   associated to $J_i$ are the following:
\begin{itemize}
    \item $p_{ij}$, the processing time of $O_{ij}$. We simply note $p_i$ if $J_i$ implies only one operation,
    \item $r_i$, the release date of $J_i$, before which $J_i$ cannot be processed,
    \item $d_i$, the due date of $J_i$ job. $J_i$ is then expected to be expected before this date, otherwise some penalty is induced depending on the problem,
    \item $w_i$, a weight associated with $J_i$. This is may be used to define the contribution to the cost function of $J_i$,
    \item $f_i(t)$, the cost function which measure the cost induced by $J_i$ when it is completed at time $t$.
\end{itemize}
Other data may be present according to the problem under consideration.

\subsubsection{Field $\alpha$ - Machine Environment}
The field $\alpha$  may contain two parameters : $\alpha=\alpha_1\alpha_2$, with $\alpha_1$ defining the machine environment and $\alpha_2$ indicating the  number of machines. $\alpha_2$ can be an integer value or $m$ which means an arbitrary fixed number or empty which means an arbitrary number. We have $\alpha_1\in\{1, P,Q,R, G,X,O,J,F\}$ with the following meaning:

\begin{itemize}
    \item if $\alpha_1 = 1$, then all jobs must be processed on one single machine, hence $\alpha_2=\emptyset$,
    \item if $\alpha_1 \in \{P,Q,R\}$,  then each job has a single operation 
            \begin{itemize}
                \item if $\alpha_1 =P$, then we have identical parallel machines, i.e. the processing time of $J_i$ on $M_j$ is $p_{ij}=p_i$, $\forall j=1..m$,
                \item if $\alpha_1 =Q$, then we have uniform parallel machines, i.e. $p_{ij}=p_i/s_j$ with $s_j$ the processing speed of $M_j$, $\forall j=1..m$,
                \item if $\alpha_1 =R$, then we have unrelated parallel machines, i.e. $p_{ij}=p_i/s_{ij}$ with $s_{ij}$ the speed of $M_j$ for processing $J_i$, $\forall j=1..m$,
            \end{itemize}
    \item if $\alpha_1 \in \{G,X,O,J,F\}$,  then each job $J_i$ is made up of multiple operations $O_{i1},...,O_{in_i}$, each operation being processed by a single dedicated machine. There may also have  precedence relations between operations of the same job. This model in general is called General Shop by setting $\alpha_1=G$
             \begin{itemize}
                \item if $\alpha_1 =J$, then we have a jobshop problem, in which we have precedence relations of the form of $O_{i1}\rightarrow O_{i2}\rightarrow...\rightarrow O_{in_i}$. Also, we often assume that consecutive operations are not processed on the same machine.
                \item if $\alpha_1 =F$, then we have a flowshop problem, which is a special jobshop in which $n_i=m$ and $\mu_{ij}={M_j}$, $\forall j=1..m$.
                \item if $\alpha_1 =O$, then we have an openshop problem, which is a flowshop without precedence relations between operations of the same job,
                \item if $\alpha_1 =X$, then we have a mixed shop problem which combines jobshop and openshop.
            \end{itemize}
\end{itemize}

\subsubsection{Field $\beta$ - Job Characteristic}
The field $\beta$ defines job characteristics or extra problem parameters. Most common values are listed below, and note that multiple values can appear together.
\begin{itemize}
    \item Preemption. If $pmtn$ appears in $\beta$, preemption is allowed, i.e. a job being processed may be interrupted and resumed later,
    \item Precedence. If $prec$ appears, there are precedence relations between jobs that can be represented as a directed acyclic graph $G=(V,A)$ with $V$ the job set and $A$ the arc set whose direction indicate the precedence relation. Replace $prec$ by other keywords to provide more precise indication : $intree(outtree)$ if $G$ is a tree and the maximum outdegree(indegree) of vertices less or equal to one; $sp-graph$ if $G$ is series parallel.
    \item Release date. If $r_i$ appears, there is a release date for each job $J_i$, otherwise $r_i=0$ for all $i$,
    \item Restriction on processing times. For example $p_i=p$ indicates that jobs have equal processing times.
    \item Due date. If $d_i$ appears, a due date is set for each job $J_i$, such that $J_i$ is expected to finish at that time, otherwise some cost may arise depending on the problem.
    \item Deadline. If $\tilde{d_i}$ appears, a deadline is set for each job $J_i$ so that $J_i$ must be finished before it.
    \item Batching. In some problems jobs can be batched for processing. A setup time is needed to prepare a batch, and the processing time of the batch can be the maximum (or sum of) processing time of jobs in the batch, indicated by $p-batch\ (s-batch)$.
\end{itemize}

\subsubsection{Field $\gamma$ - Optimality Criteria}
The criteria to optimize usually depends on the completion time of jobs. We denote the completion time of $J_i$ by $C_i$ and associate a cost function  $f_i(C_i)$ to each job. Two bottleneck objectives functions can be met:
$$f_{max} = \max_{i}f_i(C_i)$$
and sum objectives
$$\sum f=\sum_i{f_i(C_i)}.$$
Various job cost functions are considered in the literature, such as:
\begin{itemize}
    \item Earliness, denoted by $E_i=max(0, d_i-C_i)$,
    \item Tardiness, denoted by $T_i=max(0, C_i-d_i)$,
    \item Lateness, denoted by $L_i= C_i-d_i$,
    \item Unit penalty, denoted by $U_i$. With $U_i=0$ if $C_i<d_i$ otherwise $U_i=1$,
    \item Lateness, denoted by $L_i= C_i-d_i$,
    \item Absolute deviation, denoted by $D_i=|C_i-d_i|$,
    \item Squared deviation, denoted by $S_i=(C_i-d_i)^2$.
\end{itemize}

Job weights $w_i$ are often considered in $f_i$. Therefore, some widely considered objective functions are $C_{max}$, $L_{max}$, $\sum T_i$, $\sum w_iT_i$, $\sum U_i$, $\sum w_iU_i$, $\sum D_i$, $\sum w_iD_i$, $\sum S_i$, $\sum w_iS_i$, $\sum E_i$, $\sum w_iE_i$. Note that more than one functions can appear in this field in case of \textit{Multicriteria Scheduling Problems}, in which case the notation is accordingly adapted \citep{t2006multicriteria}.

\subsubsection{Example of notation}
As an example, the \pbft{} problem asks to schedule jobs in a Flowshop environment with three machines to minimize the global completion time. The \pbtt{} problem asks to schedule jobs on one machine to minimize the sum of tardiness of jobs, with respect to the due date of each job.


\section{Exact Exponential Algorithms (EEA)}
When considered as an independent research area, EEA is at the intersection of several well established domains such as Combinatorial Optimization, Operations Research, Computer Science, and Complexity Theory, in which we are faced with a \textit{hard} decision or optimization problem for which we search for an exact algorithm. In the context of Complexity Theory, we consider a problem as \textit{easy} if it can be solved in polynomial time, the algorithms we search for here are hence super-polynomial and most of the time exponential in the input size. If we note the complexity of these algorithms as $\ostar{c^n}$ with $c$ a constant, then the general objective in EEA field is to propose algorithms with $c$ as small as possible. Note that we use extensively the complexity notation $\ostar{\cdot}$ to suppress polynomial multiplicative factors for simplification. That is, $f(x)=\ostar{g(x)}\Leftrightarrow f(x)=\bigo{p(x)g(x)}$, with $p(x)$ a polynomial of $x$. Recall that, $f(x)=\bigo{g(x)}$ if and only if there exists $ c>0,\ n_0>0$, such that $\forall x>n_0$, $0\leq f(x)\leq cg(x)$. Other conventional complexity notations like $\Omega(\cdot)$ and $o(\cdot)$ are also adopted. $f(x)=\Omega(g(x))$ if and only if $g(x)=\bigo{f(x)}$. $f(x)=o(g(x))$ if and only if  $ \forall c>0,\ \exists n_0>0$, such that $\forall x>n_0$, $0\leq f(x)< cg(x)$. More details can be found in the classic book of \cite{cormen2009intro}, chapter 3.1.

Why would these exponential algorithms be interesting for us ? What are the algorithm design techniques used in this domain? We discuss these points through the following subsections.

\subsection{Motivations}
The motivation for the search of \textit{exact} algorithms does not need justification: given an optimization problem, we are supposed to find an optimal solution. The question is why \textit{exponential} ? Note that except explicit statement, in the context of this thesis the complexity of an algorithm refers to the worst-case complexity (time complexity by default). This is a common practice \citep{cormen2009intro} since the worst-case complexity provides an upper bound on the running time of the algorithm, and hence ensures that the performance will not be worse. On the other hand, worst-case occurs pretty often in practice for some problems such as searching an entry that does not exist in a database.

\cite{edmonds1965paths} argued on the meaning of a \textit{good} algorithm and he claimed that a \textit{good} algorithm should have its performance algebraically related to the input size of the problem, i.e. the complexity of the algorithm should be polynomial on the input size. At the same period, the complexity class $\mathcal{P}$ started to be considered \citep{cobham1965intrinsic}. Since then, it is well accepted by the community that the search of polynomial time algorithm is the primary objective when solving any problem. Unfortunately, under the well believed $\mathcal{P} \neq \mathcal{NP}$ conjecture, polynomial algorithms do not exist for some \textit{hard} problems, which are often classified as \npc{} or \nph{} problems. EEA are proposed for these problems.

Therefore, first of all we hope it clear that we are searching for Fast Exact Algorithms, however the intrinsic difficulty of target problems implies the exponentiality of resulting algorithms. That is, if we want to find the optimal solution of a \textit{hard} optimization problem with a worst-case guarantee, we are necessarily on the road of searching an EEA, with no other choices. This is the first motivation.

Consider polynomial algorithms as \textit{fast} and exponential algorithms as \textit{slow} is fully justified from an asymptotic view, where the input size is necessarily large. However, it is not the case in practice where the input size is often bounded. It is easy to see that $n^3 > 1.0941^n\cdot n$ for $n\leq 100$. It is therefore worthy to consider exponential algorithms and is particularly interesting to derive \textit{moderately} exponential algorithms. As a canonical example, the MIS (Maximum Independent Set) problem which asks to find a largest subset of pairwise non-adjacent vertices from a given undirected graph, has been benefiting from a series of algorithmic improvements. To the authors' knowledge, the currently fastest algorithm runs in $\ostar{1.1996^n}$ \citep{xiao2017exact}, which can be thought as \textit{moderately exponential}. Therefore, the performance of an exponential algorithm in practice is not necessarily bad. This is the second motivation. 

As stated before, the high complexity of algorithms comes from the intrinsic difficulty of problems. On the contrary, the research on faster EEA helps understand better the origin of that difficulty and the possibility to tackle it. Another important observation is that all \nph{} problems are not equally hard. Some problems have been shown to have algorithms running much faster than the brute-force (e.g. MIS), while for others (e.g. Satisfiability) the enumeration of all possible solutions stays the best choice. Therefore an enhanced research on EEA for these problems will also lead to a finer complexity classification of them. 
These theoretical interests are the third motivation.

Finally, for a problem that can be solved in $\bigo{c^n}$ with $c$ a constant, improving the complexity to $\bigo{c-\epsilon}^n$ allows to enlarge the size of solvable instances within a given amount of time by a multiplicative factor, while running the $\bigo{c^n}$ algorithm on a faster computer will only increase this size by an addictive factor. Therefore research efforts are worth to be paid to minimize ``$c$''.

\subsection{Common Algorithmic Frameworks}\label{sec:ch1:algo}
The ambition to conquer \nph{} problems by designing fast exact algorithms have been existing for a long time. However, it is during recent years that EEA has been attracting more and more research intelligence than ever. Fomin and Kratsch in their monograph \textit{Exact Exponential Algorithms} \citep{fomin2010exact} attributes the source of this prosperity to the survey of  \cite{woeginger2003exact}, in which several techniques are summarized and several open problems are proposed, intriguing a number of researchers to enter this area. The monograph of \cite{fomin2010exact} is a standard textbook on this topic. Most common algorithmic frameworks are well discussed in the book, making it a handy manual when tackling hard problems. In this section, we review some of the most important algorithmic frameworks that are to be considered when facing with an \nph{} problem and we kindly refer the reader to the book \citep{fomin2010exact} for more details.

\subsubsection{Branching}\label{sss:branching}
Branching is perhaps the most natural solving approach when dealing with an optimization problem, because it is based on the simple enumeration idea. All combinatorial problems can be solved in finite time by enumerating the whole solution space. A candidate solution can be represented as a set of variables and the optimal solution has its variables' values inducing the optimal objective function value. A naive branching algorithm, at each \textit{branching}, simply chooses a variable and fix its value - we say that we branch on this variable - to create a new \textit{branch}, which corresponds to a subproblem with one less variable to consider. 
The algorithm continues in this way and get a candidate solution when all variables are fixed. And the algorithm stops when all possible values of all variables have been tested. The optimal variable assignment is returned, representing the optimal solution of the problem. Based on this initial structure, the algorithm can be enhanced in different ways to avoid enumerating the whole solution space, by making use of problem dependent or independent properties.

As an example, consider the MIS (Maximum Independant Set) problem. Given an undirected graph $G=(V,E)$, it asks to find $S\subseteq V$ with the maximum  cardinality, such that $\forall u,v \in S$, $(u,v)\notin E$. A trivial brute-force algorithm or a naive branching algorithm would just try all $2^{|V|}$ subsets of vertices, verify the constraint of non-adjacency and return the maximum subset found. However, a simple utilization of the non-adjacency property allows to drastically reduce the search space. Consider branching on a vertex $v$ for two possibilities (hence implying a boolean variable): either $v$ belongs to a maximum independent set $S$ or not. Two observations follow:
\begin{enumerate}
    \item $v\in S$: according to the non-adjacency property, all neighbors of $v$, denoted by $N[v]$, must be excluded from $S$.
    \item $v\notin S$: then at least two neighbors of $v$ are inside $S$. Otherwise by forcing $v$ into $S$ and remove $N[v]$ from $S$, we have a not-smaller independent set.
\end{enumerate}

The above observations can be embeded into the branching procedure: whenever we make decision on one vertex $v$, we are also making decisions or partial decisions on some other vertices. This boosts the branching algorithm since it is much faster than the pure enumeration on all vertices. Similar but more complex observations on this problem can further enhance the branching.

The above branching algorithm is referred to as \techbr{} by \cite{fomin2010exact}. Note that several different names are used in the literature such as search tree algorithm, backtracking algorithm, Davis-Putnam type algorithm, etc. 
Branching algorithms can be usually represented over a search tree, in which each node represents a subproblem to be solved, which is explored by the algorithm. Various problem dependent or independent techniques can be applied to \textit{prune} the tree, that is to reduce the search space and hence accelerate the solution. 
Depending on the pruning scheme, we face various names of branching algorithms like \techbb{}, which is commonly used to solve \nph{} problems. Considering a minimization problem, at each node, the algorithm evaluates a LB (Lower Bound) on the objective function value that we can expect if branch on this node. If the LB is greater than the global UB (Upper Bound) given by a candidate solution, then the current node should be cut since it will not provide a better result. Note that even though \techbb{} is widely used in practice, it is not commonly considered in the context of EEA. This is due to the difficulty to analyze tightly its worst-case complexity, which is usually determined by the chosen branching scheme. In fact, the practical average efficiency of \techbb{} algorithms largely depends on the quality of the LB computed at each node. On the theoretical side, proving that this pruning scheme removes an exponential number of nodes, in the worst-case scenario, is a hard task. Other techniques that can be categorized as branching algorithms include \techbc{}, \techbp{} etc. Integer Programming and Constraint Programming can also be seen as a problem solving approach using internal branching algorithm. Again, we insist on the fact that in the context of EEA, we are only interested in branching algorithms that can provide a worst-case running time guarantee better than that of brute-force algorithms.

\subsubsubsection{Complexity Analysis}

The complexity analysis of branching algorithms is worthy to be introduced here. Since branching algorithms are usually represented as a tree and each single branching operation is usually done in polynomial time, the time complexity of branching algorithms depends on the total number of nodes or the number of leave nodes in the tree. Let $T(n)$ be the maximum number of leaves in a branching tree of input size $n$. A recurrence relation can be established between a node and all its subnodes if the branching rule is systematic. Let us show this on the above presented MIS problem, without considering the observation 2. Let $deg(v)=|N[v]|-1$ be the degree of $v$ in $G$ and $deg(G)=\max_{v\in V}deg(v)$. At each branching, we need to decide whether a vertex $v$ belongs to a maximum independent set $S$ or not. If yes, a subnode is created which corresponds to a problem of size $(n-|N[v]|)$, with $N[v]$ denoting the set of neighbours of $v$ (including $v$), according to observation 1; otherwise if $v$ is excluded from $S$, a subnode is created corresponding to a subproblem with $v$ deleted from the graph, therefore of size $(n-1)$. The following relation is then established:
$$ T(n) \leq T(n-|N[v]|) + T(n-1)$$

If we decide to choose $v$ as the vertex with maximum degree, then $deg(v)\geq 3$ because if $deg(v)\leq 2$ we can solve it in polynomial time without need to branch \citep{fomin2010exact}. Hence, the recurrence becomes

$$ T(n) \leq T(n-4) + T(n-1)$$

This kind of linear recurrences are known to have the solution of the form $T(n)=c^n$, with $c$ the largest root of the equation

$$ x^n = x^{n-4} + x^{n-1}$$.

More generally, we may associate to this recurrence a \textit{branching vector}, which is defined as the vector of the size reduced in each branching case, hence here $\langle4,1\rangle$. Given a branching algorithm, we analyze its induced branching vector in worst-case, then solved the related equation of the above form to obtain the complexity. The equation can be solved easily by some mathematical solvers, (the one used here is WolframAlpha\footnote{\url{https://www.wolframalpha.com/}}) and the returned result is $T(n)=\bigo{1.3803^n}$. Notice that this complexity is in fact an upper bound on the worst-case time complexity of the branching algorithm. Recurrence solving is well discussed in textbooks on discrete mathematics  \citep{graham1994concrete, rosen2012discrete}.

\subsubsubsection{Measure \& Conquer}
The above complexity analysis supposes that every vertex has the same weight (which is 1) in the problem size measurement. This explains the branching vector $\langle4,1\rangle$. Nevertheless, the analysis is not necessarily tight in the sense that the vertex $v$ may have $deg(v)$ much greater than 3, but this cannot be considered in the recurrence function. More concretely, branching on a vertex $v$ with $deg(v)=3$ versus branching on vertex $u$ with $deg(u)=100$ makes a huge difference on the running of the algorithm, though both vertices are considered having the same weight (which is 1) in the previous analysis. Also, as stated above, a vertex $v$ with $deg(v)\leq 2$ do not need to be branched, hence, it has a little impact on the running time, but they are also holding the same weight (which is 1) as others.

With all these observations in mind, a tighter complexity analysis method is proposed as \textit{Measure \& Conquer}. The key is to choose the problem measure in a clever way. Three conditions need to be satisfied:

\begin{enumerate}
    \item The measure of a subproblem obtained by a branching must be smaller than the measure of the father problem.
    \item The measure of an instance is non-negative.
    \item The measure of the input should be upper bounded by some functions of natural parameters of the input. The ``natural parameter'' here means the classic measure choices such as the number of vertices in a graph. This consideration allows to compare the complexity result obtained by Measure \& Conquer to classical results.
\end{enumerate}

Consider the following algorithm :
\begin{enumerate}
    \item If $\exists v\in V$, $deg(v)\leq 1$, count $v$ into the maximum independent set $S$ 
    \item If $deg(G)=2$, solve the problem in polynomial time.
    \item If $deg(G)\geq 3$, branch on $v$ with $deg(v)=deg(G)$.
\end{enumerate}
 
 By the classical analysis, the branching vector is $\langle 1,4\rangle$, and the result is $T(n)=\bigo{1.3803^n}$ as presented above. Now consider a new weight attribution:
 
 \begin{enumerate}
    \item  $\forall v$ with $deg(v)\leq 1$, $weight(v)=0$, since these vertices do not need to be branched and they are treated in polynomial time.
    \item $\forall v$ with $deg(v)=2$, $weight(v)=0.5$
    \item $\forall v$ with $deg(v)\geq 3$, $weight(v)=1$
\end{enumerate}

 Now we analyze the situation when branching on $v$. As a starting point, consider $deg(v)=deg(G)=3$, if $v$ is selected into $S$, then the total weight of instance decreases by $1$ for the removal of $v$. The neighbors of $v$, i.e. $N(v)$ should also be removed and their initial degree is at least 0.5, therefore the total weight reduction is $1+0.5*3=2.5$. If $v$ is not selected in $S$ and removed from $V$, then the degree of $u$, $u\in N(v)$, passes from 1 to 0.5 or from 0.5 to 0, therefore the total weight reduction is $1+0.5*3=2.5$. The branching vector becomes $\langle 2.5,2.5\rangle$, inducing $T(n)=\ostar{1.3196^n}$. What about the case where $deg(G)\geq 4$? 
 Selecting or discarding $v$ implies the removal of $v$ which decreases the weight by $1$. The weight decrease on the neighbours of $v$ depends on their initial weight. Therefore let $x=|\{u: u\in N(v) \bigwedge d(u)=2\}|$, $y=|\{u: u\in N(v) \bigwedge d(u)=3\}|$ and $z=|\{u: u\in N(v) \bigwedge d(u)\geq 4\}|$. Now the following relations hold:

 $$
 \begin{cases}x+y+z=d(v) \\
 dis = 1+0.5x+0.5y \\
 sel = 1+0.5x+y+z
 \end{cases}
 $$
 
 According to Lemma 2.3 in the book of \cite{fomin2010exact}, the branching factor is smaller when it is more balanced, i.e. $\langle k,k\rangle\leq \langle i,j\rangle$ if $i+j=2k$. Hence 
 $\langle dis,sel\rangle \leq \langle 1,dis+sel-1\rangle=\langle 1,1+x+1.5y+z\rangle\leq\langle 1,1+d(v)\rangle\leq\langle 1,5\rangle$.
 
 
 That is, the branching vector is at least $\langle 1,5 \rangle$, inducing $T(n)=\ostar{1.3248^n}$. Consequently, without modifying a single part of the algorithm, we have proved that the complexity is bounded by $\ostar{1.3248^n}$ instead of $\ostar{1.3803^n}$.

For more details on Measure \& Conquer, we refer the reader to the works of  \cite{fomin2005measure, fomin2006measure,fomin2009measure}.

\subsubsection{Dynamic Programming}
Since its early appearance in the 50's (see for instance \cite{bellman1952theory,Bellman1957}) with the term coined by Richard Bellman \citep{dreyfus2002richard}, \textit{Dynamic Programming} (DP) has become rapidly one of the most important algorithmic tool for solving complex optimization problems; especially for problems whose solution can be expressed as a sequence of decisions. The elementary principle behind DP is often referred to as \textit{Principle of optimality} which can be briefly stated as follows:
\begin{quote}
\textit{The partial solutions of an optimal solution are themselves optimal for the corresponding subproblems.}
\end{quote}

Though it seems not complex to be understood today, this principle directly induces DP as a powerful methodology. For instance, until now the DP is still the algorithm having the best time complexity ($\ostar{2^n}$) for the  \textit{Travelling Salesman Problem} (TSP) since the 60's. The corresponding recursive equation, also called Bellman equation, of the DP for the TSP is of the following form \citep{bellman1962dynamic,held1962adynamic}:
\begin{equation}\label{eq:tsp}
{Opt}(S,i)=\min\{Opt(S-\{i\},j)+d(j,i): j\in S-\{i\} \}
\end{equation}

where $S=\{2,\ldots,n\}$, $Opt(S,i)$ denotes the length of the shortest path from city $1$ to $i$ passing all cities in $S$, with $Opt(\{i\},i)$ initiated to $d(1,i)$, the distance between cities $1$ and $i$. The same idea can be easily extended to general sequencing problems with regular non-decreasing cost functions, as stated by \cite{held1962adynamic,woeginger2003exact}. In fact, when seeking a worst-case guarantee, Dynamic Programming seems the most suitable choice for \nph{} scheduling problems and it yields the best worst-case complexity for a lot of them, even though it has the drawback of having an exponential space complexity. The reader will see later that our first contribution is also about a Dynamic Programming algorithm. 

\subsubsection{Sort \& Search}
The brute-force approach on \nph{} optimization problems often requires an enormous time because of the complete traverse of the whole solution space. Faster algorithms are able to eliminate part of the solution space during the search without missing the/an optimal solution. 

The idea of \techss{} is to rearrange the solution space by some preprocessing on the input data so that the search for an optimal solution is better guided. More concretely, it first sorts partial candidate solutions into some order, then later when constructing a complete global solution by several partial solutions, the construction can be done quickly by means of efficient search. 

We illustrate this technique on the \pbss{} problem. In this problem we are given a set of positive integers $a_1,a_2,...,a_n$ and $A$, and the aim is to find a a subset $I\subseteq \{a_1,...,a_n\}$ such that $\sum_{i\in I}i = A$.

The trivial enumeration algorithm traverses all subsets of $\{a_1,...,a_n\}$ and tests for each one whether its sum gives $A$. The induced worst-case time complexity is therefore $\ostar{2^n}$. Now consider the following treatment.
\begin{enumerate}
    \item Partition integers $\{a_1,...,a_n\}$ into two equal-sized sets $S_1$ and $S_2$ in an arbitrary way. Assume $n$ is even without loss of generality.
    \item Generate all subsets of $S_1$ and $S_2$ respectively. Let us note $SS_1=\{I|I\subseteq S_1\}$ and $SS_2=\{I|I\subseteq S_2\}$. Hence, $|SS_1|=|SS_2|=2^{n/2}$. Note that if $\exists I\subseteq\{a_1,...,a_n\}$ such that $\sum_{i\in I}i = A$, then $\exists I_1\subseteq SS_1, I_2\subseteq SS_2$ such that $I_1\cup I_2=I$.
    \item Sort elements of $SS_2$ according to their sum value. (Sort Phase)
    \item Enumerate elements of $SS_1$, for each $I_1\subseteq SS_1$, search in $SS_2$ an element $I_2$ such that $sum(I_2)=A-sum(I_1)$. (Search Phase)
\end{enumerate}

The correctness of the algorithm is evident. What is its complexity?
The best complexity of comparative sorting is $\bigo{N\log N}$ for $N$ elements, therefore the sorting here requires $\bigo{2^{n/2}\log 2^{n/2}}=\bigo{n 2^{n/2}}$. The searching on the sorted set $SS_2$ for each $I_1\subseteq SS_1$ requires $\log{2^{n/2}}$. Therefore the global time complexity is $\bigo{n 2^{n/2}} + 2^{n/2}\log{2^{n/2}}=\bigo{n 2^{n/2}}$ which is faster than $\bigo{2^n}$.

Note that the space complexity is $\bigo{2^{n/2}}$, induced by the storage of $SS_1$ and $SS_2$.

Similarly, \techss{} can be applied on \pbsac{} and \pbxsat{} \citep{fomin2010exact,horowitz1974computing, schroeppel1981a}. 
\cite{lente2013extension} proposed an extension of this method to a general class of problems called \textit{Multiple Constraint Problems} and use the method to derive faster exponential-time algorithms for several machine scheduling problems.

\techss{} accelerates the solution by making use of preprocessed information (sorting) at the price of more space requirement. The relation between time and space complexity of algorithms is further discussed in section \ref{sss:dc} and section \ref{sss:memo}.

\subsubsection{Divide \& Conquer}\label{sss:dc}
\techdc{} is one of the most common algorithm design paradigm. It can be applied on problems whose global optimal solution can be constructed from several optimal subsolutions of subproblems. At this point, it reminds us about \techdp{} and as what will be shown later, these two techniques do have strong inner connections.

In \techdc{}, a problem is solved recursively with three steps applied on each level of recursion \citep{cormen2009intro}:
\begin{enumerate}
    \item \textbf{Divide} the problem into a number of subproblems which are instances of the same problem but with smaller size.
    \item \textbf{Conquer} each subproblem in recursive way and when the size of subproblem is small enough, corresponding to the elementary case, solve it in a straightforward way without recursion.
    \item \textbf{Combine} the obtained solution of each subproblems to form the optimal solution of the initial problem.
\end{enumerate}

A very natural \techdc{} example as given by \cite{cormen2009intro} is the \textit{Merge Sort} algorithm. Given an unsorted array of numbers, merge sort first divide the array in the middle, sort each half array in a recursive way then combine the two sorted sequence into a global sorted one. The elementary case arises when the array to sort contains one element only.

The \textit{Divide} operation for sorting is pretty straightforward, nonetheless it is not always the case depending on the problem. Here we present an interesting application of \techdc{} on \pbhp{} problem, initially proposed by \cite{gurevich1987expected}. Given an undirected graph $G=(V,E)$, and two vertices $s,t\in V$, $|V|=n$, the problem asks whether there exists a path $P$ from $s$ to $t$ of length $(n-1)$, which pass by each vertex exactly once. 

Here is the \techdc{} thinking: if a Hamiltonian path $P=(s,v_1,v_2,...,v_{n-2}, t)$ exists, it can be seen as the concatenation of two half-paths $P_1=(s,v_1,...,v_{(n-2)/2})$ and $P_2=(v_{(n-2)/2},...,v_{n-2}, t)$. Furthermore, $P_1$ ($P_2$) is a Hamiltonian path of a subproblem where the graph is induced by the vertices $V_1$ ($V_2$) from the path. This naturally suggests a \techdc{} approach. However, unlike the sorting problem, here we cannot divide $V$ arbitrarily into $V_1$ and $V_2$, since not all divisions can lead to an correct answer. Therefore, we need to try all possible divisions.

Let $HP(V\smallsetminus\{s,t\},s,t)=true$ if there exists a hamitonian path from $s$ to $t$ passing by all $v\in V\smallsetminus\{s,t\}$. We first choose the middle vertex $v_{(n-2)/2}$ then try all possible partitions of $V_1$ and $V_2$, finally combine $P_1$ and $P_2$. The following recurrence relation holds:
$$
HP(V\smallsetminus\{s,t\},s,t)= \underset{\substack{
V_1\subset\{V\smallsetminus\{t\}\bigcup\{s\}\},
V_2\subset\{V\smallsetminus\{s\}\bigcup\{t\}\},\\
|V_1|=|V_2|=(|V|-3)/2+1,\\
V_1\bigcap V_2 = \{v_{mid}\}}
}
{\bigvee}
\left(HP(V_1,s,v_{mid}) \bigwedge HP(V_2,v_{mid},t)\right)
$$

The time complexity induced by the above recurrence is
\begin{align*}
T(n) & \leq (n-2)2^{n-3}2T((n-3)/2) \\
     & \leq n2^{n}2T(n/2) \\
     & = 2^{n(1+1/2+1/4+...)}n^{\log n} T(1)\\
     & = \ostar{4^nn^{\log n}}
\end{align*}
Since the algorithm is recursive, it required only polynomial space. We can apply \techdc{} on other problems as shown by \cite{bjorklund2008exact,bodlaender2006exact}.

\subsubsubsection{Time, space and Dynamic Programming}\label{sss:spacefortime}
The \pbhp{} problem can be solved by \techdp{} in $\ostar{2^n}$ time and space. Nevertheless, it can be unified with \techdc{} on this problem and other sequencing problems with a similar structure. Indeed, reconsider the above \techdc{} algorithm, instead of selecting the intermediate vertex for the middle position and trying all 2-partitions of the remaining vertices, we can choose the position of the intermediate vertex arbitrarily. When the intermediate vertex is on the last position, we got a Dynamic Programming algorithm.

This suggests us to consider the compromise between time and space requirement of our algorithm. Actually it is proved that for every $i\in\{0,...,\log_2n\}$, the problem TSP on $n$ cities is solvable in time $\ostar{2^{n(2-1/2^i)}n^i}$ and space $\ostar{2^{n/2^i}}$ \citep{fomin2010exact}. The idea is to first apply \techdc{} until reaching subproblems of limited size, which is chosen depending on $i$, then call Dynamic Programming to solve it.

We suggest to always consider these two techniques as a bundle whenever solving a given problem, so as to achieve a better time-space compromise.

\subsubsection{Memorization}\label{sss:memo}
We have already seen the possibility to achieve a time-space compromise by combining \techdp{} and \techdc{} (section \ref{sss:spacefortime}). Another technique allowing to achieve this is \techmemo{}. As its names suggests, memory space is used to store useful information which could accelerate the solution. More concretely, it is applied upon a base algorithm, most usually a branching algorithm, during the execution of which, identical subproblems may be encountered many times. With \techmemo{} enabled, each time when a subproblem is solved, the corresponding subsolution is ``memorized'', that is, stored into a hashtable, a database or any arbitrary data structure allowing a fast query. When the same subproblem is encountered again, instead of solving it again, it can be just retrieved from the ``memory''. 

As an example, consider the branching algorithm for the MIS problem introduced in section \ref{sss:branching}. For a given graph of $n$ vertices the algorithm generates $\bigo{1.3803^n}$ subproblems. For any size $1\leq p \leq n$, a subproblem of size $p$ can be solved in $\bigo{1.3803^p}$ time, hence the number of subproblems  of size $p$ solved during the solution of an instance of size $n$, denoted by $N_p(n)$, is at most $\frac{T(n)}{T(p)}=\bigo{1.3803^{n-p}}$. 
With \techmemo{}, each subproblem is solved only once, therefore $N_p(n)\leq{n \choose p}$. This leads to
$$N_p(n)=\bigo{\min(1.3803^{n-p}, {n \choose p}) }$$
with $N_p(n)$ reaches its maximum when the two inner terms are balanced. The equation
$$ 1.3803^{n-p} =  {n \choose p}$$
yields $p=0.0865n$. Therefore the number of subproblems of size $0.0865n$ is bounded by $1.3803^{n-p} =1.3424^n$ and the total number of subproblems of all sizes is bounded by $n1.3424^n=\ostar{1.3424^n}$ which is the running time of the algorithm with memorization. The space complexity is also $\ostar{1.3424^n}$ since the solution of all subproblems are memorized when these problems are solved at the first time. More details on this analysis can be found in the work of \cite{fomin2010exact,robson1986algorithms}.

Note that even though some algorithms also use memory space to store information for faster solving, such as \techss{}, \techmemo{} refers particularly to the use of memory for directly storing solutions of subproblems.

\subsubsection{Inclusion \& Exclusion}
The above mentioned techniques share a common property: they do better than the brute-force search by avoiding exploring the whole solution space, making use of some problem dependent/independent properties. There exists some algebraic techniques that tackle  problems from other angles. One representative technique is called \techie. For a given \npc{} problem, instead of searching directly for a certificate, it counts the number of certificates and if this number is greater than zero, the corresponding decision problem gets a YES answer. Counterintuitively, counting the number of certificates is sometimes easier than verifying their existence.

Karp seems the first person illustrating the use of \techie{} on solving hard problems.  He stated in his paper \citep{karp1982dynamic} that many problems of sequencing, packing and assignment can be approached by \techie{}. As an example, he showed how to apply it to solve \pbhp{}, a generic scheduling problem with release date and deadline and \pbbp. We summarize here the solution of \pbhp{} by \techie.

Let $G=(V,E)$ be a digraph with $V=\{s,t\}\bigcup\{1,...,n\}$. \pbhp{} asks to determine whether there is a Hamitonian path of length $(n+1)$ from $s$ to $t$, passing each vertex $\{1,...,n\}$ exactly once. Let $X$ be the number of the hamitonian paths. For $S\subseteq\{1,...,n\}$, let $N(S)$ be the number of paths of length $(n+1)$ from $s$ to $t$ that containing no vertex from $S$. Then, according to the principle of \techie{}, the following relation holds.

$$
X=\sum_{S\subseteq \{1,...,n\}}{(-1)^{|S|}N(S)}
$$

For fixed $S$, $N(S)$ can be computed by adjacency matrix multiplication \cite{fomin2010exact}. Let $A$ be the adjacency matrix of the induced graph containing vertices in $V\smallsetminus S$, the entry at row $i$ and column $j$ of the matrix $A^{n+1}$ corresponds to the number of paths of length $(n+1)$ from $i$ to $j$. Note that a vertex can appear repetitively in the path. Since matrix multiplication can be performed in polynomial time, $N(S)$ can therefore be computed in polynomial time. There are $2^n$ different values of $S$, thus the time complexity of the \techie{} algorithm is $\ostar{2^n}$. Moreover, the space requirement is polynomial. Note that $N(S)$ can also be computed via \techdp{} \citep{karp1982dynamic}.

\subsubsection{Exponential Time Hypothesis (ETH)}

Unlike the above presented techniques used to design fast exact exponential algorithms, \textit{Exponential Time Hypothesis} (ETH) is a powerful tool attacking problems from another direction: it is used to prove complexity lower bounds. 
The conjecture $\mathcal{P} \neq \mathcal{NP}$ implies that \nph{} problems cannot be solved in polynomial time. However, no more precision is given on the complexity of \nph{} problems except that they are super-polynomial. People have realized that this simple binary classification of decision problems to \textit{easy} or \textit{hard} is perhaps too coarse-grained. This encourages to analyze the complexity of problems in a finer way. On this direction,  \textit{Exponential Time Hypothesis}, proposed by \cite{Impagliazzo2001}, states that 3SAT cannot be solved in sub-exponential time ($\bigo{2^{o(n)}}$). Based on this assumption, complexity lower bounds have been deduced for many problems such as $k$-Colorability, $k$-Set Cover, \textsc{Independent Set}, \textsc{Clique}, and \textsc{Vertex Cover}. The existence of sub-exponential algorithms for any of them implies the same for the others. 
Moreover, a \textit{Strong Exponential Time Hypothesis} (SETH) has also been proposed, which assumes that $k$-CNF-SAT requires $\bigo{2^n}$ time in the worst-case when $k$ grows to infinity. This allows to derive even tighter bounds such as shown by \cite{cygan2016problems}: for every $\epsilon<1$ an $\bigo{2^{\epsilon n}}$ algorithm for \textsc{Hitting Set}, \textsc{Set Splitting} or \textsc{NAE-SAT} would violate SETH. 

Therefore, when it seems difficult to derive EEAs with a better complexity, it may because that the complexity lower bound of the problem is reached. It is hence interesting to prove that by using ETH.

\subsubsection{Other techniques}
Besides \techie{}, other interesting algebraic techniques exist, among which there is \techffc{}. We kindly refer the reader to the work of \cite{fomin2010exact, bjorklund2007fourier} for further details. The use of algebraic methods to decrease space complexity is discussed by \cite{lokshtanov2010saving}.

Some other techniques that are discussed in the book of \cite{fomin2010exact} include algorithms based on \textit{Treewidth}, \techls{}, etc. We kindly refer the reader to this book for further details. 

\subsection{Notes on Parameterized Algorithms}
\textit{Parameterized Algorithms} are algorithms whose complexity depends not only on the input size of instances but also on some extra parameters. This is a fairly new field of complexity theory, which is strongly related to \textit{Exact Exponential Algorithms}. The latter can be seen as parameterized algorithms which uses the input size as parameter. 
Therefore, it is necessary to add some notes on Parameterized Algorithms, even though no specific research efforts are paid on this during the PhD study.

Parameterized complexity theory is created as a framework which allows to measure the complexity of algorithms more finely. Conventional complexity theory measures the complexity only in terms of the input size, hence ignores the structural information about the input instances. Consequently, the resulting complexity may make the problem appear harder than it really is. With Parameterized Algorithms, the complexity depends both on the input size and another well chosen parameter, which is expected to be ``small''. 

Analogously to the classic P-NP complexity framework, a complexity hierarchy has been established to classify problems according to their hardness. The ``easiest'' problems in this hierarchy are \textit{fixed-parameter tractable} (FPT) problems, which are problems that can be solved by algorithms having a complexity of $f(k)p(n)$, where $n$ is the input size, $k$ is the chosen parameter and $p(\cdot)$ is a polynomial function. 
Since the complexity depends only polynomially on the input size, the running time is hence expected to be small when $k$ is well chosen as a parameter with small value in most instances. 

In fact, algorithms have always been analyzed in terms of different parameters, so parameterized algorithms do exist before this well defined framework is established. Nevertheless, the framework is also very important in the other direction, which is to prove the \textit{intractability} of problems. In the above mentioned parameterized complexity  hierarchy, FPT problems are easiest ones, analogously to class {P} in the P-NP system, and class W[P] is an analogue of NP. But unlike the P-NP system, many natural problems are neither in FPT nor W[P]-complete. To classify these problems, a hierarchy of classes within W[P] is established, called W-hierarchy. It is usually presented as 
$$FPT \subseteq  W[1] \subseteq  W[2] ...\subseteq W[P].$$

Parameterized complexity theory was developed by Downey and Fellows in early 90's. Their monograph \citep{downey1999} provided a picture on the foundation of this theory and it has been updated in 2013 \citep{downey2013fundamentals}. Another referential book in the area is the one of \cite{flum2006parameterized}. More literature references can be found through these books.

\section{Existing EEA for Scheduling}
Concerning \nph{} scheduling problems, some of the techniques presented in section \ref{sec:ch1:algo} have been used to construct exact algorithms with worst-case guarantee.

The most widely applicable technique on scheduling problems is \techdp{}, which is applied very naturally on sequencing problems \citep{held1962adynamic} since its early occurrence. Its classical application is on TSP which can also be considered as a scheduling/sequencing problem. \cite{held1962adynamic} discussed the application of \techdp{} on general sequencing problems with arbitrary cost function. In the remarkable survey of  \cite{woeginger2003exact}, a single machine scheduling problem under precedence constraints minimizing total completion time is discussed. For all problems discussed in the above two references, the induced time complexity is $\ostar{2^n}$ which is a common result whenever the \techdp{} is applied in the conventional way, that is, across all subsets. For the latter problem, an algorithm running in $\bigo{(2-\epsilon)^n}$ is reported by \cite{cygan2014}, based on the fact that the \techdp{} does not need to traverse all subsets: an amount of subsets are infeasible due to the precedence constraints. 
When applied on more complex problems, the structure and the complexity of the resulting \techdp{} depends largely on the problem's properties. The job shop scheduling problem can be solved by \techdp{} in $\bigo{p_{max}^{2n}(m+1)^n}$ \citep{gromicho2012}, which is exponentially faster than the brute-force which runs in $\bigo{(n!)^m}$. \cite{jansen2013} investigates the lower bounds of complexities of some scheduling and packing problems, based on ETH (\textit{Exponential Time Hypothesis}, see \cite{Impagliazzo2001,fomin2010exact}). They showed that problems including \pbss, \pbsac, \pbbp, $P2||C_{max}$ and $P2||w_jC_j$  have a lower bound of  $2^{o(n)}\cdot \|I\|^{\bigo{1}}$, with $\|I\|$ the input length. They also developed a \techdp{} framework that is able to solve a number of scheduling and packing problems in $2^{O(n)}\cdot \|I\|$ time. 
Flowshop problems can also be treated by a generalized \techdp{}, as will be discussed in Chapter \ref{ch2}. The resulting time complexity is $\ostar{c^n}$ with $c>2$ a constant. Noteworthy, for all problems that can be treated by \techdp{} in the conventional way, \techdc{} can usually also be applied which in consequence increases the time complexity but requires less space.

\techdp{} being broadly applicable on scheduling problems, it is often considered as not practical due to its exponential space complexity. As reported in Section \ref{sec:his}, the most commonly considered algorithms for solving scheduling problems in practice are \textit{Branching} algorithms, especially \techbb{}. By well choosing the branching scheme and bounding procedures, \techbb{} can often be efficient on solving randomly generated instances. However, when considering the worst case complexity, it is extremely difficult to analyze a \techbb{} algorithm because the effectiveness of the lower bounding procedure depends largely on the actual instance. In other words, it seems very hard to prove the non-existence of a worst-case instance on which all lower bounding procedures are unuseful hence the \techbb{} reduces to a simple enumerative branching algorithm.

The application of \techss{} on scheduling problems is well explored by \cite{lente2013extension}. They propose an extension of this method to a general class of problems called \textit{Multiple Constraint Problems} and use the method to derive faster exponential-time algorithms for several machine scheduling problems including $P||Cmax$, $P|d_i|w_iU_i$ and $F2||C_{max}^k$ with a worst-case complexity $\ostar{m^{\frac{n}{2}}}$ where $m$ is the number of machines. It is important to notice that the scheduling problems that are tackled by \techss{} always imply solving an assignment problem. 
However, whether \techss{} can be applied on pure sequencing problems remains unclear.

\techie{} can also be applied to solve some scheduling problems, as initially reported by \cite{karp1982dynamic}. 

The best known worst-case complexity results are synthesized in Table \ref{tab:synthesis}. 

\begin{table}[!ht]
    \centering
    \begin{tabular}{|c|c|c|c|}\hline
    Problem & Enumeration & EEA & Reference \\\hline
$1|dec|f_{max}$ & $\ostar{n!}$ & $\ostar{2^n}$ & \cite{fomin2010exact}  \\\hline
$1|dec|\sum f_{i}$ & $\ostar{n!}$ & $\ostar{2^n}$ & \cite{fomin2010exact} \\\hline
$1|prec|\sum C_{i}$ & $\ostar{n!}$ & $\ostar{(2-10^{-10})^n}$ & \cite{cygan2014}  \\\hline
$1|prec|\sum w_iC_{i}$ & $\ostar{n!}$ & $\ostar{2^n}$ & \cite{woeginger2003exact}  \\\hline
$1|d_i|\sum w_iU_i$ & $\ostar{n!}$ & $\ostar{\sqrt{2}^n}$ & \cite{lente2013extension,lente2014exponential}  \\\hline
$1|d_i|\sum T_i$ & $\ostar{n!}$ & $\ostar{2^n}$ & \cite{woeginger2003exact}  \\\hline
$1|r_i,prec|\sum w_iC_i$ & $\ostar{n!}$ & $\ostar{3^n}$ & \cite{woeginger2003exact}  \\\hline
$IntSched$ & $\ostar{2^{n\log m}}$ & \multiline{$\ostar{1.2132^{nm}}$ \\ $\ostar{2^n}$ \\ $\ostar{2^{(m+1)\log_2n}}$}  & \cite{lente2014exponential}  \\\hline
$P|dec|f_{max}$ & $\ostar{m^nn!}$ & $\ostar{3^n}$ & \cite{lente2014exponential}  \\\hline
$P|dec|\sum f_i$ & $\ostar{m^nn!}$ & $\ostar{3^n}$ & \cite{lente2014exponential}  \\\hline
$P||C_{max}$ & $\ostar{m^n}$ & $\ostar{ {\sqrt{m}}^n(\frac{n}{2})^{m+1}}$ & \cite{lente2013extension}  \\\hline
$P2||C_{max}$ & $\ostar{2^n}$ & $\ostar{{\sqrt{2}}^n}$ & \cite{lente2013extension,lente2014exponential}  \\\hline
$P3||C_{max}$ & $\ostar{3^n}$ & $\sqrt{3}^n$  & \cite{lente2013extension}  \\\hline 
$P4||C_{max}$ & $\ostar{4^n}$ & $\ostar{(1+\sqrt{2})^n}$ & \cite{lente2014exponential}  \\\hline
$P|d_i|\sum w_iU_i$ & $\ostar{(m+1)^n}$ & $\ostar{{\sqrt{m+1}}^n(\frac{n}{2})^{m+2}}$ & \cite{lente2013extension}  \\\hline
$P2|d_i|\sum w_iU_i$ & $\ostar{3^n}$ & $\sqrt{3}^n$  & \cite{lente2013extension}  \\\hline 
$F2||C^k_{max}$ & $\ostar{2^n}$ & $\ostar{{\sqrt{2}}^n}$ & \cite{lente2013extension,lente2014exponential}  \\\hline
$J||C_{max}$ & $\ostar{n!}$ & $\ostar{P^{2n}_{max}(m+1)^n}$ & \cite{gromicho2012}  \\\hline
    \end{tabular}
    \caption{Synthesis of the best known worst-case complexities}
    \label{tab:synthesis}
\end{table}

Recent years, more and more  \textit{Parameterized Complexity} results on scheduling problems are appearing, even though among them the majority of results are negative. Without being exhaustive, \cite{bodlaender1995w} showed that $k$-processor scheduling with precedence constraints is W[2]-hard. \cite{fellows2003parametric} state that scheduling unit-length jobs with deadlines and precedence constraints is W[1]-hard when parameterized by the number of tardy jobs. The results of \cite{jansen2013bin} imply that makespan minimization on $k$ parallel machines with jobs of polynomially bounded processing times is W[1]-hard. On the positive side, \cite{mnich2015scheduling} propose FPT algorithms for some classic scheduling problems including makespan minimization,  scheduling with job-dependent cost functions and with job rejection, by identifying several crucial parameters such as the number of distinct processing times, the number of distinct weights and the maximum number of rejected jobs. They also remark that some scheduling problems can be addressed by choosing as parameter the ``number of numbers'', introduced by \cite{fellows2012parameterizing}. 
Very recently, \cite{van2017} reported new parameterized complexity results of scheduling problems with release time and deadlines \citep{van2017parameterized}, precedence constraints \citep{van2016precedence} and sequence-dependent batch setup times or routing \citep{van2016completing}.  
An up-to-date list of publications on the intersection of \textit{Parameterized Complexity} and \textit{Scheduling} can be found in the community wiki: \url{http://fpt.wikidot.com/operations-research}.

\chapter{Dynamic Programming for Flowshop Scheduling Problems}\label{ch2}

\section{Introduction}
In this chapter, we present our work on a \techdp{} algorithm that solves the \pbft{} problem. 
\techdp{} is a useful tool to tackle permutation problems and it can be applied pretty naturally to solve \textit{decomposable} permutation problems \citep{tkindt2004revisiting}, in which the optimal solution of the problem can be composed by optimal solutions of subproblems. The corresponding \techdp{} scheme is sometimes referred to as \textit{Dynamic Programming across the subsets} (see \cite{woeginger2003exact}). 
Unfortunately, the \pbft{} problem is not \textit{decomposable} since an optimal solution may need to be based on a non-optimal partial solution of a subproblem. This is why we are interested on the study of the application of \techdp{} to these problems. 
Our algorithm is easily generalized to some other flowshop problems, among which we provide dedicated complexity analysis on  \pbfmax{} and \pbsumfi{} problems. After that, we also formulate a generalized framework as \textit{Pareto Dynamic Programming}, which can be applied to other problems (e.g. $1|r_i|f$, $F2||f$, etc) verifying some specific dominance conditions over the \textit{Pareto Front}. This is notably the case of $1|r_i|f$ and $F2||f$ problems.

\section{The \pbft{} problem}
\label{sec:f3}
The problem tackled in this section and denoted by \pbft, is described as follows. Consider $n$ jobs to be scheduled, each of which must be processed first on machine one, then on machine two and finally on machine three in this order. Each machine can only have one job being processed at a time. 
We note $O_{ij}$, $i\in\{1,...,n\}$, $j\in\{1,2,3\}$, the operation of job $i$ on machine $j$, with $p_{ij}$ a non-negative integer representing the corresponding processing time.
  The objective is to find an optimal job sequence that minimizes the makespan, which is defined as the completion time of the last job on the last machine in a schedule. Even though there exists a polynomial time algorithm when the number of machines is two, the problem with $m\ge 3$  machine has been proved to be \nph{} \citep{garey1976}. When $m=2$ or $m=3$, there always exists a permutation schedule which is optimal \citep{brucker2007scheduling}. A solution is a  permutation schedule if jobs are executed in the same order on all machines.  
Extensive researches have been done on the design of heuristic and metaheuristic algorithms (see \cite{ruiz2005,framinan2004review}), 
approximation algorithms (see \cite{SMUTNICKI199866,Sviridenko2004}) and also exact methods (see \cite{ladhari2005computational} and earlier results by  \cite{ignall1965, lomnicki1965branch,brown1966some,mcmahon1967flow, ashour1970, lageweg1978general,potts1980adaptive,carlier1996,cheng1997branch}). For a global view on different approaches on this problem and its numerious variants, we refer the reader to the review of \cite{reza2005flowshop}. 

{We focus on the search for an optimal permutation schedule of the problem, with a provable worst-case running time. We provide a theoretical study which aims at establishing upper bounds on the worst-case complexity. 
Clearly, the problem can be solved trivially by enumerating all possible job permutations, which yields a $\ostar{n!}$ worst-case time algorithm. 
The current fastest exact algorithms in practice is the \textit{Branch \& Bound} algorithm provided by \cite{ladhari2005computational}, which solves randomly generated instances with up to 2000 jobs. 
However, when considering worst-case scenarios, this branching algorithm can hardly be proved to be faster than the brute-force. An algorithmic framework which solves a number of scheduling problems in time $2^{\bigo{n}}\times \|I\|^{\bigo{1}}$, with $\|I\|$ being the input length, is proposed by \cite{jansen2013}. 
A dynamic programming algorithm by \cite{gromicho2012} also exists, solving the job-shop scheduling problem in $\bigo{(p_{max})^{2n}(m+1)^n}$. 

The aim of this work is to design a moderately exponential algorithm, that is, an algorithm running in $\ostar{c^n}$ time with the constant $c$ as small as possible, improving the currently best complexity result of $2^{\bigo{n}}\times \|I\|^{\bigo{1}}$ \citep{jansen2013}.} In the next section we introduce a dynamic programming algorithm to solve the \pbft{} problem, together with its worst-case complexity analysis.

\section{A Dynamic Programming for the \pbft{} problem}
\label{sec:dp}

Given a job permutation $\pi$, let $C^M_j(\pi)$ be the completion time of $\pi$ on machine $j$, $j\in\{1,2,3\}$. We omit $\pi$ whenever it does not create ambiguity. We also define a binary operator ``$.$'' which concatenates two permutations of jobs.

The proposed \dpa{} algorithm is based on Theorem~\ref{th0}.

\begin{theorem}
\label{th0}
Let $S$ be a set of $n$ jobs to be scheduled, and $S'\subset S$ a subset of $S$. Given $\pi$ and $\pi'$ as two permutations of jobs from $S'$, $\sigma$ as a permutation of jobs from $S\smallsetminus S'$, then we have:
\[
If \ 
\begin{cases}
C^M_2(\pi)\leq C^M_2(\pi') \\
C^M_{3}(\pi)\leq C^M_{3}(\pi') 
\end{cases}
 then \ 
\begin{cases}
C^M_2(\pi . \sigma)\leq C^M_2(\pi' . \sigma) \\
C^M_{3}(\pi . \sigma)\leq C^M_{3}(\pi' . \sigma) 
\end{cases}
\]
that is, the partial permutation $\pi$ dominates $\pi'$. Every schedule starting by $\pi'$ is dominated by a schedule starting by $\pi$.
\end{theorem}

\begin{proof}
The proof of Theorem \ref{th0} is straightforward and is illustrated in Figure \ref{fig:condom}. First, note that $C^M_1(\pi)= C^M_1(\pi')$ always holds since $\pi$ and $\pi'$ contain the same jobs and there is no idle time between jobs on machine 1. On the second and third machines, $\pi$ completes earlier (or equal) than $\pi'$ so $\sigma$ is always preferred to be put after $\pi$ instead of $\pi'$ in order to minimize the makespan $C^M_{max}$. As shown in Figure \ref{fig:condom}, the schedule containing $\pi$ is at least as good as the schedule containing  $\pi'$.
\end{proof}

\begin{figure}[!ht]
  \centering
    \includegraphics[width=.6\columnwidth]{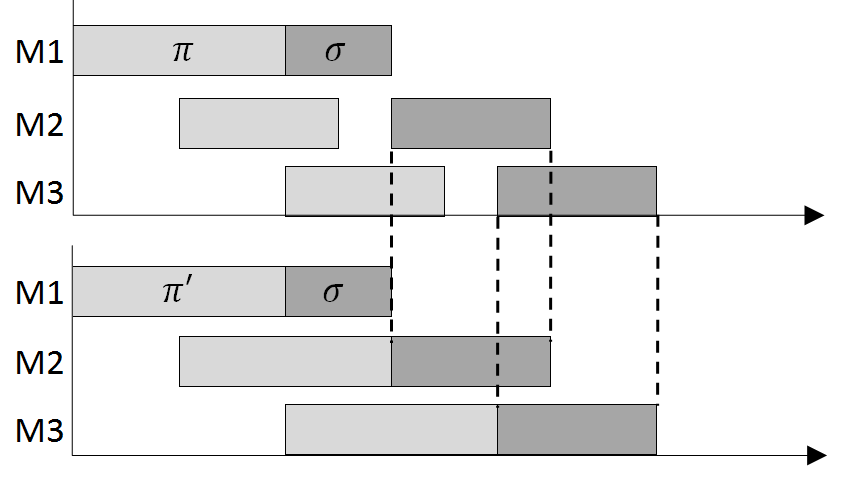}
    \caption{Illustration of Theorem \ref{th0}}
  \label{fig:condom}
\end{figure}

The idea of \dpa{} is the following. When trying to construct an optimal solution from partial solutions, only the partial solutions that are not dominated need to be considered. This can be seen by switching to the criteria space $\langle C^M_2, C^M_3\rangle $ for representation.
For a fixed jobset, we plot all its partial permutations as points in the criteria space $\langle C^M_2,C^M_3\rangle $. Then the non-dominated permutations form the \pf{} of points (Figure~\ref{fig:pf}).

\begin{figure}[!ht]
  \centering
    \includegraphics[scale=0.6]{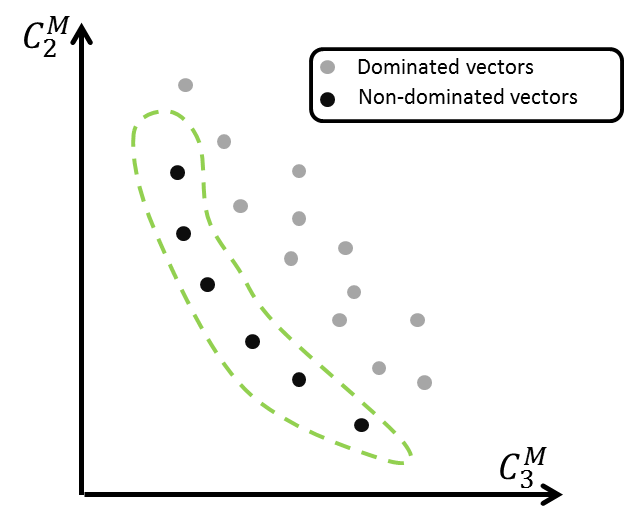}
  \caption{Partial solutions and \pf{} of a given jobset in criteria space}
  \begin{flushleft}
  \end{flushleft}
  \label{fig:pf}
\end{figure}

Now we formulate the algorithm, starting with some definitions.
\begin{definition}
Given a set of permutations of a jobset $S$, \pp{} define the subset of permutations whose associated criteria vector $\langle C^M_2,C^M_{3}\rangle $ is not dominated by the criteria vector of any other permutation from this set. 
In case that several dominating permutations have the same criteria vector, only one of them is retained, arbitrarily. 
Let $MinPerm$ be a function which takes a set of job permutations as input and returns its \pp{} set. $OptPerm(S)$ is the \pp{} of all jobs from the jobset $S$ and $OptPerm_l(S)$, $l\in \{1,...,|OptPerm(S)|\}$ is the $l$-th permutation (the numeration is arbitrary).
\end{definition}

\begin{definition}~\label{def:cp}
For a given job schedule in which all operations are executed as early as possible, its \textit{Critical Path} (see for instance \cite{brucker2007scheduling,kelley1959critical})\footnote{Instead of defining \textit{Critical Path} in the context of digraph, here we adapt the definition to make it more intuitive} $CP^j$, $j\in\{1,2,3\}$, is defined as the sequence of operations on the first $j$ machines, that are executed consecutively without any idle time and the total processing time of these operations determines $C^M_j$. More specifically, for any two consecutive operations in $CP^j$, denoted by $O_{i'j'}\to O_{i''j''}$, $i',i''\in\{1,...,n\}$, $j',j''\in\{1,...,j\}$, we must have $j'=j''$ or otherwise $i'=i''$ and $j''=j'+1$. In the latter case, job $i'$ is defined as a \textit{critical job}. 
For remaining jobs, let $CP^j_q$, $q\in\{1,...,j\}$, be the set of jobs whose operations on machine $q$ appear in $CP^j$, excluding \textit{critical jobs}. 
\end{definition}

As an example, in Figure~\ref{fig:cp_ex}, $CP^2$ consists of the sequence $(O_{11},O_{21},O_{31},O_{32},O_{42})$, with job $3$ as the critical job. 
$CP^2_1=\{1,2\}$ and $CP^2_2=\{4\}$.
It should be noticed that for a given $j$, $CP^j$ may not be unique, in that case we choose it arbitrarily.

\begin{figure}[!ht]
  \centering
    \includegraphics[width=.6\columnwidth]{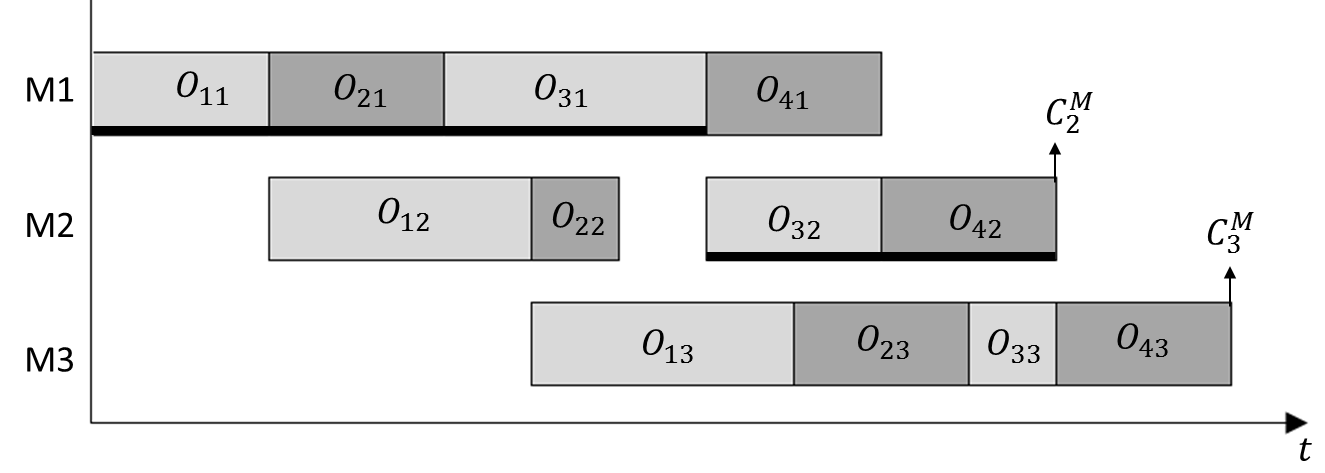}
    \caption{Example of \textit{critical paths}}
  \label{fig:cp_ex}
\end{figure}

\begin{lemma}
\label{le1}
For a jobset $S'$ with $|S'|=t$, the number of its \pp{} $|OptPerm(S')|$ is at most $\ostar{2^t}$.
\end{lemma}

\begin{proof} 
Each permutation in $OptPerm(S')$ has a corresponding criteria vector $\langle C^M_2,C^M_3\rangle $. The number of \pp{} corresponds to the number of non-dominated vectors, which is bounded by the number of possible values of $C^M_2$. 
The next observation is that the $C^M_2$ value of a given permutation is related to the critical path of job operations on the first two machines (Definition~\ref{def:cp}). Hence the problem reduces to count the number of possible critical paths. 
To do this, we first decide the critical job, for which we have $t$ choices. Each of the remaining jobs has to belong to either $CP^2_1$ or $CP^2_2$. This yields $2^{t-1}$ choices. Hence $|OptPerm(S')| \leq {t2^{t-1}} = \ostar{2^t}$.
\end{proof}

We now state in Lemma~\ref{le2} another upper bound on $|OptPerm(S')|$.
\begin{lemma}
\label{le2}
If there exists a constant $M$ such that $\forall i,j$, $p_{ij}\leq M$, then for a given jobset $S'$ of $t$ jobs, the number of non-dominated criteria vectors $\langle C^M_2,C^M_{3}\rangle $, in \pp{}, is at most ${(t+1)M}$ = $\ostar{M}$.
\end{lemma}
\begin{proof}
Lemma \ref{le2} is based on the consideration that the maximum value of $C^M_2$ will not exceed $(t+1)M$. Since $C^M_2$ is an integer, the number of possible values of $C^M_2$ is also bounded by $(t+1)M$.
\end{proof}

The recurrence function of our dynamic programming algorithm is stated in Theorem~\ref{th1}.
\begin{theorem}
\label{th1}
$OptPerm(S)$ can be computed by \dpa{} as follows:

\begin{align*}
OptPerm(S)&= MinPerm\big\lbrace OptPerm_l(S\smallsetminus \{k\}).\{k\} \\
&:k\in S, l\in \{1,...,|OptPerm(S\smallsetminus \{k\})|\} \big\rbrace
\end{align*}
\end{theorem}

Theorem \ref{th1} is the recursive formula used by the \dpa{} algorithm. It is directly based on Theorem \ref{th0}. When constructing an optimal solution starting with partial permutations, only \pp{} are kept. Consider that we already have $OptPerm(S\smallsetminus \{k\})$, $\forall k\in S$. In order to construct $OptPerm(S)$, we append job $k$ after each permutation from $OptPerm(S\smallsetminus \{k\})$ to get a new set of permutations for jobset $S$. Over all of these newly constructed permutations, we apply the operation $MinPerm$ to remove dominated ones and return $OptPerm(S)$. A more intuitive representation is given in Figure~\ref{fig:recur} where $|S|=t$.

\begin{figure}[!ht]
  \centering
    \includegraphics[width=.7\columnwidth]{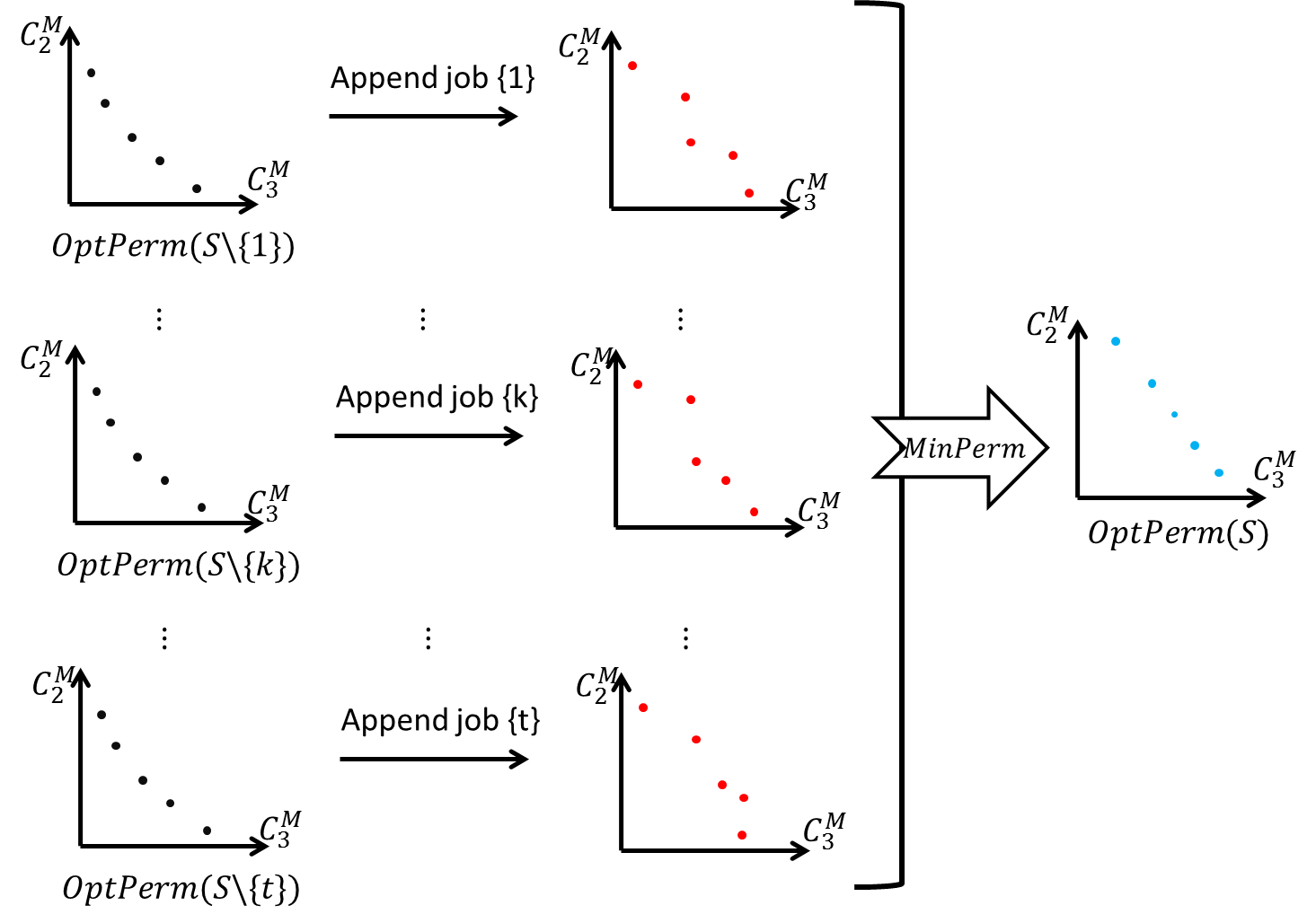}
    \caption{Recursive procedure of the \dpa{} algorithm.}
  \label{fig:recur}
\end{figure}

\subsection{Complexity Analysis}
{Consider implementing the algorithm in a recursive way according to Theorem~\ref{th1}. We keep in memory already computed $OptPerm(S)$, for all $S$, for fast access later. As initialization, when $S=\{j\}$, we set $OptPerm(S)=\{j\}$. From that point, every $OptPerm(S)$ for $|S|>1$ can be computed by $MinPerm$ from $OptPerm(S\smallsetminus\{k\})$, $\forall k\in S$.} 
The function $MinPerm$ uses an existing algorithm \citep{kung1975finding} for finding, {given a set of $N$ vectors $\langle C^M_2,C^M_{3}\rangle $}, non-dominated vectors, with a complexity of $\bigo{N\log N}$. 
{According to the proof of Lemma \ref{le1}, $N$ is bounded by ${{t(t-1)2^{t-2}}=\bigo{t^22^t}}$ where the extra factor $t$ is the number of possible values of $k$. Therefore, computing $OptPerm(S)$ from solved subproblems yields a complexity of $${\bigo{t^22^t\log(t^22^t)}=\bigo{t^32^t}=\ostar{2^t}}$$} 

The algorithm traverses across all problems defined by all subsets of jobs. Thus, the overall time complexity for calculating $OptPerm(S)$ is 
\begin{align*}
{\sum_{t=1}^n {n \choose t}\bigo{t^32^t}}& = n^3\sum_{t=1}^n {n \choose t}\bigo{2^t}1^{n-t} \\
&=\bigo{n^3(2+1)^n} \tag*{(Binomial theorem)}\\
&=  \ostar{3^n} 
\end{align*}

Alternatively, based on Lemma \ref{le2}, the time complexity can also be expressed as

$$ \sum_{t=1}^n {n \choose t}\ostar{M}=\ostar{M2^n} $$

The minimum value of $C^M_{max}$ can be retrieved from $OptPerm(S)$ with an additive cost of $\ostar{2^n}$ (or $\ostar{M}$) time, which does not change the established complexity. The space complexity is also $\ostar{3^n}$ (or $\ostar{M2^n}$) considering the storage of all necessary \pp{}. Therefore, we proved that at the cost of exponential memory usage the \pbft{} problem can be solved in $\ostar{3^n}$ time. This improves the complexity result  $2^{\bigo{n}}\times \|I\|^{\bigo{1}}$ of \cite{jansen2013}.

\subsection{Computational results}
The practical efficiency of exact exponential algorithms is also an interesting subject that we want to explore in this thesis. These algorithms may not seem efficient in practice in general, because specific treatment is performed dedicated to obtain a theoretic guarantee on worst-case instances. However, despite this possible behavior of the proposed DP algorithm, we wanted to evaluate it in practice on randomly generated instances.
Unfortunately, 
{despite the superiority of our algorithm in worst-case scenarios with respect to existing exact algorithms, preliminary experimentation \citep{morin2014} shows that the algorithm only managed to solve random instances with up to 20 jobs due to its space complexity. The results are reported in Table \ref{tab:dp}. The processing time of operations are generated uniformly from the range $[1,100]$. For each problem size, 20 instances are generated. Small instances with 5 and 10 jobs are solved instantly. Instances with 15 and 20 jobs are also solved  within an acceptable time, but the exponential growth of solution time over the size of instances can already be observed. On instances with 25 jobs, the program encountered ``out of memory'' problems on certain instances, knowing that the tests have been performed on a computer having 4G RAM and a 1.70GHz CPU. 

\begin{table}
\centering
    \begin{tabular}{|c|c|c|c|c|}\hline
N				&5 		&10 	&	15 &	20   \\ \hline
Min time (s)	&0		&0		&0.8		&42.4      \\ \hline
Max time (s)	&0		&0		&2.8		&82.5      \\ \hline
Mean time (s)	&0		&0		&1.2		&54.1      \\ \hline
    \end{tabular}
    \caption{Preliminary results of dynamic programming}\label{tab:dp}
\end{table}

This is far less than that can be solved by the \textit{Branch \& Bound} algorithm of \cite{ladhari2005computational}. However, it is of theoretical interest to have an algorithm with a proved running time upper bound faster than the brute force approach.}

\section{Generalizations to other flowshop problems}
\label{sec:3_new}
The \techdp{} algorithm can be generalized to other three-machine flowshop problems in a similar manner. In this section we present such generalizations on the \pbfmax{} and the \pbsumfi{} problems, with a dedicated complexity analysis provided.

\subsection{The \pbfmax{} problem}
The \dpa{} algorithm can be applied to a more general problem referred to as \pbfmax. This problem is similar to the \pbft{} problem but with a more general objective function $f_{max}$ to be minimized. For job $i\in\{1,\ldots,n\}$, let $C_i$ be the completion time of job $i$ on machine $3$ and let $f_i(C_i)$ be the cost associated to job $i$. We consider the $f_i$'s as non-decreasing functions. The objective is defined as 

$$ f_{max} = max\{f_i(C_i)|i=1,\ldots, n\} $$

If $f_i = f_j~\forall i\neq j$, then the \dpa{} algorithm can be applied without adaptation, and therefore with the same $\ostar{3^n}$ worst-case complexity. The \pbft{} problem can be considered as a particular case of this setting where $f(C_i)=C_i$.

If the definition of the function $f_i$ depends on $i$, for instance $f_i=w_iC_i$ with $w_i$ the weight of job $i$, then the dominance status of a given permutation is not only dependent on $\langle C^M_2, C^M_3\rangle $, but also on the corresponding $f_{max}$ value of the permutation. In other words, given two different permutations $\pi$ and $\pi'$ for a same subset of jobs, the condition $C^M_2(\pi)<C^M_2(\pi')$ and $C^M_3(\pi)<C^M_3(\pi')$, is not sufficient to discard $\pi'$ since it is possible that $f_{max}(\pi') < f_{max}(\pi)$ which means that $\pi'$ may lead to an optimal solution. Therefore, we add $f_{max}$ as a third criterion into the criteria vector. The \dpa{} acts in a similar way as the one for \pbft{} but on a three-dimension criteria space $\langle C^M_2, C^M_3, f_{max}\rangle $. 

\begin{theorem}
\label{th2}
Let $S$ be a set of $n$ jobs to be scheduled, and $S'\subset S$ a subset of $S$. Given $\pi$ and $\pi'$ as two permutations of jobs from $S'$, $\sigma$ as a permutation of jobs from $S\smallsetminus S'$, then we have:
\[
If
\begin{cases}
C^M_2(\pi)\leq C^M_2(\pi') \\
C^M_{3}(\pi)\leq C^M_{3}(\pi') \\
f_{max}(\pi)\leq f_{max}(\pi')
\end{cases}
 then
\begin{cases}
C^M_2(\pi . \sigma)\leq C^M_2(\pi' . \sigma) \\
C^M_{3}(\pi . \sigma)\leq C^M_{3}(\pi' . \sigma) \\
f_{max}(\pi. \sigma)\leq f_{max}(\pi'. \sigma)
\end{cases}
\]
that is, the partial permutation $\pi$ dominates $\pi'$. Every schedule starting by $\pi'$ is dominated by a schedule starting by $\pi$.
\end{theorem}
\begin{proof}
The proof of Theorem \ref{th2} is straightforward.
\end{proof}

The same running time analysis can be performed as before, where the critical part is to count the number of possible non-dominated vectors. This number is bounded by the number of different $\langle C^M_2,C^M_3\rangle $ values.
 Similarly to what has been done in Lemma~\ref{le1} to establish an upper bound on the number of values of $C^M_2$, we can notice that there also exists a critical path (but through three machines, see Figure \ref{fig:cp_ex_m3}) which leads to each value of $C^M_3$ and therefore the number of possible values of $C^M_3$ is upper bounded by $\ostar{3^n}$.
If we consider the number of possible $\langle C^M_2,C^M_3\rangle $ vectors as the number of combinations of $C^M_2$ values and $C^M_3$ values, then we have $2^n\times 3^n=6^n$ non-dominated criteria vectors $\langle C^M_2,C^M_3\rangle $. However it can be proved that this upper bound is far from being tight, by exploiting links between $CP^2$ and $CP^3$.

\begin{figure}[!ht]
  \centering
    \includegraphics[width=.6\columnwidth]{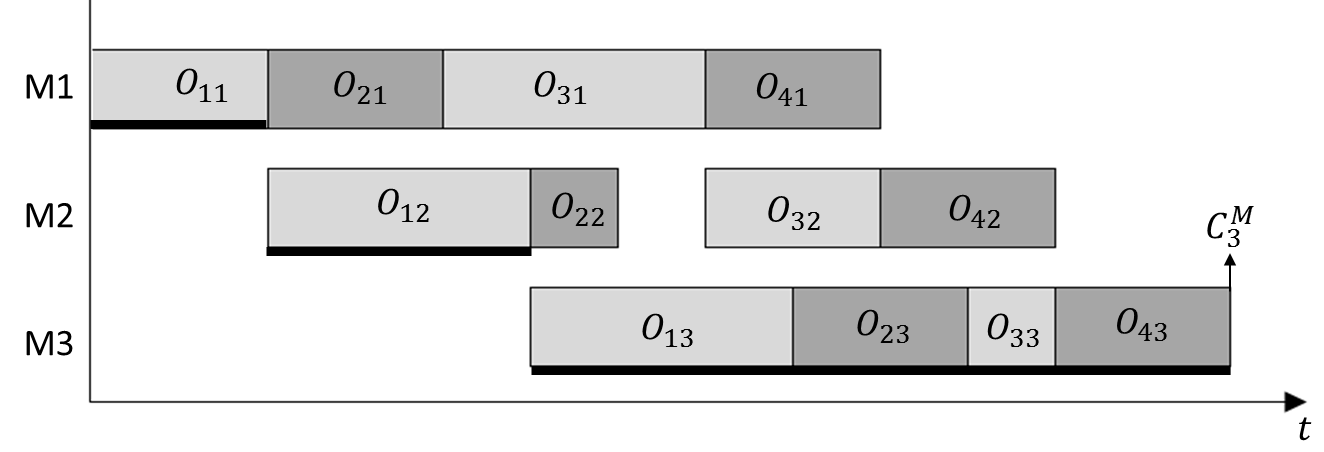}
    \caption{Example of a critical path through 3 machines}
  \label{fig:cp_ex_m3}
\end{figure}

\begin{lemma}~\label{le_cp1}
On machine 1, it is sufficient to consider critical paths such that $CP^3_1 \subseteq CP^2_1$.
\end{lemma}
\begin{proof}
Firstly it is obvious that on machine 1, either $CP^3_1 \subseteq CP^2_1$ or $CP^2_1 \subseteq CP^3_1$, since either they are empty sets or they must contain the first several operations on machine 1. Consider the situation where we have $CP^2_1 \subsetneq CP^3_1$, {$CP^3_2\subsetneq CP^2_2$ holds since $CP^2_2$ contains all jobs of $N\smallsetminus CP^2_1$ except the critical job. Now by setting $CP^3_1=CP^2_1$ and adjust $CP^3_2$ accordingly, $CP^3$ remains a critical path and $CP^3_1\subseteq CP^2_1$. 
An example is illustrated in Figure~\ref{fig:excp1}}.
\end{proof}

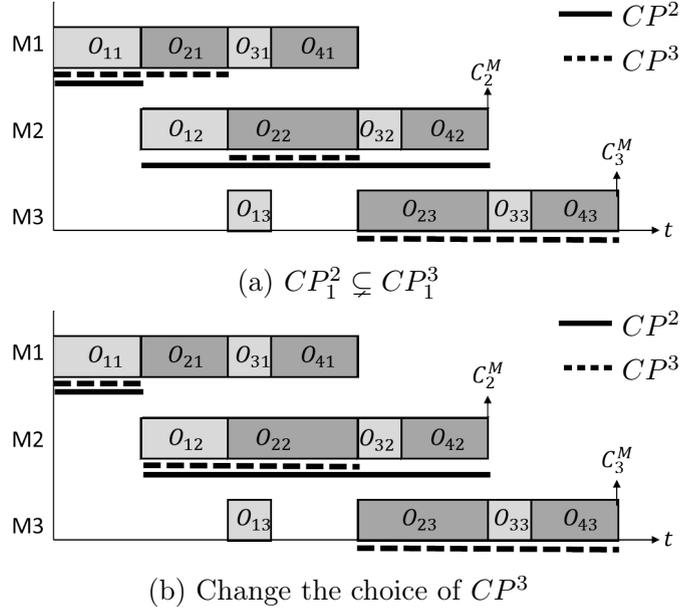
\begin{figure}[!ht]
  \centering
  \begin{tikzpicture}
        \node[anchor=north east] at (0,0) {\includegraphics[width=.6\columnwidth]{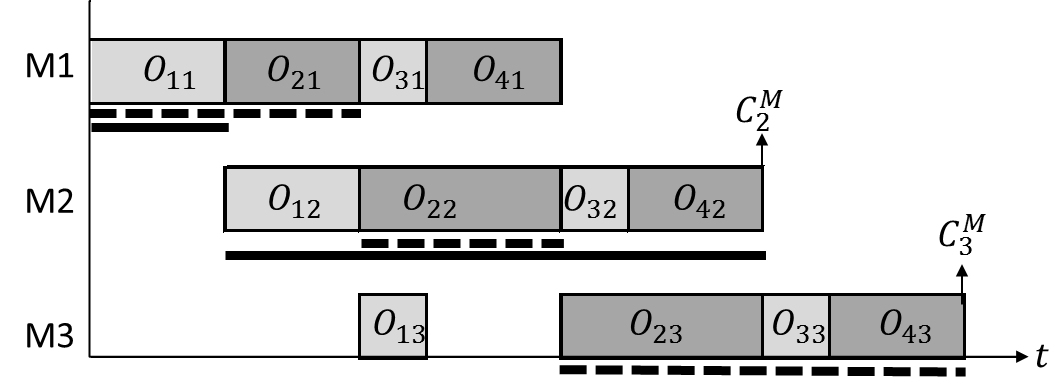}};
        \node[anchor=north east] (cp2) at (0,0) {$CP^2$};
        \node[below=0mm of cp2.south] (cp3) {$CP^3$};
        \draw[line width=2pt] (cp2.west) -- +(-0.7,0);
        \draw[line width=2pt, densely dashed] (cp3.west) -- +(-0.7,0);
  \end{tikzpicture}
  \\(a) $CP^2_1 \subsetneq CP^3_1$\\
  \begin{tikzpicture}
        \node[anchor=north east] at (0,0) {\includegraphics[width=.6\columnwidth]{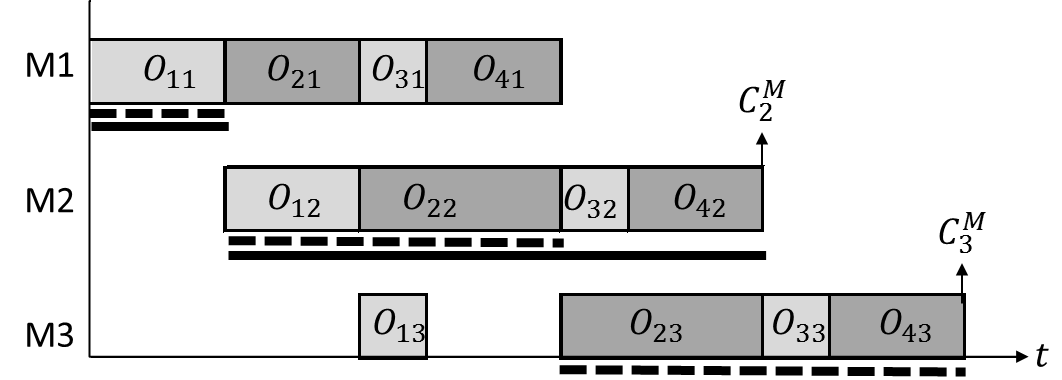}};
        \node[anchor=north east] (cp2) at (0,0) {$CP^2$};
        \node[below=0mm of cp2.south] (cp3) {$CP^3$};
        \draw[line width=2pt] (cp2.west) -- +(-0.7,0);
        \draw[line width=2pt, densely dashed] (cp3.west) -- +(-0.7,0);
  \end{tikzpicture}
  \\(b) Change the choice of $CP^3$
 \caption[An example for Lemma~\ref{le_cp1}]{An example for Lemma~\ref{le_cp1}. $CP^2$ (resp. $CP^3$) is indicated by a solid (resp. dashed) line.}
  \label{fig:excp1}
\end{figure}

\begin{lemma}~\label{le_cp2}
On machine 2, only two cases need to be considered: 
\begin{enumerate}
\item $CP^3_1 = CP^2_1$ and $CP^3_2 \subseteq CP^2_2$.
\item $CP^3_1 \subsetneq CP^2_1$ and $CP^3_2 \cap CP^2_2 = \emptyset$.
\end{enumerate}
\end{lemma}
\begin{proof}
Both cases are easy to verify with a similar consideration as in Lemma~\ref{le_cp1}. The first case is obvious, since $CP^3_1 = CP^2_1$ and $CP^2_2$ contains all jobs of $N\smallsetminus CP^2_1$ except the critical job, which is also the critical job of $CP^3$, so $CP^3_2 \subseteq CP^2_2$ holds. 
In the second case $CP^3_1 \subsetneq CP^2_1$, if $CP^3_2 \cap CP^2_2 \neq \emptyset$, let job $o\in CP^3_2 \cap CP^2_2$. This means that both critical paths $CP^2$ and $CP^3$ contain the operation $O_{o2}$, and the operation sequences preceding $O_{o2}$ are different in $CP^2$ and $CP^3$. Then similarly to Lemma~\ref{le_cp1}, $CP^2$ or $CP^3$ can then be rearranged by adopting the same operation sequences preceding $O_{o2}$ to reduce to the first case of Lemma~\ref{le_cp2}. See Figure~\ref{fig:excp2} and Figure~\ref{fig:excp3} for an example. Note that we are only interested in the values of $C^M_2$ and $C^M_3$, the actual critical paths are not our concern. 
\end{proof}
\begin{figure}[!ht]
  \centering
    \begin{tikzpicture}
        \node[anchor=north east] at (0,0) {\includegraphics[width=.6\columnwidth]{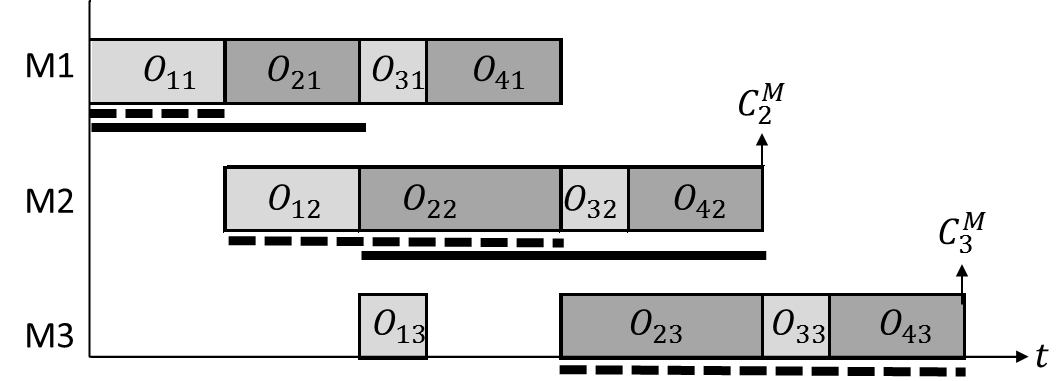}};
        \node[anchor=north east] (cp2) at (0,0) {$CP^2$};
        \node[below=0mm of cp2.south] (cp3) {$CP^3$};
        \draw[line width=2pt] (cp2.west) -- +(-0.7,0);
        \draw[line width=2pt, densely dashed] (cp3.west) -- +(-0.7,0);
  \end{tikzpicture}
    \\(a) $CP^3_1 \subsetneq CP^2_1$ and $CP^3_2 \cap CP^2_2 \neq \emptyset$
    
    \begin{tikzpicture}
        \node[anchor=north east] at (0,0) {\includegraphics[width=.6\columnwidth]{img/cp23_2.png}};
        \node[anchor=north east] (cp2) at (0,0) {$CP^2$};
        \node[below=0mm of cp2.south] (cp3) {$CP^3$};
        \draw[line width=2pt] (cp2.west) -- +(-0.7,0);
        \draw[line width=2pt, densely dashed] (cp3.west) -- +(-0.7,0);
  \end{tikzpicture}
    \\(b) Change the choice of $CP^2$
    \caption[An example for Lemma~\ref{le_cp2}]{An example for Lemma~\ref{le_cp2}. $CP^2$ (resp. $CP^3$) is indicated by a solid (resp. dashed) line.}
  \label{fig:excp2}
\end{figure}

\begin{lemma}
For a jobset $S'$ with $|S'|=t$, the number of possible pairs $\langle C^M_2,C^M_3\rangle $ is bounded by $\ostar{4^t}$.
\end{lemma}
\begin{proof}
In order to count this number, we first need to choose the critical jobs for $CP^2$ and $CP^3$, as similar to the proof of Lemma~\ref{le_cp1}. {This yields a factor of $t(t-1)=\bigo{t^2}$ such jobs.}

According to Lemma~\ref{le_cp2}, we then choose $i$ jobs that form $CP^2_1$ ($CP^2_2$ is then implicitly decided). Then, for $CP^3$:
\begin{itemize}
\item Consider Lemma~\ref{le_cp2} case 1: the $i$ jobs are also in $CP^3_1$, the remaining $(t-i)$ jobs can be put in $CP^3_2$ or $CP^3_3$, which yields $2^{t-i}$ choices.
\item Consider Lemma~\ref{le_cp2} case 2 (see Figure~\ref{fig:excp3}): these $i$ jobs can be put either in $CP^3_1$, $CP^3_2$ or $CP^3_3$. The remaining $t-i$ jobs can only be put in $CP^3_3$, since $CP^3_2 \cap CP^2_2 = \emptyset$. In total this yields $3^i$ choices.
\end{itemize}

{
\begin{align*}
t(t-1)\sum_{i=1}^{t}{t \choose i}\bigo{3^i+2^{t-i}}&=\bigo{t^2(4^t+3^t)}\\
&=\bigo{t^24^t}\\
&=\ostar{4^t}
\end{align*}}
\end{proof}
\begin{figure}[!ht]
  \centering
    \begin{tikzpicture}
        \node[anchor=north east] at (0,0) {\includegraphics[width=.6\columnwidth]{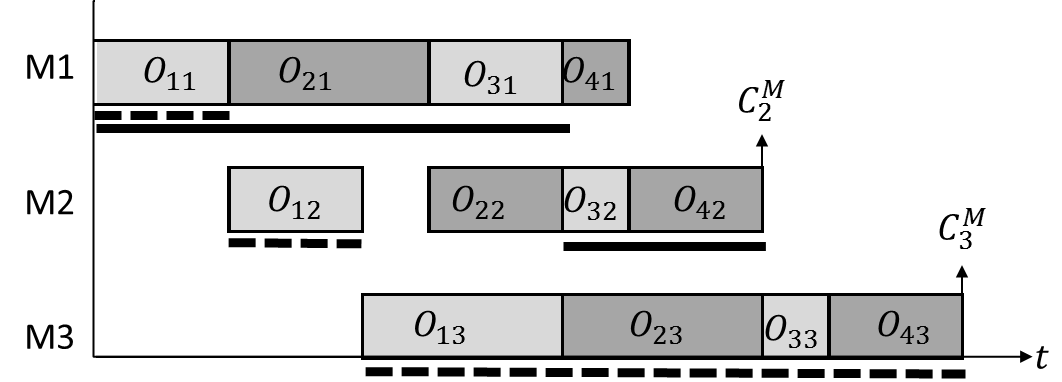}};
        \node[anchor=north east] (cp2) at (0,0) {$CP^2$};
        \node[below=0mm of cp2.south] (cp3) {$CP^3$};
        \draw[line width=2pt] (cp2.west) -- +(-0.7,0);
        \draw[line width=2pt, densely dashed] (cp3.west) -- +(-0.7,0);
  \end{tikzpicture}
    \caption[An example for Lemma~\ref{le_cp2}, case 2]{An example for Lemma~\ref{le_cp2}, case 2. $CP^2$ (resp. $CP^3$) is indicated by a solid (resp. dashed) line.}
  \label{fig:excp3}
\end{figure}
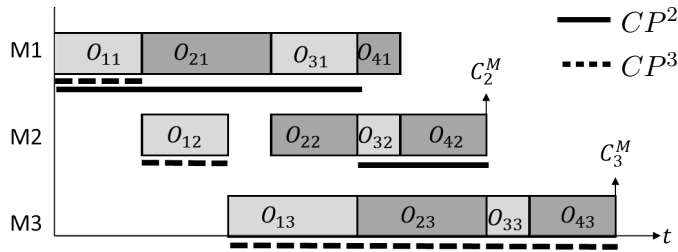

Finally, with a same analysis as for the \pbft{} problem, the complexity of our algorithm for the \pbfmax{} problem is:
$$
{\sum_{t=1}^n {n \choose t}\bigo{t\cdot t^24^t \log(t\cdot t^24^t)}=\bigo{t^44^n}=\ostar{5^n}}
$$

As in Lemma~\ref{le2}, we can also adopt a parameter $M$ such that $ p_{ij}\leq M$ into the analysis. In this case the number of possible pairs $\langle C^M_2,C^M_3\rangle $ is bounded by $(n+1)M\cdot(n+2)M=\ostar{M^2}$, and the final complexity is
$$
\sum_{t=1}^n {n \choose t}M^2=\ostar{M^22^n}
$$

 \subsection{The \pbsumfi{} problem}
Another problem that our \dpa{} can be applied to is \pbsumfi, in which the objective is to minimize the $\sum f_i(C_i)$ problem. 
When we consider $f_i$ as a non-decreasing function, this problem can be solved by \dpa{} in a similar way to the \pbfmax{} problem. 
Hence, based on a similar dominance condition between partial permutations, the criteria vector is defined as ${\langle C^M_2,C^M_3,F\rangle }$ with $F=\sum f_i(C_i)$ for a given sequence. 

\begin{theorem}
\label{th3}
Let $S$ be a set of $n$ jobs to be scheduled, and $S'\subset S$ a subset of $S$. Given $\pi$ and $\pi'$ as two permutations of jobs from $S'$, $\sigma$ as a permutation of jobs from $S\smallsetminus S'$, then we have:
\[
If
\begin{cases}
C^M_2(\pi)\leq C^M_2(\pi') \\
C^M_{3}(\pi)\leq C^M_{3}(\pi') \\
F(\pi)\leq F(\pi')
\end{cases}
 then
\begin{cases}
C^M_2(\pi . \sigma)\leq C^M_2(\pi' . \sigma) \\
C^M_{3}(\pi . \sigma)\leq C^M_{3}(\pi' . \sigma) \\
F(\pi. \sigma)\leq F(\pi'. \sigma)
\end{cases}
\]
that is, the partial permutation $\pi$ dominates $\pi'$, and there exists a permutation schedule starting by $\pi$ at least as good for the $\sum f_i$ as another starting by $\pi'$.
\end{theorem}
\begin{proof}
The proof of Theorem \ref{th3} is straightforward.
\end{proof}

Therefore with the same analysis as for the \pbfmax{} problem, the total complexity is also $\ostar{5^n}$ (or $\ostar{M^22^n}$) in time and space.

\section{A general framework and its applications}
From the applications of \techdp{} on the above flowshop problems, it can be observed that the algorithm is not limited to flowshop problems, but can be applied on any other problems  verifying similar dominance conditions over \pf. For this reason, we further formulate the idea as a general framework, then apply it onto several other scheduling problems.

We call \textit{conventional DP} the classic algorithmic pattern of \techdp{} when applied to sequencing problems including TSP in 1960's \citep{bellman1962dynamic, held1962adynamic}. This pattern is also referred to as \textit{Dynamic Programming across the subsets} according to \cite{woeginger2003exact}. 
In fact the generalization of the conventional  DP has been studied in the literature  \citep{carraway1988theory,mitten1974}. Briefly speaking, the idea of generalized DP (GDP) is still to construct optimal solutions by optimal solutions of subproblems.
It eliminates partial solutions that are not promising to be extended to an optimal global solution. Nevertheless, instead of using directly the global cost function to define the local dominance condition, as done in the conventional DP, GDP defines the dominance condition by problem-specific local preference relations. Preference is given to the decision(s) that can lead to a better global solution by considering all possible subsequent decisions.
GDP has been widely discussed on multi-criteria optimization problems or problems with multi-attribute objective functions (see for instance \cite{carraway1990generalized,villarreal1981interactive,evans1982multiobjective,bard1988techniques}).  
The local preference rule in GDP is often based on the relation of criteria or attribute vectors associated to each partial solution and this vector is 
often explicit for the aforementioned problems. 


Here we propose another formulation of generalized DP (named PDP, Pareto Dynamic Programming)  which shares the same high-level principle as GDP but its formulation has a special focus on single criterion optimisation problems with implicit dominance vectors and on minimizing the computational complexity of the algorithm in worst-case scenarios. The general idea is to identify a vector of criteria which affect the objective function value 
so that partial decision sequences whose criteria vector is not in the corresponding Pareto sets can be removed during the process of DP. 
Generating Pareto sets in DP has already been mentioned in the literature \citep{henig1983vector,rosenman1983pareto} and it is sometimes considered as impractical due to the size of Pareto sets stored during the process. However, in the context of the design of \textit{Exact Exponential Algorithms}  for \nph{} problems, PDP may provide a good performance guarantee in worst-case scenarios.


We first formulate PDP and discuss the complexity analysis of PDP in a general way. Then we apply it on some single machine scheduling problems with dynamic arrivals, for which the instantiation of PDP yields the best time complexity so far. 
Finally, we provide perspectives for its applications.

\subsection{Pareto Dynamic Programming}
\label{sec:formulation}
Now let us formulate PDP in a general way, considering a general optimization problem where we need to find the optimal order of decisions to make. 
Let $\vec{x}\in \mathbb{R}^p$ be a criteria vector referring to the current state of solution, i.e. the consequence of decisions that have already been made. The state is initiated as a all-zero vector $\vec{x_0}$ when no decision has been made. 
Let $\vec{y}\in\mathbb{R}^q$ be the vector representing the next decision to make. Making the next decision can be represented by a binary operator of concatenation $\odot$ defined as follows:
$$\vec{x}\odot\vec{y}=\vec{x'},\ \vec{x'}\in \mathbb{R}^p$$
In addition, for $A\subset\mathbb{R}^p, \vec{y}\in\mathbb{R}^q$,
$$A\odot\vec{y}=\{\vec{x}\odot\vec{y} | \vec{x}\in A  \}$$
That is, $A\odot\vec{y}$ is the set of solution states obtained by concatenating the decision $\vec{y}$ at the end of each decision state from $A$.


For $\vec{x},\vec{x'}\in \mathbb{R}^p$, the local preference between $\vec{x}$ and $\vec{x'}$ is defined as $$\vec{x}\preccurlyeq\vec{x'}\Rightarrow(\vec{x}\odot\vec{y})\leq(\vec{x'}\odot\vec{y}), \forall \vec{y}\in \mathbb{R}^q$$
The relation $\preccurlyeq$ is read as ``is preferred with respect to'', considering a minimization problem. Ties are broken arbitrarily.

Define $S=\{\vec{y_1},\ldots,\vec{y_n}\} \subset \mathbb{R}^q $ the candidate decisions at the input for a problem of size $n$. $Opt(S)$ is the set of final solution states containing the optimal ones, PDP can then be expressed by the following equation:
\begin{align}\label{eq:dpps0}
Opt(S)=\minvec\left(\underset{\vec{y}\in S}{\cup}Opt(S\smallsetminus\vec{y})\odot\vec{y}\right)
\end{align}
with $\minvec(\cdot)$ returning all non-dominated vectors from the input ones. 

The correctness of equation \ref{eq:dpps0} is obvious since the quality of a solution state can be judged directly on its associate state vector $\vec{x}\in \mathbb{R}^p$ by definition. The final optimal solution can be retrieved easily from $Opt(S)$ by evaluating the global objective function. 

\vspace{1cm}
The definition of the criteria vector $\vec{x}\in \mathbb{R}^p$ is problem dependent and it plays a key role in order to apply PDP on a specific problem. 
For each decision permutation $\pi$, its associated criteria vector $\mathcal{V}(\pi)=\langle g_1(\pi),\ldots,g_p(\pi) \rangle$ should be defined in a way such that for any permutation $\pi'$ of the same set of decisions as $\pi$, $\pi\preccurlyeq \pi'$ iff $\mathcal{V}(\pi)\leq \mathcal{V}(\pi')$. By finding (or defining) appropriate criterion $g_i$, we transfer our focus from decision sequences to a specific criteria space. 
In this context, $\minperm$ is implemented to find the permutations whose associated criteria vector is minimum and $Opt(S)$ contains these permutations.

The problems that are treatable by the conventional DP (the one solving TSP, for instance) can be seen as a special case of PDP where $\mathcal{V}=\langle f\rangle$, i.e. the global criterion is the only one that needs to be considered for the local dominance. In this case $p=1$ so that the cardinality of $Opt(S)$ is always $1$.

\subsection{Complexity analysis}
In terms of the time complexity of PDP, it is easy to see that the main operation in the algorithm is the function $\minvec$. For $|S|=t$, the number of input of $\minvec$ in equation \eqref{eq:dpps0} is upper-bounded by $t\times|\optperm(S\smallsetminus \vec{y})|$. Therefore, it is important to bound the maximum cardinality of $\optperm(S)$ for a given decision set $S$. 
According to the authors' experience, the upper-bound analysis of $\optperm(S)$ is usually problem dependent and it is the place where we can exploit largely the properties of the problem. 
In section~\ref{sec:dpapp}, we show the application of PDP on several problems together with detailed problem-dependent complexity analysis.

Suppose for $|S|=t$ the cardinality of $\optperm(S)$ can be bounded by $B(t)$ which is a problem dependent function. The running time of function $\minperm$ depends on $p$, which is the number of criteria in $\mathcal{V}$. According to  \cite{kung1975finding}, $\minperm$ runs in $\bigo{N\log N}$ for $p=2,3$ and in $\bigo{N(\log N)^{p-2}}$ for $p\geq 4$, with $N$ the number of input of $\minperm$, i.e. $tB(t-1)$. Also considering the algorithm runs across all the subsets of different cardinality $t$, the total running time can then be computed as:

\begin{numcases}{T(n)=}
\sum_{t=1}^n{n \choose t}t\cdot B(t-1)=\sum_{t=1}^n{n \choose t}t\leq 2^n             & $p=1$ \label{eq:r1}\\
\sum_{t=1}^n {n \choose t} t\cdot B(t-1)\log(B(t-1)) & $p=2,3$ \label{eq:r2}\\
\sum_{t=1}^n {n \choose t} t\cdot B(t-1)(\log(B(t-1)))^{p-2} & $p\geq 4$
\end{numcases}

Notice that for the case of equation \eqref{eq:r1}, $B(t)=1$, which explains the frequent occurrence of the complexity $\bigo{n2^n}$ of DP algorithms for many problems.

\subsection{Illustration for the {$1|r_i|\sum f_i(C_i)$ and $1|r_i|\max (f_i(C_i))$} problems}\label{sec:dpapp}
Consider a single machine scheduling problem with dynamic arrivals. Let $\{1,\ldots,n\}$ be the jobs to schedule, and each job $i$ is defined by its processing time $p_i$ and a release date $r_i$. For a given schedule, i.e. a sequence of the $n$ jobs with every job starting as early as possible after its release date, $C_i$  is defined as the completion time of job $i$. Let $f_i$ be a non-decreasing function of $C_i$, and the problem (denoted to as $1|r_i|\sum f_i(C_i)$) asks to find an optimal schedule of jobs which minimizes $\sum f_i(C_i)$. This problem is \nph{} even when $f_i=C_i$. Note that the {conventional} DP cannot be applied directly since the objective function $\sum f_i(C_i)$ alone is not sufficient to decide the dominance between two partial solutions composed by same jobs.

To apply PDP on the $1|r_i|\sum f_i(C_i)$ problem, we define $p=2$ with $g_1(\pi)$ the completion time of a partial schedule $\pi$ and $g_2(\pi)=\sum_{i\in \pi}f_i(C_i)$ its associated cost function. This is sufficient since if $\mathcal{V}(\pi)\leq \mathcal{V}(\pi')$, whatever is the job sequence after, the global solution starting with $\pi$ always dominates $\pi'$. On the other hand, the vector $\vec{y}_i\in \mathbb{R}^q$ associated to decision $i$ (job $i$, in our case) is pretty straightforward:  $\vec{y}_i=\langle p_i,r_i\rangle$. That is, a job sequence $\pi=(2,1,3)$ can be seen as the consequence of making three decisions by doing $\vec{x_0}\odot\vec{y}_2\odot\vec{y}_1\odot\vec{y}_3$, with $\vec{x_0}\in \mathbb{R}^p$ the initial state vector containing only 0's. 
Moreover, $\mathcal{V}(\pi) \odot \vec{y}_i=\langle \max(g_1(\pi),r_i)+p_i, g_2(\pi)+f_i(\max(g_1(\pi),r_i)+p_i) \rangle$. Equation \ref{eq:dpps0} can now be applied directly.

We now provide the complexity analysis of PDP on the $1|r_i|\sum f_i(C_i)$ problem. 
The question to answer first is, for a given jobset $S$, $|S|=t$, what is the value of $B(t)$, i.e. the maximum cardinality of $\{\mathcal{V}(\pi) | (\pi,\pi'\text{ are permutations of }$S$)\ \wedge\ (\nexists \pi',\ 
\mathcal{V}(\pi')<\mathcal{V}(\pi))\}$. A simple observation is that $B(t)$ is upper-bounded by the number of possible values of $g_1(\pi)$, which is bounded by the number of completion times of $\pi$. In order to compute this, we introduce Lemma~\ref{le:r}.

\begin{lemma}\label{le:r}
Assume $r_i\geq 0$, $\forall i=1,..,t$. Let $\pi=(j_1,j_2,\ldots,j_t)$ be a permutation of $t$ jobs. Then, $\exists k\in\{1,\ldots,t\}$ such that $g_1(\pi)=r_k+\sum_{h=k}^tp_h$.
\end{lemma}
\begin{proof}
Starting from the end of $\pi$, find the longest sub-sequence of jobs which do not have idle time in-between. Let $k$ be the first job of this sub-sequence. If $k$ is the first job in $\pi$, it must start at time $r_k$ since all jobs start as early as possible. Otherwise if $k$ is not the first in $\pi$, we know it still starts at time $r_k$ since there is some idle time between job $j_{k-1}$ and $j_k$ by the construction of the sub-sequence. An exemple of $\pi$ is given in Figure \ref{fig:rk}.
\end{proof}

\begin{figure}[!ht]
  \centering
    \includegraphics[width=.6\columnwidth]{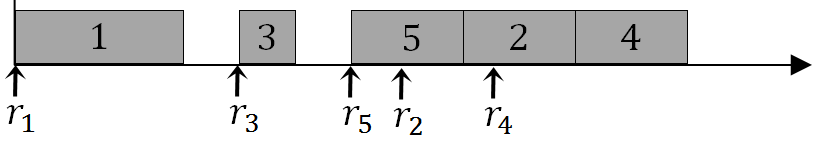}
    \caption{Example for Lemma \ref{le:r} (k=5)}
  \label{fig:rk}
\end{figure}

\begin{lemma}\label{le:rb}
$B(t)=\bigo{t2^{t-1}}=\ostar{2^{t}}$.
\end{lemma}
\begin{proof}
By Lemma\ref{le:r}, possible values of $g_1(\pi)$ can be computed as follows. First choose $k$, for which there are $t$ choices; then choose the jobs that are put after $j_k$, which yields $\bigo{2^{t-1}}$ possibilities. 
\end{proof}

\begin{theorem}\label{th:r}
Problem $1|r_i|\sum f_i(C_i)$ can be solved in $\ostar{3^n}$ time and space. 
\end{theorem}
\begin{proof}
The result can be derived directly from equation \ref{eq:r2} and Lemma\ref{le:rb}: 
$$T(n)=\sum_{t=1}^n {n \choose t} t\cdot\ostar{2^{t}}\log(\ostar{2^{t}})= \ostar{n^2\sum_{t=1}^n {n \choose t}2^t\cdot 1^{n-t})}=\ostar{3^n}$$
\end{proof}
The reasoning and the result is the same for the $1|r_i|\max (f_i(C_i))$ problem.


\FloatBarrier

\section{Auxiliary Result: complexity lower bound of \pbft{} based on ETH}

The conjecture $\mathcal{P} \neq \mathcal{NP}$ implies that \nph{} problems cannot be solved in polynomial time. However, no more precision is given on the complexity of \nph{} problems except that they are super-polynomial. People have realized that this simple binary classification of decision problems to \textit{easy} or \textit{hard} is perhaps too coarse-grained. This encourages to analyze the complexity of problems in a finer way. On this direction,  
\cite{Impagliazzo2001} proposed the \textit{Exponential Time Hypothesis} (ETH) which states that 3SAT cannot be solved in sub-exponential time ($\bigo{2^{o(n)}}$). Based on this assumption, complexity lower bounds have been deduced for many problems such as $k$-Colorability, $k$-Set Cover, \textsc{Independent Set}, \textsc{Clique}, and \textsc{Vertex Cover}. The existence of sub-exponential algorithms for any of them implies the same for the others. 
Moreover, a \textit{Strong Exponential Time Hypothesis} (SETH) has also been proposed, which assumes that $k$-CNF-SAT requires $\bigo{2^n}$ time in the worst-case when $k$ grows to infinity. This allows to derive even tighter bounds such as shown by \cite{cygan2016problems}: for every $\epsilon<1$ an $\bigo{2^{\epsilon n}}$ algorithm for \textsc{Hitting Set}, \textsc{Set Splitting} or \textsc{NAE-SAT} would violate SETH. 

Being not able to propose an algorithm with a smaller upper bound on the worst-case time complexity than the proposed dynamic programming, we prove by ETH an asymptotic lower bound, by a reduction from 3SAT to \textsc{SubsetSum} and then to \pbft. We show that \pbft{} cannot be solved in sub-exponential time, unless ETH fails. This immediately implies that the $Fm\|C_{max}$ problem, $\forall m>3$, cannot be solved in sub-exponential time either.

When establishing this result, we were not aware that a similar result  reported by \cite{jansen2013}, stated that for $f\in\{C_{max}, \sum w_jC_j\}$, the problems $ O3||f$, $ J3||f$, $ F3||f$ and $ P2||f$ cannot be solved asymptotically faster than $(2^{O(n)}\cdot \|I\|)$ unless ETH fails. The used reduction chain  3SAT $\rightarrow$ \textsc{SubsetSum}$\rightarrow$ \textsc{Partition} $\rightarrow$ \pbft, which is different from ours.

\subsubsection{From 3SAT to \pbft{}}
As stated by \cite{Impagliazzo2001}, proving lower bounds based on ETH implies designing problem reductions that preserve sub-exponential time property. For this, it is critical that resulting problem size of the reduction does not grow more than linearly. However, the time needed by the reduction is less important: it is not necessarily polynomial but should be at most sub-exponential.

\begin{definition}\label{def1}
Consider a reduction $\propto$ from problem $P_1$ with complexity parameter $m_1$ to $P_2$ with complexity parameter $m_2$. Reduction $\propto$ is \textbf{sub-exponential time preserving} if 
\begin{enumerate}
\item $\propto$ runs in sub-exponential time, i.e. for any $\epsilon>0$, $\propto$ runs in $\bigo{2^{\epsilon n}}$ and
\item $m_2(\propto(x))= \bigo{m_1(x)}$, i.e. the size of the output problem is linearly related to the size of the input problem.
\end{enumerate}
\end{definition}

In fact, if $P_2$ can be solved in sub-exponential time, so can $P_1$, since for any given instance of $P_1$ we can construct in sub-exponential time an instance of $P_2$ that can be solved in sub-exponential time.

Many existing standard problem reductions are naturally sub-exponential preserving. It is pretty straightforward to verify this since standard reductions are in polynomial time and we only need to check whether ``$m_2(\propto(x))= \bigo{m_1(x)}$''. In the case of the \pbft{} problem, it was proved as \npc{} by \cite{garey1976} by a reduction from 3-\textsc{Partition} (3PAR). The latter  was proved to be \npc{} by a series of reductions : 3SAT $\to$ 3-\textsc{Dimentional Matchining} (3DM) $\to$ 3PAR \citep{garey1975}. The classic reference of \cite{garey1979computers} also described the reduction 3SAT $\to$ 3DM  $\to$ 4PAR $\to$ 3PAR. 

Unfortunately, the described reduction chain from 3SAT to \pbft{} does not preserve sub-exponential time, since the reduction to 3DM outputs an instance of size $\Omega(n^2)$, 
i.e. asymptotically lower bounded by $n^2$. In this section, we describe a new reduction chain from 3SAT to \pbft{} which is proved to be sub-exponential time preserving and which in result allows to prove lower bounds for \pbft{} based on ETH. This chain is $3SAT\to$\textsc{SubsetSum} $\to$\pbft.

\subsubsection{3SAT$\to$\textsc{SubsetSum}}
For this  step, the standard NP-completeness reduction can be adopted directly (see for instance \cite{cormen2009intro}). However, in order to prove that the resulting instance has its size linear in that of the input instance, we need to apply here the \textsc{Sparsification Lemma} of  \cite{Impagliazzo2001}. For the sake of clarity we start by a brief description of the reduction (R1). The desired input and output of R1 are described as follows:

\textbf{Input problem} (\textsc{3SAT}): Given a boolean formula $F$ in 3-CNF form containing $n$ variables and $m$ clauses and each clause having exactly 3 literals, the question is whether there exists an assignment of variables that makes $F=true$.

\textbf{Output problem} (\textsc{SubsetSum}): Given $S$ as a set of $n$ integers, $t$ an integer parameter, the question is whether there exists a subset $S'\subset S$, s.t. $\sum_{i\in S'} i=t$.

R1 constructs the target \textsc{SubsetSum} as follows. Each element in $S$ is represented by a $(m+n)$-digit base-$7$ number. The first $n$ digits correspond to variables and the last $m$ digits are labeled by clauses. We note $C(i)=\{c: (c$ is a clause in $F) \wedge (v_i\in c)\}$. Without loss of generality, we assume that each variable must appear in at least one clause and no clause contains a variable and its negation.

\begin{enumerate}
\item The target value $t$ is a $(m+n)$-digit base-$7$ number which has $1$ on all variable digits and $4$ on all clause digits.
\item For each variable $v_i$ in $F$, we add 2 numbers into $S$: $s_i$ and $s'_i$ with the $v_i$-digit set to $1$. $\forall c\in C(i)$, the $c$-digit of $s_i$ (resp. $s'_i$) is set to $1$ if $c$ is satisfied when $v_i=true$ (resp. $false$). The other digits are set to $0$.
\item For each clause $c_i$ in $F$, we add 2 numbers $a_i$ and $a'_i$ into $S$. $a_i$ (resp. $a'_i$) is created by setting the $c_i$-digit to $1$ (resp. $2$), and by setting the other digits  to $0$. Note that these are slack numbers which will help us to add to $4$ in the target digit.
\end{enumerate}

Clearly, in the digits of $t$, the $v_i-digit$ is $1$ if and only if when $s_i$ or $s'_i$ is chosen, which corresponds to a true or false assignment to $v_i$. The $c_i-digit$ is $4$ only when at least one variable assignment satisfies $c_i$ since otherwise we would have at most $3$ by adding up $a_i$ and $a'_i$. Also note that by encoding the numbers in base $7$, no carries can occur. The correctness of the reduction is therefore straightforward.

The resulting \textsc{SubsetSum} instance has $2(m+n)$ numbers, each has $(m+n)$ digits, the whole construction can be done in polynomial time (even when considering $m=n^3$). The size of the resulting instance is $2(m+n)$, but according to Definition \ref{def1}, we should have $\bigo{n}$. This is the place where the \textsc{Sparsification Lemma} applies.

\begin{lemma}\label{lemma:sl}
\textsc{Sparsification Lemma}\citep{Impagliazzo2001} \\
$\forall \epsilon>0,k\geq 2$, a k-SAT formula $F$ with $n$ variables and $m$ clauses can be transformed into a k-SAT formula $F'$ which is the disjunction of at most $2^{\epsilon n}$ sub-formulas, each with $\bigo{n}$ clauses. The transformation runs in $\bigo{poly(n)2^{\epsilon n}}$ time. Alternatively speaking, k-SAT has a $2^{o(n)}$ algorithm if and only if it has a $2^{o(m)}$ algorithm.
\end{lemma}

This allows the reduction R1 to produce for each sub-formula with $n'$ variables and $m'$ clauses, a \textsc{SubsetSum} instance of size $2(m'+n')=\bigo{n'}$. Therefore R1 is \textit{sub-exponential time preserving} since if \textsc{SubsetSum} can be solved in sub-exponential time then for any instance of 3SAT and any $\epsilon>0$, we can solve it in sub-exponential time by first applying the \textsc{Sparsification Lemma}  then applying R1 and solve one by one the resulting \textsc{SubsetSum} sub-instances.

\subsubsection{\textsc{SubsetSum} $\to$\pbft}
Now we focus on the reduction (R2) from \textsc{SubsetSum} to \pbft{}. The desired input and output of R2 are described as follows:

\textbf{Input problem} (\textsc{SubsetSum}): Given $S$ as a set of $n$ integers, $t$ an integer parameter, the question is whether there exists a subset $S'\subset S$, s.t. $\sum_{i\in S'} i=t$.

\textbf{Output problem} (\pbft{}): Given $n'$ jobs to be processed on $3$ machines, each job must be processed on machine $1,2,3$ in this order without overlapped processing. The processing time of job $j$ on the $3$ machines is denoted as $p_j=\langle p_{j1},p_{j2},p_{j3} \rangle$. Each machine can only handle one job at one time. We ask whether there exists an ordering of jobs such that the whole processing can finish before time instant $T$.

We note $s$ as the sum of all integers in $S$. The reduction (R2) constructs the set of jobs $J$ as follows.
\begin{enumerate}
\item Add to $J$ job $j_{n+1}$: $\langle t,t,s-t\rangle$
\item For each element $a_i\in S$, add job $i$ : $\langle 0,a_i,0\rangle$
\item Let $T=s+t$
\end{enumerate}

Since the execution of $j_{n+1}$ requires already $s+t$ time, 
as shown in Figure~\ref{fig:insf3}, the remaining $n$ jobs must be perfectly filled into the line-shaded area without leaving any idle time in order to not excede $T$ on machine $2$. Since that area is cut into two parts of size $t$ and $(s-t)$ respectively, a perfect fill-in of the first part implies the existence of $S'\subset S$ such that $\sum S'=t$.

\begin{figure}
\centering
\includegraphics[width=0.6\textwidth]{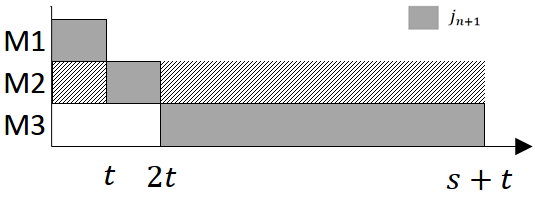}
\caption{Scheduling overview of the \pbft{} instance reduced from \textsc{SubsetSum}}\label{fig:insf3}
\end{figure}

\begin{lemma}
Assuming ETH, $\nexists$ algorithm such that for any $\epsilon>0$, it solves \pbft{} in $\bigo{2^{\epsilon n}}$ time. In other words, the \pbft{} problem cannot be solved in sub-exponential time unless ETH fails.
\end{lemma} 
\begin{proof}
It is easy to check that R2 runs in linear time and the resulting \pbft{} instance has size $n+2$ which is nearly the same as the input instance. As a consequence, R2 is sub-exponential time preserving. Combining with R1, the lemma is proved.
\end{proof}

\section{Chapter Summary}
\label{sec:4}
In this work we focus on the design of exact algorithms for several sequencing problems considering worst-case scenarios. The first problem that is tackled is the \pbft{} problem for which we propose a \dpa{} algorithm, improving the best-known time complexity in the literature {from $2^{O(n)}\times \|I\|^{\bigo{1}}$ (which can be proved to be $\ostar{c^n}$, $c>3$, when applied to this problem) to $\ostar{3^n}$}. The generalizations of the algorithm on the \pbfmax{} and \pbsumfi{} problems are then presented with worst-case time and space complexities in $\ostar{5^n}$ or $\ostar{M^22^n}$. 
Since the idea is based on the consideration of dominance conditions for non-decreasing cost functions and on the consideration of \pf{} of corresponding criteria vectors, it can be easily generalized to other problems that have similar structures. We formulate the generalized algorithm as a framework named \textit{Pareto Dynamic Programming}, and then apply it to the $1|r_i|\sum f_i(C_i)$ and $1|r_i|\max (f_i(C_i))$ problems. More generally, for all problems that possess a dominance condition based on multiple criteria, our \dpa{} is likely to be applicable. Table \ref{tab:summary} summaries the complexity results of PDP when applied to several problems. Recall that $B(n)$ is the upper bound of the number of optimal solutions returned after solving an instance of size $n$.



\begin{table}[!ht]
    \centering
    \begin{tabular}{|c|c|c|c|}
        \hline \makecell{Problem $(\alpha|\beta|f)$  \\$\beta\neq pmtn$ \\ $f\in\{f_{max},\sum f_i\}$}  & Criteria Vector & B(n) & Complexity \\ \hline
        $1||f$ &  $\langle f\rangle$ & $1$ & $\ostar{2^n}$ \\ \hline
        $1|r_i|f$ &  $\langle C_{max},f\rangle$ & $\ostar{2^n}$ & $\ostar{3^n}$ \\ \hline
        $F2\|f$ &  $\langle C_{max},f\rangle$ & $\ostar{2^n}$& $\ostar{3^n}$ \\ \hline
        $F3\|f$ &  $\langle C^M_2,C_{max},f\rangle$ & $\ostar{4^n}$& $\ostar{5^n}$ \\ \hline
    \end{tabular}
    \caption{Application of PDP on sequencing problems}
    \label{tab:summary}
\end{table}

The main work in this chapter has been performed together with Christophe Lenté, Mathieu Liedloff and Vincent T'Kindt. The results on the \pbft, \pbfmax{} and \pbsumfi{} problems have been published in \textit{Journal of Scheduling} \citep{shang2017jos}, \textit{MISTA} 2015 conference \citep{shang2015f3cmax} and \textit{ROADEF} 2016 conference \citep{shang2016roadef}.

\chapter{\techbm{} on the Single Machine Total Tardiness Problem}\label{ch3}
 
\section{Introduction}\label{sec:bm:intro}
In this chapter, we report new results on exact exponential algorithms solving a pure sequencing problem, the single machine total tardiness problem, denoted by \pbtt. In this problem, a job set $N=\{1,2, \dots, n\}$ of~$n$ jobs
must be scheduled on a single machine.
For each job~$i$, a processing time~$p_i$
and a due date~$d_i$ are given.
The problem asks for arranging the job set
into a sequence $\sigma$ so as to minimize
$T(N,\sigma)=\sum_{i=1}^n T_i = \sum_{i=1}^n \max\{C_i-d_i,0\}$,
where $C_i$ is the completion time of job $j$ in sequence $\sigma$.
The~\pbtt{} problem is NP-hard
in the ordinary sense~\citep{du1990minimizing}. It
has been extensively studied in the literature
and many exact procedures 
\citep{lawler1977pseudopolynomial,potts1982decomposition,della1998new,szwarc2001algorithmic} 
have been proposed. The current state-of-the-art exact method of~\cite{szwarc2001algorithmic} dates back to $2001$
and solves to optimality instances with up to~$500$ jobs.
All these procedures are search tree approaches, but dynamic programming algorithms were also considered.
On the one hand, \cite{lawler1977pseudopolynomial} proposed a pseudo-polynomial dynamic programming
algorithm running with complexity~$\bigo{n^4\sum p_i}$. 
On the other hand, the standard technique of doing dynamic programming across the subsets (see, for instance,~\cite{fomin2010exact}) applies and runs with complexity $\bigo{n^2 2^n}$ both in time and space.
Latest theoretical developments for the problem, including both exact and heuristic approaches can be found in the survey of ~\cite{koulamas2010single}.

We keep the use of the $\ostar{\cdot}$ notation for measuring worst-case complexities. $T(n)$ is still defined as the time required in the worst-case to exactly solve a  problem of size $n$. 
To the best of our knowledge, no available exact algorithm for this problem
running in $\ostar{c^n}$ ($c$ being a constant) time and polynomial space has been reported in the literature. Actually,  the complexity of the state-of-the-art algorithm of \cite{szwarc2001algorithmic} was not discussed by the authors, but our analysis shows that it runs in $\ostar{2.4143^n}$ time and polynomial space (see section \ref{sec:comp_bb2001}). 
One could also possibly apply a divide-and-conquer approach 
as described by \cite{gurevich1987expected} and \cite{bodlaender2012exact}. The idea is to determine the half of  jobs that should be processed first by enumerating all the 2-partitions of the input jobset. Then, for each 2-partition, solve the two corresponding subproblems recursively and finally combine the obtained solutions of subproblems to get the optimal solution. 
This would lead to an $\ostar{4^n}$ complexity in time 
requiring polynomial space. 
The aim of this work is to present an improved exact algorithm exploiting known decomposition properties of the problem. Different versions of the proposed approach are described in section \ref{sec:br}. A final version making use of a new technique called \techbm{}
that avoids the solution of several equivalent subproblems in the branching tree is presented in section \ref{sec:bm}. Finally we present some additional results in section \ref{sec:bmadditional}. 

First, let use recall some known properties on the $1||\sum T_j$ problem.


Given the job set $N=\{1,2,\dots,n\}$, let
$(1,2,\dots,n)$ be a~LPT (Longest Processing Time first) sequence, where $i < j$
whenever $p_i > p_j$ (or $p_i = p_j$ and $d_i \leq d_j$). Let also
$([1],[2],\dots,[n])$ be an~EDD (Earliest Due Date first) sequence, where $i < j$
whenever $d_{[i]} < d_{[j]}$ (or $d_{[i]} = d_{[j]}$ and $p_{[i]} \leq p_{[j]}$).
As the cost function is a regular performance measure, we know that
in an optimal solution, the jobs are processed with no interruption
starting from time zero.
Let~$B_i$ and~$A_i$ be the sets of jobs that precede and follow job~$i$ in an optimal sequence. 
Consequently, in this optimal sequence being constructed, 
$C_i = \sum_{\ell \in B_i}p_\ell + p_i$.
Similarly, if job $j$ is assigned to position $k$, we denote by $C_i(k)$ the corresponding completion time of $i$
and by $B_i(k)$ and $A_i(k)$ the sets of predecessors and successors of $i$, respectively.

The main known theoretical properties are the following.

\begin{proper}\label{Emmons69}~{\citep{emmons1969one}}
Consider two jobs~$i$ and~$j$ with $p_i < p_j$. Then, $i$ precedes $j$ in an optimal schedule if $d_i
\leq \max\{d_j,C_j\}$. Else $j$ precedes $i$ in an optimal schedule if $d_i + p_i > C_j$.
\end{proper}
\begin{proper}\label{dec1}~{\citep{lawler1977pseudopolynomial}}
Let job~$1$ in LPT order correspond to job~$[k]$ in EDD order.
Then, job~$1$ can be set only in positions $h\geq k$ and
the jobs preceding and following job~$1$ are uniquely determined as
$B_1(h) = \{[1],[2],\dots,[k-1],[k+1],\dots,[!ht]\}$ and $A_1(h) =
\{[h+1],\dots,[n]\}$. 
\end{proper}


\begin{proper}\label{Elim}~{\citep{lawler1977pseudopolynomial,potts1982decomposition,szwarc1993single}}
Consider $C_1(h)$, for $h \geq k$. 
Job~$1$ 
cannot be set in  positions $h \geq k$ if:
\begin{description}
\item[(a)] $C_1(h) \geq d_{[h+1]}$, $\forall h<n$;
\item[(b)] $C_1(h) < d_{[r]}+p_{[r]}$, for some
  $r=k+1,\dots,h$.  
\end{description}
\end{proper}

\begin{proper}\label{Elim3}~{\citep{szwarc1996decomposition}}
For any pair of adjacent positions $(i,i+1)$ that can be assigned to job $1$, at least one of them is eliminated by Property \ref{Elim}.
\end{proper}

In terms of complexity analysis, we recall (see, for instance,~\cite{eppstein2001improved}) 
that, if it is possible to bound above $T(n)$ by a recurrence expression 
of the type~$T(n) \leq \sum_{i=1}^h T(n-r_i) +  \bigo{p(n)}$, then we have~$\sum_{i=1}^h T(n-r_i) +  \bigo{p(n)} = \ostar{\alpha(r_1,\ldots,r_h)^n}$ 
where~$\alpha(r_1,\ldots,r_h)$
is the largest root of the function~$f(x) = 1 - \sum_{i=1}^h x^{-r_i}$. This observation will be useful later on to analyze worst-case complexities. 


\section{A \techbr{} approach}\label{sec:br}\label{sec:prelim}

We first focus on two direct \techbr{} algorithms based on the structural properties stated in the previous section. \techbr{} algorithms, as introduced in section \ref{sss:branching}, are search tree based algorithms which are elaborated upon two main components: a branching strategy and a reduction rule. The later relies on a mathematical condition which enables to take optimal decisions in the worst-case scenario whenever some other decisions have been taken by the branching. For instance, it can be a rule which states that even in the worst-case, if we branch on a job $j$ there always exists another job $i$ which can be scheduled without branching on it. 

\subsection{A first \techbr{} algorithm}
A basic \techbr{} algorithm TTBR1 (Total Tardiness \techbr{} version 1) can be designed by exploiting Property \ref{dec1}, which enables to decompose the problem into two smaller subproblems when the position of the longest job $\ell$ is given (see Figure \ref{fig:dec}). 

\begin{figure}
\centering
\includegraphics[width=0.6\columnwidth]{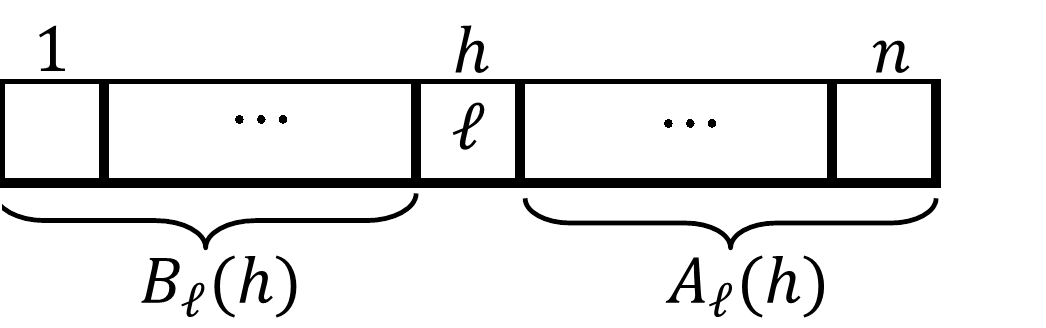}
\caption{The situation when the longest job $\ell$ is put in position $h$}\label{fig:dec}
\end{figure}

The basic idea is to iteratively branch by assigning job $\ell$ to every eligible branching position and correspondingly decomposing the problem. Each time job $\ell$ is assigned to a certain position $i$, two different subproblems are generated, corresponding to schedule the jobs before $\ell$ (inducing subproblem ${B_\ell(i)}$) or after $\ell$ (inducing subproblem ${A_\ell(i)}$), respectively. 
The algorithm operates by applying to any given job set $S$ starting at time $t$
function $TTBR1(S,t)$ that computes the corresponding optimal solution.
With this notation, the original problem is indicated by ${N}=\{1,...,n\}$
and the optimal solution is reached when function $TTBR1(N,0)$ is computed.

The algorithm proceeds by solving the subproblems along the branching tree according to a depth-first strategy and runs until all the leaves of the search tree have been reached. Finally, it provides the best solution found as an output.
Algorithm~\ref{algo:ttbr1} summarizes the structure of this approach, while Proposition~\ref{propos:ttbr1} states its worst-case complexity. Figure \ref{fig:br} illustrates the branching operation of the algorithm.

\begin{algorithm}
\caption{Total Tardiness \techbr{} version 1 (TTBR1)}\label{algo:ttbr1}
\begin{algorithmic}[1]
\Require ${N}=\{1,...,n\}$ is the problem to be solved 
\Function{TTBR1}{${S,t}$}
\State $seqOpt \gets$ the EDD sequence of jobs $S$
\State $\ell \gets$ the longest job in $S$
\For{any eligible position $i$ (Property \ref{dec1})}
    \State Branch by assigning job $\ell$ in position $i$ 
    \State $seqLeft \gets$ TTBR1(${B_\ell(i)},t$)
    \State $seqRight \gets$ TTBR1(${A_\ell(i)},t+\sum_{k \in B_\ell(i)}p_k + p_\ell$)
    \State $seqCurrent \gets$ concatenation of $seqLeft$, $\ell$ and $seqRight$
    \State $seqOpt \gets$ best solution between $seqOpt$ and $seqCurrent$
\EndFor
\State  \Return $seqOpt$
\EndFunction
\end{algorithmic}
\end{algorithm}

\begin{figure}
\centering
\includegraphics[width=0.9\columnwidth]{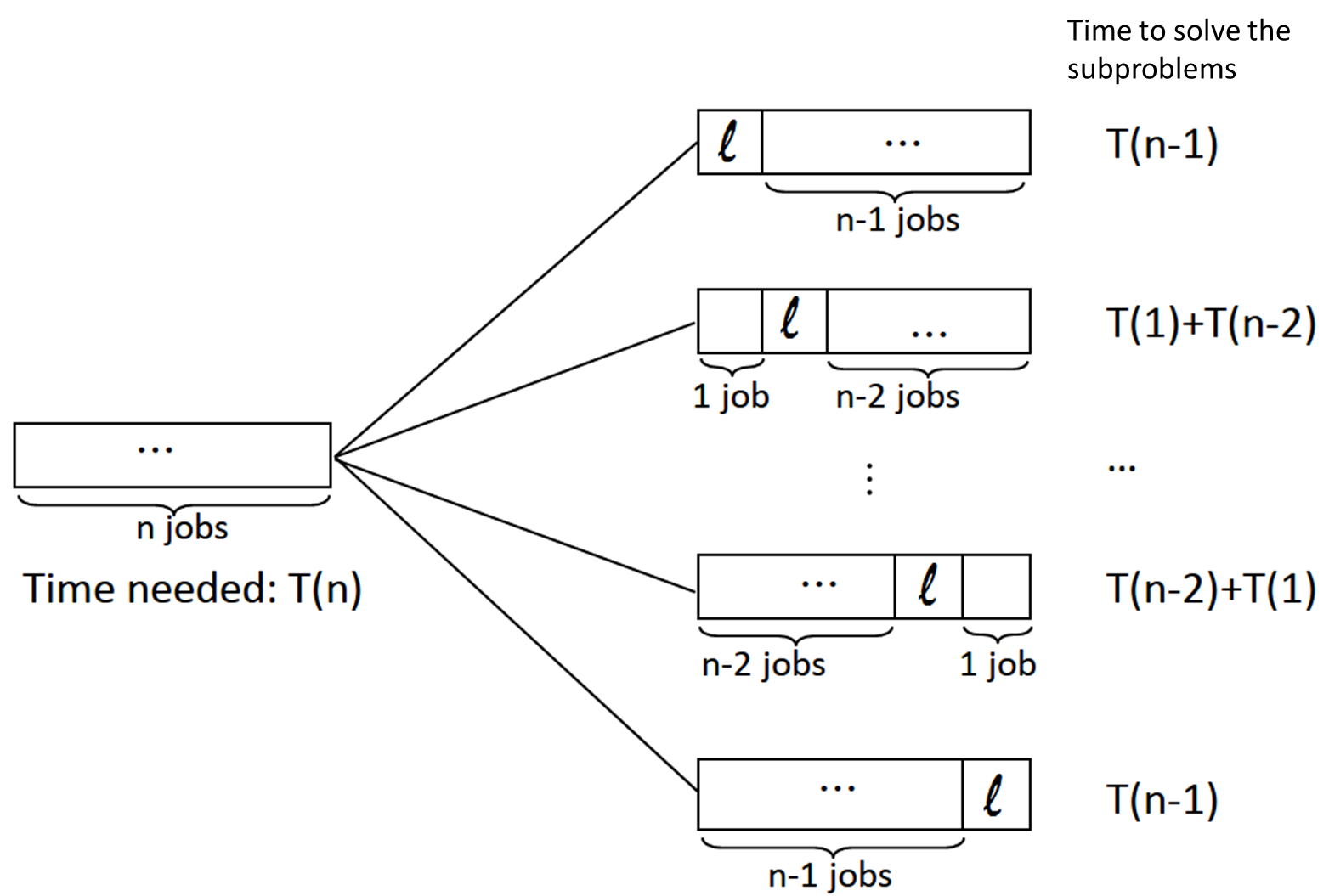}
\caption{Branching scheme of TTBR1}\label{fig:br}
\end{figure}

\begin{propos}\label{propos:ttbr1}
Algorithm TTBR1 runs in $\ostar{3^n}$ time and polynomial space in the worst case.
\end{propos}
\begin{proof}
According to Property \ref{dec1}, the worst-case scenario arises when the ordering LPT is the same as EDD, in which case all positions are eligible for branching on the longest job. 
Whenever the longest job $1$ is assigned to the first and the last position of the sequence, two subproblems of size $(n-1)$ are generated. 
For each $2 \leq i \leq n-1$, two subproblems with size $(i-1)$ and $(n-i)$ are generated.
Hence, the total number of generated subproblems is $(2n-2)$ and the time cost related to computing the best solution of size $n$ starting from these subproblems is $\bigo{p(n)}$.
This induces the following recurrence for the running time $T(n)$ required by TTBR1:
\begin{equation}
\label{eq:tn}
T(n) = 2T(n-1) + 2T(n-2) +... + 2T(2) + 2T(1) + \bigo{p(n)}
\end{equation}
By replacing $n$ with $(n-1)$, the following expression is derived: 
\begin{equation}
\label{eq:tn1}
T(n-1) = 2T(n-2) +... + 2T(2) + 2T(1) + \bigo{p(n-1)}
\end{equation}
Expression~\ref{eq:tn1} can be used to simplify the right hand side of expression~\ref{eq:tn} leading to:
\begin{equation}
\label{eq:tn1b}
T(n) = 3 T(n-1) + \bigo{p(n)}.
\end{equation}
By solving the corresponding equation $x^n={3}{x^{n-1}}$ and take its root as the exponential base, this recurrence yields the result $T(n) = \ostar{3^n}$. 
The space requirement is polynomial since the search tree is explored according to a depth-first strategy. 
\end{proof}
\bigskip

\subsection{A second \techbr{} algorithm}
An improved version of the algorithm TTBR1 is defined by taking into account Property~\ref{Elim} and Property~\ref{Elim3}, which state that for each pair of adjacent positions $(i,i+1)$, at least one of them can be discarded.
The worst case occurs when the largest possible subproblems are kept when branching, since otherwise the complexity can be easily proved as smaller.
This corresponds to solving problems with size $(n-1),(n-3),(n-5), \ldots$, that arise by branching in positions $i$ and $(n-i+1)$ with $i$ odd.
The resulting algorithm is referred to as TTBR2 (Total Tardiness Branch and Reduce version 2). Its structure is similar to the one of TTBR1 depicted in Algorithm~\ref{algo:ttbr1}, but lines 5-9 are executed only when $\ell$ can be set in position $i$ according to Property~\ref{Elim}. The complexity of the algorithm is discussed in Proposition~\ref{propos:ttbr2}.

\begin{propos}\label{propos:ttbr2}
Algorithm TTBR2 runs in $\ostar{(1+\sqrt{2})^n}  = \ostar{2.4143^n}$ time and polynomial space in the worst case.
\end{propos}
\begin{proof}
The proof is close to that of Proposition~\ref{propos:ttbr1}. We refer 
to problems where $n$ is odd, but the analysis for
$n$ even is substantially the same. The algorithm induces a recursion of the type:
\begin{equation}
\label{eq:p2tn1}
T(n) = 2T(n-1) + 2T(n-3) +... + 2T(4) + 2T(2) + \bigo{p(n)}
\end{equation}
as the worst case occurs when we keep the branches that induce the largest possible subproblems. 
Analogously to Proposition~\ref{propos:ttbr1}, we replace $n$ by $(n-2)$ in the previous recurrence and we obtain: 

\begin{equation}
\label{eq:p2tn2}
T(n-2) =  2T(n-3) + 2T(n-5) +... + 2T(4) + 2T(2) + \bigo{p(n-2)}.
\end{equation}

Again, we plug the latter expression  into the former one and obtain the recurrence:

\begin{equation}
\label{eq:tn2b}
T(n) = 2T(n-1) + T(n-2)  + \bigo{p(n)}.
\end{equation}
By solving the corresponding equation $x^n=2{x^{n-1}}+x^{n-2}$ and take its largest positive root as the exponential base, the recurrence  induces  
$T(n)=\ostar{(1+\sqrt{2})^n} = \ostar{2.4143^n}$.  The space complexity remains polynomial as for TTBR1. 
\end{proof}
\bigskip

\section{A \techbm{} Algorithm}\label{sec:bm}

In this section, we describe how to get an algorithm running with complexity arbitrarily close to $\ostar{2^n}$ in time and polynomial space by integrating a node-merging procedure into TTBR1. The idea of such a procedure comes from the observation that in the search tree created by the branching scheme (Property \ref{dec1}), a lot of identical subproblems are explored for nothing. 

We recall that in TTBR1 the branching scheme is defined by assigning the longest unscheduled job to each available position and accordingly divide the problem into two subproblems. 
To facilitate the description of the algorithm, we focus on the scenario where the LPT sequence $(1,...,n)$ coincides with the EDD sequence $([1],...,[n])$, for convenience we write $LPT=EDD$. 

We provide the algorithmic details of the node-merging procedure on this scenario to facilitate its understanding. The resulting branch-and-merge algorithm has its time complexity tend to $\ostar{2^n}$. We prove by Lemma~\ref{lemma:BMWC} that the case where $LPT=EDD$ is the worst-case scenario, hence, it follows that the problem \pbtt{} can be solved in time complexity tending to $\ostar{2^n}$. We leave to the reader the generalization of the node-merging procedure to the general case.

Figure~\ref{fig:scheme0} shows how an input problem $\{1,...,n\}$ is decomposed by the branching scheme of TTBR1. Each node is labelled by the corresponding subproblem $P_j$ ($P$ denotes the input problem).
Notice that from now on $P_{j_1,j_2,...,j_k}, 1\leq k\leq n$, denotes the problem (corresponding to a node in the search tree) induced by the
branching scheme of TTBR1 when the largest processing time job $1$ is in position $j_1$,
the second largest processing time job $2$ is in position $j_2$ and so on till
the k-th largest processing time job $k$ being placed in position $j_k$. {Notice that $P_{j_1,j_2,...,j_k}$ results from a series of branching operations according to Property \ref{dec1}, and may contain more than one subproblems to be solved. Consider the problem $P_3$ illustrated in Figure \ref{fig:scheme0}: both jobsets $\{2,3\}$ and $\{4,..,n\}$ are to be scheduled and can be considered as two independents subproblems. Nevertheless, the notation $P_{j_1,j_2,...,j_k}$ can largely simplify the presentation of our algorithm and the fact that a node may contain more than one subproblems should not cause ambiguity to our belief. Surely, this is considered when analyzing the complexity of the algorithms.}

\begin{figure}[!ht]
\includegraphics[width=\columnwidth]{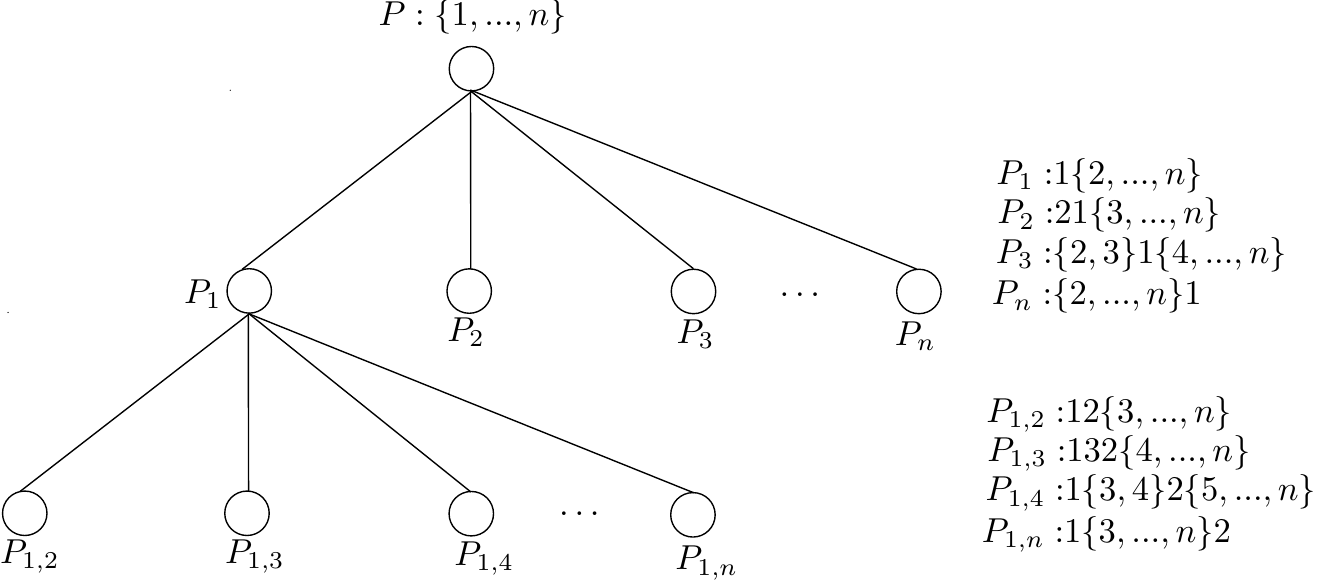}
\caption{The branching scheme of TTBR1 at the root node}
\label{fig:scheme0}
\end{figure}

To roughly illustrate the guiding idea of the merging technique introduced in this section, consider Figure \ref{fig:scheme0}. 
Noteworthy, nodes $P_2$ and $P_{1,2}$ are identical except for the initial 
subsequence ($21$ vs $12$).  This fact implies, in this particular case, that the problem of scheduling job set $\{3,...,n\}$ at time $p_1+p_2$ is solved twice. This kind of redundancy can however be eliminated by merging node $P_2$ with node $P_{1,2}$ and creating a single node in which the best sequence among $21$ and $12$ is scheduled at the beginning and the job set $\{3,...,n\}$, starting at time $p_1+p_2$, remains to be branched on. Furthermore, the best subsequence (starting at time $t=0$) between $21$ and $12$
can be computed in constant time. 
Hence, the node created after the merging operation involves a constant time preprocessing step plus the search for the optimal solution of job set $\{3,...,n\}$ to be processed starting at time $p_1+p_2$. 
We remark that, in the branching scheme of TTBR1, for any constant $k \geq 3$, the branches corresponding to $P_i$ and $P_{n-i+1}$, with $i = 2, ..., k$, are decomposed into two problems where one subproblem has size $(n-i)$ and the other problem has size $(i-1) \leq k$.
Correspondingly, the merging technique presented on problems $P_2$ and $P_{1,2}$ can be generalized to all branches inducing subproblems of size less than $k$. Notice that, by means of algorithm TTBR2, any problem of size less than $k$ requires at most  $\ostar{2.4143^k}$ time to be solved (which is constant time when $k$ is fixed). This is the central idea of the merging procedure: for a given value of $k$ as an input of the algorithm, merge nodes which involve the same subproblems of size less than $k$. As $k$ is a constant, the merging can be done in constant time.

In the remainder of the paper, for any constant $k\leq \frac{n}{2}$, we denote by left-side branches
the search tree branches corresponding to problems $P_1,...,P_k$ and by right-side branches the ones corresponding to problems $P_{n-k+1},...,P_{n}$.

In the following subsections, we show how the node-merging procedure can be systematically performed to improve the time complexity of TTBR1. Basically, two different recurrent structures hold respectively for left-side and right-side branches and allow to generate less subproblems at each level of the search tree. 
 The node-merging mechanism is described by means of two distinct procedures, called {\ttfamily LEFT\_MERGE} (applied to left-side branches) and {\ttfamily RIGHT\_MERGE} (applied to right-side branches), which are discussed in sections~\ref{sec:leftsection} and~\ref{sec:rightsection}, 
respectively. The final \techbm{} algorithm is described in section~\ref{sec:thealgorithm} and embeds both procedures into the structure of TTBR1.

\subsection{Merging left-side branches}
\label{sec:leftsection}

We  first illustrate the merging operations on the root node. The following proposition highlights two properties of problems $P_{j}$ and $P_{1,j}$ with $2 \leq j \leq k$.

\begin{lemma}\label{propos:LMLink}
For a pair of problems $P_{j}$ and $P_{1,j}$ with $2 \leq j \leq k$, the following conditions hold:
\begin{enumerate}
\item The solution of problems $P_{j}$ and $P_{1,j}$ involves the solution of a common subproblem which consists in scheduling job set $\{j+1,..., n\}$ starting at time $t = \sum_{i = 1}^jp_i$.
\item Both in $P_{j}$ and $P_{1,j}$, at most $k$ jobs have to be scheduled before job set $\{j+1,...,n\}$.
\end{enumerate}
\end{lemma}
\begin{proof}
As problems $P_{j}$ and $P_{1,j}$  are respectively defined by  $\{2,...,j\}1\{j+1,...,n\}$ and $1\{3,...,j\}2\{j+1,...,n\}$, the first part of the property is straightforward.\\
The second part can be simply established by counting the number of jobs to be scheduled before job set $\{j+1,..., n\}$ when $j$ is maximal, {\it i.e.} when $j=k$. In this case, job set $\{k+1,...,n\}$ has $(n-k)$ jobs which implies that $k$ jobs remain to be scheduled before that job set. 
\end{proof}

Each pair of problems indicated in Proposition \ref{propos:LMLink} can be merged as long as they share the same subproblem to be solved. More precisely, $(k-1)$ problems $P_{j}$
(with $2\leq j \leq k $) can be merged with the corresponding problems $P_{1,j}$.

\begin{figure}[!ht]
\begingroup
\thickmuskip=0mu
\medmuskip=0mu
\centering
        \begin{subfigure}[!ht]{\textwidth}
                \centering
                
                \includegraphics[width=\columnwidth]{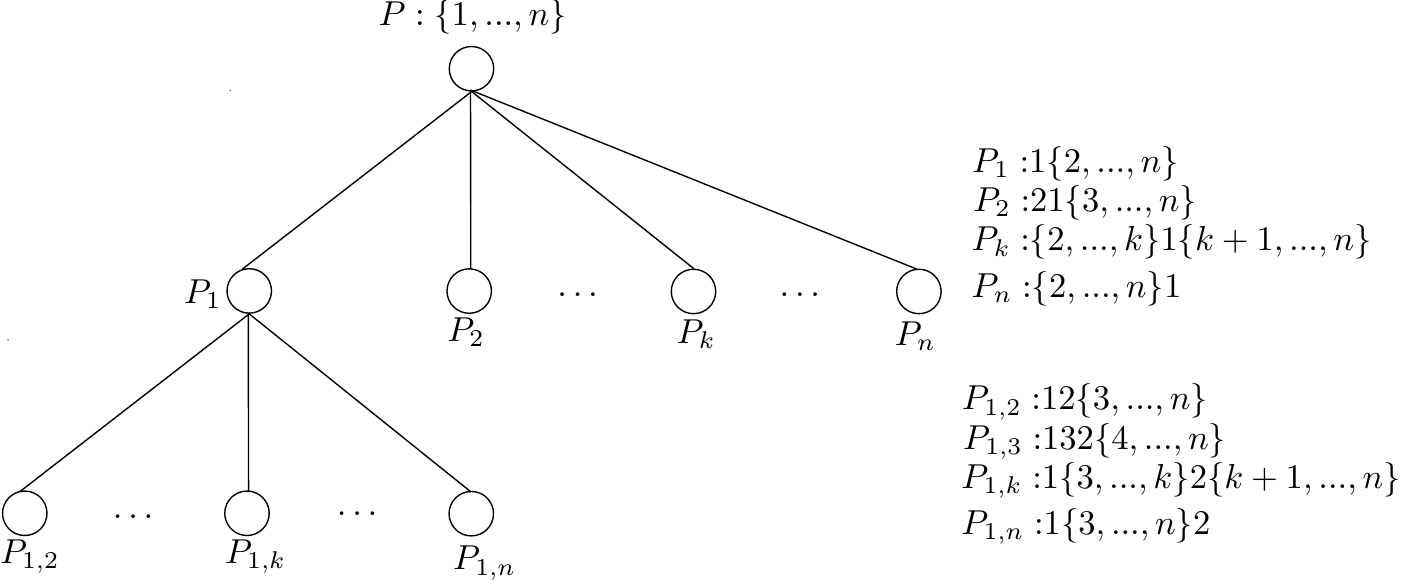}
                \caption{Left-side branches of $P$ before performing the merging operations}
                \label{fig:gull}
        \end{subfigure}%
        
\vspace{0.5cm}
        \begin{subfigure}[!ht]{\textwidth}
                \centering
                \includegraphics[width=\columnwidth]{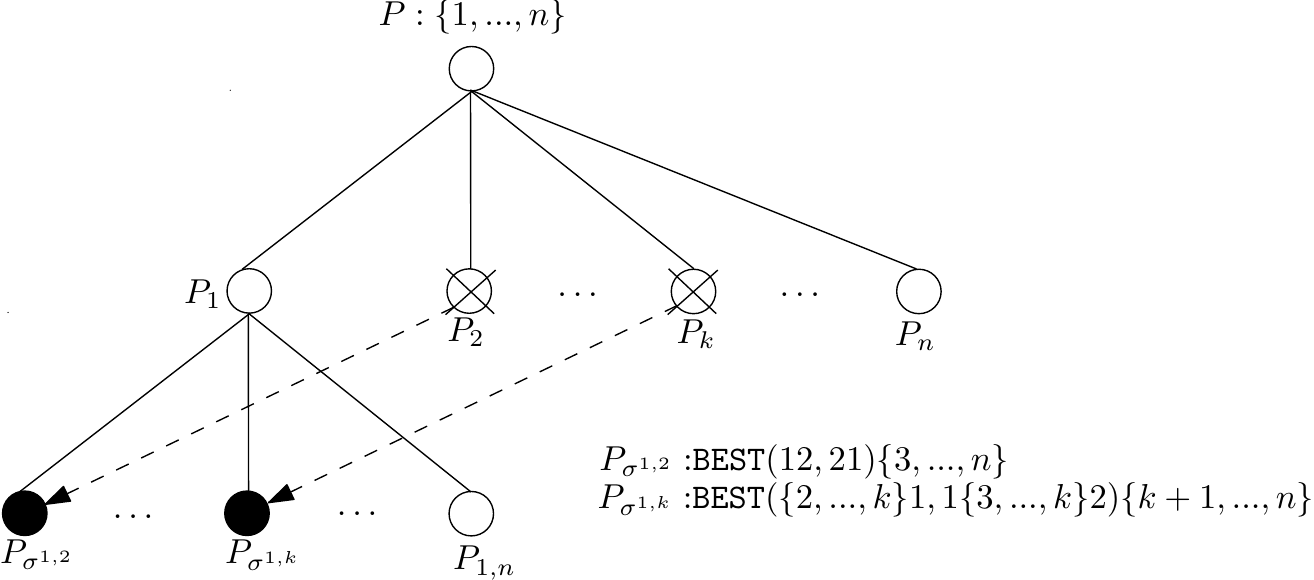}
                \caption{Left-side branches of $P$ after performing the merging operations. \\$\texttt{BEST}(\alpha,\beta)$ returns the best sequence among $\alpha$ and $\beta$.}
                \label{fig:gull2}
        \end{subfigure}%
\caption{Left-side branches merging at the root node}
\label{fig:ex_merging}
\endgroup
\end{figure}


Figure~\ref{fig:ex_merging} illustrates the merging operations performed at the root node on its left-side branches, by showing the branch tree before and after (Figure~\ref{fig:gull} and Figure~\ref{fig:gull2}) such merging operations. For any given $2 \leq j \leq k$, problems $P_j$ and $P_{1,j}$
share the same subproblem $\{j+1,...,n\}$ starting at time $t = \sum_{i =1}^j p_i$. Hence, by merging the left part of both problems which is constituted by job set $\{1,...,j\}$ having size $j\leq k$,
we can delete node $P_j$ and replace node $P_{1,j}$ in the search tree by the node $P_{\sigma^{1,j}}$ which is defined as follows (Figure \ref{fig:gull2}):  
\begin{itemize}
\item $\{j+1,...,n\}$ is the set of jobs on which it remains to branch.
\item Let $\sigma^{1,j}$ be the sequence of branching positions on which the $j$ longest jobs $1,...,j$ are branched, that leads to the best jobs permutation between $\{2,...,j\}1 $ and $ 1\{3,...,j\}2$. More concretely, the subproblem involving  $\{2,...,j\}$ from $\{2,...,j\}1 $ and the subproblem involving $\{3,...,j\}$ from $ 1\{3,...,j\}2$ must be solved (by calling TTBR2, for instance) first in order to obtain the optimal permutation of $\{2,...,j\}1 $ and $ 1\{3,...,j\}2$. 
This involves the solution of two problems of size at most $(k-1)$ (in $\ostar{2.4143^{k}}$ time by TTBR2) and the comparison of the total tardiness value of the two sequences obtained. 
\end{itemize}

In the following, we describe how to apply analogous merging operations on any node of the tree. With respect to the root node, the only additional consideration is that the children nodes of an arbitrary node may have already been affected by a previous merging. 
In Figure \ref{fig:gull2}, this is for instance the case for $P_\sigma^{1,3}$ which can be merged with $P_{\sigma^{1,2},3}$, a child node of $P_\sigma^{1,2}$ (if $k\geq 3$).

In order to define the branching scheme used with the \texttt{LEFT\_MERGE} procedure, a data structure $\mathcal{L}_{\sigma}$ is associated to a problem $P_\sigma$. It represents a list of $(k-1)$ subproblems that result from a previous merging and are now the first $(k-1)$ children nodes of $P_\sigma$. 
When $P_\sigma$ is created by branching, $\mathcal{L}_{\sigma} = \emptyset$. When a merging operation sets the first $(k-1)$ children nodes of $P_\sigma$ to $P_{\sigma^1},...,P_{\sigma^{k-1}}$, we set $\mathcal{L}_{\sigma} = \{P_{\sigma^1},...,P_{\sigma^{k-1}}\}$. This list if used to memorize the fact that the first $(k-1)$ child nodes of $P_\sigma$ will not have to be built by a branching operation. As a conclusion, the following branching scheme for an arbitrary node of the tree holds.

\begin{definition}
\label{def:leftscheme}
The branching scheme for an arbitrary node $P_\sigma$ is defined as follows:
\begin{itemize}
\item If $\mathcal{L}_{\sigma} = \emptyset$, use the branching scheme of TTBR1;
\item If $\mathcal{L}_{\sigma} \neq \emptyset$, extract problems from $\mathcal{L}_{\sigma}$ as the first $(k-1)$ branches, then branch on the longest job in $P_\sigma$ in the available positions from the $k$-th to the last according to Property~\ref{dec1}.
\end{itemize}
This branching scheme, whenever necessary, will be referred to as \textbf{improved branching}. Note that this will be extended later when ``merging right-side nodes'' is introduced.
\end{definition}

Before describing how merging operations can be applied on an arbitrary node $P_\sigma$, we highlight its structural properties by means of Proposition~\ref{propos:psigmaprop}.

\begin{propos}\label{propos:psigmaprop}
Let $P_{\sigma}$ be a problem to branch on, and $\sigma$ be the permutation of positions assigned to jobs $1,\ldots,|\sigma|$, with $\sigma$ empty if no positions are assigned. The following properties hold:
\begin{enumerate}
\item $j^*=|\sigma|+1$ is the job to branch on,
\item \par{$j^*$ can occupy in the branching process, positions $\{\ell_{b},\ell_{b}+1,\ldots,\ell_{e}\}$, where
 \[\ell_{b}=
    \begin{cases}
    |\sigma|+1 & \text{if $\sigma$ is a permutation of $\{1,\ldots,|\sigma\}$ or $\sigma$ is empty}\\
    \rho_1+1  &\text{otherwise}
    \end{cases}\] with $\rho_1=\max\{i: i> 0, \text{ positions $1,\ldots,i$ are in $\sigma$}\}$ 
 and \[\ell_{e}=
    \begin{cases}
    n & \text{if $\sigma$ is a permutation of $\{1,\ldots,|\sigma|\}$ or $\sigma$ is empty}\\
    \rho_2-1  &\text{otherwise}
    \end{cases}\]} with $\rho_2=\min\{i:\text{ $i>\rho_1$, $i\in\sigma$}\}$
\end{enumerate}
\end{propos}
\begin{proof}
According to the definition of the notation $P_\sigma$, $\sigma$ is a sequence of positions that are assigned to the longest $|\sigma|$ jobs. Since we always branch on the longest unscheduled job, the first part of the proposition is straightforward.
The second part aims at specifying the range of positions that job $j^*$ can occupy. Two cases are considered depending on the content of $\sigma$: 
\begin{itemize}
\item If $\sigma$ is a permutation of $\{1,\ldots,|\sigma|\}$, it means that the longest $|\sigma|$ jobs are set on the first $|\sigma|$ positions, which implies that the job $j^*$ should be branched in positions $|\sigma|+1$ to $n$
\item If $\sigma$ is not a permutation of $\{1,\ldots,|\sigma|\}$, it means that the longest $|\sigma|$ jobs are not set on consecutive positions. As a result, the current unassigned positions may be split into several ranges. As a consequence of the decomposition property and the assumption LPT=EDD, the longest job $j^*$ should necessarily be branched on the first range of free positions, that goes from $l_b$ to $l_e$. 
Let us consider as an example $P_{1,9,2,8}$ (see Figure \ref{fig:propro3}), whose structure is $13\{5,\ldots,9\}42\{10,\ldots,n\}$ and the job to branch on is $5$. In this case, we have: $\sigma=(1,9,2,8)$, $\ell_{b}=3$, $\ell_{e}=7$. It is easy to verify that $5$ can only be branched in positions $\{3,\ldots,7\}$ as a direct result of Property~\ref{dec1}. 
\end{itemize}
\end{proof}

\begin{figure}
\centering
 \includegraphics[width=0.8\columnwidth]{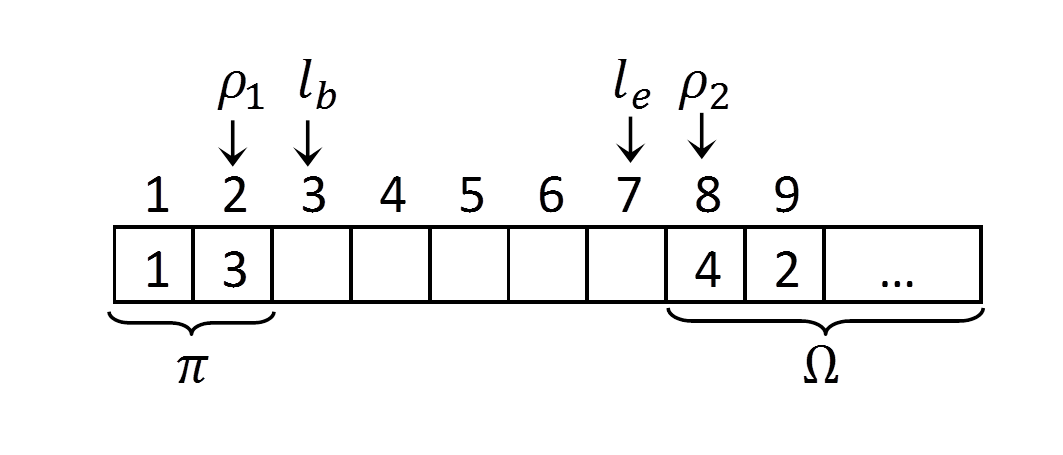}
  \caption{An example ($P_{1,9,2,8}$) for Proposition \ref{propos:psigmaprop} and Corollary \ref{coro}}\label{fig:propro3}
\end{figure}

Corollary~\ref{coro} emphasises the fact that even though a node may contain several ranges of free positions, only the first range is the current focus since we only branch on the longest job in eligible positions.
\begin{coro}\label{coro}
Problem $P_\sigma$ has the following structure:
$$ {\pi}\{j^*,\ldots,j^*+\ell_{e}-\ell_{b}\}\Omega$$ with ${\pi}$ the subsequence of jobs on the first $(\ell_{b}-1)$ positions in $\sigma$ and $\Omega$ the remaining subset of jobs to be scheduled after position $\ell_{e}$ (some of them can have been already scheduled). The merging procedure is applied on job set $\{j^*,\ldots,j^*+\ell_{e}-\ell_{b}\}$ starting at time $t_{\pi}=\sum_{i\in \Pi }p_i$ where $\Pi$ is the job set of $\pi$.
\end{coro}

The validity of merging on a general node still holds as indicated in Proposition~\ref{propos:LMLinkLeft}, which extends the result stated in Proposition~\ref{propos:LMLink}.
\begingroup
\thickmuskip=0mu
\medmuskip=0mu
\begin{propos}\label{propos:LMLinkLeft}
Let $P_{\sigma}$ be an arbitrary problem and let $\pi,j^*,\ell_{b},\ell_e,\Omega$ be computed relatively to $P_{\sigma}$ according to Corollary~\ref{coro}. If $\mathcal{L}_{\sigma} = \emptyset$ the $j$-th child node $P_{\sigma^j}$ is $P_{\sigma ,\ell_{b}+j-1}$ for $1 \leq j \leq k$. Otherwise, the $j$-th child node $P_{\sigma^j}$ is extracted from  $\mathcal{L}_{\sigma}$ for $1 \leq j \leq k-1$, while it is created as $P_{\sigma ,\ell_{b}+k-1}$ for $j=k$.  For any pair of problems $P_{\sigma^j}$ and $P_{\sigma^1,\ell_{b}+j-1}$ (the $(j-1)$-th child node of $P_{\sigma^1}$) with $2 \leq j \leq k$, 
the following conditions hold:

\begin{enumerate}
\item Problems $P_{\sigma^j}$ and $P_{\sigma^1,\ell_{b}+j-1}$ with $2 \leq j \leq k$ have the following structure:
    \begin{itemize}
        \item $P_{\sigma^j}:$
            \[\begin{cases}
            \verb?                  ?\pi^j\{j^*+j,\ldots,j^*+\ell_{e}-\ell_{b}\}\Omega & \text{$1 \leq j \leq k-1$ and $\mathcal{L}_{\sigma} \neq \emptyset$}\\
            \pi\{j^*+1,\ldots,j^*+j-1\}j^*\{j^*+j,\ldots,j^*+\ell_{e}-\ell_{b}\}\Omega & \parbox{3cm}{\mbox{~}\\($1 \leq j \leq k-1;\mathcal{L}_{\sigma} = \emptyset$)\\ or $j=k$}\\
            \end{cases}\]
        \item $P_{\sigma^1,\ell_{b}+j-1}:\\ \pi^1\{j^*+2,\ldots,j^*+j-1\}(j^*+1)\{j^*+j,\ldots,j^*+\ell_{e}-\ell_{b}\}\Omega$
    \end{itemize}
\item By solving all the problems of size less than $k$, that consist in scheduling the job set $\{j^*+1,\ldots,j^*+j-1\}$ between $\pi$ and $j^*$ and in scheduling $\{j^*+2,\ldots,j^*+j-1\}$ between $\pi^1$ and $j^*+1$, both $P_{\sigma^j}$ and $P_{\sigma^1,\ell_{b}+j-1}$ consist in scheduling $\{j^* + j,..., j^*+\ell_{e}-\ell_{b}\}\Omega$ starting at time $t_{\pi^j} = \sum_{i \in \Pi^j} p_{i}$ where $\Pi^j$ is the job set of $\pi^j$.
\end{enumerate}
\end{propos}
\begin{proof} 
The first part of the statement follows directly from Definition \ref{def:leftscheme} and simply defines the structure of the children nodes of $P_{\sigma}$.
The problem $P_{\sigma^j}$ is the result of a merging operation between a sibling node of $P_\sigma$  and the problem $P_{\sigma,\ell_{b}+j-1}$ and it could possibly coincide with $P_{\sigma,\ell_{b}+j-1}$, being the result of merging, for each $j=1,...,k-1$. 
Furthermore, $P_{\sigma^j}$ is exactly $P_{\sigma,\ell_{b}+j-1}$ for $j=k$. 
The structure of $P_{\sigma,\ell_{b}+j-1}$ is $\pi\{j^*+1,\ldots,j^*+j-1\}j^*\{j^*+j,\ldots,j^*+\ell_{e}-\ell_{b}\}\Omega$, and the merging operations preserve the job set to schedule after $j^*$. Thus, we have $\Pi^j=\Pi\cup\{j^*,...,j^*+j-1 \}$ for each $j=1,...,k-1$, and this proves the first statement.
Analogously, the structure of $P_{\sigma^1,\ell_{b}+j-1}$ is $\pi^1\{j^*+2,\ldots,j^*+j-1\}(j^*+1)\{j^*+j,\ldots,j^*+\ell_{e}-\ell_{b}\}\Omega$. Once the subproblem before $(j^*+1)$ of size less than $k$ has been solved, $P_{\sigma^1,\ell_{b}+j-1}$ consists in scheduling the job set $\{j^* + j,..., j^*+\ell_{e}-\ell_{b}\}$ at time $t_{\pi^j} = \sum_{i \in \Pi^j} p_{i}$. In fact, we have that $\Pi^j = \Pi^1 \cup \{j^*+2,\ldots,j^*+j-1\} \cup \{j^*+1\} = \Pi \cup \{j^*,\ldots,j^*+j-1\} $. Therefore, both $P_{\sigma^j}$ and $P_{\sigma^1,\ell_{b}+j-1}$ consist in scheduling $\{j^* + j,..., j^*+\ell_{e}-\ell_{b}\}\Omega$ starting at time $t_{\pi^j} = \sum_{i \in \Pi^j} p_{i}$.

\end{proof}\endgroup

\begin{figure}[!ht]
\begingroup
\thickmuskip=0mu
\medmuskip=0mu
\centering
\includegraphics[width=0.7\columnwidth]{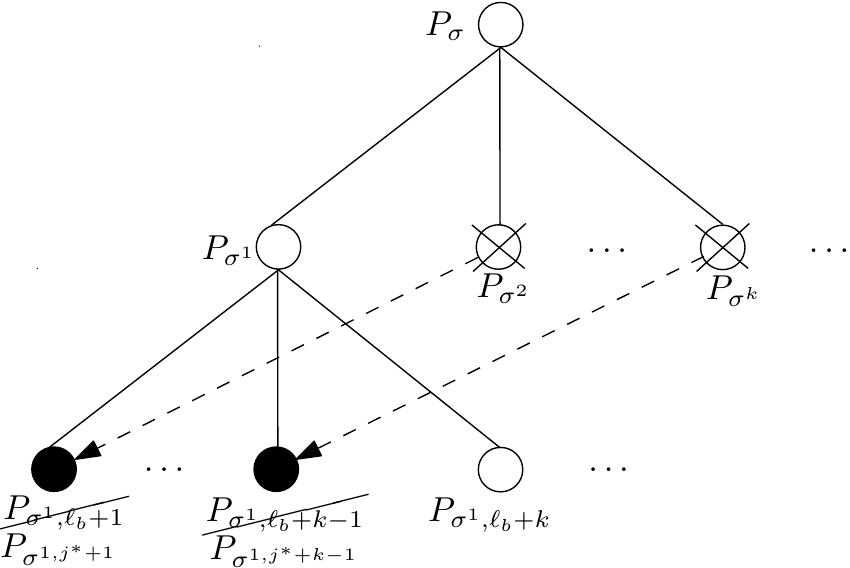}
\caption{Merging for an arbitrary left-side branch}
\label{fig:ex_merging_g}
\endgroup
\end{figure}

Analogously to the root node, each pair of problems indicated in Proposition \ref{propos:LMLinkLeft} can be merged. 
Again, $(k-1)$ problems $P_{\sigma^j}$ (with $2\leq j \leq k $) can be merged with the corresponding problems $P_{\sigma^1,\ell_{b}+ j-1}$. 
  $P_{\sigma^j}$ is deleted and $P_{\sigma^1,\ell_{b}+ j-1}$ is replaced by $P_{\sigma^{1,j^*+j-1}}$ (Figure~\ref{fig:ex_merging_g}), defined as follows:
\begin{itemize}
\item $\{j^*+ j,...,j^*+ \ell_{e}-\ell_{b}\}\Omega$ is the set of jobs on which it remains to branch on.
\item Let $\sigma^{1,j^*+j-1}$ be the sequence of positions on which the $(j^*+j-1)$ longest jobs $\{1,...,j^*+j-1\}$ are branched, that leads to the best jobs permutation between $\pi^j$ and $\pi^1\{j^*+2,\ldots,j^*+j-1\}(j^*+1)$, for $2\leq j\leq k-1$, and between $\pi\{j^*+1,\ldots,j^*+j-1\}j^*$ and $\pi^1\{j^*+2,\ldots,j^*+j-1\}(j^*+1)$, for $j=k$. This involves the solution of one or two problems of size at most $(k-1)$ (in $\ostar{2.4143^{k}}$ time by TTBR2) and the finding of the 
sequence that has the smallest total tardiness value knowing that both sequences start at time 0.
\end{itemize}

The {\ttfamily LEFT\_MERGE} procedure is presented in Algorithm \ref{algo:leftmerge}. Notice that, from a technical point of view, this algorithm takes as input one problem and produces as an output its first child node to branch on, which replaces all its $k$ left-side children nodes.

\begingroup
\thickmuskip=0mu
\medmuskip=0mu
\begin{algorithm}
\caption{{\ttfamily LEFT\_MERGE} Procedure}\label{algo:leftmerge}
\begin{algorithmic}[1]
\Require $P_{\sigma}$ an input problem of size $n$, with $\ell_{b}$, $j^*$ accordingly computed  
\Ensure $Q:$ a list of problems to branch on after merging
\Function{{\ttfamily LEFT\_MERGE}}{$P_{\sigma}$}
\State $Q \gets \emptyset$

\For{$j=1$ to ${k}$}
    \State Create $P_{\sigma^j}$ ($j$-th child of $P_\sigma$) by the improved branching with the subproblem induced by the job set $\{j^*+1,\ldots,j^*+j-1\}$ solved if $\mathcal{L}_{\sigma} =\emptyset$ or $j=k$
\EndFor
\For{$j=1$ to ${k-1}$}
    \State Create $P_{{\sigma^1}^j}$ ($j$-th child of $P_{\sigma^1}$) by the improved branching with the subproblem induced by the job set $\{j^*+2,\ldots,j^*+j-1\}$ solved if $\mathcal{L}_{\sigma^1} =\emptyset$ or $j=k$
\EndFor
\State $\mathcal{L}_{\sigma^1} \gets \emptyset$
\For{$j=1$ to ${k-1}$}
    \State $\mathcal{L}_{\sigma^1} \gets \mathcal{L}_{\sigma^1} \cup \texttt{BEST}(P_{\sigma^{j+1}},P_{{\sigma^1}^j})$
\EndFor
\State $Q \gets Q \cup P_{\sigma^1}$ 
 \State \Return $Q$
 \EndFunction
\end{algorithmic}
\end{algorithm}\endgroup


\begin{lemma} \label{lemma:leftmerge}
The {\ttfamily LEFT\_MERGE} procedure returns one node to branch on in $\bigo{n}$ time and polynomial space. The corresponding problem is of size $(n-1)$.
\end{lemma}
\begin{proof}
The creation of problems $P_{\sigma^1,\ell_{b}+j-1}$, $\forall j=2,\ldots,k$, can be done in $\bigo{n}$ time. The call of TTBR2 costs constant time. The \texttt{BEST} function called at line 8 consists in computing then comparing the total tardiness value of two known sequences of jobs starting at the same time instant: it runs in $\bigo{n}$ time. The overall time complexity of {\ttfamily LEFT\_MERGE} procedure is then bounded by $\bigo{n}$ 
time as $k$ is a constant. 
Finally, as only node $P_{\sigma^1}$ is returned, its size is clearly $(n-1)$ when $P_{\sigma}$ has size $n$. 
\end{proof}

In the final part of this section, we discuss the extension of the algorithm in the case where $LPT\neq EDD$. 
In this case, Property \ref{dec1} 
allows to discard subproblems associated to branching in some positions. 
Notice that if a problem $P$ can be discarded according to this property, then we say that $P$ does not exist and its associated node is empty. 

\begin{lemma}
\label{lemma:LMWC}
When left merging is incorporated into the branching scheme of TTBR, instances such that $LPT=EDD$ correspond to worst-case instances. 
\end{lemma}
\begin{proof}
We prove the result by showing that the time reduction obtained from left merging and Property \ref{dec1} in the case LPT=EDD is not greater than that of any other cases. Let us consider the improved branching scheme. 
The following exhaustive cases hold:
\begin{enumerate}
\item $1=[1]$ and $2=[2]$;
\item $1= [j]$ with $j \geq 2$;
\item $1=[1]$ and $2=[j]$ with $j\geq 3$.
\end{enumerate}

We first sketch the idea of the proof. For each of the 3 cases above, we analyse the time reduction that can be obtained on one single branching and merging and we show that the reduction corresponding to case 1, which covers the case LPT=EDD, is the smallest among all the 3 cases. In fact, for case 2 and case 3, some nodes are not created due to Property 2, and the resulting time reduction is not less than that of case 1. 


Let $T(n)$ be the time needed to solve an instance of size $n$ in general. From Lemma 2, we can deduce that $T(n)>2T(n-1)$  because for instances with $LPT=EDD$, on each branching, a node of size $(n-1)$ is returned by left merging and another node of size $(n-1)$ exists due to the last child node of $P_\sigma$. This statement is also valid in the worst case if no merging is done, due to the branching scheme. The inequality $T(n)>2T(n-1)$ induces that $T(n)=\omega(2^n)$, which will be used below to prove the lemma.

In order to be general, consider the current node as $P_{\sigma}$, as shown in Figure \ref{fig:ex_merging_g}. The time reduction of the 3 cases are denoted respectively by $TR1$, $TR2$ and $TR3$. We also note $TR_{LPT=EDD}$ the time reduction corresponding to the case $LPT=EDD$.

In case 1, no nodes are eliminated by Property 2, hence, the merging can be done as described for the case LPT=EDD (Figure \ref{fig:ex_merging_g}). Therefore, $TR1=TR_{LPT=EDD}=T(n-2)+T(n-3)+..+T(n-k)$ according to Lemma \ref{lemma:leftmerge} when \texttt{LEFT\_MERGE} is executed.

In case 2, the subproblem of $P_\sigma$ corresponding to branching the longest job on the first position, is eliminated directly by Property \ref{dec1}. Therefore, $TR2\geq T(n-1)$. 

In case 3, let $\ell_b$ be the first free position in $P_\sigma$, as defined in Proposition \ref{propos:psigmaprop}. Some child nodes of $P_{\sigma^1}$, as in Figure \ref{fig:ex_merging_g}, corresponding to branch job 2 on positions $\{\ell_b+1,..,\ell_b+j-1\}$, are eliminated due to Property \ref{dec1}. For these nodes, the time reduction that could have been achieved by merging is already ensured, while the nodes that are not eliminated, notably those corresponding to branch job 2 in positions $\{\ell_b+j,..,\ell_b+k-1\}$ can still be merged pairwise with nodes $\{P_{\sigma^{j+1}},..,P_{\sigma^k}\}$. More reduction can be gained if $j>k$. 
Therefore, $TR3\geq T(n-2)+T(n-3)+..+T(n-k)$. 

Since $TR1\leq TR3$, this brings us to compare $TR1$ and $TR2$. Suppose $TR1>TR2$, i.e., $T(n-1)<T(n-2)+T(n-3)+..+T(n-k)$, then we have $T(n-1)<T(n-2)+T(n-3)+..+T(1)$. By solving this recurrence relation we get $T(n)=o(2^n)$ which is in contradiction with the fact that $T(n)=\omega(2^n)$, as proved above. Therefore, $TR1< TR2$, i.e.,  on each recursion of the algorithm, the time reduction obtained in case 1 is not greater than any other cases. 
Since $TR_1=TR_{LPT=EDD}$, this proves that $LPT=EDD$ is the worst-case scenario, in which the {\ttfamily LEFT\_MERGE} procedure returns one node of size $n-1$ to branch on.
\end{proof}

\subsubsection{A working example for Left Merge}\label{sec:leftex}
In order to better illustrate the merging operations on left-side branches, a detailed example is provided in this section. The input data is in Table \ref{tab:lmex} and the value of $k$ is chosen as $2$.

\begin{table}[!ht]
\centering
\begin{tabular}{|c|c|c|c|c|}\hline
$i$ & 1 & 2 &3&4 \\ \hline
$p_i$& 7&4&2&1\\ \hline
$d_i$&2&4&6&8\\ \hline
\end{tabular}
\caption{A sample instance}\label{tab:lmex}
\end{table}

The solution of the instance is depicted in Figures \ref{subfig:leftex1}, \ref{subfig:leftex2},\ref{subfig:leftex3},\ref{subfig:leftex4},\ref{subfig:leftex5}. On each node, the jobs to schedule are surrounded by accolades and the indicated partial total tardiness value ($tt=\sum T_j$) is computed on the jobs that are fixed before the unscheduled jobs.
The applied algorithm is TTBR1 but with \texttt{LEFT\_MERGE} integrated. 
The algorithm runs in depth-first. 
\texttt{LEFT\_MERGE} is firstly called on the root node $P$. Lines 1-8 in Algorithm \ref{algo:leftmerge} results the situation in Figure \ref{subfig:leftex1}. This firstly (lines 3-5) involves the generation of the first $k$ child nodes of $P$ which are $P_1$ and $P_2$, with $\mathcal{L}=\{P_1\}$, and the computation of total tardiness of the fixed job sequences. For instance, the sequence $(2,1)$ in $P_2$ has a total tardiness of $9$. 
Similarly, by lines 6-8, the first $(k-1)$ child nodes of $P_1$ are generated and the total tardiness of partially fixed jobs is also computed. The sequence $(1,2)$ in $P_{1,2}$ has a total tardiness of 12. 
Then, by comparing the {total tardiness} of the fixed parts of $P_2$ and $P_{1,2}$, line 11 merges node $P_2$ to the position of $P_{1,2}$ by setting the child list of $P_1$ $\mathcal{L}_1=\{P_2\}$, since $9<11$ (see Figure \ref{subfig:leftex2}). The result of merging is $P_2$ which is renamed as $P_{\sigma^{1,2}}$ according to our notation. Line 14 returns the next node to open, which is $P_1$. This completes the first call of \texttt{LEFT\_MERGE}.

Then the algorithm continues in the same way by applying \texttt{LEFT\_MERGE} on the node $P_1$. In a similar way,  the node $P_{1,3}$ and the first child node of $P_{\sigma^{1,2}}$ are generated and their partial total tardiness values are computed. 
By comparing the partial total tardiness of these two nodes, $P_{1,3}$ is cut since $17>16$ (Figure \ref{subfig:leftex3}). This also explains why we decide to \textit{merge} two nodes instead of just cutting the dominated one. Actually, by putting the merged node in a specific position, the node can further participate in subsequent merging operations. More concretely, if $P_2$ was not moved to the position of $P_{1,2}$, the merging in Figure \ref{subfig:leftex3} would not happen.

The algorithm continues its exploration in depth-first order, each time a leaf node is reached, the algorithm updates the current best solution. Another merging is done between $P_{4,2}$ and $P_{4,1,2}$, leading to prune $P_{4,2}$ (Figure \ref{subfig:leftex4}). Then, the algorithm goes on until all leaf nodes are explored and it returns the optimal solution given by $P_{4,1,2,3}$, corresponding to the sequence of $(2,3,4,1)$ (Figure \ref{subfig:leftex5}). 

Notice that some nodes involved in merging have only one job to schedule, like $P_{4,2}$, but we still consider them as a non-solved subproblems in the exemple in order to illustrate the merging. Also notice that the \texttt{LEFT\_MERGE} procedure should only be called on subproblems of enough size (greater than $k$). All the three nodes $P$, $P_1$ and $P_4$ have more than 2 jobs to schedule.

\FloatBarrier
\begin{figure}
\centering
\begin{subfigure}[!ht]{\textwidth}\centering
\includegraphics[width=0.5\columnwidth]{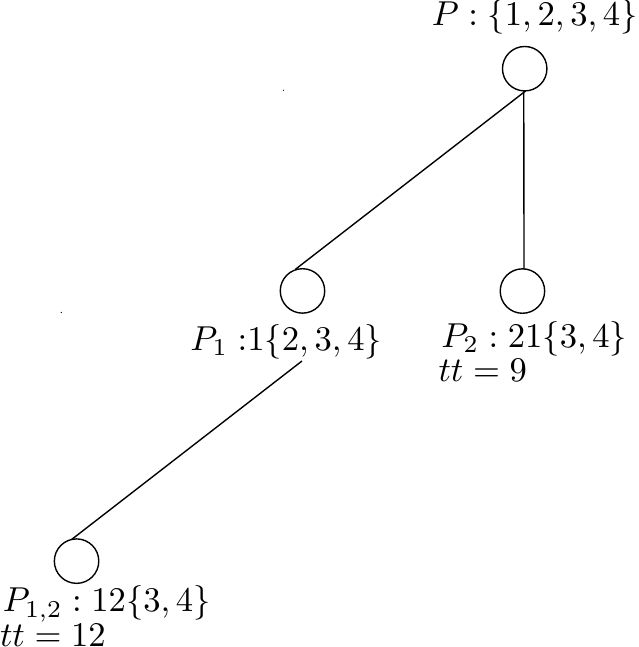}\caption{}
\label{subfig:leftex1}
\end{subfigure}
\end{figure}

\begin{figure}
\ContinuedFloat
\begin{subfigure}[!ht]{\textwidth}\centering
\includegraphics[width=0.5\columnwidth]{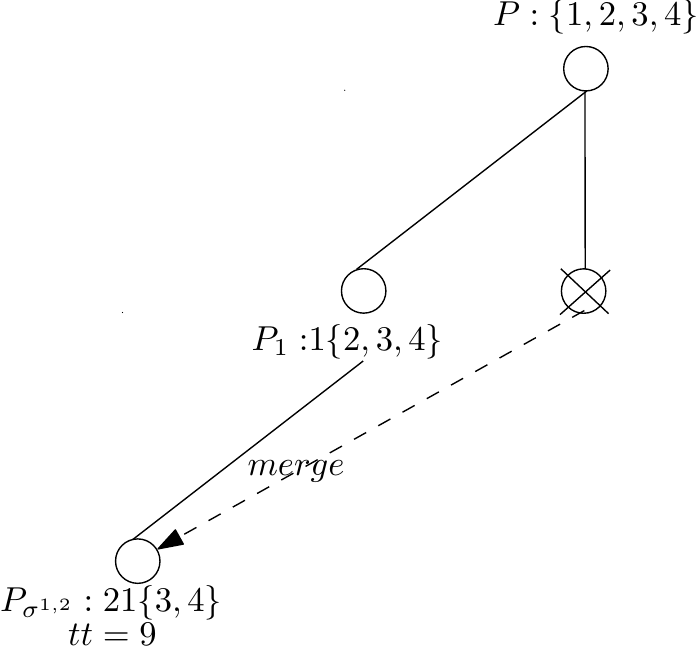}\caption{}
\label{subfig:leftex2}
\end{subfigure}
\end{figure}

\begin{figure}
\ContinuedFloat
\begin{subfigure}[!ht]{\textwidth}\centering
\includegraphics[width=0.5\columnwidth]{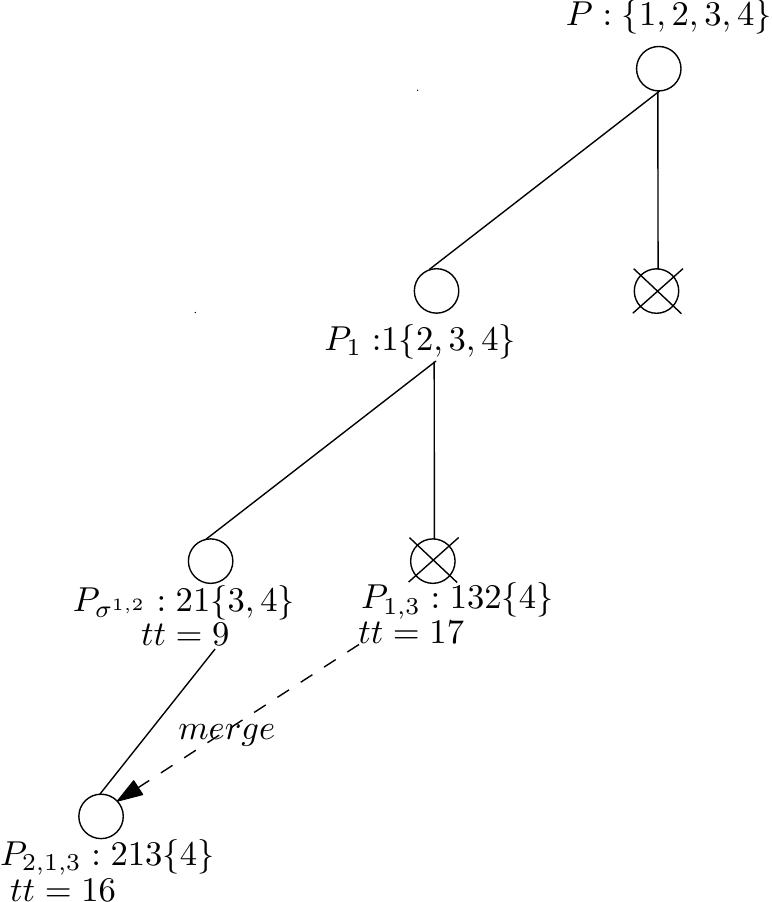}\caption{}
\label{subfig:leftex3}
\end{subfigure}
\end{figure}

\begin{figure}
\ContinuedFloat
\begin{subfigure}[!ht]{\textwidth}\centering
\includegraphics[width=\columnwidth]{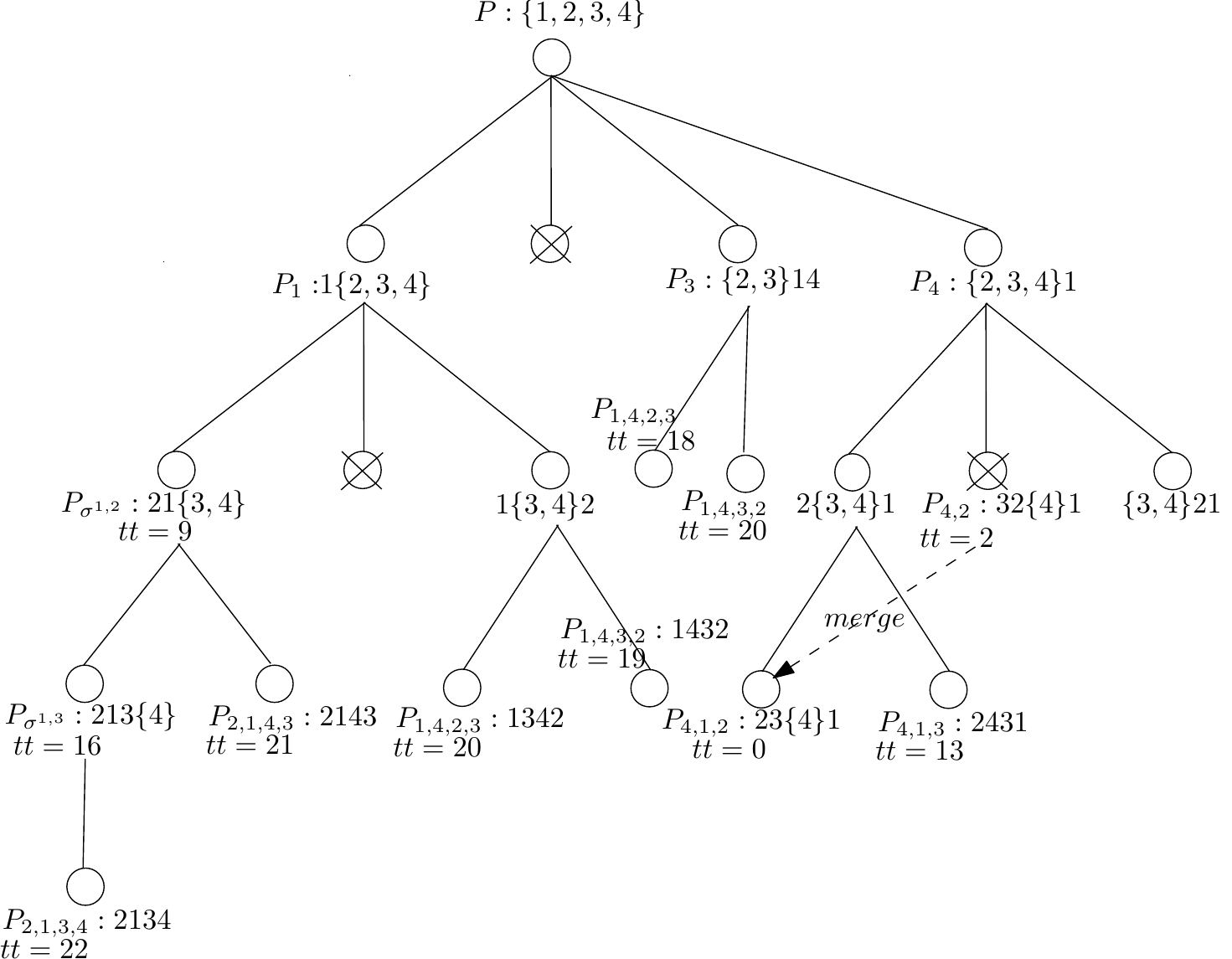}\caption{}
\label{subfig:leftex4}
\end{subfigure}
\end{figure}

\begin{figure}
\ContinuedFloat
\begin{subfigure}[!ht]{\textwidth}\centering
\includegraphics[width=\columnwidth]{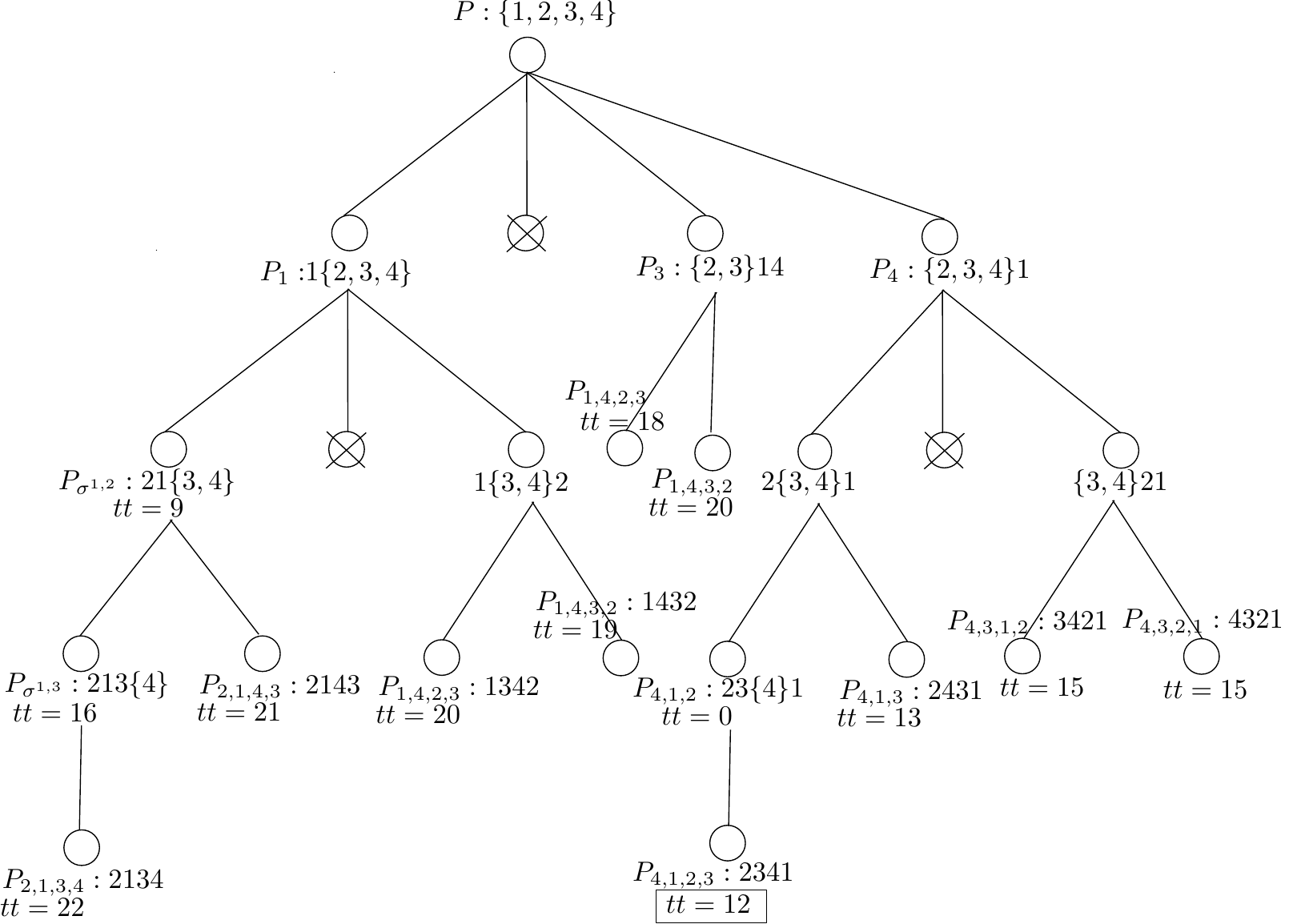}\caption{}
\label{subfig:leftex5}
\end{subfigure}
\caption{A sample instance solved by TTBR1 with left-merging integrated}
\label{fig:leftex}
\end{figure}
\FloatBarrier
\FloatBarrier
\subsection{Merging right-side branches}~\label{sec:rightsection}

Due to the branching scheme, the merging of right-side branches involves a more complicated procedure than the merging of left-side branches. In the merging of left-side branches, it is possible to merge some nodes associated to problems $P_\ell$ with children nodes of $P_1$, while for the right-side branches, it is not possible to merge some nodes $P_\ell$ with children nodes of $P_n$. To see this, consider the situation depicted in Figure \ref{fig:right_observe}. If the merging can be applied in a symmetric way as on left-size branches, then the node $P_{n-1}$ and the node $P_{n,n-1}$ should be concerned by a merging operation. However, due to Property \ref{dec1}, these two nodes do not share any common subproblem. Therefore, they can not be merged.

\begin{figure}[!ht]
\centering
\includegraphics[width=0.8\columnwidth]{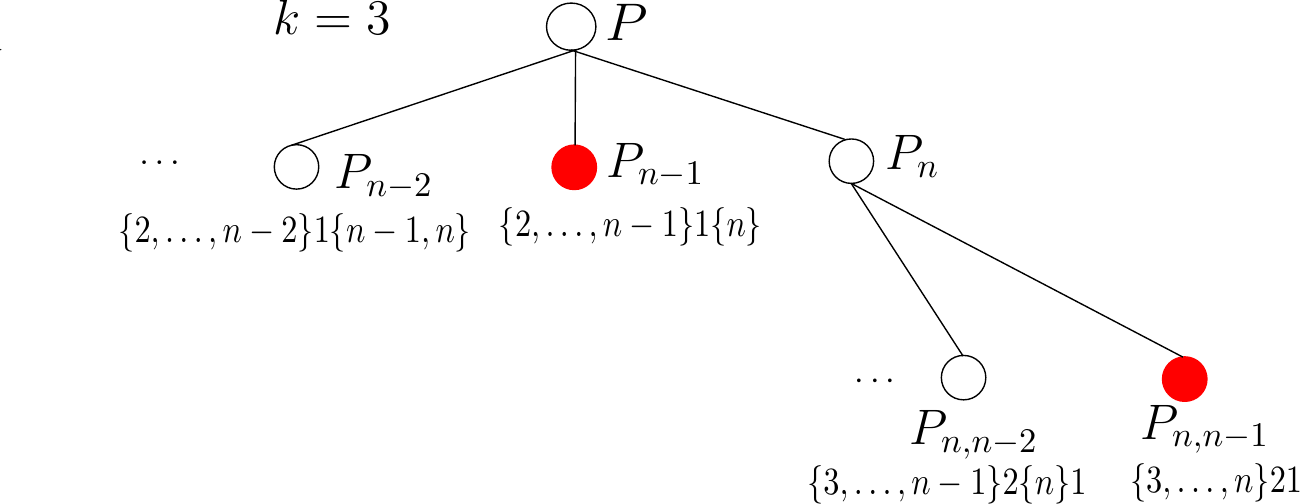}
\caption{An failed attempt on merging right-side branches}
\label{fig:right_observe}
\end{figure}

However, another structure of merging exist: the children nodes of $P_\ell$ can be merged  with the children nodes of $P_n$. Let us more formally introduce the right merging procedure and, again, let $k\leq \frac{n}{2}$ be the same constant parameter as used in the left merging. 

Figure~\ref{fig:rk3_abstract_diffshapes} shows an example on the structure of merging for the $k$ right-side branches with $k=3$. The root problem $P$ consists in scheduling job set $\{1,\ldots,n\}$. Unlike left-side merging, the right-side merging is done horizontally for each level. 
Nodes that are involved in merging are colored.
For instance, the black square nodes at level 1 can be merged together. This concerns three nodes: $P_{n-2,n-3}$, $P_{n-1,n-3}$, $P_{n,n-3}$. As indicated in the figure, all these three nodes contain an identical subproblem which consists in scheduling jobs $\{3,..,n-2\}$ from time $0$. Also notice that apart from this subproblem, the number of remaining jobs to schedule is less than $k$. Therefore, these small number of jobs are fixed first, which then enables the merging of all the three nodes. The result of the merging should take the position of the right most node, which is $P_{n,n-3}$.  
Similarly, the black circle nodes at level 1 can be merged,
the grey square nodes at level 2 can be merged and
the grey circle nodes at level 2 can be merged.  
Notice that each right-side branch of $P$ is expanded to a different depth which depends on the branching decision: for the purpose of right merging, a branch is no more expanded whenever the largest subproblem of the first child node is of size $(n-k-1)$. For example,  in Figure \ref{fig:rk3_abstract_diffshapes}, consider node $P_{n-2}$ and its first child node $P_{n-2,1}$ which defines the subproblem $2\{3,..,n-2\}1\{n-1,n\}$. The largest subproblem in $P_{n-2,1}$ conerns jobs $\{3,..,n-2\}$ which has $(n-4)=(n-k-1)$ jobs. Then, the expansion of the branch rooted by $P_{n-2,1}$ is stopped. Notice that in the implementation of the right merge, it is not mandatory to have the whole expansion of the search tree as indicated in Figure \ref{fig:rk3_abstract_diffshapes}. But, this expansion eases the presentation of the mechanism and the analysis of the final complexity.



\begin{figure}[!ht]
\centering
\includegraphics[width=\columnwidth]{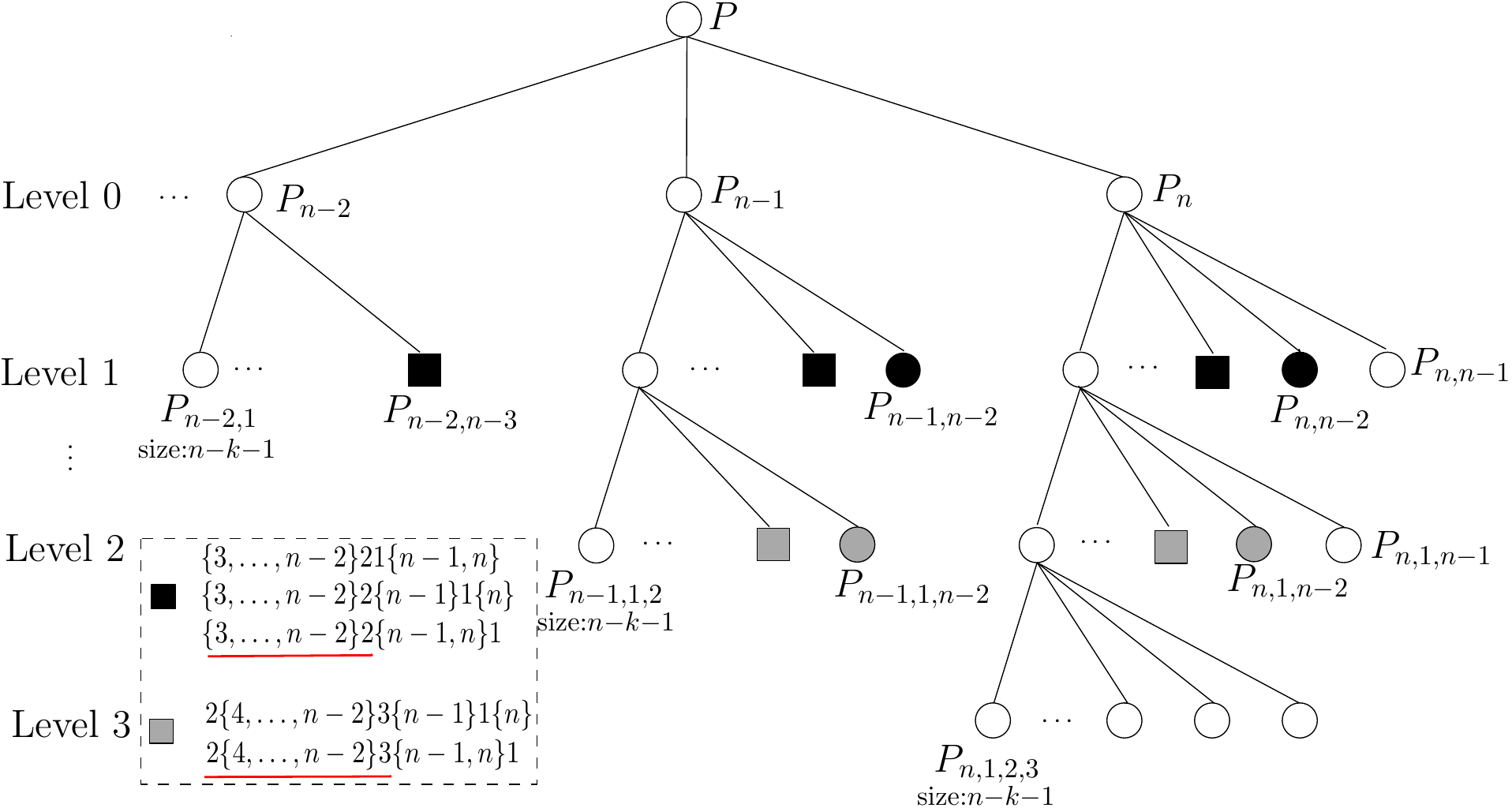}
\caption{An example of right-side branches merging for $k=3$}
\label{fig:rk3_abstract_diffshapes}
\end{figure}

More generally, Figure~\ref{fig:rm1} shows the right-side search tree and the content of the nodes involved in the merging for an arbitrary value of $k$. \\

\begin{figure}[!ht]
\begingroup
\thickmuskip=0mu
\medmuskip=0mu
\centering
\includegraphics[width=\columnwidth]{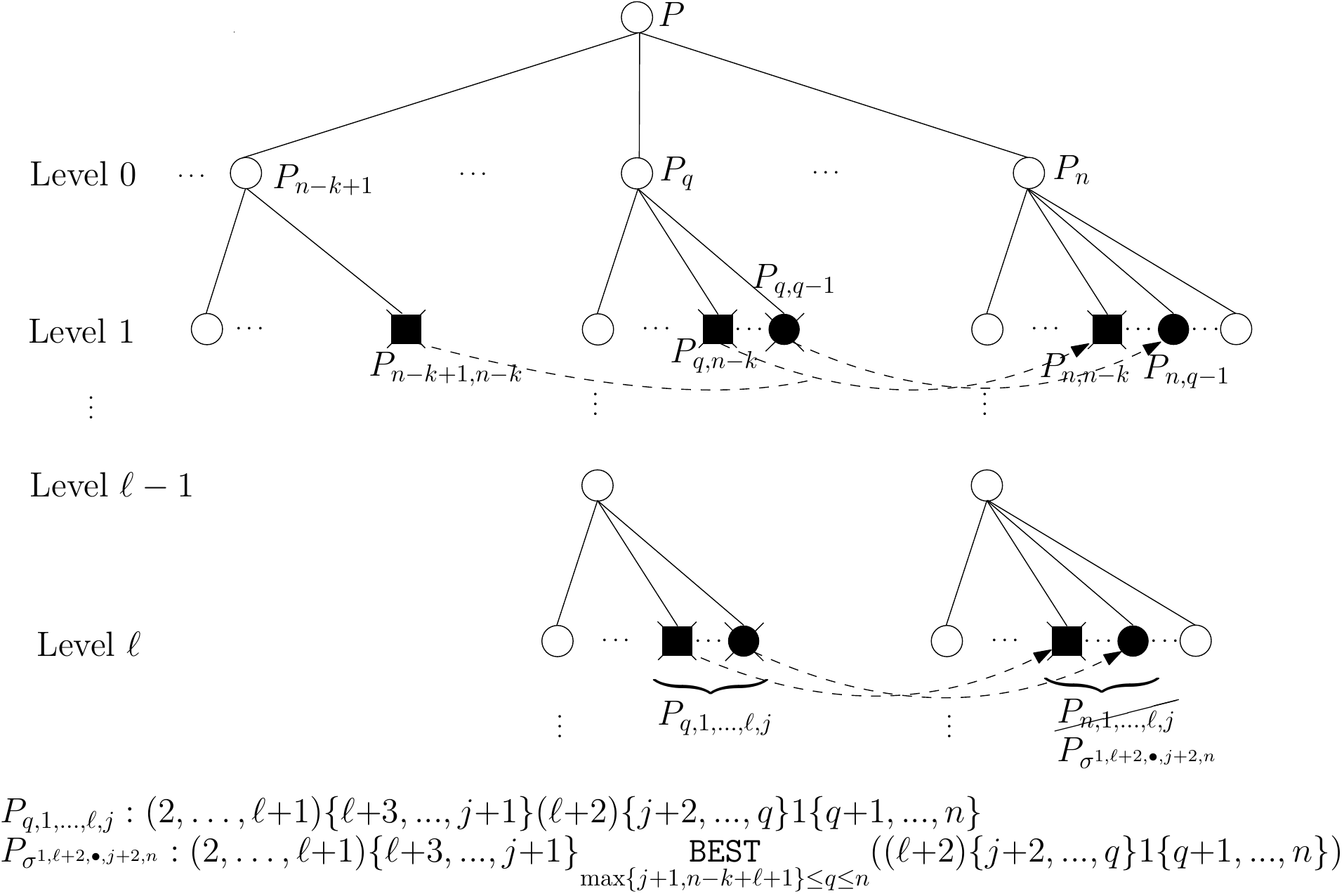}
\caption{Generic right-side merging at the root node}
\label{fig:rm1}
\endgroup
\end{figure}

The rest of this section intends to describe the merging by following the same lines as for left merging. We first extend the notation $P_\sigma$ in the sense that $\sigma$ 
may now contain placeholders. The $i$-th element of $\sigma$ is either the position assigned to job $i$ if $i$ is fixed, or $\bullet$ if job $i$ is not yet fixed. 
The $\bullet$ sign is used as placeholder, with its cardinality below indicating the number of consecutive $\bullet$. As an example, the problem \texttt{$\{2,\ldots,n-1\}1n$} can now be denoted by $P_{n-1,\phn{n-2},n}$. 
The cardinality of $\bullet$ may be omitted whenever it is not important for the presentation or it can be easily deduced as in the above example. Note that this adapted notation eases the presentation of right merge while it has no impact on the validity of the results stated in the previous section.


\begin{propos}\label{propos:rightstruct}
Let $P_{\sigma}$ be a problem to branch on.  
Let $j^*, \ell_b, \ell_e,\rho_1$ and $\rho_2$ be defined as in Proposition~\ref{propos:psigmaprop}. Extending Corollary~\ref{coro}, problem $P_\sigma$ has the following structure:
$$ {\pi}\{j^*,\ldots,j^*+\ell_{e}-\ell_{b}\}\gamma\Omega'$$ 
where $\pi$ is defined as in Corollary~\ref{coro} and ${\gamma}$ is the sequence of jobs in positions $\rho_2,\ldots,\rho_3$ with $\rho_3=\max\{i: i\geq\rho_2, positions\text{ $\rho_2,\ldots,i$ are in $\sigma$}\}$ and $\Omega'$ the remaining subset of jobs to be scheduled after position $\rho_3$ (some of them can have been already scheduled). The merging procedure is applied on job set $\{j^*,\ldots,j^*+\ell_{e}-\ell_{b}\}$ preceded by a sequence of jobs $\pi$ and followed by $\gamma \Omega'$.
\end{propos}
\begin{proof}
The problem structure stated in Corollary~\ref{coro} is refined on the part of $\Omega$. $\Omega$ is split into two parts: $\gamma$ and $\Omega'$ (see Figure \ref{fig:propro5}). 
The motivation is that $\gamma$ will be involved in the right merging, just like the role of $\pi$ in left merging.
\end{proof}

\begin{figure}
\centering
 \includegraphics[width=0.7\columnwidth]{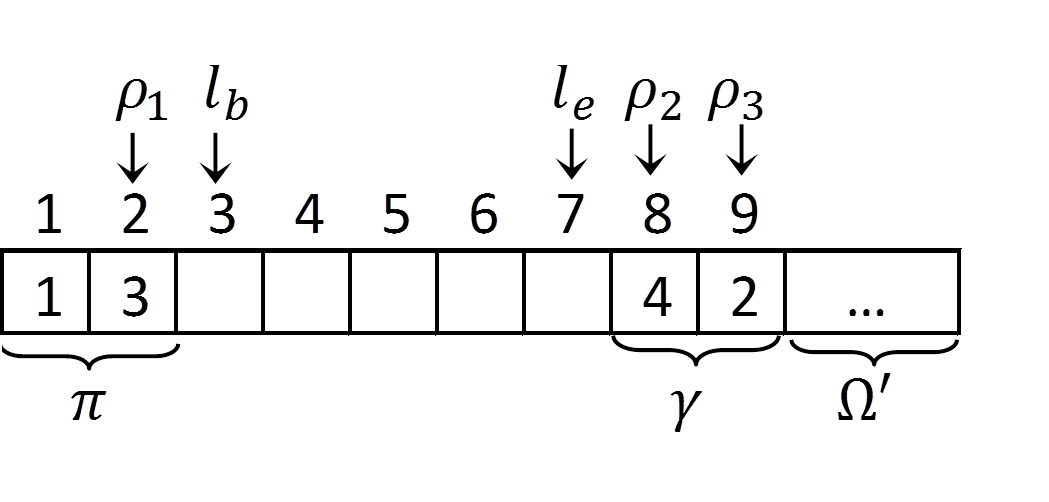}
  \caption{An example ($P_{1,9,2,8}$) for Proposition \ref{propos:rightstruct}}\label{fig:propro5}
\end{figure}

Proposition~\ref{propos:RMLink} shed lights on how to merge the right side branches originated from the root node. We first introduce Definition \ref{def:slj} which will be used to identify a set of nodes to merge.

\begin{definition}\label{def:slj}
Let $\mathcal{S}_{\ell,j}$ be a set of nodes, defined as
$$\mathcal{S}_{\ell,j} = \left\{P_\sigma:~~ \parbox{5cm}{$|\sigma|=\ell+2$,\\ 
$ \max\{j+1,n-k+\ell+1\}\leq \sigma_1\leq n$,\\ 
$\sigma_i=i-1$, $\forall i\in\{2,\ldots,\ell+1\},\\ $
$\sigma_{\ell+2}=j$}
\right\}\footnote{Placeholders do not count in the cardinality of $\sigma$}$$

with $0\leq \ell \leq k-1$, $n-k\leq j \leq n-1$, and with $\sigma_i$ referring to the position of job $i$ in $\sigma$. 
\end{definition}

{To understand the proposition intuitively, the reader can always refer to Figure \ref{fig:rk3_abstract_diffshapes} and Figure \ref{fig:rm1}. Each set $\mathcal{S}_{\ell,j}$ represents a set of nodes with the same color and shape in Figure \ref{fig:rk3_abstract_diffshapes}, that is, they can be merged together. The subscript $\ell$ in $\mathcal{S}_{\ell,j}$ indicates the level of nodes, more precisely the number of jobs branched on in $\pi$, and the subscript $j$ determines the ``shape'' of nodes at that level, more precisely the position of the first job in $\gamma$. As an example, all black squares in Figure \ref{fig:rk3_abstract_diffshapes} are in $\mathcal{S}_{0,n-3}=\{P_{n,n-3},P_{n-1,n-3},P_{n-2,n-3}\}$.}

We now introduce Proposition \ref{propos:RMLink}.
\begingroup
\thickmuskip=0mu
\medmuskip=0mu
\begin{propos}\label{propos:RMLink}
For each problem in the set $\mathcal{S}_{\ell,j}$, we have the two following properties:
\begin{enumerate}
\item The solution of problems in $\mathcal{S}_{\ell,j}$ involves the solution of a common subproblem which consists in scheduling job set $\{\ell+3,..., j+1\}$ starting at time $t_{\ell} = \sum_{i=2}^{\ell+1} p_i$.
\item For any problem in $\mathcal{S}_{\ell,j}$, at most $k+1$ jobs have to be scheduled after job set $\{\ell+3,..., j+1\}$.
\end{enumerate}
\end{propos}
\begin{proof}

As each problem $P_{\sigma}\in \mathcal{S}_{\ell,j}$ is defined by  $(2,\ldots,\ell+1)\{\ell+3,...,j+1\}(\ell+2)\{j+2,...,\sigma_1\}1\{\sigma_1+1,...,n\}$, the first part of the proposition is straightforward.\\
Besides, the second part can be simply established by counting the number of jobs to be scheduled after
job set $\{\ell+3,...,j+1\}$ when $j$ is minimal, {\it i.e.} when $j=n-k$. In this case, $(\ell+2)\{j+2,...,\sigma_1\}1\{\sigma_1+1,...,n\}$ contains $k+1$ jobs. 
\end{proof}
\endgroup

The above proposition highlights the fact that some nodes can be merged as soon as they share the same initial subproblem to be solved. More precisely, at most $(k-\ell-1)$ nodes associated to problems $P_{q,1..\ell,j}$, $\max\{j+1,n-k+\ell+1\}\leq q\leq (n-1)$, can be merged with the node associated to problem $P_{n,1..\ell,j}$, $\forall j=(n-k),..., (n-1)$. The node $P_{n,1..\ell,j}$ is replaced in the search tree by the node $P_{\sigma^{1,\ell+2,\bullet, j+2,n}}$ defined as follows (Figure~\ref{fig:rm1}):
\begin{itemize}
\item $\{\ell+3,...,j+1\}$ is the set of jobs on which it remains to branch.
\item Let $\sigma^{1,\ell+2,\bullet, j+2,n}$ be the sequence containing positions of jobs $\{1,\ldots,\ell+2,j+2,\ldots,n\}$ and placeholders for the other jobs, that leads to the best jobs permutation among $(\ell+2)\{j+2,...,q\}1\{q+1,...,n\}$, $\max\{j+1,n-k+\ell+1\}\leq q\leq n$. 
This involves the solution of at most $k$  
problems of size at most $k+1$ (in $\ostar{k\times 2.4143^{k+1}}$ time by TTBR2) and the determination of the best of the computed sequences knowing that all of them start at time $t$, namely the 
sum of the jobs processing times in $(2,\ldots,\ell+1)\{\ell+3,...,j+1\}$.
\end{itemize}

The merging process described above is applied at the root node, while an analogous merging can be applied at any node of the tree.  
With respect to the root node, the only additional consideration is that the right-side branches of a general node may have already been modified by previous merging operations. As an example, let us consider Figure~\ref{fig:aftermerge}. 
It shows that, due  to the merging operations performed from $P$, the right-side branches of $P_n$ may not be the subproblems induced by the branching scheme. However, it can be shown in a similar way as for left-merge, that the merging can still be applied. 

\begin{figure}[!ht]
\centering
\includegraphics[width=0.7\columnwidth]{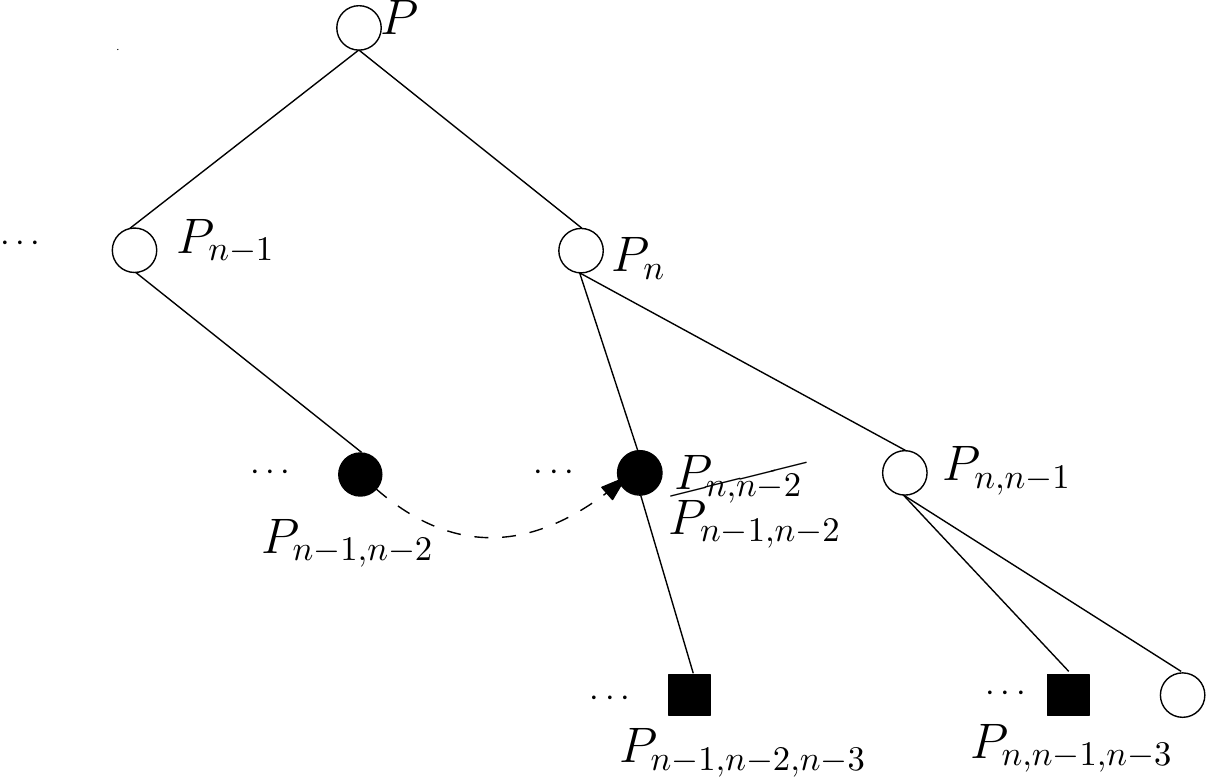}
\caption{The right branches of $P_n$ were modified by the right-merging from $P$}
\label{fig:aftermerge}
\end{figure}

In order to define the branching scheme used with the \texttt{RIGHT\_MERGE} procedure, a data structure $\mathcal{R}_{\sigma}$ is associated to a problem $P_\sigma$. It represents a list of subproblems that result from a previous merging and are now the $k$ right-side children nodes of $P_\sigma$. 
When a merging operation sets the $k$ right-side children nodes of $P_\sigma$ to $P_{\sigma^{n-k+1}},...,P_{\sigma^{n}}$, we set $\mathcal{R}_{\sigma} = \{P_{\sigma^{n-k+1}},...,P_{\sigma^{n}}\}$, otherwise we have $\mathcal{R}_{\sigma} = \emptyset$. As a conclusion, the following branching scheme for an arbitrary node of the tree is defined. It is an extension of the branching scheme introduced in Definition~\ref{def:leftscheme} for the left merging.

\begin{definition}
\label{def:lrscheme}
The branching scheme for an arbitrary node $P_\sigma$ is defined as follows:
\begin{itemize}
\item If $\mathcal{R}_{\sigma} = \emptyset$, use the branching scheme defined in Definition~\ref{def:leftscheme};
\item If $\mathcal{L}_{\sigma} = \emptyset$ and $\mathcal{R}_{\sigma} \neq \emptyset$, branch on the longest job in the available positions from the $1$st to the $(n-k)$-th, then extract problems from $\mathcal{R}_{\sigma}$ as the last $k$ branches.
\item If $\mathcal{L}_{\sigma} \neq \emptyset$ and $\mathcal{R}_{\sigma} \neq \emptyset$, extract problems from $\mathcal{L}_{\sigma}$ as the first $(k-1)$ branches, then branch on the longest job in the available positions from the $k$-th to the $n-k$-th, finally extract problems from $\mathcal{R}_{\sigma}$ as the last $k$ branches.
\end{itemize}
This branching scheme generalizes and replaces, the one introduced in Definition~\ref{def:leftscheme} and it will be referred to as \textbf{improved branching} from now on.  
\end{definition}

Proposition~\ref{propos:rg} states the validity of merging a general node, which extends the result in Proposition~\ref{propos:RMLink}.
\begingroup
\thickmuskip=0mu
\medmuskip=0mu
\begin{propos}\label{propos:rg}
Let $P_{\sigma}$ be an arbitrary problem and let $\pi,j^*,\ell_{b},\ell_e,\gamma, \Omega'$ be computed relatively to $P_{\sigma}$ according to Proposition~\ref{propos:rightstruct}. 
If $\mathcal{R}_{\sigma} = \emptyset$, the right merging on $P_\sigma$ can be easily performed by considering $P_\sigma$ as a new root problem. 
Suppose $\mathcal{R}_{\sigma} \neq \emptyset$, 
the $q$-th child node $P_{\sigma^q}$ is extracted from $\mathcal{R}_{\sigma}$, $\forall n'-k+1 \leq q \leq n'$, where $n'=\ell_e-\ell_b+1$ is the number of children nodes of $P_{\sigma}$. The structure of $P_{\sigma^q}$ is $\pi\{j^*+1,...,j^*+q-1\}\gamma^{q}\Omega'$.\\

For $0\leq \ell \leq k-1$ and $n'-k\leq j \leq n'-1$, the following conditions hold:
\begin{enumerate}
\item Problems in $\mathcal{S}_{\ell,j}^\sigma$ have the following structure:\\
$\pi(j^*+1,\ldots,j^*+\ell)\{j^*+\ell+2,...,j^*+j\}(j^*+\ell+1)\{j^*+j+1,...,j^*+q-1\}\gamma^{q}\Omega'$ with $q$ varies from $\max\{j+1,n-k+\ell+1\}$ to $n'$. 
\item The solution of all problems in $\mathcal{S}_{\ell,j}^\sigma$ involves the scheduling of a job set $\{j^*+j+1,...,j^*+q-1\}$, $\max\{j+1,n-k+\ell+1\}\leq q \leq n'$, which is of size less than $k$. Besides, for all problems in $\mathcal{S}_{\ell,j}^\sigma$ it is required to solve a common subproblem made of job set $\{j^*+\ell+2,...,j^*+j\}$ starting after $\pi(j^*+1,\ldots,j^*+\ell)$ and before $(j^*+\ell+1)\{j^*+j+1,...,j^*+q-1\}\gamma^{q}\Omega'$. 
\end{enumerate}
\end{propos}
\begin{proof} 
The proof is similar to the one of Proposition~\ref{propos:LMLinkLeft}. 
The first part of the statement follows directly from Definition~\ref{def:lrscheme} and simply defines the structure of the children nodes of $P_{\sigma}$. 
For the second part, it is necessary to prove that $\{j^*+j+1,...,j^*+q-1\}\gamma^{q}$ consists of the same jobs for any valid value of $q$. Actually, since right-merging  only merges nodes that have common jobs fixed after the unscheduled jobs, the jobs present in $\{j^*+j+1,...,j^*+q-1\}\gamma^{q}$ and the jobs present in $\{j^*+j+1,...,j^*+q-1\}j^*\{j^*+q,...,j^*+n'-1\}\gamma$, $\max\{j+1,n-k+\ell+1\} \leq q \leq n'$, must be the same, which proves the statement.
\end{proof}

Analogously to the scenario at the root node, given the values of $\ell$ and $j$, all the problems in $\mathcal{S}_{\ell,j}^\sigma$ can be merged. More precisely, we rewrite $\sigma$ as $\alpha\phn{n'}\beta$ where $\alpha$ is the sequence of positions assigned to jobs $\{1,\ldots,j^*-1\}$, $\phn{n'}$ refers to the job set to branch on and $\beta$ contains the positions assigned to the rest of jobs. At most ${(k-\ell-1)}$ nodes associated to problems $P_{\alpha,\ell_b+q-1,\ell_b..\ell_b+\ell-1,\ell_b+j-1,\ph,\beta}$, with $\max\{j+1,n'-k+\ell+1\}\leq q\leq n'-1$, can be merged with the node associated to problem $P_{\alpha,\ell_e,\ell_b..\ell_b+\ell-1,\ell_b+j-1,\ph,\beta}$. 

Node $P_{\alpha,\ell_e,\ell_b..\ell_b+\ell-1,\ell_b+j-1,\ph,\beta}$ is replaced in the search tree by node $P_{\alpha,\sigma^{\ell,\ell_b,j},\ph,\beta}$ defined as follows:
\begin{itemize}
\item $\{j^*+\ell+2,...,j^*+j\}$ is the set of jobs on which it remains to branch. 
\item Let $\sigma^{\ell,\ell_b,j}$ be the sequence of positions among $$\{(\ell_b+q-1,\ell_b..\ell_b+\ell-1,\ell_b+j-1)\ :\ \max\{j+1,n'-k+\ell+1\}\leq q\leq n'-1\}$$ associated to the best permutation of jobs in $(j^*+\ell+1)\{j^*+j+1,...,j^*+q-1\}\gamma^{q}$, $\forall \max\{{j+1},n'-k+\ell+1\} \leq q\leq n'$. 
This involves the solution of $k$ problems of size at most $k+1$ (in $\ostar{k\times 2.4143^{k+1}}$ time by TTBR2) 
and the determination of the best of the computed sequences
knowing that all of them start at time $t$, i.e. the 
sum of the jobs processing times in  $\pi(j^*+1,\ldots,j^*+\ell)\{j^*+\ell+2,...,j^*+j\}$.

\end{itemize}\endgroup

The {\ttfamily RIGHT\_MERGE} procedure is presented in Algorithm \ref{algo:rightmerge}. Notice that, similarly to the {\ttfamily LEFT\_MERGE} procedure, this algorithm takes as input one problem $P_{\sigma}$ and provides as an output a set of nodes to branch on, which replaces all its $k$ right-side children nodes of $P_\sigma$. It is important to notice that the {\ttfamily LEFT\_MERGE} procedure is also integrated.

A procedure {\ttfamily MERGE\_RIGHT\_NODES} (Algorithm \ref{algo:mrn}) is invoked to perform the right merging for each level $\ell=0,...,k-1$ in a recursive way. The initial inputs of this procedure (line 13 in {\ttfamily RIGHT\_MERGE}) are the problem $P_{\sigma}$ and the list of its $k$ right-side children nodes, denoted by \textit{rnodes}. They are created according to the improved branching (lines 4-12 of Algorithm~\ref{algo:rightmerge}). 
Besides, the output is a list $Q$ containing the problems to branch on after merging. 
In the first call to {\ttfamily MERGE\_RIGHT\_NODES}, the left merge is applied to the first element of \textit{rnodes} (line 2), all the children nodes of nodes in \textit{rnodes} not involved in right nor left merging, are added to $Q$ (lines 3-7). This is also the case for the result of the right merging operations at the current level (lines 8-11). 
In Algorithm~\ref{algo:mrn}, the value of $r$ indicates the current size of \textit{rnodes}. It is reduced by one at each recursive call and the value $(k-r)$ identifies the current level with respect to $P_{\sigma}$. As a consequence, each right merging operation consists in finding the problem with the best total tardiness value on its fixed part, among the ones in set $\mathcal{S}^\sigma_{k-r,j}$. This is performed by the \texttt{BEST} function (line 10 of {\ttfamily MERGE\_RIGHT\_NODES}) which extends the one called in Algorithm~\ref{algo:leftmerge} by taking at most $k$ subproblems as input and returning the dominating one. 

The {\ttfamily MERGE\_RIGHT\_NODES} procedure is then called recursively 
on the list containing the first child node of the $2$nd to $r$-th node in \textit{rnodes} (lines 13-17). Note that the procedure {\ttfamily LEFT\_MERGE} is applied to every node in \textit{rnodes} except the last one. 
In fact, for any specific level, the last node in \textit{rnodes} belongs to the last branch of $P_{\sigma}$, which is $P_{\sigma,l_b+n-1,\ph,\beta}$. Since $P_{\sigma,l_b+n-1,\ph,\beta}$ is put into $Q$ at line 14 of {\ttfamily RIGHT\_MERGE}, it means that this node will be re-processed later and {\ttfamily LEFT\_MERGE} will be called on it at that moment. 
Since the recursive call of {\ttfamily MERGE\_RIGHT\_NODES} (line 18) will merge some nodes to the right-side children nodes of $P_{\alpha,\ell_b,\underset{n_r-1}{\bullet},\beta^r}$, the latter one must be added to the list $\mathcal{L}$ of $P_{\alpha,\underset{n_r}{\bullet},\beta^r}$ (line 19). In addition, since we defined $\mathcal{L}$ as a list of size either $0$ or $(k-1)$, lines 20-24 add the other $(k-2)$ nodes to $\mathcal{L}_{\alpha,\underset{n_r}{\bullet},\beta^r}$. 

It is also important to notice the fact that a node may have its $\mathcal{L}$ or $\mathcal{R}$ structures non-empty, if and only if it is the first or last child node of its parent node. A direct result is that only one node among those involved in a merging may have its $\mathcal{L}$ or $\mathcal{R}$ non-empty. In this case, these structures need to be associated to the resulting node. The reader can always refer to Figure~\ref{fig:rk3_abstract_diffshapes} for a more intuitive representation. 

\begingroup
\medmuskip=0mu
\begin{algorithm}
\caption{{\ttfamily RIGHT\_MERGE} Procedure}\label{algo:rightmerge}
\begin{algorithmic}[1]
\Require $P_\sigma=P_{\alpha,\phn{n},\beta}$ a problem of size $n$, with $\ell_{b}$, $j^*$ computed according to Proposition~\ref{propos:psigmaprop}  
\Ensure $Q:$ a list of problems to branch on after merging
\Function{{\ttfamily RIGHT\_MERGE}}{$P_{\sigma}$}
\State $Q \gets \emptyset$
\State $nodes \gets \emptyset$
\If {$\mathcal{R}_{\sigma} =\emptyset$}
    \For{$q=n-k+1$ to ${n}$}
        \State Create $P_{\alpha,\ell_{b}+q-1,\ph,\beta}$ by branching
        \State $\delta \gets$ the sequence of positions of jobs $\{j^*+q,\ldots,j^*+n-1\}$ fixed by TTBR2
        \State $nodes\gets nodes \cup P_{\alpha,\ell_{b}+q-1,\ph,\delta,\beta}$
    \EndFor
\Else 
    \State $nodes\gets \mathcal{R}_\sigma$
\EndIf
\State $Q\gets Q\cup \mathtt{MERGE\_RIGHT\_NODES}(nodes,P_\sigma)$
\State $Q\gets Q\cup nodes[k]$  \Comment{The last node will be re-processed}
\State \Return $Q$
\EndFunction
\end{algorithmic}
\end{algorithm}
\endgroup

\begin{algorithm}
\caption{{\ttfamily MERGE\_RIGHT\_NODES} Procedure}\label{algo:mrn}
\begin{algorithmic}[1]
\Require $rnodes=[P_{\alpha,\underset{n_1}{\bullet},\beta^1},\ldots,P_{\alpha,\underset{n_r}{\bullet},\beta^r}]$, ordered list of $r$ last children nodes with $\ell_b$ defined on any node in $rnodes$. $|\alpha|+1$ is the job to branch on and $n_r=n_1+r-1$.
\Ensure $Q$, a list of problems to branch on after merging
\Function{{\ttfamily MERGE\_RIGHT\_NODES}}{$rnodes, P_\sigma$}
\State $Q \gets \mathtt{LEFT\_MERGE}(P_{\alpha,\underset{n_1}{\bullet},\beta^1})$
\For{$q=1$ to $r-1$}
    \For{$j=\ell_b+ \min(k,\lfloor{n/2}\rfloor)$ to $\ell_b+n_1-1$}
        \State $ Q\gets Q\cup P_{\alpha,j,\underset{n_q-1}{\bullet},\beta^q}$
    \EndFor
\EndFor

\For{$j=\ell_b+n_1$ to $\ell_b+n_r$}
    \State Solve all the subproblems of size less than $k$ in $\mathcal{S}^\sigma_{k-r,j}$
    \State $\mathcal{R}_{\alpha,\underset{n_r}{\bullet},\beta^r} \gets \mathcal{R}_{\alpha,\underset{n_r}{\bullet},\beta^r}\cup \mathtt{BEST}( \mathcal{S}^\sigma_{k-r,j})$ 
\EndFor
\If{$r > 2$}  
    \State $newnodes\gets \emptyset$
    \For{$q=2$ to $r-1$}
        \State $newnodes\gets newnodes\cup\mathtt{LEFT\_MERGE}(P_{\alpha,\underset{n_q}{\bullet},\beta^q})$ 
    \EndFor
    \State $newnodes\gets newnodes\cup P_{\alpha,\ell_b,\underset{n_r-1}{\bullet},\beta^r}$
    \State $Q\gets Q\cup \mathtt{MERGE\_RIGHT\_NODES}(newnodes,P_\sigma)$
    \State $\mathcal{L}_{\alpha,\underset{n_r}{\bullet},\beta^r}\gets P_{\alpha,\ell_b,\underset{n_r-1}{\bullet},\beta^r}$ 
    \For{$q=2$ to ${k-1}$} 
        \State Create $P_{\alpha,\ell_b+q-1,\underset{n_r-1}{\bullet},\beta^r}$ by branching
        \State $\delta \gets$ the sequence of positions of jobs $\{|\alpha|+2,\ldots,|\alpha|+q\}$ fixed by TTBR2
        \State $\mathcal{L}_{\alpha,\underset{n_r}{\bullet},\beta^r}\gets \mathcal{L}_{\alpha,\underset{n_r}{\bullet},\beta^r}\cup  P_{\alpha,\ell_b+q-1,\delta,\underset{n_r-1}{\bullet},\beta^r}$
    \EndFor
\EndIf
\State \Return $Q$
\EndFunction
\end{algorithmic}
\end{algorithm}

\begin{lemma}
\label{lemma:rightmerge}
The {\ttfamily RIGHT\_MERGE} procedure returns a list of $\bigo{n}$ nodes 
in polynomial time and space.\\
The solution of the  problems associated to these nodes involves the solution of $1$ subproblem of size $(n-1)$, of $(k-1)$ subproblems of size $(n-k-1)$, and subproblems of size $i$ and $(n_q-(k-r)-i-1)$, $\forall r=2, ...,k; q=1,...,(r-1); i=k,...,(n-2k+r-2)$, where $n_q$ is defined in Algorithm \ref{algo:mrn} (\texttt{MERGE\_RIGHT\_NODES}).\\
\end{lemma}
\begin{proof}
The first part of the result follows directly from Algorithm~\ref{algo:rightmerge}. The only lines where nodes are added to $Q$ in \texttt{RIGHT\_MERGE} are lines 13-14. In line 14, only one problem is added to $Q$, thus it needs to be proved that the call on \texttt{MERGE\_RIGHT\_NODES} (line 13) returns $\bigo{n}$ nodes. This can be computed by analysing the lines 2-7 of Algorithm~\ref{algo:mrn}. Considering all recursive calls, the total number of nodes returned by \texttt{MERGE\_RIGHT\_NODES} is $(\sum_{i=1}^{k-1}(k-i)(n-2k-i)) + k-1$ which yields $\bigo{n}$. 
The number of all nodes considered in right merging is bounded by a linear function on $n$. Furthermore, all the operations associated to the nodes (merging, creation, etc) have a polynomial cost. As a consequence, Algorithm~\ref{algo:rightmerge} runs in polynomial time and space.

Regarding the sizes of the subproblems returned by \texttt{RIGHT\_MERGE}, the node added in line 14 of Algorithm~\ref{algo:rightmerge} contains one subproblem of size $(n-1)$, corresponding to branching the longest job on the last available position. 
Note that even though the subtree of this node has been modified by merging operations, its solution time is still bounded by the time for solving a subproblem of size $(n-1)$. 

Then, the problems added by the call to \texttt{MERGE\_RIGHT\_NODES} are added to $Q$.
In line 2 of Algorithm~\ref{algo:mrn}, the size of the problem returned by \texttt{LEFT\_MERGE} 
is reduced by one unit when compared to the input problem which is of size $(n-k-(k-r))$. Note that
$(k-r)$ is the current level with respect to the node tackled by Algorithm~\ref{algo:mrn}.
As a consequence, the size of the resulting subproblem is $(n-k-(k-r)-1)$. Note that this line is executed $(k-1)$ times, $\forall r=k,\ldots,2$, corresponding to the number of calls to \texttt{MERGE\_RIGHT\_NODES}. In line 5 of Algorithm~\ref{algo:mrn}, 
the list of nodes which are not involved in any merging operation are added to $Q$. 
This corresponds to pairs of problems of size $i$ and $(n_q-(k-r)-i-1),~\forall i=k,...,(n-k-1)$ and this proves the last part of the lemma.
\end{proof}

Lemma \ref{lemma:RMWC} considers non-worst-case scenarios.

\begin{lemma}
\label{lemma:RMWC}
When right merging is incorporated into the branching scheme of TTBR, instances such that $LPT=EDD$ correspond to worst-case instances.
\end{lemma}
\begin{proof}
The proof follows a similar reasoning to the proof of Lemma~\ref{lemma:LMWC}. We analyze the time reduction obtained from the last $k$ subtrees of the current node. 
The following cases are considered:
\begin{enumerate}
    \item $1=[j]$, $j\leq (n-k+1)$ and $2=[j']$, $j'\leq (n-k)$;
    \item $1=[j]$, $j\leq (n-k+1)$ and $2=[j']$, $j' > (n-k)$; 
    \item $1=[j]$, $j > (n-k+1)$.
\end{enumerate}

We refer the reader to Figure~\ref{fig:rm1} for a better understanding of the proof. Since the structure of right merging is the same at different levels (except level 0) of the tree, it is sufficient to consider  the time reduction from level 0 and level 1. We denote the resulting time reduction by $TR1$, $TR2$ and $TR3$ for each of the three cases, respectively. We also note $TR_{LPT=EDD}$ the time reduction corresponding to the case $LPT=EDD$.

In case 1, all black nodes at level 1 of Figure~\ref{fig:rm1} are created and merged to one. Therefore, the corresponding time reduction $TR1\geq TR_{LPT=EDD}=\sum_{i=1}^{k-1}iT(n-i-2)$.

In case 2, some black nodes at level 1 of Figure~\ref{fig:rm1} are not created due to Property \ref{dec1}, while the black nodes that remain can still be merged to the last subtree. Therefore, $TR2\geq TR1$.

In case 3, the subtrees rooted by  $\{P_1,...,P_{j-1}\}$ are not created due to Property \ref{dec1}. This is obviously a stronger reduction than only merging some nodes from inside these subtrees. In addition, for subtrees that remain except the last one, i.e., those rooted by $\{P_j,...,P_{n}\}$, time reduction is still guaranteed by, either right merging or the non-creation of some nodes due to Property \ref{dec1}, depending on the position of job $2$ in EDD ordering.

In other words, if $LPT\neq EDD$ then the number of nodes in $\mathcal{S}_{\ell,j}^\sigma$ (defined in Proposition~\ref{propos:rg}) can be less, since some nodes may not be created due to Property~\ref{dec1}. However, all the nodes inside $\mathcal{S}_{\ell,j}^\sigma$ can still be merged to one except when $\mathcal{S}_{\ell,j}^\sigma$ is empty. In either case, we can achieve at least the same reduction as the case of $LPT = EDD$. This reasoning obviously holds when extending the consideration to all levels of the tree and to all recursions. Therefore, $LPT=EDD$ is the worst-case scenario.
\end{proof}

\subsubsection{A working example for Right Merge}
In order to better illustrate the merging operation on right-side branches, an example is provided in this section. The input data ($n=6$) is given in Table \ref{tab:rmex} and the value of $k$ is chosen as $3$.

\begin{table}[!ht]
\centering
\begin{tabular}{|c|c|c|c|c|c|c|}\hline
$i$ & 1 & 2 &3&4&5&6 \\ \hline
$p_i$& 11&10&7&4&2 & 1\\ \hline
$d_i$&10&12&14&16&18&19\\ \hline
\end{tabular}
\caption{A sample instance}\label{tab:rmex}
\end{table}

The solution of the instance is depicted in Figure \ref{fig:rightex}. The indicated partial total tardiness value ($tt=\sum T_j$) is computed on the jobs that are fixed before and after the unscheduled jobs.
The applied algorithm is TTBR1 with \texttt{RIGHT\_MERGE} integrated. 
At the beginning, the root node $P$ is the current node, and \texttt{RIGHT\_MERGE} is called on it. Line 1-12 in Algorithm \ref{algo:rightmerge} creates child nodes in order to prepare the merging, as shown in \ref{subfig:rightex1}. Notice that the small subproblem involving jobs $\{5,6\}$ inside $P_4$ is solved directly to obtain the optimal sequence which is $(6,5)$. 
Line 13 calls the \texttt{MERGE\_RIGHT\_NODES} procedure with arguments $(P_4,P_5,P_6)$. 

In \texttt{MERGE\_RIGHT\_NODES} procedure (Algorithm \ref{algo:mrn}), line 2 means to apply \texttt{LEFT\_MERGE} on $P_4$. This is not performed in our exemple since $P_4$ only have 3 jobs to schedule and $k$ is supposed to verify $2\leq k \leq n/2$. Therefore, we add directly $P_{4,1}$ to $\mathcal{Q}$. 
Line 3-7 means to add all ``middle'' nodes that are not involved by merging hence should be opened later as independent subproblems. This includes $P_{4,2}$ only. Notice that $\ell_b=1$ at line 4. 
Lines 8-11 do the merging. As shown in Figure \ref{subfig:leftex2}, nodes $P_{4,3}$,  $P_{5,3}$ are merged with $P_{6,3}$ and the node $P_{5,4}$ is merged with $P_{6,4}$. The dominant nodes are $P_{6,3}$ and $P_{6,3}$, hence nodes nodes $P_{4,3}$,  $P_{5,3}$ and $P_{5,4}$ are  cut. 
After this, at line 15, \texttt{LEFT\_MERGE} is applied on $P_5$ (Figure \ref{subfig:rightex3}) then the \texttt{MERGE\_RIGHT\_NODES} procedure is called in a recursive manner on nodes $P_{5,1}$ and $P_{6,1}$. The nodes $P_{5,1,3}$ and $P_{5,1,4}$ are finally cut due to merging (Figure \ref{subfig:rightex4}).

Then the algorithm continues in a recursive way. In depth-first order, the node $P_1$ is the next to open. 



\FloatBarrier
\begin{figure}[ht!]
\centering
\begin{subfigure}[ht!]{\textwidth}\centering
\includegraphics[width=\columnwidth]{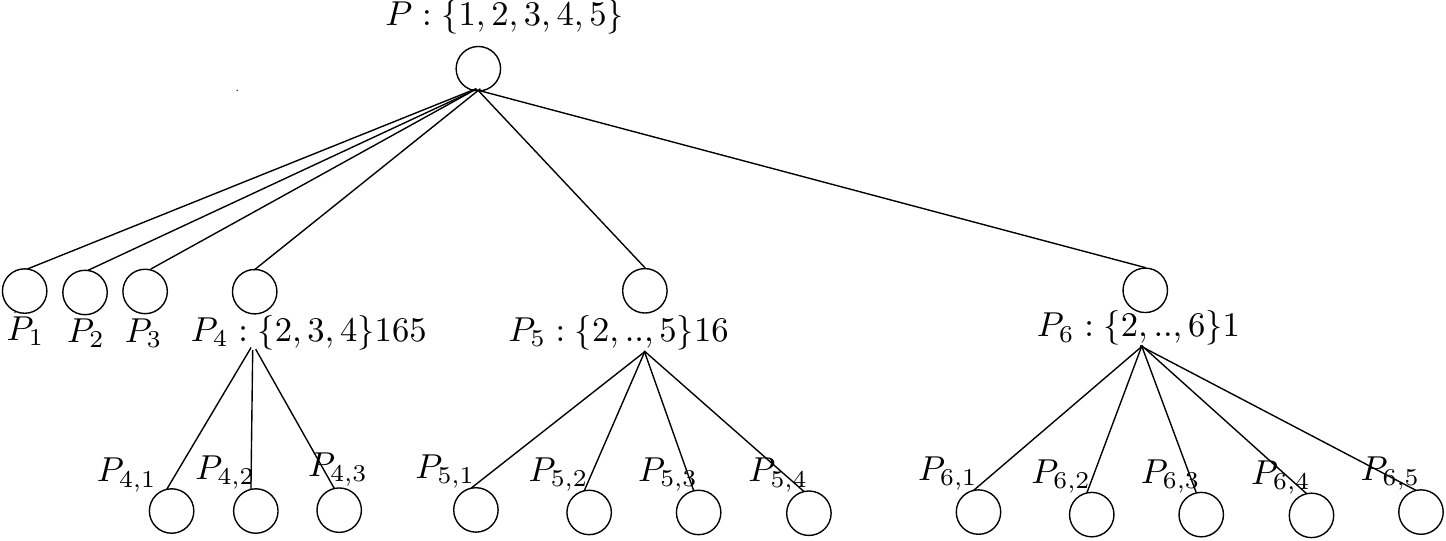}\caption{}
\label{subfig:rightex1}
\end{subfigure}
\end{figure}

\begin{figure}[ht!]
\ContinuedFloat
\begin{subfigure}[!ht]{\textwidth}\centering
\includegraphics[width=\columnwidth]{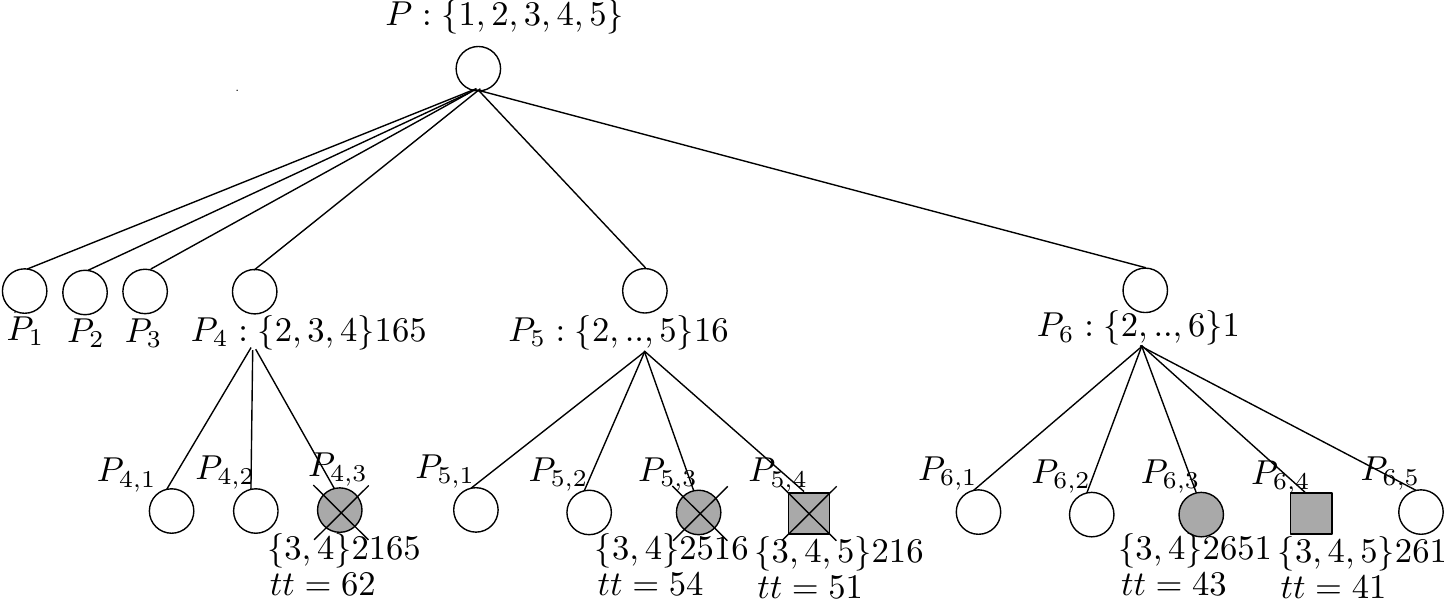}\caption{}
\label{subfig:rightex2}
\end{subfigure}
\end{figure}

\begin{figure}[ht!]
\ContinuedFloat
\begin{subfigure}[!ht]{\textwidth}\centering
\includegraphics[width=\columnwidth]{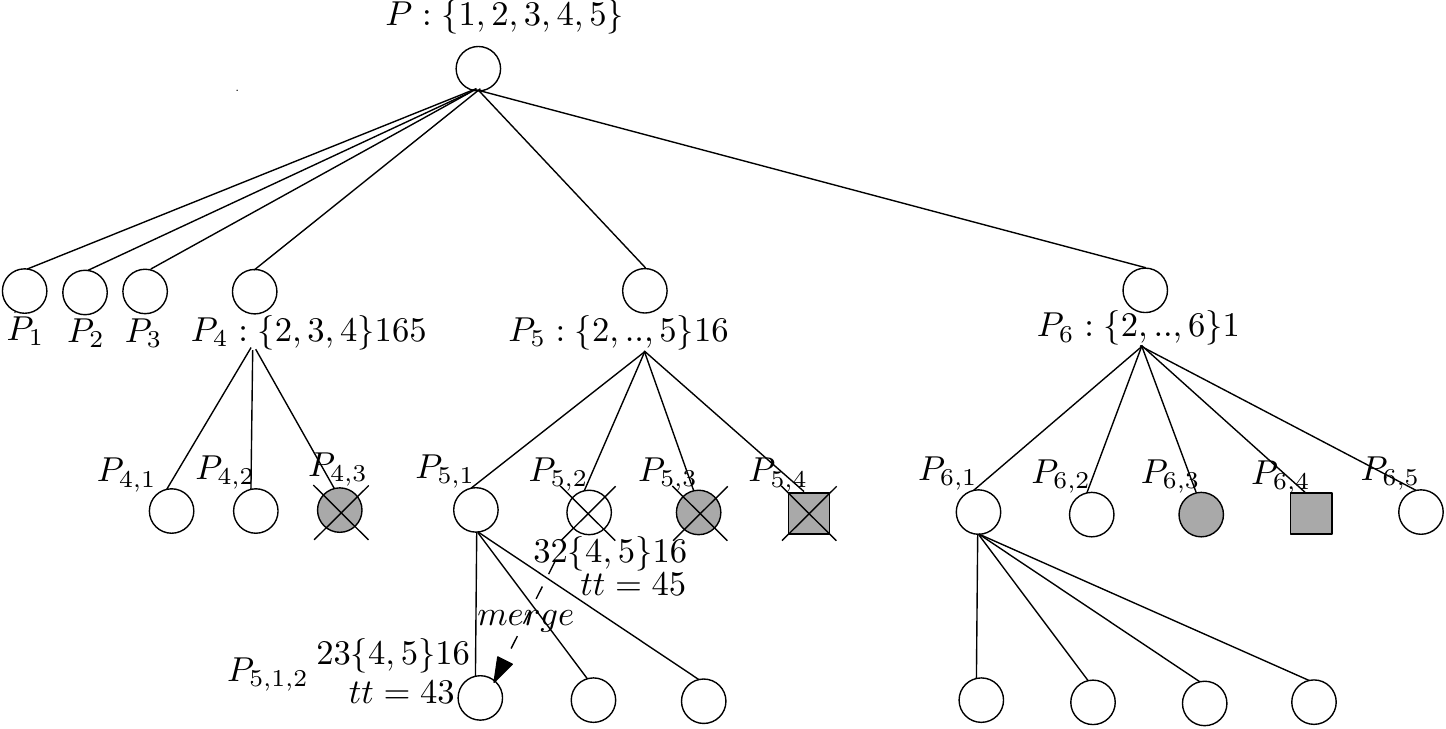}\caption{}
\label{subfig:rightex3}
\end{subfigure}
\end{figure}

\begin{figure}[ht!]
\ContinuedFloat
\begin{subfigure}[!ht]{\textwidth}\centering
\includegraphics[width=0.9\columnwidth]{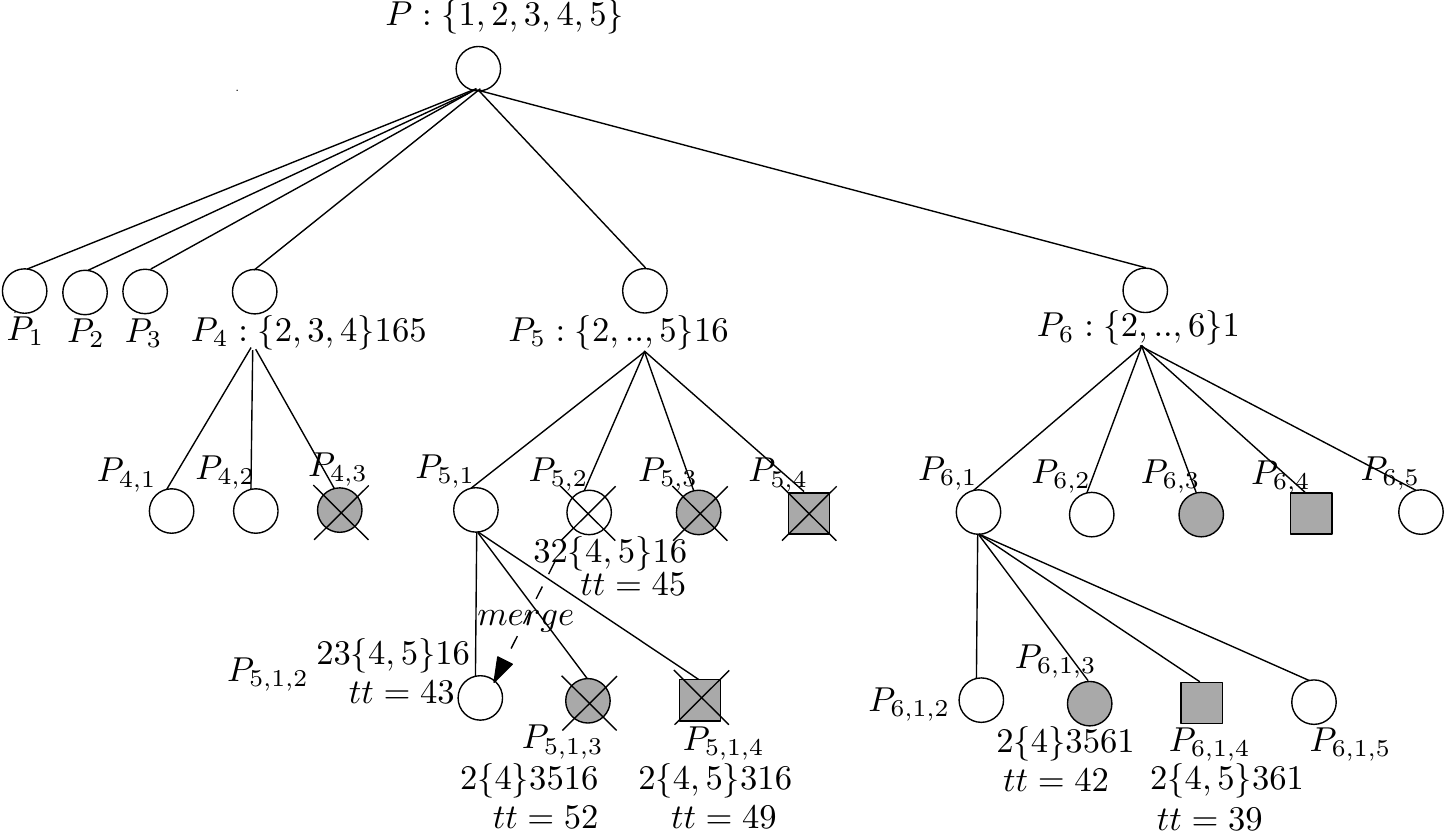}\caption{}
\label{subfig:rightex4}
\end{subfigure}
\caption{A sample instance solved by TTBR1 with right-merging integrated}
\label{fig:rightex}
\end{figure}
\FloatBarrier

\FloatBarrier
\subsection{Complete algorithm and analysis}~\label{sec:thealgorithm}
We are now ready to define the main procedure {\ttfamily TTBM} (Total Tardiness \techbm{}), stated in Algorithm~\ref{algo:ttbm} which is called on the initial input problem $P : \{1,...,n\}$. 
The algorithm has a similar recursive structure as TTBR1. However, each time a node is opened, the sub-branches required for the merging operations are generated, the subproblems of size less than $k$ are solved and the procedures {\ttfamily LEFT\_MERGE} and {\ttfamily RIGHT\_MERGE} are called. Then, the algorithm proceeds recursively by extracting the next node  from $Q$ with a depth-first strategy and terminates when $Q$ is empty. 

\begin{algorithm}
\caption{Total Tardiness Branch and Merge (TTBM)}\label{algo:ttbm}
\begin{algorithmic}[1]
\Require $P:\{1,...,n\}$: input problem of size $n$ \par $\frac{n}{2}\geq k\geq 2$: an integer constant 
\Ensure $seqOpt$: an optimal sequence of jobs
\Function{TTBM}{$P$,$k$}
\State $Q \gets P$
\State $seqOpt \gets$ the EDD sequence of jobs
\While{$Q\neq \emptyset$}
    \State $P^*\gets $ extract next problem from $Q$ (depth-first order)
    \If{the size of $P^*<2k$}
        \State Solve $P^*$ by calling TTBR2
    \EndIf
    \If{all jobs $\{1,...,n\}$ are fixed in $P^*$}
        \State $seqCurrent \gets$ the solution defined by $P^*$
        \State $seqOpt \gets$ best solution between $seqOpt$ and $seqCurrent$
    \Else
        \State $Q \gets Q\cup \mathtt{LEFT\_MERGE}(P^*)$    \Comment{Left-side nodes}
        \For{$i =k+1,...,n-k$} 
            \State Create the $i$-th child node $P_i$ by branching scheme of TTBR1 
            \State $Q \gets Q\cup P_i$   
        \EndFor
        \State $Q \gets \mathtt{RIGHT\_MERGE}(P^*)$ \Comment{Right-side nodes}
    \EndIf
\EndWhile
\State  \Return $seqOpt$
\EndFunction
\end{algorithmic}
\end{algorithm}

The complexity of the algorithm depends on the value $k$. The higher it is, the more subproblems can be merged and the better is the worst-case time complexity of the algorithm. Figure \ref{fig:cutsafter} demonstrates the reduction that can be obtained by TTBM on each recursive call. It can help complexity analysis of the algorithm. 
Proposition~\ref{propos:merge} provides details about the time complexity of the proposed algorithm. 

\begin{figure}
\centering
\includegraphics[width=0.9\columnwidth]{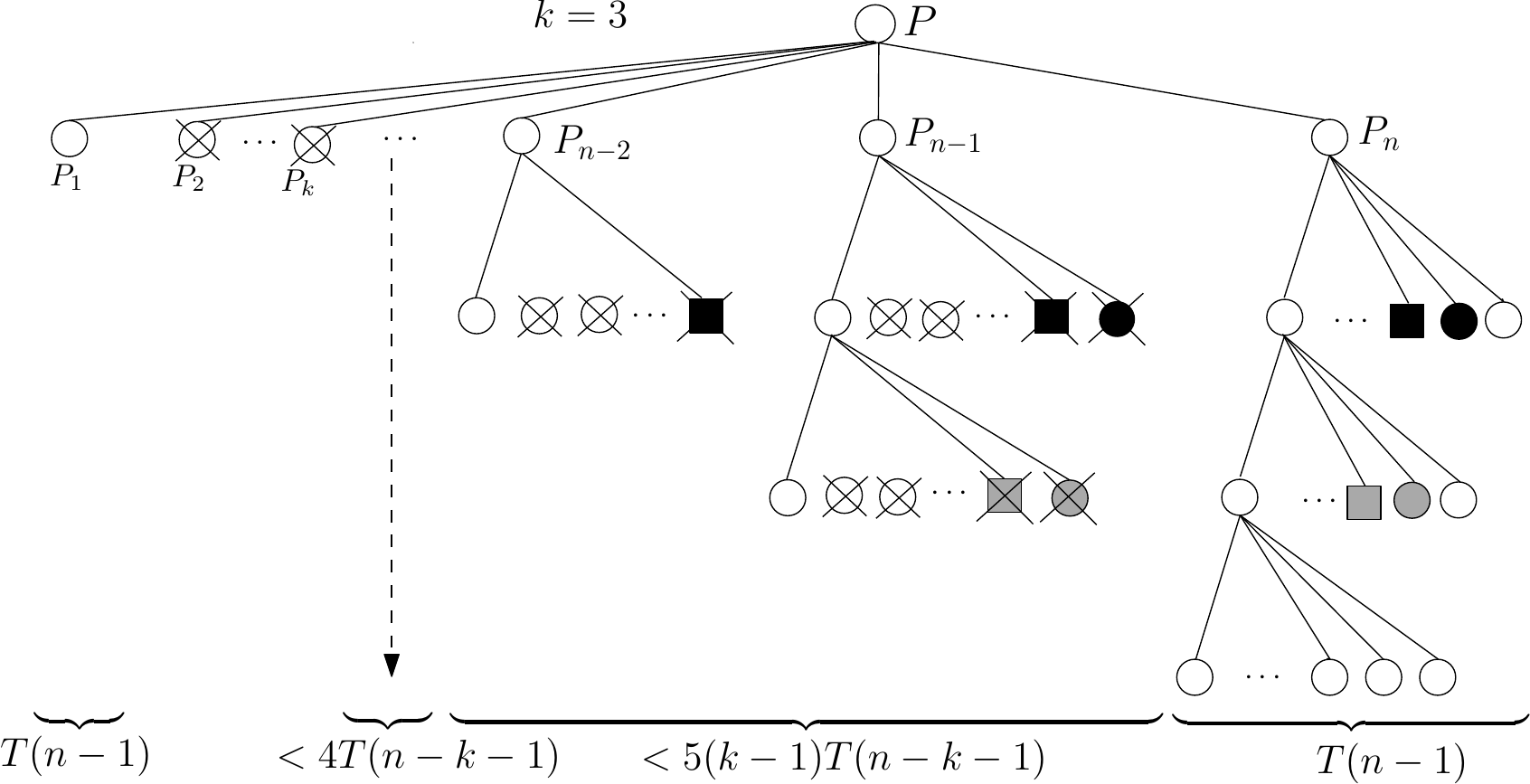}
\caption{Reduction obtained by merging}\label{fig:cutsafter}
\end{figure}

\FloatBarrier
{\begin{propos}\label{propos:merge}
Algorithm TTBM runs in $\ostar{(2+\epsilon)^n}$ time and polynomial space, where $\epsilon \to 0$ when $k \to \infty$. 
\end{propos}
\begin{proof}
The proof is based on the analysis of the number and the size of the subproblems put in $Q$ when a single problem $P^*$ is expanded. 
As a consequence of Lemma~\ref{lemma:LMWC} and Lemma~\ref{lemma:RMWC}, TTBM induces the following recursion:
\begin{align*}
T(n) = 
    &2 T(n-1) + 2T(n-k-1) + ... + 2T(k) \\
    &+ \sum_{r=2}^k\sum_{q=1}^{r-1}\sum_{i=k}^{n_1-(k-r)-2}(T(i)+T(n_q-(k-r)-i-1))\\
    &+(k-1)T(n_1-1)  + \bigo{p(n)}
\end{align*}
First, a simple lower bound on the complexity of the algorithm can be derived by the fact that the procedures \texttt{RIGHT\_MERGE} and \texttt{LEFT\_MERGE} provide (among  others) two subproblems of size $(n-1)$, based on which the following inequality holds:
\begin{equation}
\label{eq:lbound2n}
    T(n)> 2T(n-1)
\end{equation} 
By solving the recurrence, we obtain that $T(n) = \omega(2^n)$. Here the $\omega()$ notation is adopted to express an asymptotic lower bound of the complexity. $f(x)=\omega(g(x))$ if and only if for any positive constant $c$, $\exists n'$ such that $\forall x>n'$, $f(x)> cg(x)\geq 0$. 
As a consequence, the following inequality holds:
\begin{equation}
\label{eq:lboundtn}
    T(n)> T(n-1)+\ldots+T(1)
\end{equation} 
In fact, if it does not hold, we have a contradiction on the fact $T(n) = \omega(2^n)$. 
Now, we consider the summation $\sum_{i=k}^{n_1-(k-r)-2}(T(n_q-(k-r)-i-1))$.
Since $n_q=n_1 + q - 1$, we can simply expand the summation as follows:
$$\sum_{i=k}^{n_1-(k-r)-2}(T(n_q-(k-r)-i-1)) = T(q) + ... +T(n_1 -(k-r)+q-k-2)$$.
We know that $k\geq q$, then $q-k \leq 0$ and the following inequality holds:
$$T(q) + ... +T(n_1 -(k-r)+q-k-2) \leq \sum_{i=q}^{n_1-(k-r)-2} T(i) $$.

As a consequence, we can bound above $T(n)$ as follows: 

\begin{align*}
T(n) = 
    &2 T(n-1) + 2T(n-k-1) + ... + 2T(k) \\
    &+ \sum_{r=2}^k\sum_{q=1}^{r-1}\sum_{i=k}^{n_1-(k-r)-2}(T(i)+T(n_q-(k-r)-i-1))\\
&\leq 2 T(n-1) + 2T(n-k-1) + ... + 2T(k) \\
    &\qquad+\sum_{r=2}^k\sum_{q=1}^{r-1}\sum_{i=q}^{n_1-(k-r)-2}2T(i)+(k-1)T(n_1-1) + \bigo{p(n)}\\
&\leq 2 T(n-1) + 2T(n-k-1) + ... + 2T(k) \\
    &\qquad+\sum_{r=2}^k\sum_{q=1}^{r-1}\sum_{i=1}^{n_1-(k-r)-2}2T(i)+(k-1)T(n_1-1) + \bigo{p(n)}
\end{align*}

By using Equation~\ref{eq:lboundtn}, we obtain the following:

\begin{align*}
T(n)
&\leq 2 T(n-1) + 2T(n-k-1) + ... + 2T(k) \\
    &\qquad+\sum_{r=2}^k\sum_{q=1}^{r-1}\sum_{i=1}^{n_1-(k-r)-2}2T(i)+(k-1)T(n_1-1) + \bigo{p(n)}\\
&\leq 2 T(n-1) + 2T(n-k-1) + ... + 2T(k) \\
    &\qquad+\sum_{r=2}^k\sum_{q=1}^{r-1}2T(n_1-(k-r)-1) +(k-1)T(n_1-1)+ \bigo{p(n)}\\
\end{align*}
Finally, we apply some algebraic steps and we use the equality $n_1 = n-k$ to derive the following upper limitation of $T(n)$:
\begin{align*}
T(n) &\leq 2 T(n-1) + 2T(n-k-1) + ... + 2T(k) \\
    &\qquad+\sum_{r=2}^k(r-1)2T(n_1-(k-r)-1) +(k-1)T(n_1-1)+ \bigo{p(n)}\\
    &\leq 2 T(n-1) + 2T(n-k-1) + ... + 2T(k) +2(k-1)T(n_1-1) \\
    &\qquad+\sum_{r=2}^{k-1}(r-1)2T(n_1-(k-r)-1) +(k-1)T(n_1-1)+ \bigo{p(n)}\\
    &\leq 2 T(n-1) + 2T(n-k-1) + ... + 2T(k) \\
    &\qquad+(k-1)4T(n_1-1) +(k-1)T(n_1-1)+ \bigo{p(n)}\\
&\leq 2 T(n-1) + 4 T(n-k-1) + 5 (k-1) T(n-k-1)  + \bigo{p(n)}\\ 
&= 2 T(n-1) + (5k-1) T(n-k-1)  + \bigo{p(n)}\\
\end{align*}


Note that $\bigo{p(n)}$ includes the cost for creating all nodes for each level and the cost of all the merging operations, performed in constant time.

The recursion $T(n) =  2 T(n-1) + (5k-1) T(n-k-1) + \bigo{p(n)}$ is an upper limitation on the running time of TTBM. Recall that its solution is $T(n) = \ostar{c^n}$ where $c$ is the largest root of the function:

\begin{equation}
\label{eq:fk}
f_k(x) = 1 - \frac{2}{x} - \frac{5k-1}{x^{k+1}}
\end{equation}. 

As $k$ increases, the function $f_k(x)$ converges to $1 - \frac{2}{x}$, which induces a complexity of $\ostar{2^n}$.
Table~\ref{tab:kvalues} shows the time complexity of TTBM obtained by solving Equation~\ref{eq:fk} for all the values of $k$ from $3$ to $20$.\footnote{SageMath (\url{http://www.sagemath.org}) is used for the computation.} The base of the exponential is computed by solving Equation~\ref{eq:fk} by means of a mathematical solver and rounding up the fourth digit of the solution. The table shows that the time complexity is $\ostar{2.0001^n}$ for $k \geq 20$. 
\end{proof}}

\begin{table}
\centering
\begin{tabular}{|c|c|}
\hline
$k$     & $T(n)$  \\ \hline
$3$ 	& $\ostar{2.5814^n}$  \\ 
$4$ 	& $\ostar{2.4302^n}$  \\ 
$5$ 	& $\ostar{2.3065^n}$  \\ 
$6$ 	& $\ostar{2.2129^n}$  \\ 
$7$ 	& $\ostar{2.1441^n}$  \\ 
$8$ 	& $\ostar{2.0945^n}$  \\ 
$9$	    & $\ostar{2.0600^n}$  \\ 
$10$	& $\ostar{2.0367^n}$  \\ 
$11$	& $\ostar{2.0217^n}$  \\ 
$12$	& $\ostar{2.0125^n}$  \\ 
$13$	& $\ostar{2.0070^n}$  \\ 
$14$	& $\ostar{2.0039^n}$  \\ 
$15$	& $\ostar{2.0022^n}$  \\ 
$16$	& $\ostar{2.0012^n}$  \\ 
$17$	& $\ostar{2.0007^n}$  \\ 
$18$	& $\ostar{2.0004^n}$  \\ 
$19$	& $\ostar{2.0002^n}$  \\ 
$20$	& $\ostar{2.0001^n}$  \\ \hline

\end{tabular}
\caption{The time complexity of TTBM for values of $k$ from $3$ to $20$}
\label{tab:kvalues}
\end{table}
\begin{lemma}
\label{lemma:BMWC}The problem \pbtt{} can be solved in $\ostar{(2+\epsilon)^n}$ time and polynomial space, where $\epsilon>0$ can be arbitrarily small. 
\end{lemma}
\begin{proof}
Lemma \ref{lemma:LMWC} and lemma \ref{lemma:RMWC} proved that $LPT=EDD$ is the worst-case scenario for left merging and right merging. Since $k\leq \frac{n}{2}$, the time reduction obtained from left merging and right merging, when both are incorporated into TTBM, can be combined. Thus, lemma \ref{lemma:LMWC} and lemma \ref{lemma:RMWC} together prove that instances with $LPT=EDD$ are the worst-case instances for TTBM. Therefore, the current lemma is proved according to Proposition \ref{propos:merge}.
\end{proof}

\FloatBarrier
\section{Experimental results}
The whole mechanism of \techbm{} has been implemented and tested on instances generated in the same way as in \cite{szwarc2001algorithmic}. 
200 instances are generated randomly for each problem size using the same generation scheme as in the paper of  \cite{szwarc2001algorithmic}. Processing times are integers generated from an uniform distribution in the range $[1, 100]$ and due dates $d_i$ are integers from a uniform distribution in the range $[p_iu, p_iv]$ where $u=1-T-R/2$ and $v=1-T+R/2$. Each due date is set to zero whenever its generated value is negative. Twenty combinations $(R,T)$ are considered where $R\in\{0.2, 0.4, 0.6,0.8,1\}$, and $T\in\{0.2, 0.4, 0.6,0.8\}$. 
Ten instances are generated for each combination and the combination $(R=0.2, T=0.6)$ yields the hardest instances as reported in the literature (see \cite{szwarc1999solution}) and confirmed by our experiments. The tests are performed on a HP Z400 work station with 3.07GHz CPU and 8GB RAM.
 
Before comparing \techbm{} to the state-of-the-art algorithm in \cite{szwarc2001algorithmic}, we first describe the latter one more clearly. {That  algorithm, named BR-SPM}, is based on the branching structure of BR, with the following three extra {features} integrated.

\begin{enumerate}
    \item \textit{Split}, which decomposes a problem according to precedence relations; 
    \item \textit{PosElim}, which eliminates non-promising branching positions before each branching;
    \item \textit{Memorization}, which avoids solving a subproblem more than once by saving its solution to a database (the ``memory'') and retrieve it whenever the same subproblem appears again. 
\end{enumerate}

We implemented BR-SPM in the following way: on a given problem $P$, BR-SPM first {tries} to find the solution in the \textit{memory}; if failed, it applies \textit{Split} to decompose $P$ into subproblems then solve them recursively; if \textit{Split} does not decompose, \textit{PosElim} is called which returns the list of positions on which the longest job can be branched on. The branching then occurs and the resulting subproblems are then solved recursively. Each time we solve a subproblem by branching, the solution sequence is saved to the memory for further query. Our implementation of BR-SPM successfully solves instances with up to 900 jobs in size, knowing that the original program was limited to instances with up to 500 jobs due to memory size limit.

We now provide our experimental results. In order to verify the concept of \textit{merging} without {extra features}, we first 
 compare \techbr{} with \techbm{} on the hardest subset of instances of each size.
{Table~\ref{tab:brbm1}, depicting (minimum, average and maximum) CPU time, 
average number of merging hits and total number of explored nodes, shows that the merge mechanism strongly accelerates the solution. 
However \techbm{} is still limited to $50$ jobs in size only.}

\begin{table}[!ht]\centering\resizebox{0.8\textwidth}{!}{\begin{tabular}{|c|c|c|c|c|c|} \hline 
& TMin& TAvg& TMax& \#Merge& \#Nodes  \\ \hline 
BR & 52.0 & 1039.0 & 3127.0  & 0 & 1094033204\\ \hline 
BM & 3.0 & 67.6 & 319.0  & 11277311  & 47143367  \\ \hline
\end{tabular}}
\centering
\caption{ {Results for instances of size 40}}
\label{tab:brbm1}
\end{table}

To improve the performances, we enable \textit{Split} and \textit{PosElim}. The resulting algorithms are called BR-SP and BM-SP. 
Now both algorithms can handle instances with up to 300 jobs (see Table~\ref{tab:brbm2}). Surprisingly, however, even with a considerable number of {merged nodes}, BM-SP turns out to be slower than BR-SP.

\begin{table}[!ht]\centering\resizebox{0.8\textwidth}{!}{\begin{tabular}{|c|c|c|c|c|c|} \hline 
& TMin& TAvg& TMax&  \#Merge&  \#Nodes  \\ \hline 
{BR-SP} & 504.0 & 3000.8 & 7580.0 &  0 & 634569859 \\ \hline 
{BM-SP} & 521.0 & 3097.9 & 7730.0 &  608986 &  508710322 \\ \hline 
\end{tabular}}\caption{Results for instances of size 300}\label{tab:brbm2}\end{table}
Auxiliary tests show that \textit{Split} and \textit{PosElim} negatively affect the merge mechanism. Solving a small problem by $Split$, which sometimes finds directly the solution sequence according to precedence relations, may be faster than merging two nodes. 
\textit{PosElim} is also powerful as the average number of branching positions at each node after its application is 
approximately 2, i.e. most positions are already eliminated before merging. This implies that the search tree explored by BR-SP may be even smaller than a binary tree, as there are also many nodes that are not counted: they only have a single child node. These observations show that it is not straightforward to combine {directly} the current merge scheme with existing solving techniques. The theoretical effectiveness and running time guarantee of merging stays valid, however we need to find a new way to apply it in practice. This is more developed in chapter \ref{ch4}.

\FloatBarrier
\section{Additional Results}~\label{sec:bmadditional}
In this section, two additional results are described. The first one is the complexity analysis of the algorithm of \cite{szwarc2001algorithmic}. The second one is TTBM with right-merge removed, but with Property \ref{Elim3} considered.

\subsection{The complexity of the algorithm of \cite{szwarc2001algorithmic}}\label{sec:comp_bb2001}
The current state-of-the-art algorithm described by \cite{szwarc2001algorithmic}, noted hereafter as BB2001, is a branch and bound algorithm having a similar structure as that of TTBR1. The main difference is that in BB2001, in addition to the decomposition rule given in Property \ref{dec1} (as in TTBR1), another decomposition rule, based on Property \ref{dec2}, is applied simultaneously when branching on a node. We provide in Proposition \ref{propos:bb2001}  the time complexity analysis of BB2001.

\begin{proper}\label{dec2}~{\citep{della1998new}}
Let job~$k$ in LPT sequence correspond to job~$[1]$ in EDD sequence.
Then, job~$k$ can be set only in positions $h\leq (n-k+1)$ and
the jobs preceding job~$k$ are uniquely determined as
$B_k(h)$, where $B_k(h)\subseteq \{k+1,k+2,\dots,n\}$ and $\forall i\in B_k(h), j\in \{n,n-1,\dots,k+1\}\smallsetminus B_k(h)$, $d_i\leq d_j$
\end{proper}

\begin{propos}\label{propos:bb2001}
Algorithm BB2001 runs in $\ostar{2.4143^n}$ time and polynomial space in the worst case. 
\end{propos}
\begin{proof}
Before branching on a node, BB2001 first computes the possible positions for the longest job and the job with smallest due date. Then a new branch is created by assigning a pair of compatible positions to these two jobs. We consider two cases as follows.

Firstly, consider the case where job $1=[n]$. The two decomposition rules become identical and if this condition is also verified in all subproblems, then the time complexity is $\ostar{2.4143^n}$ as proved in Proposition \ref{propos:ttbr2}.

In the case where $1\neq[n]$, the worst case occurs when $1=[2]$ and $[n]=2$, since in this case we have a maximum number of available branching positions: job $[n]$ can be branched in position $i\in\{1,...,n-1\}$ and job $1$ can be branched in position $j\in\{2,...,n\}$, with $i<j$ for each branching. Moreover, we recall that Property \ref{Elim3} remains valid.

Three subproblems (left, middle and right) are created on each double branching (zero-sized problems are counted). For the sake of simplicity, we note $T(l,m,r)=T(l)+T(m)+T(r)$. The following recurrence holds.

\begin{align}
T(n) &=\sum_{\substack{i=1\\ i\ is\ odd}}^{n-1}\sum_{\substack{j=i+1\\j\ is\ even}}^n(T(i-1,j-i-1,n-j))+\bigo{p(n)}\\
          &=T(0,0,n-2)+T(0,2,n-4)+T(0,4,n-6)+...+T(0,n-2,0)+\\
&\phantom{ffffffffffffff}T(2,0,n-4)+T(2,2,n-6)+...+T(2,n-4,0)+\\
&\phantom{ffffffffffffffffffffffffffffffffffffffffffffffffff}...\\
&\phantom{fffffffffffffffffffffffffffff}T(n-4,0,2)+T(n-4,2,0)+\\
&\phantom{ffffffffffffffffffffffffffffffffffffffff}T(n-2,0,0)+\\
&\phantom{fffffffffffffffffffffffffffffffffffffffffffff}\bigo{p(n)}\\
&=3*(T(n-2)+2T(n-4)+4T(n-6)+...+\frac{n}{2}T(0))+\bigo{p(n)} \\
\end{align}
By applying a similar process of simplification as in the proof of Proposition \ref{propos:ttbr2}, the following result is finally derived:
\begin{equation}
T(n)=5T(n-2)-T(n-4).
\end{equation}
Correspondingly, we have $T(n)=\ostar{\sqrt{\frac{5+\sqrt{21}}{2}}^n}=\ostar{2.1890^n}$.
Therefore the worst case occurs when the two decomposition rules overlap, and the resulting time complexity is the same as TTBR2, namely $\ostar{2.4143^n}$. The space complexity of BB2001 is also polynomial since depth-first exploration is adopted.

\end{proof}

\subsection{Algorithm TTBM-L}
An intermediate algorithm has been derived based on TTBR1, with Left-Merge enabled, Right-Merge disabled, and Property \ref{Elim3} included. The algorithm is named as TTBM-L and it is stated in Algorithm~\ref{algo:ttbml}. This is basically the algorithm applied on the example problem in section \ref{sec:leftex} without considering Properties \ref{Elim} and \ref{Elim3}.

Each time a node is opened, the sub-branches required for left-merging operations are generated, the subproblems of size less than $k$ are solved and the procedure {\ttfamily LEFT\_MERGE} is called. Then, the algorithm proceeds recursively by extracting the next node  from $Q$ with a depth-first strategy and terminates when $Q$ is empty. 

\begin{algorithm}[tb]
{\footnotesize
\caption{Total Tardiness Branch and Merge (Left Only) (TTBM-L)}\label{algo:ttbml}
\begin{algorithmic}[1]
\Require $P:\{1,...,n\}$: input problem of size $n$ \par $k\geq 2$: an integer constant 
\Ensure $seqOpt$: an optimal sequence of jobs
\Function{TTBM-L}{$P$,$k$}
\State $Q \gets P$
\State $seqOpt \gets$ a random sequence of jobs
\While{$Q\neq \emptyset$}
    \State $P^*\gets $ extract next problem from $Q$ (depth-first order)
    \If{(the size of $P^*<k$)} Solve $P^*$ by calling TTBR1 \EndIf
    \If{all jobs $\{1,...,n\}$ are fixed in $P^*$}
        \State $seqCurrent \gets$ the solution defined by $P^*$
        \State $seqOpt \gets$ best solution between $seqOpt$ and $seqCurrent$
    \Else
        \State $Q \gets Q\cup \mathtt{LEFT\_MERGE}(P^*)$
        \For{$i =k+1,...,n$} 
            \State Create child node $P_i$ like in TTBR1 
            \If{$P_i$ is not eliminated by Property~\ref{Elim}} $Q \gets Q\cup P_i$ \EndIf
        \EndFor
    \EndIf
\EndWhile
\State  \Return $seqOpt$
\EndFunction
\end{algorithmic}
}
\end{algorithm}

{\begin{propos}\label{propos:ttbml}
Algorithm TTBM-L runs in $\ostar{(2.247+\epsilon)^n}$ time and polynomial space, where $\epsilon \to 0$ when $k\to \infty$. 
\end{propos}
\begin{proof}
Starting from Algorithm \ref{algo:ttbml}, we can derive that for a given problem $P$ of size $n$, the $(k-1)$ first children nodes $P_2$ to $P_k$ are merged with children nodes of $P_1$. Consequently, among these nodes, only node $P_1$ remains as a child node of $P$. For the other $(n-k)$ children nodes, Property \ref{Elim} is applied eliminating by the way one node over two. The worst-case is achieved when $n$ is odd and $k$ is even and we have the following recurrence:
$$T(n)=T(n-1)+(T(n-k-1)+T(k)) + (T(n-k-3)+T(k+2)) +...$$ 
       $$+(T(2)+T(n-3))+T(n-1)+\bigo{p(n)}$$
which can be reformulated as
$$T(n)=2T(n-1)+T(n-3)+...+T(n-k+1)+2T(n-k-1)+...+2T(2)+\bigo{p(n)}$$
Following the same approach used in the proof of Proposition \ref{propos:ttbr2}, we plug $T(n-2)$ into the formula and we have
$$T(n)=2T(n-1)+T(n-2)-T(n-3)+T(n-k-1)+\bigo{p(n)}-\bigo{p(n-2)}$$
The solution of this recurrence is $T(n)=\ostar{c^n}$ with $c$ the largest root of $$1=\frac{2}{x}+\frac{1}{x^2}-\frac{1}{x^3}+\frac{1}{x^{k+1}}$$ 
When $k$ is large enough, the last term in the equation can be ignored, leading to a value of $c$ which tends towards $2.247$ as $k$ increases. 
\end{proof}}

\FloatBarrier
\section{Chapter summary}

This chapter focuses on the design of exact branching algorithms for the single machine total tardiness problem. By exploiting some inherent properties of the problem, we first proposed two \techbr{} algorithms: TTBR1 and TTBR2. The former runs in $\ostar{3^n}$ time, while the latter achieves a better time complexity in $\ostar{2.4143^n}$. 
The space requirement is polynomial in both cases.  Furthermore, a technique called merging, is presented and applied onto TTBR1 in order to improve its performance. The final achievement is a new algorithm (TTBM) with time complexity converging towards $\ostar{2^n}$ and requiring polynomial space. The same technique can be tediously  to improve the performance of TTBR2, but the resulting algorithm achieves the same asymptotic time complexity as TTBM, and thus it was omitted.  To the best of authors' knowledge, TTBM is the polynomial space algorithm that has the best worst-case time complexity for solving the $1||\sum T_j$ problem.

Beyond the new established complexity results, the main contribution of the paper is the \techbm{} algorithm. The basic idea is very simple, and it consists of 
speeding up branching algorithms by avoiding to solve identical problems. The same goal is traditionally pursued by means of \textit{Memorization} \citep{fomin2010exact}, where the 
solution of already solved subproblems are stored and then queried when an identical subproblem appears. This is at the cost of an exponential space requirement. In contrast, 
\techbm{} also discards identical subproblems but by appropriately merging, in polynomial time and space, nodes involving the solution of common subproblems. When applied 
systematically in the search tree, this technique enables to achieve a good worst-case time bound. 
On the computational side, it should be noticed that the  merging operation can be relaxed to obtain more efficiency in practice. Instead of comparing nodes at specific positions and solving in $\ostar{2.4143^k}$ subproblems at merged nodes, as described in \techbm{}, we may just compare active nodes with already branched nodes and cut the dominated ones, keeping a polynomial space usage. This can also be seen as memorization but with a fixed size memory used to store already explored nodes. This leads to a weaker worst-case time bound but early works of \cite{tkindt2004revisiting} have shown that this can lead to substantially good practical results, at least on some scheduling problems.

As a further development of this work, our aim is twofold. First, we aim at applying the \techbm{} algorithm to other combinatorial optimization problems in order to
establish its potential generalization to other problems. Second, we want to explore further the practical efficiency of this algorithm on the single machine total tardiness problem, with a different implementation, in a similar way as done by \cite{szwarc2001algorithmic} and \cite{tkindt2004revisiting}. The first aim has not yet been achieved at this point (the end of this thesis) because a direct application of the \techbm{} algorithm requires the target problem to verify certain properties like Property \ref{dec1}. However, much work has been performed on the second direction, which is about the practical efficiency of the algorithm. The Chapter \ref{ch4} is dedicated to this topic.

The main work in this chapter has been performed together with Michele Garraffa, Federico Della Croce from \textit{Politechnico di Torino} (Italy) and Vincent T'Kindt. 
The \techbr{} algorithms TTBR1 and TTBR2 have been reported in  the \textit{MISTA} 2015 conference \citep{della2015smtt} and an Italian conference \textit{AIRO} 2015 \citep{della2015airo}. 
The main results on \techbm{} have been reported at the international conference on \textit{Project Management and Scheduling} (PMS 2016) \citep{shang2016bm} and the French  conference \textit{ROADEF} 2017 \citep{shang2017roadef}. The additional results have been reported at the  \textit{International Symposium on Parameterized and Exact Computation} (IPEC 2017) \citep{shang2017ipec}. A journal paper submitted to \textit{Theoretical Computer Science} is currently under review and the paper can be accessed on HAL \citep{garraffa2017bm}. 
\chapter{The Memorization Paradigm: \textit{Branch \& Memorize} Algorithms for the Efficient Solution of Sequencing Problems}\label{ch4}

\section{Introduction}\label{intro}
The algorithm \techbm{} presented in chapter \ref{ch3} is proved to be not efficient enough in practice. We therefore turned our interest to another technique called \techmemo, which was adopted by \cite{szwarc2001algorithmic} to efficiently solve the \pbtt{} problem. It was also used, though not in a typical way, by \cite{tkindt2004revisiting} to solve three sequencing problems, notably the \pbrisumc, the \pbdtilde{} and the \pbfsumc{} problems.  In this chapter we study, experiment, compare and analyze the efficiency of this method on the four above mentioned problems.

\techmemo{} as an algorithm design technique, allows to speed up algorithms at the price of more space usage. {Typically in branching algorithms, on lower branching levels, isomorphic subproblems may appear exponentially many times and the idea of Memorization is to avoid repetitive solutions as they correspond to identical subproblems. 
The method was first applied on the \textit{Maximum Independent Set} problem by \cite{robson1986algorithms}. By exploiting graph theoretic properties and by applying \techmemo{} to avoid solving identical subproblems, Robson proposed an algorithm with a worst-case time complexity in $\bigo{1.2109^n}$. It has remained the exact exponential algorithm with the smallest worst-case time complexity until 2013, when it was improved by the $\bigo{1.1996^n}$ algorithm of \cite{xiao2017exact}. \techmemo{} is sometimes used to speed up branching algorithms \citep{chandran2005refined,fomin2005some,fomin2010exact} in the context of EEA (\textit{Exact Exponential Algorithms}), where the objective is to conceive exact algorithms that can provide a best possible worst-case running time guarantee.}

{Despite of the fact that a typical \techmemo{} algorithm memorizes solutions of subproblems that appear repeatedly, we prefer to interpret the idea in a more general way.} 
What we call the \textit{Memorization Paradigm} can be formulated as ``Memorize and learn from what have been done so far, to improve the next decisions''. 
In the literature, various algorithms can be classified as procedures embedding memorization techniques, though the implementation could be quite different depending on the problem structure and the information to store. For instance, \textit{Tabu Search} \citep{glover1989tabu,glover1990tabu} is a metaheuristic, which memorizes recently visited solutions in order to avoid returning back to these solutions again during the search. SAT solvers deduce and then memorize conflict clauses during the tree search in order to perform non-chronological backtracking (\textit{Conflict Driven Clause Learning}) \citep{biere2009conflict,zhang2001efficient}. Similar ideas also appear in \textit{Artificial Intelligence} area  as \textit{Intelligent Backtracking} or \textit{Intelligent Back-jumping}.


From a theoretical point of view, the drawback relies on the memory consumption of \techmemo{} which can be exponential. 
This drawback turns out to limit the quantity of memorized information like in \textit{Tabu Search} or SAT solvers. 
In this chapter we instantiate the \textit{Memorization Paradigm} in a way similar to what is done in the field of EEA, i.e. we set up a \techmemo{} framework for search tree based exact algorithms but with a control on the memory usage. We had the intuition that a \techmemo{} with limited memory could already dramatically accelerate  the solution  in practice. By embedding a simple \techmemo{} technique into their \techbb{} algorithm,  \cite{szwarc2001algorithmic} solve the single machine total tardiness problem on instances with up to 500 jobs in size. Other works presenting standard memorization techniques applied to sequencing problems have been published by \cite{tkindt2004revisiting} where the benefit of such technique is well shown.

As introduced in section \ref{sss:branching}, branching algorithms are based on the idea of enumerating all possibilities via a search tree created by a branching mechanism. 
For each decision variable, the algorithm \textit{branches} on all possible values, each time creating a new subproblem (a node in the search tree) of a reduced size. The algorithm continues recursively and returns the global optimal solution. The basic structure being simple, the critical question is how to prune the search tree so as to avoid exploring unpromising nodes.  Dominance conditions are commonly used to cut nodes: at a node, if it is proved that a more promising node exists or can be easily found, then the current one can be abandoned. This is also the case for \techbb{}, in which at each node, the bounding procedure provides an optimistic estimation of the solution quality of that node. If the estimation value is not better than the currently best solution found, in other words, the current node is dominated by the incumbent solution, then the node is cut without being further developed. Just like the bounding procedure in \techbb, \techmemo{} can be seen as another procedure which can help in pruning the search tree.  In branching algorithms, especially on lower branching levels, isomorphic subproblems may appear exponentially many times and  \techmemo{} can be used to avoid solving identical problems multiple times. 

Memorization, apparently, has not yet been systematically considered when designing search tree based algorithms, as the bounding procedure in \techbb. It is at least rare in sequencing problems, to the authors' knowledge. The aim of this work is to promote a systematic integration of \techmemo{} into search tree based algorithms in order to better prune the search tree. In the following sections, we first describe a general framework of \techmemo{} (section \ref{frame}), followed by some guidelines on the implementation (section \ref{impl}). Then, we apply the framework to four scheduling problems including \pbrisumc{} (section \ref{ri}), \pbdtilde{} (section \ref{dtilde}), \pbfsumc{}  (section \ref{fsumc}) and \pbtt{} (section \ref{tt}). Finally, we conclude this work in section \ref{con}.


\section{A general framework for \techmemo{} in search trees}\label{frame}
For a given minimization problem, the application of \techmemo{} depends on several components of the search tree based algorithm such as the branching scheme, the search strategies and also the characteristics of the problem. 
In this section, we consider possible scenarios that may appear for sequencing problems. Then, we present the possible schemes of \techmemo{} and how to choose the right scheme depending on the scenario.

Even though the general idea of \techmemo{} can be generalized and applied to any combinatorial optimization problems, we focus on sequencing problems in the context of the thesis. Consider a generic sequencing problem where $n$ jobs $J=\{1,...,n\}$ are to be scheduled. 
Each job $i$ is defined by a set of features like a processing time $p_i$, a due date $d_i$, etc, which depends on the problem under consideration. Some resources are available for the execution of jobs and an ordering of jobs must be found to minimize some cost function, usually depending on the jobs completion times. We adopt an intuitive way to represent the content of a node or a subproblem: as an example $123\{4,...,n\}$ represents a subproblem in which jobs $\{1,2,3\}$ are already fixed by branching, to the first three positions of the sequence, while the jobs to be scheduled after are $\{4,...,n\}$.


No matter the branching scheme, at any iteration of the algorithm, by \an{} we denote the nodes that are created but not yet developed, and by \en{} the nodes that have already been branched on (children nodes have been created). We also adopt the notion of \textit{decomposable problems} defined by \cite{tkindt2004revisiting}. Typically for single machine scheduling problems, this often implies that the completion time of the prefixed job sequence of a node is constant no matter of the order of jobs inside (it is defined as the sum of processing times of the jobs in that sequence).

\begin{definition}\label{def:dec}
Let $\{1,...,i\}\{i+1,...,n\}$ be a problem to be solved. It is \textit{decomposable} if and only if the optimal solution of the subproblem $\{1,...,i\}$ (resp. $\{i+1,...,n\}$) can be computed independently from $\{i+1,...,n\}$ (resp. $\{1,...,i\}$), i.e. without knowing the optimal sequence of $\{i+1,...,n\}$ (resp. $\{1,...,i\}$).
\end{definition}


\subsection{Branching schemes}\label{sec:bs}
In common search tree based algorithms for scheduling (sequencing) problems, the branching operation consists in assigning a job to a specific position in the sequence. A \textit{Branching Scheme} defines, at a node, how to choose this job and the positions to occupy.
We consider three classic branching schemes, namely \fb, \bb{} and \db.

When applying \fb{} at a given node, each eligible free job is assigned to the first free position. For example, the nodes at the first level of the search tree correspond to the following subproblems: $1\{2,...,n\}$, $2\{1,3,...,n\}$, $...$, $n\{1,...,n-1\}$.

When applying \bb{} at a given node, each eligible free job is assigned to the last free position. For example the nodes at the first level of the search tree correspond to the following subproblems: $\{2,...,n\}1$, $\{1,3,...,n\}2$, $...$, $\{1,...,n-1\}n$. This scheme is symmetric with \fb, hence for the sake of simplicity we only discuss \fb{}  and add extra remarks on \bb{} whenever necessary.

When applying \db{} at a given node, the job that is being considered to branch is called a \textit{decomposition job}. When a \textit{decomposition job} is assigned to a position, two subproblems are generated implied by the free positions before and after the \textit{decomposition job}. Certainly one may determine the jobs that should be scheduled before and after this position by enumerating all 2-partitions of jobs as the \textit{Divide \& Conquer} technique introduced by \cite{fomin2010exact}, but here we restrict our study to the situation where the two subproblems can be uniquely determined in polynomial time making use of some specific problem properties. 
As an example, the nodes at the first level of the search tree could contain $\{2,3,4\}1\{5,...,n\}$, if job $1$ is the \textit{decomposition job} which is assigned to position 4 and generates two subproblems corresponding to jobsets $\{2,3,4\}$ and $\{5,...,n\}$,
respectively. 
This situation occurs, for instance, to
the \pbtt{} problem as already introduced in chapter \ref{ch3}.

\subsection{Search strategies}
During the execution of search tree based algorithms, when two or more nodes are active, a strategy is needed to determine the next node to branch on. The classic search strategies are \df, \best{} and \wf.

\textit{Depth first} is the most common strategy: the node to explore is an arbitrary active node at the lowest search tree level. The advantage of this strategy is that it only requires polynomial space.

\textit{Breadth first} selects an active node with the highest search tree level. This leads to an exponential space usage since the search tree is explored level by level.

\textit{Best first} 
 chooses the node to explore according to its lower bound. The space usage in the worst case is therefore also super-polynomial like in \wf.

It seems conventional that when constructing a search tree based algorithm, the \df{} is adopted. However, this choice is strongly questionable according to \cite{tkindt2004revisiting}.

\subsection{Memorization schemes} 
The \textit{memorization} presented by \cite{robson1986algorithms} stores the optimal solution of each subproblem of
a predetermined limited size  and reuses that solution whenever such subproblem appears again during the tree search. However different memorization approaches can be used. The differences rely on the choice of the information to store and the way in which the stored information is used. We discuss below three different memorization schemes that are helpful to efficiently solve some sequencing problems. 

Taking account of the branching schemes introduced in section \ref{sec:bs}, any node of the search tree can be defined by $\sigma_1 S_1\sigma_2S_2...\sigma_kS_k$, the ${\sigma_j}'s$ being partial sequences of jobs and the ${S_j}'s$ being subproblems which remain to be scheduled. For the sake of simplicity, we explain the \techmemo{} schemes in the case of \fb, i.e. $k=1$, and a node corresponds to a problem $\sigma S$.


\subsubsection{Solution memorization}\label{solmemo}

\begin{figure}[!ht]
\centering
    \includegraphics[width=0.7\columnwidth]{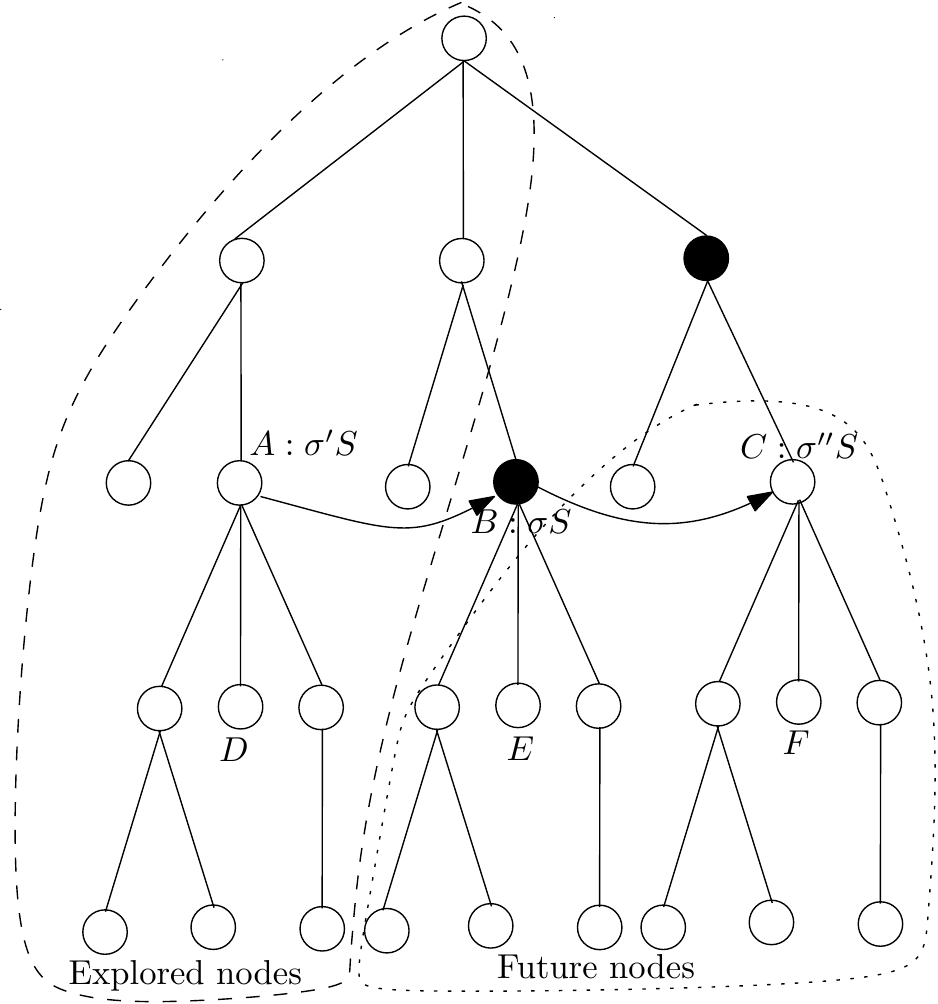}
    \caption{Solution Memorization}\label{memo1}
\end{figure}
Consider the situation illustrated in Figure \ref{memo1}, where active nodes are colored in black. Node $B$ is the current node,
while $\sigma$, $\sigma'$ and $\sigma''$ are different permutations of the same jobset. In other words, nodes $A$, $B$ and $C$ may contain the same subproblem to solve, implied by $S$. In that case, if $A$ has already been solved (consider for instance a depth-first search) and the optimal sequence of $S$ has been memorized, then it can be used directly to solve nodes $B$ and $C$ and it is no longer necessary to branch on these nodes. 
{Note that, in order to successfully perform memorization, we must guarantee that the solution of $S$ memorized at node $A$ is optimal. Depending on the branching algorithm implementation, this may not be obvious: for instance in \techbb{} algorithms, the leaf node corresponding to the optimal solution of node $A$ may be missed if one of its ascendant node is cut due to a dominance condition. Looking at Figure \ref{memo1}, assume that node $D$ should have led to the optimal solution of  problem $S$ but has been cut by a dominance condition. Applying \fmemo{} may then lead to memorize another solution $\beta$ to $S$, which is not optimal with respect to $S$. Troubles may appear if the global optimal solution to the original problem (associated to the root node) is, for instance, given by node $E$. Solution memorization may imply not exploring node $B$ and directly replacing $S$ by the ``best'' solution found from node $A$. As a consequence, the global optimal solution is missed. This situation occurs whenever the dominance condition which has pruned node $D$ would not have pruned node $E$: in the remainder, this kind of conditions are refined to as \textit{context dependent dominance conditions} since they depend on the context of each node (typically, the initial partial sequence $\sigma$, $\sigma'$ and $\sigma''$). By opposition, a \textit{context independent dominance condition} would have pruned node $D$, $E$ and $F$. 
A direct way to fix this is to disable dominance conditions whenever \fmemo{} is applied.}

{However, if these context dependent conditions are playing a very important role in the algorithm then this may slow down the algorithm even if \textit{solution memorization} works. 
Another approach to manage context dependent dominance conditions is to extend the memorization from ``solutions'' to ``lower bounds'' when the branching algorithm involves a bounding mechanism. In that version of \techmemo, we assume that all dominance conditions are kept in the algorithm. 
When node $A$ is created, a lower bound is computed, which represents the best solution value we may expect from the subtree of $A$. This lower bound is based on the cost function value of the sequence $\sigma'$ which is already fixed, and an evaluation on the unsolved part $S$. 
When branching down the subtree of $A$, jobs in $S$ are fixed gradually, hence the evaluation on the remaining unscheduled jobs also becomes more and more precise. 
When all leaf nodes of the subtree of $A$ are explored, this value finally becomes tighter (higher) than the initial value computed at node $A$. Since the objective function value of $\sigma'$ is known, we can then deduce the lower bound value corresponding to $S$ when scheduled after $\sigma'$, and memorize it. Now when node $B$ is opened, instead of computing its lower bound, we can get it by finding the lower bound of $S$ directly from the memory and then add the objective function value of $\sigma$. In this way, the lower bound we get is tighter, and node $B$ is more likely to be cut. Moreover, the lower bound computation at node $B$, which may be time costly, is saved. Notice that, for nodes cut by context dependent dominance conditions, their lower bound values still need to be computed and considered (hence introduces an extra cost). Lower bound memorization can be a good alternative to \fmemo{} with context dependent dominance conditions turned off as long as these conditions are efficient in pruning the search tree.}

{Note that the memorization of lower bounds is compatible with the memorization of optimal solutions: whenever in a subtree no nodes are cut by context dependent dominance conditions and the global upper bound is updated by some nodes from this subtree, the optimal solution of this subtree is memorized. Otherwise, the lower bound is memorized. 
We denote the described memorization technique including the memorization of optimal solutions and the memorization of lower bounds as \fmemo{} since  both are related to the memorization of the ``best solution'' of the problem associated to a node.} 

\subsubsection{Passive node memorization}
At any node $\sigma S$, another information that can be memorized is the partial sequence $\sigma$.  Unlike \fmemo{} where the memorized sequences can be used to solve a node, \pbmemo{} is only used to cut nodes. 

\begin{figure}[!ht]
\centering
    \includegraphics[width=0.6\columnwidth]{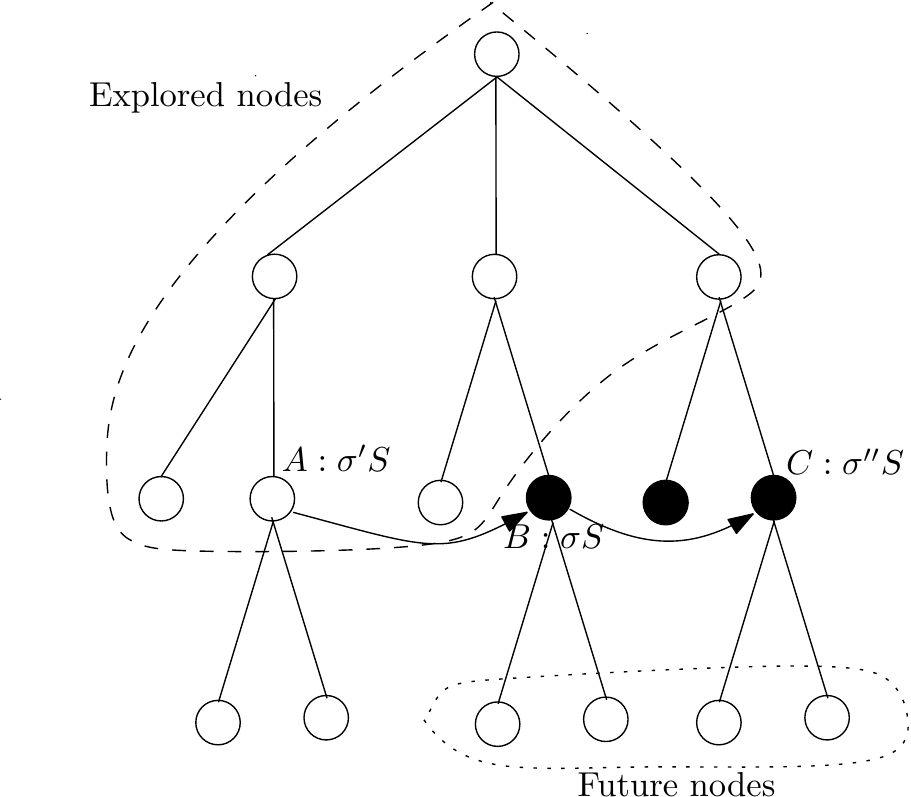}
    \caption{Passive node memorization)}\label{memo2}
\end{figure}

Consider the branching situation depicted by Figure \ref{memo2}. Again active nodes are black-colored and $B$ is the current node. Assume a node $A$ exists among explored node, with $\sigma'$ being a different permutation of the same jobset used in $\sigma$. If the partial sequence $\sigma'$ has been memorized then two situations may occur. If $\sigma'$ dominates $\sigma$ then $B$ can be cut since it cannot lead to a solution better than $A$. If no such $\sigma'$ dominating $\sigma$ is available, then $\sigma$ can be memorized in order to 
possibly prune a future node like $C$. Note that \fmemo{} and \pbmemo{} may possibly intersect. Consider the previous example and nodes $A$, $B$ and $C$. If the optimal solution of subproblem $S$ has been obtained from the exploration of node $A$, then at node $B$ and $C$ both \fmemo{} and \pbmemo{} imply not to branch on these nodes if $\sigma'$ dominates $\sigma$ and $\sigma''$.

The dominance test between sequences can be implemented as a function $check(\sigma,\sigma')$ which returns \texttt{1} if $\sigma'$ dominates $\sigma$, as introduced by \cite{tkindt2004revisiting}. The check must be done on two different sequences of the same jobset, having the same starting time and its implementation is problem dependent. Since the memorized sequence results from branching decisions, we call it \pbmemo.

Any node, ready to be branched on, must be compared to explored and/or to active nodes depending on the search strategy. Additionally, it may be necessary to perform the check twice: 
first once the node is created, then at the time of branching. Memorizing the partial sequence, when the node is created, ensures that the best sequence is kept before any exploration of nodes. Then, rechecking the dominance when branching on a node enables that node to be cut if a dominant partial sequence has been found meanwhile.  

In the following, we introduce  Property \ref{proper:con} which relates the lower bounding mechanism of search tree based algorithms to the $check$ function. {When this property is answered, the current node only needs to be
compared to explored nodes instead of all nodes when best first is chosen as the search strategy, as detailed in Section \ref{sec:dg}.}

\begin{proper}(Concordance property)\label{proper:con}
Let LB(A) be the lower bound value computed at node $A$. The search tree based algorithm satisfies the concordance property if and only if, for any node $A=\sigma S$ and $B=\pi S$, $LB(A)<LB(B) \Leftrightarrow check(\pi,\sigma)=1$.
\end{proper}

\subsubsection{Predictive node memorization}\label{sec:prenodememo}
\textit{Predictive node memorization} relies on the same idea as \pbmemo, but with additional operations. As illustrated in Figure \ref{memo3}, at a given node $B=\sigma S$, we first check, like in \pbmemo{}, if the current node can be cut by $\sigma'$ memorized at node $A$. If not, instead of directly memorizing $\sigma$, we search for an improving sequence $\pi$. Notice that, by the way, the improving sequence necessarily belongs to a part of the search tree not yet explored when dealing with the node $\sigma S$. There may be many ways to compute $\pi$. For instance, we may perform some local search on $\sigma$, searching for a neighbor sequence $\pi$ that dominates $\sigma$. Alternatively, we may focus on a short subsequence of $\sigma$ and solve it to optimality (in a brute-force way, for instance). The latter idea appears in the work of  \cite{jouglet2004branch} as \textit{Dominance Rules Relying on Scheduled Jobs}. 
We may also make use of an exact algorithm to optimize a part of $\sigma$ to get $\sigma'$, as far as this algorithm is fast. Notice that this idea is strongly related to the merging mechanism introduced in chapter \ref{ch3} (see also \cite{garraffa2017bm}), which is designed to provide good worst-case time complexities. 
If such a sequence $\pi$ can be built, then the current node $\sigma S$ is cut and node $\pi S$ is memorized.
\begin{figure}[!ht]
\centering
    \includegraphics[width=0.6\columnwidth]{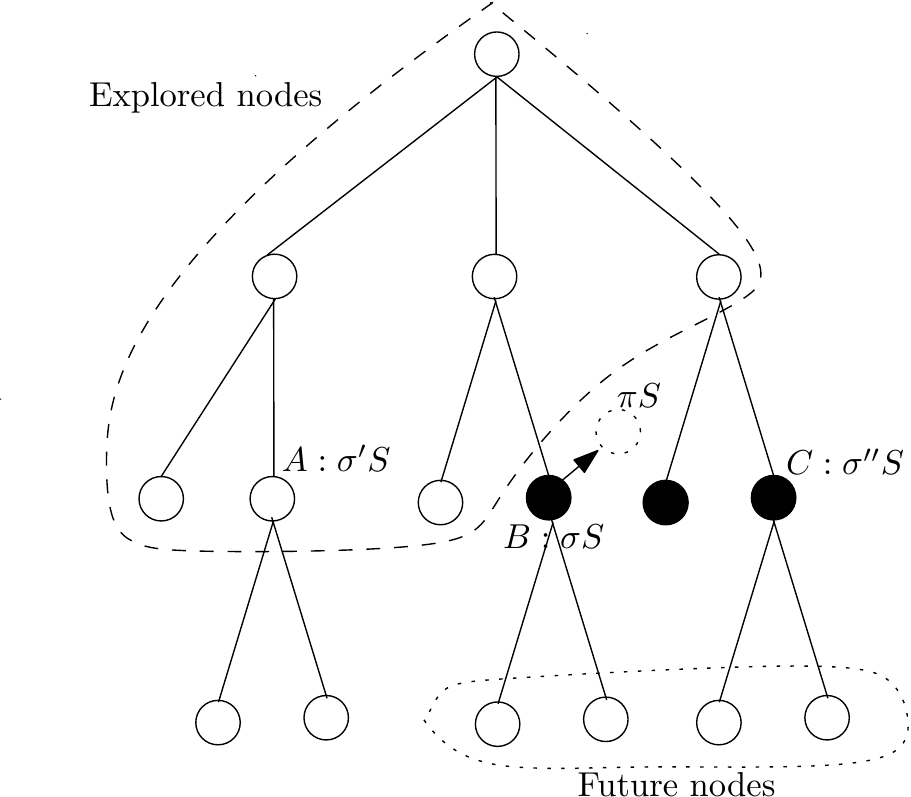}
    \caption{Predictive node memorization}\label{memo3}
\end{figure}
Note that node $\pi S$ has not yet been encountered in the search tree when dealing with node $\sigma S$ (consider, $\pi=\sigma''$). So, it is important when applying \abmemo{} to remember that $\pi S$ still needs to be branched on. 
Also, the extra cost of generating $\pi$ must be limited in order to avoid excessive CPU time consumption.

\subsection{Decision guidelines}\label{sec:dg}
In this section we provide some guidelines on how to choose the appropriate memorization scheme according to the branching scheme and the search strategy. The main results are summarized in the decision tree in Figure \ref{fig:dt}.

\begin{figure}
\centering \includegraphics[width=\columnwidth]{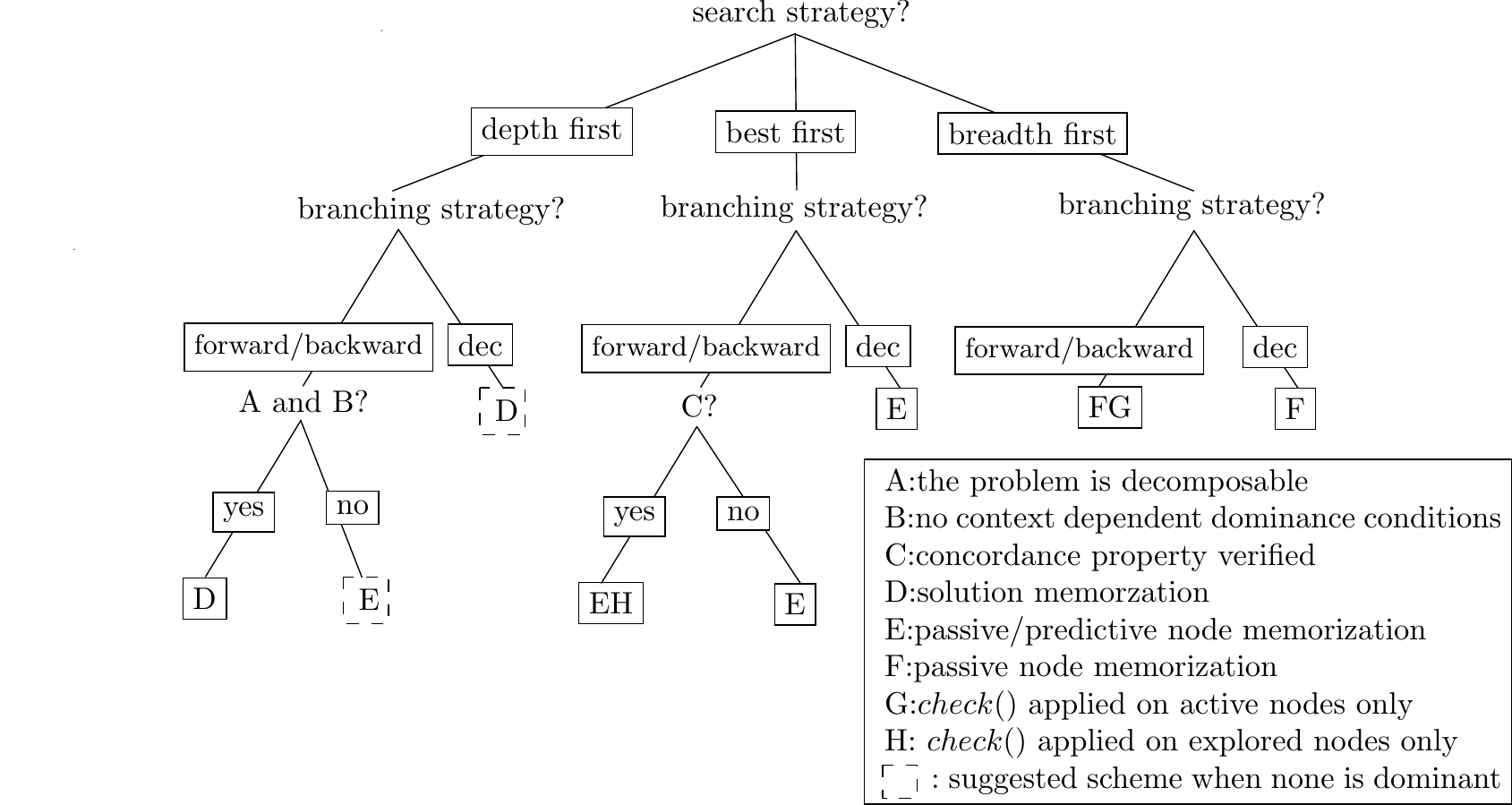}
\caption{Decision tree for choosing the memorization scheme}\label{fig:dt}
\end{figure}

\subsubsection{Forward branching and depth first search strategy}\label{fbdf}
In \fb, any node of the search tree can be defined as $\sigma S$. When \textit{depth first} is used as the search strategy we can state the following property.

\begin{proper}\label{proper:fbdf}
With forward branching and depth first, if the problem is decomposable and  \fmemo{}  memorizes optimal solutions, then \fmemo{} dominates both \pbmemo{} and \abmemo. 
\end{proper}
\begin{proof}
Any node deletion that can be achieved by \pbmemo{} and \abmemo{} can also be achieved by \fmemo, but not conversely. 
Consider nodes $A=\sigma S$ and $B=\pi S$ with $\sigma$ and $\pi$ two permutations of the same jobset. 
As the problem is decomposable, solving subproblem $S$ at node $A$ is equivalent to solving it at node $B$. Without loss of generality, we assume that 
$A$ appears before $B$ during the solution. In \pbmemo{} if $check(\pi,\sigma)=1$, i.e. sequence $\sigma$ dominates $\pi$, then $B$ can be pruned. 
However, also in \fmemo{}, node $B$ can be pruned since the optimal solution of jobset $S$ has already been memorized from node $A$.

Now consider the case where $check(\sigma,\pi)=1$, i.e. sequence $\pi$ dominates $\sigma$. This implies that with \pbmemo, node $B$ will not be pruned. However, as explained above, with \fmemo, node $B$ is pruned. With \abmemo, the conclusion is the same since we have no guarantee that starting from node $\pi S$ another node $\alpha S$ dominating $\pi S$ can be generated. Besides, even if such $\alpha S$ is generated and $\pi S$ is pruned, the same issue occurs to node $\alpha S$ when it is generated.
\end{proof}

If the problem is not decomposable, or \textit{context dependent dominance conditions} are used in the algorithm, then \fmemo{} memorizes lower bounds, and then which memorization scheme is dominant can not be determined. However, in practice, \pbmemo{} may be preferred to \fmemo. Notably, as the problem is not decomposable, then it may be necessary to solve the  subproblem consisting of jobset $S$ both at nodes $A$ and $B$. However, with \pbmemo, node $B$ may be pruned whenever $\pi$ is dominated by $\sigma$. 

\subsubsection{Forward branching and best first search strategy}
We can state the following property.
\begin{proper}\label{prop:fbbf}
With forward branching and best first strategy, \fmemo{} does not apply. \textit{Passive node memorization} or \abmemo{} can be applied only to \en{} if the concordance property (Property \ref{con}) is answered.
\end{proper}
\begin{proof}
To apply \fmemo{} at a given node, the subproblem concerning $S$ must be solved first in order to memorize its optimal solution. This is not compatible with \best{} search strategy. In fact, if \best{} reaches a leaf node, then also the optimal solution is reached and no sequence has been stored before. 

When \pbmemo{} and \abmemo{} are applied, the search, at a given node, of a dominant subsequence needs only to be done in the set of explored nodes whenever the {concordance} property holds. 
As the \best{} search strategy always consider for branching the node with the lowest lower bound value, the {concordance} property implies that no active node can dominate it.
\end{proof} 
When the {concordance} property does not hold, then node memorization techniques are required to consider both explored and active nodes for node pruning. 

Besides, no dominance can be deduced a priori between
\abmemo{} and \pbmemo{}. It depends on how the search for an improving subsequence is applied in \abmemo. Generally speaking, both memorization schemes should be considered and compared to find the best one.

\subsubsection{Forward branching and breadth first search strategy}
With forward branching and breadth first we can state the following property. 
\begin{proper}\label{prop:breadth_fmemo}
With forward branching and breadth first strategy, \fmemo{} does not apply. 
\textit{Passive node memorization} should be chosen and should be applied to active nodes.
\end{proper}
\begin{proof}
Under this configuration, \fmemo{} is useless since leaf nodes are reached only at the end of the search tree.
\textit{Passive node memorization} can be applied to \an{} only. An active node $A$ is selected for branching when all the nodes at the same level have been created, hence all other active nodes dominated by $A$ are discarded. If in turn $A$ it is dominated by another node, then it is pruned. 
There is no need to consider explored nodes since explored nodes on higher levels have less fixed jobs, therefore they are not comparable with the current node. 
Also, \textit{predictive node memorization} can not do better than \pbmemo{} since \pbmemo{} already keeps the best node at each level. 
\end{proof}


\subsubsection{Decomposition branching and depth first search strategy}\label{dbdf}
With \db, at each level of the search tree a decomposition job can be put on any free position by the branching operation. 

Under this configuration, no dominance can be deduced among the memorization schemes. In fact, we can imagine situations 
where either \fmemo{} or \pbmemo{} or \abmemo{} is dominant. Consider nodes $A=\sigma S_1 j_1 S_2$ and $B=\pi S_1 j_2 S_3$ with $A$ being explored before $B$. In both nodes, the current subproblem concerns scheduling jobset $S_1$ after $\sigma$ or $\pi$. Suppose $\sigma$ and $\pi$ contain different jobs but have the same completion time, which means that the subproblem defined by $S_1$ is identical in $A$ and $B$. Then, the optimal sequence for $S_1$ found when solving $A$ can be reused on $B$ by \fmemo, while \pbmemo{} cannot handle this case since $\sigma$ and $\pi$ contain different jobs hence, are incomparable. \textit{Predictive node memorization} may or may not cut $B$ depending whether a dominant prefix can be generated or not.

On the other hand, we may also imagine the case where $A=\sigma S_1 j_1 S_2$ and $B=\pi S_3 j_2 S_4$. Suppose $\sigma$ and $\pi$ are different permutations of the same jobset. If $check(\pi,\sigma)=1$, then node $B$ can be cut by \pbmemo{} or \abmemo, while this is not the case for \fmemo{} because subproblems $S_1$ and $S_3$ do not consist of the same jobs.

In practice, even though every memorization scheme could be dominant in some cases, the memory limitation does not allow to apply all of them and our experience suggests to prefer \fmemo. This is due to the special structure of nodes $\sigma_1 S_1...\sigma_k S_k$, which makes the prefixed jobs much spread out {(they are separated by $S_i$)}, and prevents the application of successful \pbmemo{} and \abmemo. Moreover, the case with nodes $\sigma_1 S\sigma_2$ and $\pi_1 S\pi_2$, where $\sigma_1$ and $\pi_1$ have the same completion time but contain different jobs, may occur pretty often for large instances if the jobs processing times do not present a large variance.

\subsubsection{Decomposition branching and best first search strategy}\label{dbbf}

\begin{proper}\label{prop:dec_best}
With decomposition branching and best first strategy,  \fmemo{} does not apply. \textit{Passive node  memorization} and \abmemo{} must only be applied to explored nodes whenever the concordance property holds  and the $check$ function comparing two nodes $\sigma_1S_1...\sigma_k S_k$ and $\sigma_1'{S_1}'...{\sigma_{k'}}' {S_{k'}}'$ only works on $\sigma_1$ and ${\sigma_1}'$. Otherwise, \pbmemo{} and \abmemo{} must be applied to explored and active nodes.
\end{proper}
\begin{proof}
Similar to that of Property \ref{prop:fbbf}.
\end{proof}

\subsubsection{Decomposition branching and breadth first search strategy}\label{dbwf}
 \begin{proper}\label{prop:dbwf}
With decomposition branching and breadth first strategy, \fmemo{} does not apply. Whether \textit{node memorization} should be applied to active nodes only depends on the definition of the $check$ function.
\end{proper}
\begin{proof}
This configuration discourages \fmemo{} for the same reason as in Property \ref{prop:breadth_fmemo}.
If the $check$ function is defined in a way such that the explored nodes are not comparable to active nodes, then \pbmemo{} and \abmemo{} should be applied to \an{} only, otherwise they should be applied to all nodes.
\end{proof}


\FloatBarrier
\section{Implementation guidelines}\label{impl}
In this section we discuss efficient implementations of the memorization schemes, providing, when necessary, choices specific to the sequencing problems tackled in the remainder. The key point is to have a fast access to memorized partial solutions. Henceforth, we implement a database as a hashtable which contains all the memorized solutions. 
By well choosing the hash function, a hashtable supports querying in $\bigo{1}$ time to find the corresponding elements given a hash key.

For \fmemo{}, at a given node, the database is queried with $\langle t_0, S \rangle$, where $t_0$ is the starting time of the subproblem and $S$ is the related jobset. The returned result should be $\langle \pi, opt(\pi | t_0)\rangle$ which is the optimal sequence associated to $S$ when starting at time $t_0$, and its corresponding objective function value. So, $\langle \pi, opt(\pi | t_0)\rangle$ defines the elements which are memorized in the database. We define the hash key $h$ as a combination of $t_0$ and $|S|$: seeing $h$ as a set of bits, $t_0$ occupies the higher bits in $h$ while $|S|$ occupies the lower bits. The aim is to have a unique hash key for each given pair $\langle t_0, S \rangle$, even if this is not necessarily bijective: i.e., two elements in the database with the same hash key may correspond to different pairs $\langle t_0, S \rangle$. 
As a consequence, when a list of elements is returned for a pair $\langle t_0, S \rangle$, 
it is also necessary to verify that the returned sequence is a sequence of jobset $S$. This takes $\bigo{|S|}$ operations for each returned sequence. We may also include the sum of job id's of $S$ into $h$ in order to have a more exact key, but this correspondingly increases the time needed to construct the key, without preventing from checking whether a returned sequence $\pi$ is a permutation of jobset $S$ or not.

For \textit{passive} and \textit{predictive node memorization}, implementation decisions are more dependent on the problem and on the check function used to compare two partial sequences $\sigma$ and $\pi$ of the same jobset. 
For any such $\sigma$ and $\pi$, a general definition of $check()$ could be:
\begin{equation}\label{eq:check}
check(\pi,\sigma)=
\begin{cases}
1,\ if\ C_{max}(\sigma)\leq \max(C_{max}(\pi); E_{min}(\pi))\ and\ opt(\sigma | t_0)\leq opt(\pi| t_0)\\
0,\ otherwise
\end{cases}
\end{equation}

with $C_{max}$ referring to the makespan of a partial sequence, and $E_{min}(\pi)$ referring to the earliest starting time of the jobs scheduled after $\pi$. {It is not difficult to see that if $check(\pi,  \sigma)=1$ then node $\sigma S$  dominates node $\pi S$. Indeed, for any regular  objective function to minimize, with respect to the fixed jobs, $opt(\sigma | t_0)\leq opt(\pi| t_0)$ ensures that $\sigma$ yields a smaller cost than $\pi$. Moreover, $C_{max}(\sigma)\leq \max(C_{max}(\pi); E_{min}(\pi))$ guarantees that the starting time of  jobset $S$ at node $\sigma S$ is not higher than  in  node $\pi S$. Therefore, $\sigma S$ dominates $\pi S$.}

Consequently, a database element is a tuple $\langle \sigma,  C_{max}(\sigma), E_{min}(\sigma), opt(\sigma|t_0), ExpAct\rangle$ with $ExpAct$ being a flag indicating whether this element corresponds to an explored or an active node. Notice that $t_0$ is not included since it appears in the hash key used for querying. Also, when the problem is decomposable, the $check$ function reduces to:
$$
check(\pi,\sigma)=
\begin{cases}
1,\ if\ opt(\sigma| t_0)\leq opt(\pi| t_0)\\
0,\ otherwise
\end{cases}
$$
where only $\langle \sigma,  opt(\sigma|t_0), ExpAct\rangle$ need to be stored.
For node memorization techniques the hash key, at a given node, is computed in a way similar to solution memorization. Consider,
for example, \fb{}: let $\sigma_1 S_1$ be the current node. As the dominance of another node is checked on $\sigma_1$, the database is queried with $\langle0, S_{\sigma_1} \rangle$ with $S_{\sigma_1}$ referring to the set of jobs in $\sigma_1$. Then, only $S_{\sigma_1}$ needs to be binary encoded into the hash value.

With respect to the database management, notice that when an element is added, in node memorization techniques, then the elements dominated by the added one are removed. Besides, due to memory limitation on the computer used for testing, we may need to clean the database when it is full on some instances. More precisely, in our experiments, the RAM is of 8Gb and hence the database size is also limited to 8Gb.

A cleaning strategy is needed to remove unpromising elements, i.e. those that are expected not to be used for pruning the search tree. As it is not clear which elements are unpromising,  several strategies have been tested. We have implemented the following ones during our experimentations.

\subsubsubsection{FIFO: First In First Out}
This is one of the most common database cleaning strategy: when the memory is full, we first remove the first added elements. An extra structure is needed to record the order of elements according to the time when they are added. When the database if full and a long sequence is waiting to be inserted, it may be necessary to remove more than one elements in order to free enough space.

\subsubsubsection{BEFO: Biggest Entry First Out}
This cleaning strategy suggests to remove from the database the biggest elements (longest sequences) in order to free enough continuous memory for storing new elements. 
For \fmemo{} it means removing nodes at higher levels in the search tree. An intuition of the impact of this cleaning strategy on \fmemo{} can be sketched from Figure \ref{fig:db1} which presents the number of sequences memorized per size for an instance of the \pbtt{} scheduling problem with 800 jobs. It can be seen that sequences with ``large number of jobs'' (let's say more than 500 jobs) are not often used to prune nodes, and even if some large nodes could have been useful for node pruning, we may still expect that the solution of its subproblems generated by one or several branching can be found from the memory.

\begin{figure}[!ht]
    \centering
    \includegraphics[width=0.7\textwidth]{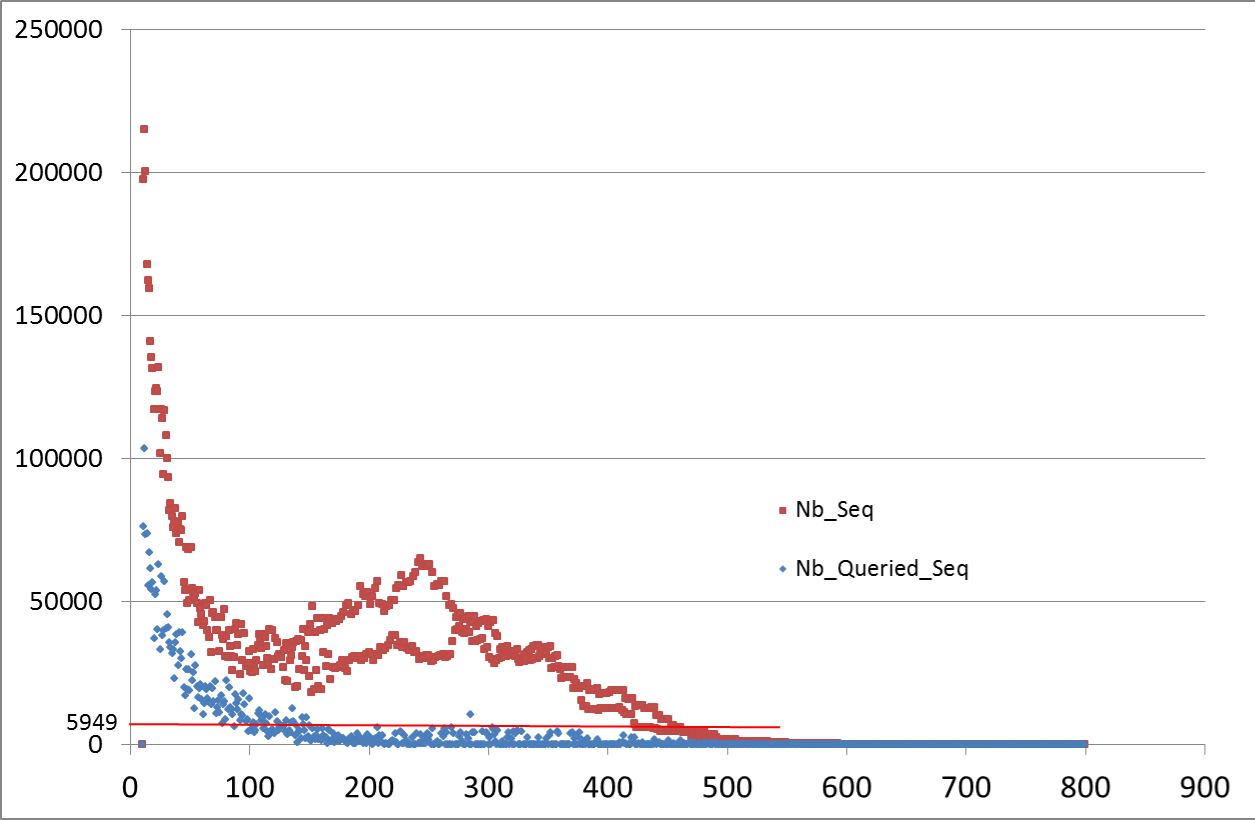}
    \begin{flushleft}
    \footnotesize
    Nb\_Seq: number of sequences of a given size, stored in the memory.\\
    Nb\_Queried\_Seq: number of sequences of a given size, that are used to avoid solving twice identical problems.
    \end{flushleft}
    \caption{Number of solutions and useful solutions in memory for an instance of \pbtt{} with 800 jobs}
    \label{fig:db1}
\end{figure}

However, for \textit{passive and predictive node memorization}, the strategy means removing nodes at lower levels of the search tree. These
nodes refer to subproblems with many jobs already fixed (and memorized) and few jobs to schedule. It may be possible that the extra cost of memorizing a long fixed partial sequence is not inferior to solving the corresponding small subproblem directly without memorization. Since this is not obvious from a theoretic point of view, some preliminary experiments were performed in order to investigate whether it is better also to remove nodes at higher levels of the search tree in \textit{node memorization}. Computational testing confirms that removing longest elements is always preferred, at least on problems \pbrisumc, \pbdtilde{} and \pbfsumc.

At each cleaning, we also tend to clean up a large amount of space in order to decrease the time cost induced by frequent cleaning operations.

\subsubsubsection{LUFO: Least Used First Out}
Figure \ref{fig:db1} also suggests another cleaning strategy since a lot of sequences are never used to prune nodes in the search tree. These sequences can be removed from the database to save space. To implement the LUFO cleaning strategy we 
keep an usage counter for each database element. The counter is incremented by $1$ each time the element is queried and used to prune a node in the search tree, and it is decremented by $1$ when a cleaning operation is performed. Elements whose counter is zero are removed by the cleaning operation. Note that in \textit{node memorization}, when a database element is replaced by a new one, the latter should inherit the counter value of the old one. This is because that the counter value reflects the usefulness of a solution and the counter value of a newly added solution should not be smaller than the counter values of solutions that are dominated by the new one.

Preliminary results, not reported here, show that FIFO strategy is not efficient for the considered scheduling problems. BEFO strategy works better than FIFO, but its efficiency is not high enough to make a difference in the computational results. LUFO strategy is proved to be surprisingly efficient.

\section{Application to the \pbrisumc, \pbdtilde{} and \pbfsumc{} problems}
In order to experiment the effectiveness of \textit{Memorization} on  scheduling problems, we first test it on three problems that were considered by \cite{tkindt2004revisiting}. In that work the authors used memory to apply the so-called \textit{DP property} over nodes in order to prune the search tree. According to the memorization framework as defined in this chapter, what they have done is \pbmemo{} with a database cleaning strategy which replaces the shortest stored sequence by the new one when the database is full. 
The aim of that paper was also on choosing the most suitable search strategy when trying to solve these problems efficiently. In this section, for each of these three problems we apply the previously defined framework of \techmemo{} with various considerations and discuss the  obtained results. 
For each problem, we compare several \techbb{} algorithms which are named according to their features: Depth-, Best- and Breadth- refer to \techbb{} algorithms with the corresponding search strategies and no memorization included. Depth\_X, Best\_X and Breadth\_X refer to a \techbb{} algorithm with corresponding search strategies and \textit{memorization} X used, with $X=S$ representing \fmemo, $X=Pa$ representing \pbmemo{} and $X=Pr$ the \abmemo. 
For \abmemo, we use \textit{k-perm} heuristic to search for new sequences, as described in section \ref{sec:prenodememo}.

\textit{k-perm} heuristic also refers to a ``dominance condition relying on scheduled jobs'' as introduced by \cite{jouglet2004branch}. At a given node $\sigma S$, 
assume that $\sigma=\sigma^0\sigma^k$ with $\sigma^k$ the subsequence of the $k$ last jobs in $\sigma$, $k$ being an input parameter. The \textit{k-perm} heuristic consists in enumerating all permutations of jobs in $\sigma^k$ to obtain sequence $\sigma^\ell$. Then, the first found sequence $\sigma^0\sigma^\ell$ dominating $\sigma^0\sigma^k$, if it exists, is used to prune node $\sigma S$. 
The dominating sequence can be memorized. The notion of dominance between sequences is the one used to define the $check$ function in \textit{node memorization}. Preliminary tests suggest us to choose $k=5$ in our implementations in order to have the most efficient \abmemo{} \textit{scheme}.

Notice that \textit{k-perm} search is not performed when \textit{breadth first} strategy is used, since the memorization applied on active nodes already covers the effect of \textit{k-perm}. 

The algorithms proposed by \cite{tkindt2004revisiting} in 2004 are also tested on the same dataset and they are named as Depth\_Pa\_04, Best\_Pa\_04 and Breadth\_Pa\_04, respectively.  {Compared to our algorithms Depth\_Pa, Best\_Pa and Breadth\_Pa, the main differences are that the RAM usage is limited to  450M for Depth\_Pa\_04, Best\_Pa\_04 and Breadth\_Pa\_04, in order to obtain similar results to that  reported in 2004. Also, in our algorithms 
LUFO is chosen as the database cleaning strategy.}

The test results on instances of certain sizes are marked as $OOT$ (out of time) if any of the instances is not solved after 5 hours. Analogously, with the application of \textit{Memorization}, memory problems may occur and the limit on RAM usage may be reached, reported as $OOM$ (out of memory). Note that according to our experiments, even when memory cleaning strategies are applied, $OOM$ may still occur due to the fragmentation of the memory after a number of cleanings. Also note that LUFO is chosen as the cleaning strategy according to preliminary experimentations. 

All tests have been done on a HP Z400 work station with 3.07GHz CPU and 8GB RAM. The datasets of all tests are open for public access from \url{http://www.vincent-tkindt.net/index.php/resources}.

\subsection{Application to the \pbrisumc{} problem}\label{ri}
The \pbrisumc{} problem asks to schedule $n$ jobs on one machine to minimize the sum of completion times. Each job $i$ has a processing time $p_i$ and a release date $r_i$ before which the job cannot be processed. The problem is NP-hard in the strong sense and it has been widely studied in the literature with both exact and heuristic algorithms considered. 
The referential computational results so far are done by \cite{tkindt2004revisiting}, in which with \textit{forward branching} and \textit{best first} and the application of a so-called \textit{DP Property} the algorithm is able to solve instances with up to 130 jobs. A mixed integer programming approach  is also reported by \cite{kooli2014mixed}, which enables to solve instances with up to 140 jobs. However only 5 instances are generated in their experiments for each set of parameters. This makes their result less convincing due to the fact that the hardness of instances varies a lot even when generated with the same parameters, as observed during our study. 
Consequently, we consider in this section that the \techbb{} algorithm provided by \cite{tkindt2004revisiting} is at least as efficient as the approach used by \cite{kooli2014mixed}. 
In contrast, 30 instances are generated by \cite{tkindt2004revisiting} for each set of parameters , which leads to 300 instances for each size. 

The work of \cite{tkindt2004revisiting} uses the \techbb{} algorithm of \cite{chu1992branch} as a basis, and so \fb{} is adopted as the branching strategy. 
With respect to search strategies, \textit{depth first}, \textit{best first} and \textit{breadth first} were all tested by \cite{tkindt2004revisiting}, aiming to explore the impact of different search strategies on the efficiency of the algorithm. The lower bounds and dominance conditions of \cite{chu1992branch} are kept. A so called \textit{DP Property}, added as a new feature in the algorithm of \cite{tkindt2004revisiting}, is actually equivalent to \pbmemo{} in our terminology. The $check()$ function is based on a dominance condition given by \cite{chu1992branch} and it was defined by \cite{tkindt2004revisiting} as follows:

\begin{equation}\label{check:1ri}    
check(\pi,\sigma)=
\begin{cases}
1,\ if\ opt(\sigma | 0)\leq opt(\pi| 0)\ and\ \\
\hspace{0.5cm} opt(\sigma | 0)+|\Omega|*E_{min}(\sigma)\leq opt(\pi | 0)+|\Omega|*E_{min}(\pi)\ \\
0,\ otherwise
\end{cases}
\end{equation}

with $\Omega$  the jobs that remain to be scheduled after sequence $\sigma$ and $\pi$. We also have $E_{min}(\sigma)=\max(C_{}(\sigma), \min_{\substack{r\in \Omega}}r_i)$, with $C_{}(\sigma)$ the completion time of $\sigma$. The item stored into memory is a tuple $\langle \sigma,  C_{}(\sigma), opt(\sigma|0), ExpAct\rangle$ and $E_{min}(\sigma)$ can be computed when needed. Note that this definition of $check$ is an adaption of the general  Equation \ref{eq:check} and if the $check$ in  Equation \ref{eq:check} return 1 then this $check$ also returns 1.

\subsubsection{Application of the memorization framework and improved results}
The problem is not \textit{decomposable} due to the existence of release dates. Therefore, with the choice of \fb, \textit{node memorization} should be chosen according to the decision tree in Figure \ref{fig:dt}. 
{The lower bound used in the algorithm is based on the SRPT (Shortest Remaining Processing Time) rule. Together with the $check()$ function defined in Equation \ref{check:1ri}, it is not clear whether the concordance property is answered. 
Hence, when \pbmemo{} is applied upon \textit{best first}, all nodes need to be considered for the comparisons, while when it is applied with \textit{breadth first}, only active nodes need to be considered.} 
Therefore the choices made by \cite{tkindt2004revisiting} with respect to memorization are kept. The $check()$ function also remains the same, as defined in Equation \ref{check:1ri}.




Here we refresh the computational results of \cite{tkindt2004revisiting} on new randomly generated input and also add  results for \abmemo{} and \fmemo. The input is generated following the way described by \cite{chu1992branch}, i.e. the processing times are generated uniformly from  $[1,100]$ and release dates are generated between 0 and $50.5\cdot n\cdot r$ with $r$ belonging to $\{0.2, 0.4, 0.6, 0.8, 1.0, 1.25, 1.50, 1.75, 2.0, 3.0\}$. 30 instances are generated for each value of $r$, hence leading to 300 instances for each size $n$ from 70 to 140. 

We first ran the algorithms of \cite{tkindt2004revisiting} (Depth\_Pa\_04, Best\_Pa\_04 and Breadth\_Pa\_04) on these newly generated instances. 
All these three algorithms are able to solve instances with up to 110 jobs. On the running time, Best\_Pa\_04 is the fastest one, with a maximum running time of 27 seconds for instances with 110 jobs. This value becomes 40 seconds for Depth\_Pa\_04 and 1082 seconds for Breadth\_Pa\_04. It was reported by \cite{tkindt2004revisiting} that the algorithms with best first and breadth first can solve respectively instances with 130 jobs and 120 jobs which is different from what we obtain. 
This reveals that the newly generated instances are harder than that of \cite{tkindt2004revisiting}, knowing that the computational power of our test environment is much better than  in 2004. Moreover, we  also observed that for Depth\_Pa\_04, the hardest instance of 110 jobs is solved in 40 seconds but the hardest instance of 90 jobs is solved in 1691 seconds. This is why we think that the hardness of instances vary a lot when generated randomly.

Results related to \textit{node memorization} are put in Table \ref{tab:risumc}. 
For all the three search strategies, \pbmemo{} enables to solve much larger instances with respect to the versions without memorization. This is sufficient to prove the power of memorization on this problem.

Depth\_Pa, Best\_Pa and Breadth\_Pa are all more powerful than their counterparts of 2004, i.e. Depth\_Pa\_04, Best\_Pa\_04 and Breadth\_Pa\_04. This is especially visible on \textit{depth first}, where Depth\_Pa\_04 solves instances with up to 110 jobs while Depth\_Pa solves instances with up to 130 jobs. This makes Depth\_Pa  the global best algorithm and shows that when more physical memory is available and a larger database with an appropriate cleaning strategy is set, the memorization can be further boosted and the gain can be important.

The impact of \textit{k-perm} search on this problem is very limited: \abmemo{} basically leads to the same result as \pbmemo. 
{In addition, we also tested \textit{solution memorization} on this problem since no theoretical dominance between the memorization schemes can be established for this problem. Since context dependent dominance conditions are enabled in the algorithm, we first disabled them in order to obtain the optimal solution of each node. But this turned out to be very inefficient. Therefore, we also implemented the memorization of lower bounds, as described in section \ref{solmemo}. 
However, the resulting algorithm can only solve instances with up to 80 jobs, hence not competitive compared to \textit{node memorization}, as predicted according to the decision tree in Figure \ref{fig:dt}.}

It is also worth to be mentioned that the database cleaning strategy LUFO enables a faster solution of large instances. As an example, we found an instance with 140 jobs is solved in 1.6 hours by Depth\_Pa with LUFO, while it needs 14 hours to be solved when the cleaning strategy of \cite{tkindt2004revisiting} is kept instead. However, due to the hardness of another instance with 140 jobs, the algorithm Depth\_Pa is finally out of time.

\begin{table}[!ht]
    \centering
\resizebox{\textwidth}{!}
{
    \begin{tabular}{|c|c|c|c|c|c|c|c|c|c|}
    \hline
    &n &70&80&90&100&110&120&130&140 \\ \hline
\multirow{4}{*}{Depth-}&Navg&141247.8&1778751.2&OOT&&&&&                                                                        \\ 
&Nmax&17491232&276190737&&&&&&                                                                       \\ 
&Tavg&1.8&22.4&&&&&&                                                                                     \\ 
&Tmax&217&3238&&&&&&                                                                                 \\ \hline
\multirow{4}{*}{Depth\_Pa}&Navg&2583.4&5756.2&18639.9&26827.4&48502.9&174545.5&192409.4&OOT                        \\ 
&Nmax&147229&314707&2253897&644151&1281097&16575522&7742714&    \\ 
&Tavg&0.0&0.0&0.3&0.7&1.3&7.1&9.1&                                          \\ 
&Tmax&2&7&64&27&41&754&295&                                    \\ \hline
\multirow{4}{*}{Depth\_Pr}&Navg&1771.1&4455.1&12625.7&19621.7&30380.4&117865.6&128277.5&OOT           \\
&Nmax&                          82765	&267416	&1455743	&588429	&1096520	&11126694&5132228&   \\
&Tavg&0.0&0.0&0.3&0.5&0.9&4.7&6.6&\\
&Tmax&                          1	&7	&46	&28	&39	&488&252&                                                                          \\\hline
Best-&&OOT&&&&&&&                                                                   \\\hline
\multirow{4}{*}{Best\_Pa}&Navg&1230.5&3299.4&5235.1&9494.8&13658.5&38574.5&43986.9&OOT \\
&Nmax&36826&256534&292929&216293&228848&2675337&1449900&                                  \\
&Tavg&0.0&0.2&0.2&0.4&0.6&15.3&11.8&                                                                    \\
&Tmax&0&46&38&27&25&3595&1630&                                                            \\\hline
\multirow{4}{*}{Best\_Pr}&Navg&1229.6&3298.2&5229.0&9490.7&13545.7&38560.1&43989.8&OOT           \\
&Nmax&36826	&256529	&292927	&216037	&228832	&2674776&1449872&\\.
&Tavg&0.0&0.2&0.2&0.4&0.7&15.4&11.9&                                                                           \\
&Tmax&1	&47	&39	&28	&25	&3579&1636&                                                                         \\\hline
Breadth-&&OOT&&&&&&&                                                                   \\\hline
\multirow{4}{*}{Breadth\_Pa}&Navg&1947.7&6745.0&9893.8&21308.5&27383.1&OOT&&                         \\
&Nmax&90494&709607&733980&575430&1209481&&&                                               \\
&Tavg&0.0&4.6&3.4&5.3&5.7&&&                                                                        \\
&Tmax&9&1319&897&483&935&&&    \\\hline
    \end{tabular}}
    \caption{Results of new algorithms on the \pbrisumc{} problem}
    \label{tab:risumc}
\end{table}

\FloatBarrier
\subsection{Application to the \pbdtilde{} problem}\label{dtilde}
The \pbdtilde{} problem asks to schedule $n$ jobs on a single machine. Each job $i$ has a processing time $p_i$, a weight $w_i$ and a deadline $\tilde{d_i}$ which has to be answered. The objective is to minimize the total weighted completion time $\sum w_iC_i$. The problem is NP-hard in the strong sense and has been solved by \techbb{} algorithms \citep{posner1985minimizing, potts1983algorithm}, with the algorithm of \cite{posner1985minimizing} being slightly superior. The basic algorithm described by \cite{tkindt2004revisiting} is a combination of algorithms of \cite{posner1985minimizing, potts1983algorithm} by incorporating the lower bound and the dominance condition of \cite{posner1985minimizing} into the \techbb{} algorithm of \cite{potts1983algorithm}. 
With respect to search strategies, all the three strategies, i.e. \df, \best{} and \wf{} were considered by \cite{tkindt2004revisiting} and \bb{} is adopted as the branching scheme as in the algorithms of \cite{posner1985minimizing, potts1983algorithm}. Similarly to what is done on the \pbrisumc{} problem, the \textit{DP Property} is also considered by \cite{tkindt2004revisiting}, which is actually \pbmemo. The $check()$ function is defined as follows, where $\Omega$ is the set of jobs to be scheduled before $\sigma$ and $\pi$. 

\begin{equation}\label{check:dtilde}    
check(\pi,\sigma)=
\begin{cases}
1,\ if\ opt(\sigma|\sum_{i\in \Omega}p_i)\leq opt(\pi|\sum_{i\in \Omega}p_i) \\
0,\ otherwise
\end{cases}
\end{equation}

The items stored in the database are $\langle \sigma, opt(\sigma|\sum_{i\in \Omega}p_i), ExpAct\rangle$. In 2003,  \cite{pan2003improved} proposed another \techbb{} algorithm with reported experiments showing that it can solve to optimality all instances with up to 90 jobs in size. As the testing protocol is identical to the one used in 2004 by \cite{tkindt2004revisiting}, we can conclude that the algorithm of Pan is outperformed by the best one proposed by \cite{tkindt2004revisiting} (which is reported as being able to solve instances with up to {130} jobs in size).

\subsubsection{Application of the memorization framework and improved results}

This problem is \textit{decomposable} according to Definition \ref{def:dec}. 
From the decision tree in Figure \ref{fig:dt} we can derive that with the depth first search strategy, {\textit{solution memorization} should be considered, even though its superiority over \textit{node memorization} depends on the presence of \textit{context dependent dominance conditions} in the algorithm}. In the paper of \cite{tkindt2004revisiting}, \textit{node memorization} was implemented with that strategy. Consequently, in this section we compare four \techbb{} algorithms: the three versions of \cite{tkindt2004revisiting}, i.e. \textit{node memorization} applied to the three search strategies and a version based on depth first with \textit{solution memorization}. 

The concordance property is answered (see Proposition \ref{propo:condtilde}) and hence the \pbmemo{} only considers explored node when the search strategy is \textit{best first}, and only active nodes need to be considered in \textit{breadth first}. For \fmemo, the items stored into the memory are $\langle \pi, opt(\pi|0) \rangle$. 
For \textit{node memorization}, the $check()$ function and the items stored remain the same as in the work of \cite{tkindt2004revisiting}, as described in the previous section. 

{About \fmemo, context dependent dominance conditions  are enabled in the algorithm. Their removal has been experimentally proved to lead to an inefficient algorithm. Therefore, lower bounds are memorized during the \fmemo, as described in section \ref{solmemo}. 
}
\begin{propos}\label{propo:condtilde}
With the $check()$ function defined in Equation \ref{check:dtilde}, our algorithms verify the concordance property (Property \ref{proper:con}).
\end{propos}
\begin{proof}
Consider two nodes $S\sigma $ and $S\pi$. First notice that the subproblem to solve in both nodes are the same, which consists in scheduling  jobs from $S$ starting from time $0$. 
The lower bound used in the algorithm (see \cite{posner1985minimizing,potts1983algorithm}) returned on the subproblems on $S$ are the same for the two nodes. Therefore, if $check(\pi,  \sigma)=1$, which means $opt(\sigma|\sum_{i\in \Omega}p_i)\leq opt(\pi|\sum_{i\in \Omega}p_i)$, then $LB(S\sigma)\leq LB(S\pi)$.

With the same reasoning, if $LB(S\sigma)\leq LB(S\pi)$, it can be deduced that the relation $opt(\sigma|\sum_{i\in \Omega}p_i)\leq opt(\pi|\sum_{i\in \Omega}p_i)$ must hold, and hence $check(\pi,  \sigma)=1$.
\end{proof}

Following the test plan described by \cite{potts1983algorithm}, for each job $i$, its processing time $p_i$ is an integer generated randomly from the uniform distribution $[1,100]$ and its weight $w_i$ is generated uniformly from $[1,10]$. The total processing time $P=\sum_{i=1}^np_i$ is then computed and for each job $i$ an integer deadline $d_i$ is generated from the uniform distribution $[ P( L - R/2), P( L +R/2)]$, with $L$ increase from $0.6$ to $1.0 $ in steps of $0.1$ and $R$ increases from $0.2$ to $1.6$ in steps of $0.2$. In order to avoid generating infeasible instances, a $(L,R)$ pair is only used when $L + R / 2 > 1$, hence only 20 $(L,R)$ pairs are actually used, for each of which 10 feasible instances are generated,  yielding a total of 200 instances for each value of $n$ from 40 to 140. 
We first present the results of \pbmemo{} algorithms (Depth\_Pa\_04, Best\_Pa\_04 and Breadth\_Pa\_04) from  \cite{tkindt2004revisiting}. 
Both Depth\_Pa\_04 and Best\_Pa\_04 are stated by \cite{tkindt2004revisiting} to solve instances with up to 110 jobs. However, they are only capable of solving instances with 70 jobs on the newly generated instances, with a maximum solution time 11 seconds and 285 seconds, respectively. Breadth\_Pa\_04 was reported to be able to solve instances with up to 130 jobs in 2004 but this falls down to 100 jobs in our tests with a maximum solution time of 36 seconds. 
This difference is not negligible and it reveals the fact that the newly generated instances seem much harder than those generated by \cite{tkindt2004revisiting}.

The results of the new algorithms are presented in Table \ref{tab:dtilde}. On \df, without \textit{memorization} the program is ``out of time'' on instances with 50 jobs, while both \fmemo{} and \pbmemo{} enable to solve instances with up to 100 jobs, with \pbmemo{} running faster. With the activation of \textit{k-perm} search, Depth\_Pr enables to solve 30 more jobs than Depth\_Pa. This strongly proves the power of all the three memorization schemes. It also worth to be noticed that Depth\_Pa solves instances with 30 more jobs with respect to Depth\_Pa\_04, knowing that the only differences between these two algorithm are that the database size in Depth\_Pa is larger 
and the database cleaning strategy is different. 

For \best, the same phenomenon can be observed, that is, Best\_Pr is more efficient than Best\_Pa, which is better than Best- and Best\_Pa\_04. Best\_Pr can also solve instances with up to 130, and faster than Depth\_Pr.

Breadth\_Pa is the most powerful algorithm among all. It is surprising to see that without \textit{memorization} Breadth- cannot even solve all instances of 40 jobs, while with \pbmemo{} instances of 130 jobs are all solved in an average solution time of 27 seconds. Again, as for the \pbrisumc{} problem, LUFO allows to accelerate the solution but it did not enable to solve larger instances.

\begin{table}[!ht]
    \centering
    \resizebox{\textwidth}{!}
{
    \begin{tabular}{|c|c|c|c|c|c|c|c|c|c|c|c|c|}
    \hline
    &n &40&50&60&70&80&90&100&110&120&130&140 \\ \hline
\multirow{4}{*}{Depth-}&Navg&116827.7&OOT&&&&&&&&&\\
&Nmax&14536979&&&&&&&&&&\\
&Tavg&1&&&&&&&&&&\\
&Tmax&74&&&&&&&&&&\\ \hline
\multirow{4}{*}{Depth\_S}&Navg&772.0		&2718.0	&6706.0	&28463.0	&114970.0		&139382.0	&563209.0   &OOT&&&\\
                          &Nmax&17699	&60462	&137207	&1660593&6180097	&2803714&	12335703   &&&&\\
                          &Tavg&0.4		&0.5.0		&1.0		&3.0		&15.0			&12.0		&113.0   &&&&\\
                          &Tmax&1		&2		&6		&275	&1544		&474	&	5346   &&&&\\ \hline
\multirow{4}{*}{Depth\_Pa}&Navg&559.1&2091.3&5240.1&20068.3&75727.4&139429.8&376206.9&OOT&&&\\
&Nmax&11963&83075&94189&1004546&1960891&4321070&5549747&&&&\\
&Tavg&0.4&0.4&0.5&0.9&2.9&5.5&17.0&&&&\\
&Tmax&1&1&2&39&157&312&515&&&&\\ \hline

\multirow{4}{*}{Depth\_Pr}&Navg&326.4&901.7&2184.1&6825.3&20429.0&32531.0&90375.3&266689.8&574824.4&1397463.6&OOT\\
&Nmax&3431	&17447	&28677	&187425	&665376	&768802 &1781123&14713483&11236833&103699138&\\
&Tavg&0.4&0.4&0.4&0.6&1.0&1.0&3.8&14.4&37.9&108.2&\\
&Tmax&0	&1	&1	&5	&30	&21 &51&901&1255&8732&\\ \hline
Best-&&OOT&&&&&&&&&&\\ \hline
\multirow{4}{*}{Best\_Pa}&Navg&334.6&879.3&1859.0&7159.9&16581.8&27259.5&60349.0&OOT&&&\\
&Nmax&3800&20889&25574&440623&547165&1252600&798372&&&&\\
&Tavg&0.4&0.4&0.4&0.7&1.3&1.8&4.3&&&&\\
&Tmax&0&1&1&30&73&130&91&&&&\\ \hline
\multirow{4}{*}{Best\_Pr}&Navg&292.0&708.9&1486.9&4312.7&10301.7&14642.9&31891.2&145203.1&239837.4&330474.1&{OOT}\\
&Nmax&                         2435	&11762	&18051	&120259	&319068	&276507&332022&10659343&7578570&6712266&\\
&Tavg&0.4&0.4&0.4&0.5&0.8&1.0&1.9&45.6&20.0&26.7&\\
&Tmax&                         0		&1		&1		&4		&23		&13&25&5008&1137&716&\\ \hline
Breadth-&&OOM&&&&&&&&&&\\ \hline
\multirow{4}{*}{Breadth\_Pa}&Navg&348.6&940.9&1833.1&6533.4&14964.4&23725.0&53309.2&102633.5&239512.5&329902.3&OOT\\
&Nmax&4701&16952&24559&437697&453506&868876&789310&5975094&7577492&6702080&\\
&Tavg&0.0&0.0&0.0&0.2&0.4&0.8&2.0&10.1&20.0&26.7&\\
&Tmax&0&0&0&9&15&31&36&1353&1135&718&\\ \hline
    \end{tabular}
    }
    \caption{Results of the new algoritihms on the \pbdtilde{} problem}
    \label{tab:dtilde}
\end{table}

\FloatBarrier
\subsection{Application to the \pbfsumc{} problem}\label{fsumc}
The \pbfsumc{} problem asks to schedule $n$ jobs are to be scheduled on two machines $M_1$ and $M_2$. Each job $i$ needs first to be processed on $M_1$ for $p_{1,i}$ time units then be processed on $M_2$ for $p_{2,i}$ time. The objective is to minimize the sum of completion times of jobs. 
We restrict to the set of permutation schedules in which there always exist an optimal solution. A permutation schedule is a schedule in which the jobs sequences 
on the two machines are the same. 
The problem is NP-hard in the strong sense. 
Up to 2016, the best exact algorithm was the \techbb{} algorithm proposed by \cite{tkindt2004revisiting} and based on the \techbb{} algorithm of \cite{della2002improved}. Recently,  \cite{detienne2016two} proposed a new and very efficient \techbb{} algorithm capable of solving instances with up to 100 jobs in size. This is definitely the state-of-the-art exact method for solving the \pbfsumc{} problem. However, in order to evaluate the impact of using \techmemo{} in a \techbb{} algorithm we make use of the algorithms described by \cite{tkindt2004revisiting} since their code was directly available to us.

The adopted branching scheme in this algorithm is \fb{} and all the three search strategies were considered. The \textit{DP Property} is also considered by \cite{tkindt2004revisiting}, which is actually \pbmemo. The $check()$ function is based on a result reported by \cite{della2002improved} and is defined as follows:

\begin{equation}\label{check:fsumc}    
check(\pi,\sigma)=
\begin{cases}
1,\ if\ opt(\sigma|0)\leq opt(\pi|0)\ and \ |\Omega|*(C_{2}(\sigma)-C_{2}(\pi))\leq opt(\pi)-opt(\sigma) \\
0,\ otherwise
\end{cases}
\end{equation}

 where $\Omega$ is the set of jobs to be scheduled after $\sigma$ and $\pi$, $C_2(\cdot)$ is the completion time of a given sequence on the second machine. The items stored into the database are $\langle \sigma, C_2(\sigma), opt(\sigma|0), ExpAct\rangle$.

\subsubsection{Application of the memorization framework and improved results}
This problem is not \textit{decomposable} since 
given a partial solution of the form $\sigma S$ with $\sigma$ a fixed sequence, the optimal solution of subproblem $S$ depends on the order of jobs in $\sigma$. 
From the decision tree in Figure \ref{fig:dt} we can derive that with the depth first search strategy,  \textit{solution memorization} should be considered, even though its superiority over \textit{node memorization} depends on the presence of \textit{context dependent dominance conditions}. \textit{Node memorization} was implemented with that strategy by \cite{tkindt2004revisiting}. Consequently, in this section we compare four \techbb{} algorithms: the three versions of \cite{tkindt2004revisiting}, i.e. \textit{node memorization} applied to the three search strategies and a version based on depth first with \textit{solution memorization}. 

{With the $check()$ function defined in Equation \ref{check:fsumc} and the lower bound (a Lagrangian Relaxation based lower bound) used in the algorithm,   the concordance property is not answered. We performed experiments to look for the case where for two nodes $\sigma S$ and $\pi S$, $check(\pi, \sigma)=1$ but $LB(\pi)<LB(\sigma)$ and we found it. Therefore, the concordance property is not verified and both active and explored nodes need to be considered for \textit{best first} strategy. For \textit{breadth first} strategy, only active nodes need to be considered.}

{For \fmemo{}, since context dependent dominance conditions  are enabled in the algorithm, and they are important for a fast solution of the problem, lower bounds are memorized during the \fmemo, as described in section \ref{solmemo}. 
The items stored into the memory are $\langle \pi, t_1, t_2, C_2(\pi), opt(\pi|(t_1,t_2)) \rangle$  where $t_1$ is the actual starting time of $\pi$ on the first machine and $t_2$ is the actual starting time of $\pi$ on the second machine. Besides, $opt(\pi|(t_1,t_2))$ is the sum of completion times of jobs in $\pi$, when  $\pi$ starts at time $t_1$ on the first machine and time $t_2$ on the second machine.  For \textit{node memorization}, the $check()$ function and the stored item remain the same as in the work of \cite{tkindt2004revisiting}, as described in the previous section.}

30 instances are generated for each size $n$ from 10 to 40, with the processing
times generated randomly from an uniform distribution in $[1, 100]$. Again, we ran the algorithms of \cite{tkindt2004revisiting} on these newly generated instances.  
Depth\_Pa\_04 is able to solve instances of 40 jobs which is 5 jobs more than reported by \cite{tkindt2004revisiting}, with a maximum solution time about 3.4 hours. Both Best\_Pa\_04 and Breadth\_Pa\_04 solve instances with up to 35 jobs, as  reported by \cite{tkindt2004revisiting}, with maximum solution times of 43 seconds and 806 seconds, respectively.

Other results are given in Table \ref{tab:fsumc}.
Depth- is able to solve instances with 35 jobs. Best- is able to solve instances with 30 jobs and Breadth- can only solve up to 25 jobs. 
With \pbmemo{} enabled, Depth\_Pa solves instances with 5 more jobs than Depth-. Best\_Pa and Breadth\_Pa solve instances with 10 more jobs than the versions without \techmemo. 
With respect to algorithms X\_Pa\_04, algorithms X\_Pa use a larger database, more precisely, the maximum number of solutions that can be stored is set to 6000000 instead of 350000. This enables Best\_Pa to solve 5 more jobs than Best\_Pa\_04. However, Depth\_Pa and Breadth\_Pa are not able to  solve larger instances with respect to Depth\_Pa\_04 and Breadth\_Pa\_04, 
even though they solve instances faster. Notice that the \pbfsumc{} problem is a really hard problem, certainly more difficult than the two other problems previously tackled.

Also, the LUFO strategy is adopted for database cleaning but it did not enable to solve larger instances without having an ``Out of Time'' problem.

\textit{Predictive node memorization} is not more efficient than \pbmemo: in fact no nodes are cut by undertaking a \textit{k-perm} search. The result is hence even slightly slower due to the time consumed  by the call to the  \textit{k-perm} heuristic. 
Depth\_S solve instances with less nodes generated compared to Depth-. However, its efficiency is even less than Depth-, due to the processing of lower bound memorization.

From a global point of view, the power of memorization is also  illustrated on this problem, since we always have benefits in using it. As a perspective for this problem it could be interesting to evaluate the contribution of \techmemo{} when embedded into the state-of-the-art algorithm of \cite{detienne2016two}.


\begin{table}[!ht]
    \centering
    \resizebox{\textwidth}{!}
{
    \begin{tabular}{|c|c|c|c|c|c|c|c|c|c|}
    \hline
    &n&10&15&20&25&30&35&40&45\\ \hline
\multirow{4}{*}{Depth-}&Navg&23.7&255.6&4137.7&21460.4&317102.0&3615780.0&{OOT}&\\
&Nmax&84&2367&83863&311742&3097479&53187978&&\\
&Tavg&0.0&0.0&0.1&0.8&26.0&423.0&&\\
&Tmax&0&0&2&17&248&6128&&\\ \hline
\multirow{4}{*}{Depth\_S}&Navg&24.0&228.0&3561.0&19733.0&294355.0&3425633.0 &OOT&\\
&Nmax&					  84	&1735	&68070	&273146	&2712580&49360565&&\\
&Tavg&						0.0&0.0&0.1&1.0&29.0&497.0                     &&\\
&Tmax&						0	&0	&2	&15	&248	&6933                &&\\ \hline
\multirow{4}{*}{Depth\_Pa}&Navg&22.8&187.2&1573.0&8205.0&61337.0&337194.0&1894037.2&OOT\\
&Nmax&80&1083&17114&48459&291750&1568506&15472612&\\
&Tavg&0.0&0.0&0.0&0.1&4.1&35.0&328.3&\\
&Tmax&0&0&0&2&21&163&3627&\\ \hline
\multirow{4}{*}{Depth\_Pr}&Navg&22.8&187.2&1573.0&8205.0&61361.3&337194.0&1894037.0&OOT\\
&Nmax&80&1083&17114&48459&291016&1568506&15472612&\\
&Tavg&0.0&0.0&0.0&0.1&4.1&32.8&332.8&\\
&Tmax&0&0&0&2&23&173&3664&\\ \hline
Best-&Navg&23.7&249.3&3993.1&21717.7&291131.9&OOM&&\\
&Nmax&84&2253&83863&311742&2451152&&&\\
&Tavg&0.0&0.0&0.1&0.7&19.1&&&\\
&Tmax&0&0&2&17&197&&&\\\hline
\multirow{4}{*}{Best\_Pa}&Navg&20.9&139.5&957.3&4780.7&28957.0&112229.8&495186.5&OOM\\
&Nmax&72&624&6646&21022&152797&426641&3617824&\\
&Tavg&0.0&0.0&0.0&0.0&1.2&7.8&80.6&\\
&Tmax&0&0&0&1&4&43&1253&\\\hline
\multirow{4}{*}{Best\_Pr}&Navg&20.9&139.5&957.3&4780.7&28957.0&112229.8&495186.5&OOM\\
&Nmax&                         72&624&6646&21022&152797&426641&3617824&\\
&Tavg&                        0.0&0.0&0.0&0.0&1.4&8.3&83.1&\\
&Tmax&                         0&0&0&1&5&45&1283&\\\hline
Breadth-&Navg&23.9&266.1&5181.8&39303.6&OOT&&&\\
&Nmax&84&2360&83863&311742&&&&\\
&Tavg&0.0&0.0&0.1&1.6&&&&\\
&Tmax&0&0&2&17&&&&\\\hline
\multirow{4}{*}{Breadth\_Pa}&Navg&21.0&148.8&1369.5&8889.1&115219.2&345109.6&OOT&\\
&Nmax&72&692&9927&63485&2242263&2357023&&\\
&Tavg&0.0&0.0&0.0&0.2&26.1&54.2&&\\
&Tmax&0&0&0&3&711&665&&\\\hline
    \end{tabular}
    }
    \caption{Results of new algorithms on the \pbfsumc{} problem}
    \label{tab:fsumc}
\end{table}

\FloatBarrier

\section{Application to the \pbtt{} problem}\label{tt}

In this section, we report the results of the application of \techmemo{} on solving the \pbtt{} problem, already tackled in chapter \ref{ch3} by \techbm{} on the theoretical aspect.  
We first recall some main properties of the problem, then determine parameters for \techmemo{} and finally report the computational results.




\subsection{Preliminaries}
The \pbtt{} problem asks to schedule a set of $n$ jobs $N=\{1,2, \dots, n\}$ on a single machine to minimize the sum of tardiness. 
The current state-of-the-art exact method in practice is a \techbb{} algorithm (named as BB2001 in this chapter) which solves to optimality problems with up to~$500$ jobs in size~\citep{szwarc2001algorithmic}. The main properties of the problem have already been introduced in section \ref{sec:bm:intro}. Some of them are reminded below for the ease of reference. 

Let $(1,2,\dots,n)$ be a~LPT (Longest Processing Time first) sequence and $([1],[2],\dots,[n])$ be an~EDD (Earliest Due Date first) sequence of all jobs. 

We first introduce two important decomposition properties, which are the same as  Property \ref{dec1ch4} and Property \ref{dec2ch4},  respectively, from chapter \ref{ch3}.

\begin{dec}\label{dec1ch4}~{\citep{lawler1977pseudopolynomial}}
(Lawler's decomposition)
Let job~$1$ in LPT sequence correspond to job~$[k]$ in EDD sequence.
Then, job~$1$ can be set only in positions $h\geq k$ and
the jobs preceding and following job~$1$ are uniquely determined as
$B_1(h) = \{[1],[2],\dots,[k-1],[k+1],\dots,[h]\}$ and $A_1(h) =
\{[h+1],\dots,[n]\}$. 
\end{dec}




\begin{dec}\label{dec2ch4}~{\citep{szwarc1999solution}}
Let job~$k$ in LPT sequence correspond to job~$[1]$ in EDD sequence.
Then, job~$k$ can be set only in positions $h\leq (n-k+1)$ and
the jobs preceding job~$k$ are uniquely determined as
$B_k(h)$, where $B_k(h)\subseteq \{k+1,k+2,\dots,n\}$ and $\forall i\in B_k(h), j\in \{n,n-1,\dots,k+1\}\smallsetminus B_k(h)$, $d_i\leq d_j$
\end{dec}

The two above decomposition rules can be applied simultaneously to derive a decomposing branching scheme called \textit{Double Decomposition} \citep{szwarc2001algorithmic}. At any node, let $S_i$ be a set of jobs to schedule. Note that some other jobs may have already been fixed on positions before or after $S_i$, implying a structure like $\sigma_1S_1\sigma_2S_2...\sigma_iS_i...\sigma_kS_k$ over all positions, but a node only focuses on the solution of one subproblem, induced by one subset of jobs ($S_i$ here). 
With depth first, which is the search strategy  retained in the \techbb{}  BB2001, the \textit{Double Decomposition} is always applied on $S_1$. This works as follows. 
First find the longest job $\ell$ and the earliest due date job $e$ in $S_1$. Then apply Decomposition \ref{dec1ch4} (resp. Decomposition \ref{dec2ch4}) to get the lists $L_\ell$ (resp. $L_e$) of positions, on which $\ell$ (resp. $e$) can be branched on. As an example, suppose $L_e=\{1,2\}$ and $L_\ell=\{5,6\}$. Then, a double branching can be done by fixing job $e$ on position $1$ and fixing job $\ell$ on position $5$, decomposing the jobset $S_i$ to three subsets (subproblems): the jobs before jobs job $e$, which is $\emptyset$; the jobs between $e$ and $\ell$; and finally the jobs after $\ell$. In the same way, the other 3 branching can be performed by fixing jobs $e$ and $\ell$ in all compatible position pairs: $(1,6)$, $(2,5)$ and $(2,6)$.

When branching from a node, another particular decomposition may occur as described in Property \ref{split}. Assume that a given subset of jobs $S$ is decomposed into two disjoint subsets $B$ and $A$, with $B\cup A=S$ and all jobs of $B$ are scheduled before that of $A$ in an optimal schedule of $S$: $(B,A)$ is then called an optimal block sequence and Property \ref{split} states when does such decomposition occur. In that case Decomposition \ref{dec1ch4} and Decomposition \ref{dec2ch4} are not applied but two child nodes are rather created each one corresponding to one block of jobs ($A$ or $B$), following Property \ref{split} (also called the \textit{split} property).

Let $E_j$ and $L_j$ be the earliest and latest completion times of job $j$, that is if $B_j$ (resp. $A_j$) is the currently known jobset that  precedes (resp. follow) job $j$, then $E_j=p(B_j)+p_j$, and $L_j = p(N\smallsetminus A_j)$.

\begin{proper}\label{split}~{\citep{szwarc1999solution}}
(\textit{Split})\\
($B,A$) is an optimal block sequence if $\max_{i\in B}L_i\leq \min_{j\in A}E_j$. 
\end{proper}

The value of $E_i$ and $L_i$ of each job $i$ can be obtained by applying Emmons' conditions \citep{emmons1969one} following the $\bigo{n^2}$ procedure provided by \cite{szwarc1999solution}. 

An initial version of \fmemo{} has been already implemented in BB2001, even though it was called \textit{Intelligent Backtracking} by the authors. Remarkably, lower bounds are not used in this \techbb{} algorithms due to 
the ``Algorithmic Paradox'' (Paradox \ref{para1}) found by  \cite{szwarc2001algorithmic}.  This one shows that the power of \techmemo{} largely surpasses the power of the lower bounding procedures in the algorithm.

\begin{para}\label{para1}
``...deleting a lower bound drastically improves the performance of the algorithm...''
\end{para}

Paradox \ref{para1} is simply due to the fact that a lot of identical subproblems occur during the exploration of the search tree. The computation time required by lower bounding procedures to cut these identical problems are much higher than simply solving that subproblems once, memorizing the solution and reusing it whenever the subproblem appears again. Besides, pruning nodes by the lower bound may negatively affect memorization since the nodes that are cut cannot be memorized.

The BB2001 algorithm uses a depth first strategy and for each node to branch on, the following procedure is applied:
\begin{enumerate}
\item Search the solution of the current problem, defined by a set of jobs and a starting time of the schedule, in ``memory'', and return the solution if found; otherwise go to 2.
\item Use Property \ref{split} to split the problem and solve each new subproblem recursively starting from step 1. If no split can be done, go to step 3.
\item Combine Decompositions \ref{dec1ch4} and \ref{dec2ch4} to branch on the longest job and the smallest-due-date job to every candidate positions. For each branching case, solve subproblems recursively, then store in memory the best solution among all branching cases and return it.
\end{enumerate}

The complexity of this algorithm has been analyzed in section \ref{sec:comp_bb2001} as an additional result. Note that due to Paradox \ref{para1}, all lower bounding procedures are removed, which makes the \techbb{} algorithm a simple branching algorithm. Notice that \fmemo{} can be implemented in BB2001 as suggested in section \ref{impl}. In BB2001, when the database of stored solutions is full, no cleaning strategy is used and no more partial solutions can be stored. The memory limit of this database in BB2001 is not mentioned by \cite{szwarc2001algorithmic}.

\subsection{Application of the memorization framework and improved results}
We take the reference algorithm BB2001 as a basis, in which \db{} and \fmemo{} are already chosen. 
The \db{} has been proved to be very powerful, and there is no evidence that other branching schemes like \fb{} or \textit{backward branching} can lead to a better algorithm (see \cite{szwarc2001algorithmic}). The problem is \textit{decomposable} according to Definition \ref{def:dec}. The main discussion relies on the relevancy of considering \textit{node memorization} instead of \fmemo. As already mentioned in section \ref{dbdf}, it is not obvious to implement \textit{node memorization}, for a decomposing branching scheme, which could outperform the \fmemo. Here a node is structured as $\sigma_1S_1...\sigma_kS_k$ with the ${\sigma_i}'s$ being the partial sequences to memorize in \textit{node memorization}. Assume we have two nodes $\sigma_1S_1...\sigma_kS_k$ and $\pi_1S_1'...\pi_\ell S_\ell'$, it is not apparent to find $\sigma_i$ and $\pi_j$, $i\in\{1,..,k\}$, $j\in\{1,..,\ell\}$, such that $\sigma_i$ and $\pi_j$ are of same jobs and have the same starting time. Moreover it seems complicated to design an efficient $check()$ function deciding which of these two nodes is dominating the other. We found no way to implement \textit{node memorization} which could hopefully lead to better results than those obtained with \fmemo. Consequently, \fmemo{} only is considered and, as sketched in sections \ref{dbbf} and \ref{dbwf}, there is no interest in considering \textit{best first} or \textit{breadth first} search strategies.

Henceforth, the choices done by \cite{szwarc2001algorithmic} with respect to memorization were good choices. In the remainder we investigate limitations of the memorization technique as implemented by \cite{szwarc2001algorithmic} and propose improvements which significantly augment the efficiency of the algorithm.




Our algorithm is based on BB2001, with two main changes.

Since the memory usage was declared as a bottleneck of BB2001, we firstly retest BB2001 on our machine: a HP Z400 work station with 3.07GHz CPU and 8GB RAM. 
200 instances are generated randomly for each problem size using the same generation scheme as in the paper of  \cite{potts1982decomposition}. Processing times are integers generated from an uniform distribution in the range $[1, 100]$ and due dates $d_i$ are integers from a uniform distribution in the range $[p_iu, p_iv]$ where $u=1-T-R/2$ and $v=1-T+R/2$. Each due date is set to zero whenever its generated value is negative. Twenty combinations $(R,T)$ are considered where $R\in\{0.2, 0.4, 0.6,0.8,1\}$, and $T\in\{0.2, 0.4, 0.6,0.8\}$. 
Ten instances are generated for each combination and the combination $(R=0.2, T=0.6)$ yields the hardest instances as reported in the literature (see \cite{szwarc1999solution}) and confirmed by our experiments. Again, all data can be found from \url{http://www.vincent-tkindt.net/index.php/resources}.

Table \ref{tab:tt} presents the results we obtain when comparing different algorithms. For each version we compute the average and maximum CPU time $T_{avg}$ and $T_{max}$ in seconds for each problem size. The average and maximum number of explored nodes $N_{avg}$ and $N_{max}$ are also computed. The time limit for the solution of each instance is set to 4 hours, and the program is considered as OOT (Out of Time) if it reaches the time limit. Also, when memorization is enabled without a database cleaning strategy, the physical memory may be saturated by the program, in which case the program is indicated as OOM (Out of Memory).
 
Our implementation of BB2001 solves instances with up to 900 jobs in size as reported in Table \ref{tab:tt}, with an average solution time of 764s and a maximum solution time of 9403s for 900-job instances, knowing that the original program, as tested in 2001 was limited to instances with up to 500 jobs due to memory size limit. Their tests were done on a Sun Ultra-Enterprise Station with a reduced CPU frequency (\textless 450MHz) and a RAM size not stated. It is anyway interesting to see that with just the computer hardware evolution, \techmemo{} is augmented to solve instances with 400 jobs more.
 
BB2001 is out of time (\textgreater 4h) for instances with 1000 jobs, and the memory size seems no longer to be the bottleneck. The first improvement we propose presume on the vein of Paradox \ref{para1}.

\begin{para}\label{para2}
Removing \textit{Split} procedure (Property \ref{split}) from BB2001 drastically accelerate the solution.
\end{para}

The effect of Paradox \ref{para2} is astonishing. The resulting algorithm \textit{NoSplit} solves instances with 700 jobs with an average solution time 20 times faster: from 192 seconds to 9 seconds (see Table \ref{tab:tt}). In fact, \textit{Split} is performed based on precedence relations between jobs, induced by the computation of the ${E_j}'s$ and ${L_j}'s$. The computation of these precedence relations is time consuming in practice. Moreover, as already claimed, many identical problems appear in the search tree and the \textit{Split} procedure in BB2001 is run each time. 
When \textit{Split} is removed, identical problems are solved needing more time when first met, but then never solved twice thanks to \fmemo{}. However, the disadvantage is also clear: more solutions are added to the database and hence the database is filled faster than when \textit{Split} is enabled. This is why \textit{NoSplit} encounters memory problems on instances with 800 jobs. 
This was not considered by \cite{szwarc2001algorithmic} because \textit{Split} is a very strong component of the algorithm and the computer memory at that time also discourages this tentative.
$S_i,SDD_2$

At this point, we have a better understanding of the power of \fmemo{} on this problem and we become curious on the effectiveness of memorized solutions. In other words, what are the stored solutions that are effectively used? 
To answer this question, we test cleaning strategies as defined in section \ref{impl}, to remove useless solutions when the database memory is ``full''. The most efficient strategy is proved to be LUFO by preliminary experiments not reported here. Embedding such a memory cleaning strategy is our second contribution to BB2001 algorithm.

\bigskip
In Table \ref{tab:tt}, the final implementation of the memorization mechanism within the \techbb{} algorithm for the \pbtt{} problem is referred to as NoSplit\_LUFO.
All 200 instances with 1200 jobs are solved, with an average solution time of 192 seconds, while BB2001 is limited to instances with 900 jobs. 

\begin{table}[!ht]
    \centering
    \resizebox{\textwidth}{!}
{
\begin{tabular}{|c|c|c|c|c|c|c|c|c|c|c|c|c|}\hline
&n &300&400&500&600&700&800&900&1000&1100&1200&1300 \\ \hline
\multirow{4}{*}{Depth-}&Navg&	46046201&	OOT&	&	&	&&&&&&\\
&Nmax&	2249342615&	&	&&&&&&&&	\\
&Tavg&	155&	&	&	&&&&&&&\\
&Tmax&	6499&	&	&	&&&&&&&\\ \hline
\multirow{4}{*}{BB2001}&Navg&	61501&	136452&	290205&	560389&	880268&	1534960&	    2053522&	OOT&&&	\\
&Nmax&	663268&	1884993&	3585456&	5784871&	9802077&	18199764&	19352429&	&&&\\
&Tavg&	2&	9&	31&	85&	192&	469&	763&	&&&\\
&Tmax&	33&	193&	580&	1263&	2963&	6817&	9403&	&&&\\ \hline
\multirow{4}{*}{NoSplit}&Navg&	202970&	457918&	985235&	1934818&	3053648&	OOM&	&	&&&\\
&Nmax&	2156144&	6027604&	13028651&	20285112&	33977553&	&	&&&&\\
&Tavg&	0&	0&	2&	4&	9&	&	&&&&\\
&Tmax&	4&	13&	34&	59&	114&	&&&&&	\\ \hline
\multirow{4}{*}{NoSplit\_LUFO}&Navg&	202970&	457918&	985235&	1934818&	3086620&	5408511&	    7697810&	12578211&	19100285&	28223766&	OOT\\
&Nmax&	2156144&	6027604&	13028651&	20285112&	36853477&	60151076&	88909109&	139698961&	332937242&	420974965&\\	
&Tavg&	0&	0&	2&	5&	9&	20&	31&	61&	112&	192&	\\
&Tmax&	4&	13&	34&	61&	136&	275&	429&	832&	2504&	3763&	\\ \hline
\end{tabular}
}
    \caption{Results for the \pbtt{} problem}
    \label{tab:tt}
\end{table}
 
The experiments presented so far have shown that correctly tuning the memorization mechanism, notably by considering a cleaning strategy and studying interference with other components of the algorithm may lead to serious changes of its efficiency. However, the striking point of these experiments relates on the comparison between the version of BB2001 without the memorization mechanism (algorithm Depth-) and NoSplit\_LUFO. Table \ref{tab:tt} highlights the major contribution of memorization: Depth- being limited to instances with up to 300 jobs while NoSplit\_LUFO is capable of solving all instances with 1200 jobs. It is evident that memorization is a very powerful mechanism.

\FloatBarrier
\section{Chapter summary}
In this chapter we focus on the application of \techmemo{} within branching algorithms for the efficient solution of sequencing problems. A framework of \techmemo{} is established with several memorization schemes defined. Advices are provided to choose the best memorization approach according to the branching scheme and the search strategy of the algorithm. Some details on the efficient implementation of \techmemo{} are also discussed. 

The application of the framework has been done on four scheduling problems. Even if the impact of \techmemo{} depends on the problem, for all the tackled problems it was beneficial to use it. Table \ref{tab:con} provides a summary of the conclusions obtained.
\begin{table}[!ht]
    \centering
    \begin{tabular}{|c|c|c|c|c|}\hline
    \multirow{2}{*}{Problem} & \multicolumn{2}{c|}{Largest instances solved} & \multirow{2}{*}{\multiline{Features of the best algorithm\\ with memorization}} & \multirow{2}{*}{ \multiline{Best in \\literature?}} \\\cline{2-3}
    &  \multiline{Without \\memorization} & \multiline{With \\memorization} & & \\\hline
    \pbrisumc & 80 jobs & 130 jobs & \multiline{depth first+ \\\abmemo} & yes \\\hline
    \pbdtilde & 40 jobs & 130 jobs & \multiline{breadth first+ \\\pbmemo} & yes \\\hline
    \pbfsumc& 30 jobs & 40 jobs & \multiline{best first+ \\\pbmemo} & no \\\hline
    \pbtt & 300 jobs & 1200 jobs & \multiline{depth first+ \\\fmemo} & yes \\\hline
    \end{tabular}
    \caption{Conclusions on the tested problems}\label{tab:con}
\end{table}

Fundamentally, what we call the \textit{Memorization Paradigm} relies on a simple but potentially very efficient idea: avoid solving multiple times the same subproblems by memorizing what has already been done or what can be done in the rest of the solution process. The contribution of this paradigm strongly relies on the branching scheme which may induce more or less redundancy in the exploration of the solution space. 
Noteworthy, the four scheduling problems dealt with in this chapter mainly serve as applications illustrating how memorization can be done in an efficient way. But, it is also clear that it can be  applied to other hard combinatorial optimization problems, by the way making this contribution interesting beyond scheduling theory. To our opinion, the memorization paradigm should be embedded into any branching algorithm, so creating a new class of branching algorithms called \techbmemo{} algorithms. They may have a theoretical interest by offering the possibility of reducing the worst-case time complexity with respect to \techbb{} algorithms. And they also have an interest from an experimental viewpoint, as illustrated in this chapter.


The main work in this chapter has been performed together with Vincent T'Kindt and Federico Della Croce (\textit{Politechnico di Torino}, Italy). The computational results on the \pbtt{} problem have been reported at the international conference on \textit{Industrial Engineering and Systems Management} (IESM 2017) \citep{shang2017iesm}, being selected as the first prize of the \textit{Best Student Paper Award}. It has also been communicated in the \textit{ROADEF} 2017 conference \citep{shang2017roadef} and the joint working days (26th/27th, September 2017) of the french GOThA and Bermudes groups. 
The complete results have been submitted to the \textit{European Journal of Operational Research} and it can be found on HAL \citep{shang2017memohal}. 
\chapter*{Conclusion}\label{con}
\markright{\MakeUppercase{Conclusion}}

Three main contributions are worked out during this thesis. The first one, presented in chapter \ref{ch2}, is about a \techdp{} algorithm which solves the \pbft{} problem in $\ostar{3^n}$ time and space. The algorithm is based on a local dominance condition, consisting of two inequalities of the completion time on the second and the third machine. For any two partial solution sequences satisfying the condition, the dominated one is eliminated, therefore no longer considered during the later construction of the optimal global solution. The algorithm can be seen as a generalization of the conventional DP solving TSP-like problems, with the difference that in conventional DP the partial solutions of an optimal solution is optimal with respect to their corresponding subproblems, while this is not the case for  flowshop problems. 
The practical efficiency of the algorithm has been evaluated. Without much surprise, the algorithm is proved to be impractical due to its exponential space complexity. This is a common drawback of \techdp{} algorithms over all subsets. Its running time on  randomly generated instances is also far from competitive compared to state-of-the-art \techbb{} algorithms. Actually, the algorithm is too much designed for the worst-case scenario that it is not powerful on solving common instances. 
The algorithm is easily generalized to solve other flowshop problems including $F2||f$ and $F3||f$, and single machine problems like $1|r_i|f$, with $f\in\{f_{max}, f_i\}$. Since it seems pretty difficult to derive even faster EEAs for \pbft, we prove a complexity lower bound of this problem, based on ETH, as an additional result. It shows that unless ETH fails, the \pbft{} problem cannot be solved in subexponential time. 
Despite our extra research effort, we do not have much perspective on the search of a faster EEA for the \pbft{} problem. Perhaps a wise choice is to walk out a little from the context of traditional complexity theory and try to include extra parameters into the complexity measure of algorithms, that is, try to derive parameterized algorithms.

The second contribution, presented in chapter \ref{ch3}, is a search tree method called \techbm{} which solves  the \pbtt{} problem with the time complexity converging to $\ostar{2^n}$ and in polynomial space. The work started by proposing a \techbr{} algorithm based on the well known Lawler's property, which decomposes the problem to two subproblems each time the job with the longest processing time is fixed to a position. This yields a structure with some similarity to \techdc, which can be represented as a search tree. We then analyzed the search tree carefully over the worst-case instance and found that many identical subproblems appear in different part of the tree, that is, they are generated repetitively due to Lawler's property. After a careful analysis, an operation called merging is derived. The idea is to merge all identical nodes to one whenever possible. More concretely, at each branching, a part of the search tree is developed, keeping the space usage polynomial, and all identical nodes are identified and merged to one. 
Notice that we are not just cutting dominated nodes but also put the resulting node of merging in a specific position in the tree, so that subsequent mergings become possible. 
The algorithm goes in depth-first hence uses only polynomial space. By a careful analysis of the running time, we show that its time complexity converges to $\ostar{2^n}$, depending on a constant parameter $k$ which controls the extent of merging. 

The algorithm \techbm{} is implemented and tested on randomly generated instances. The effectiveness of merging is proved since \techbm{} runs faster than \techbr. However, surprisingly, when we add other well known problem properties, notably another rule that decomposes the problem and some rules eliminating the positions to consider when branching on the longest job, \techbm{} does not show any superiority over the state-of-the-art \techbb{} algorithm. Further analysis and preliminary experimentation revealed that the added properties have a negative effect on the efficiency of merging. In fact, for non-worst-case instances, most branching positions are already eliminated by some rules and only very few identical nodes can still be identified and merged. The corresponding gain is not larger than the extra cost caused by the merging operation itself. Nevertheless, the \techbm{} framework provides an algorithmic complexity view of the solution of the \pbtt{} problem and it is generalizable to other problems benefiting similar decomposition properties, even though it seems difficult to find another scheduling problem having such property. On the improvement of practical efficiency, the work continued and finally led to the results in chapter \ref{ch4}.

Chapter \ref{ch4} presents our third contribution, which is an extended work of chapter \ref{ch3}. When trying to improve the practical efficiency of \techbm, we realized that a better way to implement the idea is to use \techmemo, that is, instead of merging nodes in a structural way, we simply memorize the optimal solution of subproblems when they are solved at the first time, and reuse that solution whenever the same problem appears again. This is actually what has been done in the state-of-the-art \techbb{} algorithm of \cite{szwarc2001algorithmic}. When the \textit{memory} is suitably implemented, the method is more efficient than \techbm{} even though more memory needs to be used. \techmemo{} has been implemented and tested on the \pbtt{} problem. By the finding of a new  \textit{paradox} and the implementation of a memory cleaning strategy, the method succeeded to solve instances with 300 more jobs with respect to the state-of-the-art algorithm. Then the same treatment is extended to another three problems (\pbrisumc, \pbdtilde{} and \pbfsumc) addressed by \cite{tkindt2004revisiting}. The four problems all together show the power of the \techmemo{} paradigm when applied on sequencing problems. We call it \techbmemo{} to promote a systematic consideration of \techmemo{} as an essential building block in branching algorithm like \techbb. The method can certainly be used on other problems as well.

The memory cleaning strategy implemented in our \techmemo{} has a surprising effect on the algorithm's efficiency. In fact, the most items contained in the \textit{memory} are not useful at all. By cleaning these items out, the memory is used much more efficiently, which leads to an augmentation of the performance. However, it is not always easy to decide which solutions should be kept in the memory and which ones should be eliminated. As a perspective, we consider to use some \textit{Machine Learning} techniques to make this decision. This is a so-called classification problem. Each item, corresponding to a subproblem and its optimal solution, is a \textit{feature} that should be suitably represented to  allow possible characteristic extraction. 
An AI system can then be trained to judge whether a given item is worth being memorized or whether an already memorized item should be removed. 
Since lots of different instances are needed for a good training, each feature should also encode its corresponding instance, since an item that is useful for one instance is not necessarily equally valuable for other instances. 

Another important point that we would have realized is to consider the fact that the \textit{memory} used for solving one instance may be very useful for solving other instances. At the end of the solution of an instance, the \textit{memory} contains a large number of already solved small problems. These results remains valid for later instances. In other words, without \techmemo, there might be some subproblems that are solved repetitively over different instances. 

\vspace{1.5cm}
At the end of this research, it is interesting and also important to review all these three results and to explore the inner link between the techniques: \techdp, \techbm{} and \techmemo (or \techbmemo). First of all, it worth being noticed that the conventional \techdp{} (of TSP) across all subsets can also be expressed as a search tree. Consider the search tree of the brute-force algorithm, in which the level $\ell$ of the tree is composed of all partial solutions of length $\ell$. The number of nodes at level $\ell$ is hence ${n \choose \ell}!$. Among all these partial solutions, many consist of the same jobs (or operations, to be generic) but with a different ordering. For all partial solutions that are different orderings of the same jobs, \techdp{} allows to keep only the dominant one. This is very similar to the merging operation in \techbm{} that solves the \pbtt{} problem. Note that the \pbtt{} problem can also be solved by \techdp, but by exploiting specific properties, \techbm{} can perform the job better and only in polynomial space. 

The \techdp{} algorithm can often be implemented in two ways: top-down or bottom-up along the search tree. The top-down way is a recursive implementation and the bottom-up way is actually \techmemo. Therefore, \techdp{} and \techmemo{} are also tightly related. Note that the ideas behind these two techniques are not identical: \techdp{} lays on the fact that partial solutions of an optimal solution are themselves optimal for the implied subproblems, while the idea of \techmemo{} is to avoid solving a subproblem more than once, though sometimes these two ideas may coincide on a concrete problem. 
The link between \techbm{} and \techmemo{} (or \techbmemo) is pretty clear, since the latter is finally adopted in practice as a more efficient way to implement the idea of \techbm. 

All the three techniques, are superior with respect to the brute-force algorithm. This has been achieved, very reasonably, by avoiding solving identical subproblems or by eliminating dominated partial solutions. However, as introduced in section \ref{sec:ch1:algo}, there exists other algebraic techniques, such as \techie{} and \textit{Subset Convolution} which work differently. Using these techniques to tackle scheduling problems is an interesting direction to explore, though some attempt of the author has been fruitless.

\vspace{1.5cm}
In the end, we would like to insist that \textit{Scheduling} problems are central problems in Operation Research. It is essential to gain more knowledge about the hardness and the complexity of these problems. We hope that our works during this thesis have some positive influence on this (the publications of our works are summarized in Table \ref{tab:pub}). On the practical side, we also believe that EEAs have the potential to become practical, at least for moderate input size. It seems that more and more results are popping up, on the design of algorithms with better complexity for scheduling problems. This is more visible from a point of view of \textit{Parameterized Complexity Theory}. We believe that such research will bring a new angle for solving scheduling problems and more new interesting techniques and theories will appear as a result.

\begin{table}[!ht]
    \centering
    \begin{tabular}{|c|c|c|c|}\hline
         Chapter & Publications   \\ \hline
          \ref{ch2} & \multiline{J. of Scheduling\citep{shang2017jos}\\
          MISTA'15a\citep{shang2015f3cmax}\\
          ROADEF'16\citep{shang2016roadef}}       \\\hline
        \ref{ch3}& \multiline{TCS (under review \citep{garraffa2017bm})\\
        IPEC'17\citep{shang2017ipec}\\
        PMS'16\citep{shang2016bm}\\
        MISTA'15b\citep{della2015smtt}\\
        ROADEF'17\citep{shang2017roadef}\\
        AIRO'15\citep{della2015airo}} \\\hline
         \ref{ch4} & 
         \multiline{EJOR (under review \citep{shang2017memohal})\\
         IESM'17\citep{shang2017iesm} (Best Student Paper Award, 1st place)\\
         ROADEF'17\citep{shang2017roadef}}\\\hline
    \end{tabular}
    \caption{Summary of publications}
    \label{tab:pub}
\end{table}

\newpage
\bibliographystyle{abbrvnat}
\bibliography{biblio}
\markright{\MakeUppercase{Bibliography}}



\newpage
\thispagestyle{empty}
\newpage $\ $
\newpage
\thispagestyle{empty}
\section*{Résumé :}
Cette thèse synthétise les travaux de recherches réalisés pendant les études doctorales de l'auteur. L'objectif de ces travaux est de proposer des algorithmes exacts qui ont une meilleure complexité, temporelle ou spatiale, dans le pire des cas pour des problèmes d'ordonnancement qui sont $\mathcal{NP}$-difficiles. En plus, on s'intéresse aussi à évaluer leurs performances en pratique.

Trois contributions principales sont rapportées. La première concerne un algorithme du type \techdp{} qui résout le  problème \pbft{}  en $\ostar{3^n}$ en temps et en espace. L'algorithme est généralisé facilement à d'autres problèmes du type Flowshop, y compris les problèmes $F2||f$ et $F3||f$, et aux problèmes d'ordonnancement à une seule machine tels que les problèmes $1|r_i|f$, avec $f\in\{f_{max}, f_i\}$. 

La seconde contribution porte sur l'élaboration d'une méthode arborescente appelée \techbm{} pour  résoudre le problème  \pbtt{} en $\ostar{(2+\epsilon)^n}$ en temps avec $\epsilon>0$ arbitrairement petit et en espace polynomial. Le travail se base sur l'observation que de nombreux sous-problèmes identiques apparaissent répétitivement pendant la résolution du problème global. A partir de ça, une opération appelée \textit{merge} est proposée, qui fusionne les sous-problèmes (les noeuds dans l'arbre de recherche) identiques autant que possible. Cette méthode doit pouvoir être généralisée à d'autres problèmes.

Le but de la troisième contribution est d'améliorer les performances en pratique des algorithmes exacts procédant par parcours d'un arbre de recherche. D'abord nous avons aperçu qu'une meilleure façon d'implémenter l'idée de  \techbm{} est d'utiliser une technique appelée \techmemo. Avec la découverte d'un nouveau \textit{paradoxe} et la mis en place d'une stratégie de nettoyage de mémoire,  notre algorithme a résolu les instances qui ont 300 tâches de plus par rapport à l'algorithme de référence pour le problème \pbtt. 
Avec ce succès, nous avons testé \techmemo{} sur trois autres problèmes d'ordonnancement notés \pbrisumc, \pbdtilde{} et \pbfsumc, précédemment traités par \cite{tkindt2004revisiting}. Les résultats finaux des quatre problèmes ont montré la puissance de  \techmemo{} appliquée aux problèmes d'ordonnancement. Nous nommons ce paradigme \techbmemo{} afin de  promouvoir la considération systématique de l'intégration de \techmemo{} dans les algorithmes de branchement comme \techbb, en tant qu'un composant essentiel. 
La méthode peut aussi être généralisée pour résoudre d'autres problèmes qui ne sont pas forcément des problèmes d'ordonnancement.

\section*{Mots clés :}
Algorithmes exponentiels, ordonnancement, brancher et réduire, branch and merge, mémorisation, branch and memorize, programmation dynamique, flowshop, somme de retards

\newpage
\section*{Abstract :}
This thesis summarizes the author's PhD research works on the design of exact algorithms that provide a worst-case (time or space)
guarantee for \nph{} scheduling problems. Both theoretical and practical aspects are considered with three main results reported. 

The first one is about a \techdp{} algorithm which solves the \pbft{} problem in $\ostar{3^n}$ time and space. 
The algorithm is easily generalized to other flowshop problems including $F2||f$ and $F3||f$, and single machine scheduling problems like $1|r_i|f$, with $f\in\{f_{max}, f_i\}$. 

The second contribution is about a search tree method called \techbm{} which solves the \pbtt{} problem with the time complexity converging to $\ostar{2^n}$ and in polynomial space. The work is based on the observation that many identical subproblems appear during the solution of the input problem. An operation called \textit{merge} is then derived, which merges all identical nodes to one whenever possible and hence yields a better complexity.

Our third contribution
 aims to improve the practical efficiency of exact search tree algorithms solving scheduling problems. First we realized that a better way to implement the idea of  \techbm{}  is to use a technique called \techmemo. By the finding of a new  \textit{algorithmic paradox} and the implementation of a memory cleaning strategy, the method succeeded to solve instances with 300 more jobs with respect to the state-of-the-art algorithm for the \pbtt{} problem. Then the treatment is extended to another three problems \pbrisumc, \pbdtilde{} and \pbfsumc{} previously addressed by \cite{tkindt2004revisiting}. The results of the four problems all together show the power of the \techmemo{} paradigm when applied on sequencing problems. We name it \techbmemo{} to promote a systematic consideration of \techmemo{} as an essential building block in branching algorithms like \techbb. The method can surely also be used to solve other problems, which are not necessarily scheduling problems.

\section*{Keywords :}
Exact exponential algorithms, scheduling, branch and reduce, branch and merge, memorization, branch and memorize, dynamic programming, flowshop, total tardiness

\end{document}